\documentclass{lmcs} %


\usepackage{hyperref}
\theoremstyle{plain} %

\usepackage{array}
\usepackage{adjustbox}
\usepackage{tikz-cd}
\usepackage{tikz}
\usetikzlibrary{cd,arrows}
\usepackage{mathpartir}
\usepackage{caption}
\usepackage{amsmath}
\usepackage[renew-matrix,renew-dots]{nicematrix}
\usepackage{float}
\usepackage{enumitem}
\usepackage{xcolor}
\usepackage{booktabs}
\usepackage{rotating}
\usepackage{multirow}
\usepackage{comment}
\usepackage{graphicx}
\usepackage{braket}
\usepackage{enumitem}
\usepackage{xparse}
\usepackage{mathtools}
\usepackage{xifthen}
\usepackage{makecell}
\usepackage{cleveref}
\usepackage{multicol}
\usepackage{mathtools}
\usepackage{changepage}

\usepackage{quiver}

\newcommand{\lAngle}{\langle\!\langle}
\newcommand{\rAngle}{\rangle\!\rangle}

\usepackage{tabularx}

\usepackage{stmaryrd}

\colorlet{tapeBg}{red!30}
\colorlet{tapeBorder}{red!60}

\DeclareMathSymbol{\mathinvertedexclamationmark}{\mathclose}{operators}{'074}
\DeclareMathSymbol{\mathexclamationmark}{\mathclose}{operators}{'041}

\makeatletter
\newcommand{\raisedmathinvertedexclamationmark}{%
  \mathclose{\mathpalette\raised@mathinvertedexclamationmark\relax}%
}
\newcommand{\raised@mathinvertedexclamationmark}[2]{%
  \raisebox{\depth}{$\m@th#1\mathinvertedexclamationmark$}%
}
\begingroup\lccode`~=`! \lowercase{\endgroup
  \def~}{\@ifnextchar`{\raisedmathinvertedexclamationmark\@gobble}{\mathexclamationmark}}
\mathcode`!="8000
\makeatother

\usepackage{etoolbox}

\makeatletter
\patchcmd{\endalign}{\restorealignstate@}{\global\let\df@label\@empty\restorealignstate@}{}{}
\makeatother

\newlength{\leftstackrelawd}
\newlength{\leftstackrelbwd}
\def\leftstackrel#1#2{\settowidth{\leftstackrelawd}%
{${{}^{#1}}$}\settowidth{\leftstackrelbwd}{$#2$}%
\addtolength{\leftstackrelawd}{-\leftstackrelbwd}%
\leavevmode\ifthenelse{\lengthtest{\leftstackrelawd>0pt}}%
{\kern-.5\leftstackrelawd}{}\mathrel{\mathop{#2}\limits^{#1}}}

\makeatletter
\newcommand{\mylabel}[2]{\def\@currentlabel{#2}\label{#1}}
\makeatother

\makeatletter
\newcommand\newtag[2]{#1\def\@currentlabel{#1}\label{#2}}
\makeatother

\newcommand{\op}[1]{#1^{\dagger}}
\newcommand{\co}[1]{\text{\reflectbox{\(#1\)}}}
\newcommand{\opcat}[1]{#1^{\mathsf{op}}}

\newcommand{\N}{\mathbb{N}}

\newcommand{\dcomp}{\mathbin{;}}

\newcommand{\adjunction}[4]{#1 \colon #2 \leftrightarrows #3 \co{\colon} #4}

\newcommand{\lbbd}{\left \llbracket}
\newcommand{\rbbd}{\right \rrbracket}
\newcommand{\dsem}[1]{\lbbd #1 \rbbd}

\newcommand{\CBdsem}[1]{\dsem{#1}}
\newcommand{\TCBdsem}[1]{\dsem{#1}}

\newcommand{\dsemRel}[1]{\langle #1\rangle_\interpretation}

\newcommand{\id}[1]{{id}_{#1}}

\newcommand{\perG}{\odot} %
\newcommand{\unoG}{I} %

\newcommand{\piu}{\oplus} %
\newcommand{\per}{\otimes} %

\NewDocumentCommand\Piu{oom}{\bigoplus_{\IfNoValueTF{#1}{}{#1}}^{\IfNoValueTF{#2}{}{#2}}{#3}} %
\NewDocumentCommand\Per{oom}{\bigotimes_{\IfNoValueTF{#1}{}{#1}}^{\IfNoValueTF{#2}{}{#2}}{#3}} %

\NewDocumentCommand\PiuPar{oom}{\Piu[#1][#2] {(#3)}} %
\NewDocumentCommand\PerPar{omm}{\Per[#1][#2] {(#3)}} %

\NewDocumentCommand\PiuL{oom}{\bigoplus\limits_{\IfNoValueTF{#1}{}{#1}}^{\IfNoValueTF{#2}{}{#2}}{#3}} %
\NewDocumentCommand\PerL{oom}{\bigotimes\limits_{\IfNoValueTF{#1}{}{#1}}^{\IfNoValueTF{#2}{}{#2}}{#3}} %

\NewDocumentCommand\PiuParL{oom}{\PiuL[#1][#2] {(#3)}} %
\NewDocumentCommand\PerParL{omm}{\PerL[#1][#2] {(#3)}} %

\newcolumntype{C}[1]{>{\centering\let\newline\\\arraybackslash\hspace{0pt}}p{#1}}
\newcolumntype{L}[1]{>{\raggedright\let\newline\\\arraybackslash\hspace{0pt}}p{#1}}
\newcolumntype{R}[1]{>{\raggedleft\let\newline\\\arraybackslash\hspace{0pt}}p{#1}}

\newcommand{\sort}{\mathcal{S}}
\newcommand{\sign}{\Sigma}

\newcommand{\gen}{s}

\newcommand{\Cat}[1]{\mathbf{#1}} %
\newcommand{\fun}[1]{\mathsf{#1}} %

\newcommand{\CMon}{\Cat{CMon}} %
\newcommand{\CAT}{\Cat{Cat}} %
\newcommand{\CMonCat}{\CMon\CAT} %
\newcommand{\fbCat}{\Cat{FBC}} %
\newcommand{\FBC}{\Cat{FBC}} %
\newcommand{\Mat}[1]{{\Cat{Mat}(#1)}} %
\newcommand{\MatFun}{\fun{Mat}} %
\newcommand{\SMC}{\Cat{SMC}}

\newcommand{\Rel}{\Cat{Rel}}

\newcommand{\KTCB}{\Cat{KCT}_{\sign}}
\newcommand{\KTCBI}{\Cat{KCT}_{\sign,\basicR}}

\newcommand{\KTCBP}{\Cat{KCT}_{\sign,\mathbb{P}}}

\newcommand{\KBicat}{\Cat{KBicat}} %
\newcommand{\TKAlg}{\Cat{TKAlg}} %
\newcommand{\UTr}{\fun{UTr}} %
\newcommand{\freeTr}{\fun{Tr}} %

\newcommand{\TSMC}{\Cat{TrSMC}} %
\newcommand{\UTSMC}{\Cat{UTSMC}} %
\NewDocumentCommand\symm{gg}{\sigma\IfNoValueTF{#1}{}{_{#1,#2}}} %

\newcommand{\assocp}[3]{\alpha^\piu_{#1,#2,#3}} %
\newcommand{\lunitp}[1]{\lambda^\piu_{#1}} %
\newcommand{\zero}{\ensuremath{0}} %
\NewDocumentCommand\symmp{gg}{\sigma^\piu\IfNoValueTF{#1}{}{_{#1,#2}}} %

\newcommand{\Iassocp}[3]{\alpha^{-\piu}_{#1,#2,#3}} %

\newcommand{\assoct}[3]{\alpha^\per_{#1,#2,#3}} %
\newcommand{\runitt}[1]{\rho^\per_{#1}} %
\newcommand{\uno}{1} %
\NewDocumentCommand\symmt{gg}{\sigma^\per\IfNoValueTF{#1}{}{_{#1,#2}}} %

\newcommand{\dl}[3]{\delta^l_{#1,#2,#3}} %
\newcommand{\dr}[3]{\delta^r_{#1,#2,#3}} %

\newcommand{\annl}[1]{\lambda^\bullet_{#1}} %
\newcommand{\annr}[1]{\rho^\bullet_{#1}} %

\newcommand{\Idl}[3]{\delta^{-l}_{#1,#2,#3}} %
\newcommand{\Idr}[3]{\delta^{-r}_{#1,#2,#3}} %

\newcommand{\trace}{\mathsf{tr}}
\newcommand{\kstar}[1]{#1^{\ast}}

\NewDocumentCommand\Dl{gg}{\Delta^l\IfNoValueTF{#1}{}{_{#1,#2}}}
\NewDocumentCommand\Dr{gg}{\Delta^r\IfNoValueTF{#1}{}{_{#1,#2}}}
\NewDocumentCommand\IDl{gg}{\Delta^{-l}\IfNoValueTF{#1}{}{_{#1,#2}}}
\NewDocumentCommand\IDr{gg}{\Delta^{-r}\IfNoValueTF{#1}{}{_{#1,#2}}}
\NewDocumentCommand\D{gg}{\Delta\IfNoValueTF{#1}{}{_{#1,#2}}}
\NewDocumentCommand\ID{gg}{\Delta^{-}\IfNoValueTF{#1}{}{_{#1,#2}}}

\newcommand{\copier}[1]{\blacktriangleleft_{#1}} %
\newcommand{\discharger}[1]{\mbox{\textnormal{!}}_{#1}} %

\newcommand{\cocopier}[1]{\blacktriangleright_{#1}} %
\newcommand{\codischarger}[1]{\mbox{\textnormal{!`}}_{#1}} %

\newcommand{\gcopier}[1]{{\color{gray}\copier{#1}}}

\newcommand{\gcocopier}[1]{{\color{gray}\cocopier{#1}}}

\newcommand{\diag}[1]{\lhd_{#1}}
\newcommand{\bang}[1]{\rotatebox{90}{\rotatebox{90}{\rotatebox{90}{$\multimap$}}}_{#1}}

\newcommand{\codiag}[1]{\rhd_{#1}}
\newcommand{\cobang}[1]{\rotatebox{90}{$\multimap$}_{#1}}

\newcommand{\ob}[1]{Ob(#1)} %

\newcommand{\precongR}[1]{\,{\leq}_{#1}\,}

\newcommand{\basicR}{\mathbb{I}}

\newcommand{\basicKC}{\mathbb{KC}}

\newcommand{\precongB}{\precongR{{\basicR}}}

\newcommand{\precongKC}{\precongR{{\basicKC}}}

\newcommand{\Ucong}[1][]{\approx^{\scriptscriptstyle{#1}}}
\newcommand{\Kcong}[1][]{\sim^{\scriptscriptstyle{#1}}}
\newcommand{\Kleq}[1][]{\lesssim^{\scriptscriptstyle{#1}}}
\newcommand{\idcong}[1][]{\stackrel{.}{\sim}^{\scriptscriptstyle{#1}}}
\newcommand{\interpretation}{\mathcal{I}} %

\newcommand{\LW}[2]{\fun{L}_{#1}({#2})} %
\newcommand{\RW}[2]{\fun{R}_{#1}({#2})} %
\newcommand\defeq{\stackrel{\text{\tiny def}}{=}}

\newcommand{\axeq}[1]{\stackrel{({#1})}{=}}
\newcommand{\axsubeq}[1]{\stackrel{({#1})}{\leq}}

\newcommand{\tape}[1]{\overline{\underline{\; #1 \vphantom{d_A} \;}}} %
\newcommand{\tapeFunct}[1]{\underline{\smash{\overline{\; #1 \vphantom{t} \;}}}}

\newcommand{\tapesymm}[2]{\tape{\sigma_{#1,#2}}} %
\renewcommand{\t}{\mathfrak{t}}
\newcommand{\s}{\mathfrak{s}}

\newcommand{\ar}{\mathit{ar}}
\newcommand{\coar}{\mathit{coar}}

\newtheorem{notation}{Notation}

\newcommand{\CR}{\mathsf{CR}}

\newcommand{\encoding}[1]{\mathcal{E}(#1)}

\usepackage{tikz}
\usetikzlibrary{cd,arrows,backgrounds,shapes.geometric,shapes.misc,matrix,arrows.meta,decorations.markings}
\pgfdeclarelayer{background}
\pgfdeclarelayer{edgelayer}
\pgfdeclarelayer{nodelayer}
\pgfsetlayers{background,edgelayer,nodelayer,main}

\tikzstyle{none}=[inner sep=0pt]
\tikzstyle{label}=[fill=none, draw=none, shape=circle, scale=0.6]
\tikzstyle{dots}=[fill=none, draw=none, shape=circle, scale=0.4]
\tikzstyle{black}=[circle, draw=black, fill=black, inner sep=0pt, minimum size=2.5pt]
\tikzstyle{white}=[circle, draw=black, fill=white, inner sep=0pt, minimum size=2.5pt]
\tikzstyle{reg}=[draw, fill=white, rounded rectangle, rounded rectangle left arc=none, minimum height=0.8em, minimum width=1em, node font={\tiny}]
\tikzstyle{coreg}=[draw, fill=white, rounded rectangle, rounded rectangle right arc=none, minimum height=0.8em, minimum width=1em, node font={\tiny}]
\tikzstyle{box}=[draw, fill=white, rectangle, minimum height=1em, minimum width=1em, node font={\tiny}]
\tikzstyle{stringbox}=[fill=white, draw=black, chamfered rectangle, chamfered rectangle corners={north east}, chamfered rectangle sep=2pt, inner sep=0pt, minimum height=1.4em, minimum width=1.4em, scale=0.6]
\tikzstyle{stringlongbox}=[fill=white, draw=black, chamfered rectangle, chamfered rectangle corners={north east}, chamfered rectangle sep=2pt, inner sep=0pt, minimum height=2.2em, minimum width=1.4em, scale=0.6]
\tikzstyle{triang}=[draw, fill=white, isosceles triangle, isosceles triangle apex angle=60, minimum size=0.5em, node font={\tiny}]

\tikzstyle{bbox}=[fill=white, draw=black, chamfered rectangle, chamfered rectangle corners={north east}, chamfered rectangle sep=2pt, inner sep=0pt, minimum height=1.4em, minimum width=1.4em, scale=0.6]
\tikzstyle{bboxOp}=[fill=white, draw=black, chamfered rectangle, chamfered rectangle corners={north west}, chamfered rectangle sep=2pt, inner sep=0pt, minimum height=1.4em, minimum width=1.4em, scale=0.6]
\tikzstyle{corefl}=[fill=white, draw=black, rectangle, rounded corners, minimum height=1.2em, minimum width=1.2em, inner sep=2pt, scale=0.6]

\tikzstyle{tapeFill} = [fill=tapeBg]
\tikzstyle{tapeNoFill} = [draw=tapeBorder, line width=0.5pt]
\tikzstyle{tape} = [fill=tapeBg, draw=tapeBorder, line width=0.5pt]

\tikzcdset{arrow style=tikz}

\newcommand{\comonoid}[1]{
    \begin{tikzpicture}
        \begin{pgfonlayer}{nodelayer}
            \node [style=none] (81) at (-0.75, 0) {};
            \node [style=black] (107) at (0, 0) {};
            \node [style=none] (108) at (0.5, -0.375) {};
            \node [style=none] (109) at (0.5, 0.375) {};
            \node [style=none] (110) at (0.75, -0.375) {};
            \node [style=none] (111) at (0.75, 0.375) {};
            \ifthenelse{\isempty{#1}}{}{
            \node [style=label] (112) at (1.075, -0.375) {$#1$};
            \node [style=label] (113) at (1.075, 0.375) {$#1$};
            \node [style=label] (114) at (-1.075, 0) {$#1$};
            }
        \end{pgfonlayer}
        \begin{pgfonlayer}{edgelayer}
            \draw [bend right] (109.center) to (107);
            \draw [bend right] (107) to (108.center);
            \draw (81.center) to (107);
            \draw (111.center) to (109.center);
            \draw (108.center) to (110.center);
        \end{pgfonlayer}
    \end{tikzpicture}           
 }

 \newcommand{\counit}[1]{
    \begin{tikzpicture}
        \begin{pgfonlayer}{nodelayer}
            \node [style=none] (81) at (-0.75, 0) {};
            \node [style=black] (107) at (0, 0) {};
            \ifthenelse{\isempty{#1}}{}{
            \node [style=label] (114) at (-1.075, 0) {$#1$};
            }
        \end{pgfonlayer}
        \begin{pgfonlayer}{edgelayer}
            \draw (81.center) to (107);
        \end{pgfonlayer}
    \end{tikzpicture}      
}

 \newcommand{\monoid}[1]{
    \begin{tikzpicture}
        \begin{pgfonlayer}{nodelayer}
            \node [style=none] (81) at (0.75, 0) {};
            \node [style=black] (107) at (0, 0) {};
            \node [style=none] (108) at (-0.5, -0.375) {};
            \node [style=none] (109) at (-0.5, 0.375) {};
            \node [style=none] (110) at (-0.75, -0.375) {};
            \node [style=none] (111) at (-0.75, 0.375) {};
            \ifthenelse{\isempty{#1}}{}{
            \node [style=label] (112) at (-1.075, -0.375) {$#1$};
            \node [style=label] (113) at (-1.075, 0.375) {$#1$};
            \node [style=label] (114) at (1.075, 0) {$#1$};
            }
        \end{pgfonlayer}
        \begin{pgfonlayer}{edgelayer}
            \draw [bend left] (109.center) to (107);
            \draw [bend left] (107) to (108.center);
            \draw (81.center) to (107);
            \draw (111.center) to (109.center);
            \draw (108.center) to (110.center);
        \end{pgfonlayer}
    \end{tikzpicture}                     
 }

\newcommand{\unit}[1]{
    \begin{tikzpicture}
        \begin{pgfonlayer}{nodelayer}
            \node [style=none] (81) at (-0.325, 0) {};
            \node [style=black] (107) at (-1.075, 0) {};
            \ifthenelse{\isempty{#1}}{}{
            \node [style=label] (114) at (0, 0) {$#1$};
            }
        \end{pgfonlayer}
        \begin{pgfonlayer}{edgelayer}
            \draw (81.center) to (107);
        \end{pgfonlayer}
    \end{tikzpicture}         
}

\newcommand{\wire}[1]{
    \begin{tikzpicture}
        \begin{pgfonlayer}{nodelayer}
            \node [style=label] (8) at (1.5, 0) {$#1$};
            \node [style=label] (11) at (-1.5, 0) {$#1$};
        \end{pgfonlayer}
        \begin{pgfonlayer}{edgelayer}
            \draw (11) to (8);
        \end{pgfonlayer}
    \end{tikzpicture}    
}

\newcommand{\Tcomonoid}[1]{
    \ifthenelse{\equal{#1}{\uno}}{
        \begin{tikzpicture}
            \begin{pgfonlayer}{nodelayer}
                \node [style=none] (1) at (1, 1) {};
                \node [style=none] (2) at (1, -1) {};
                \node [style=none] (6) at (-1.5, 0.5) {};
                \node [style=none] (7) at (-1.5, -0.5) {};
                \node [style=none] (8) at (1, -1.5) {};
                \node [style=none] (9) at (1, 1.5) {};
                \node [style=none] (10) at (1, 0.5) {};
                \node [style=none] (11) at (0.5, 0) {};
                \node [style=none] (12) at (1, -0.5) {};
                \node [style=none] (13) at (-0.5, -0.5) {};
                \node [style=none] (14) at (-0.5, 0.5) {};
                \node [style=none] (15) at (-1.5, 0) {};
            \end{pgfonlayer}
            \begin{pgfonlayer}{edgelayer}
                \path [style=tape] (9.center)
                     to [bend right] (14.center)
                     to (6.center)
                     to (7.center)
                     to (13.center)
                     to [bend right] (8.center)
                     to (12.center)
                     to [bend left=45] (11.center)
                     to [bend left=45] (10.center)
                     to cycle;
            \end{pgfonlayer}
        \end{tikzpicture}        
    }
    {
        \begin{tikzpicture}
            \begin{pgfonlayer}{nodelayer}
                \node [style=none] (0) at (0, 0) {};
                \node [style=none] (1) at (1, 1) {};
                \node [style=none] (2) at (1, -1) {};
                \ifthenelse{\isempty{#1}}{}{
                    \node [style=label] (3) at (-2, 0) {$#1$};
                    \node [style=label] (4) at (1.5, 1) {$#1$};
                    \node [style=label] (5) at (1.5, -1) {$#1$};
                }
                \node [style=none] (6) at (-1.5, 0.5) {};
                \node [style=none] (7) at (-1.5, -0.5) {};
                \node [style=none] (8) at (1, -1.5) {};
                \node [style=none] (9) at (1, 1.5) {};
                \node [style=none] (10) at (1, 0.5) {};
                \node [style=none] (11) at (0.5, 0) {};
                \node [style=none] (12) at (1, -0.5) {};
                \node [style=none] (13) at (-0.5, -0.5) {};
                \node [style=none] (14) at (-0.5, 0.5) {};
                \node [style=none] (15) at (-1.5, 0) {};
            \end{pgfonlayer}
            \begin{pgfonlayer}{edgelayer}
                \path [style=tape] (9.center)
                     to [bend right] (14.center)
                     to (6.center)
                     to (7.center)
                     to (13.center)
                     to [bend right] (8.center)
                     to (12.center)
                     to [bend left=45] (11.center)
                     to [bend left=45] (10.center)
                     to cycle;
                \draw [bend right=45] (0.center) to (2.center);
                \draw [bend left=45] (0.center) to (1.center);
                \draw (0.center) to (15.center);
            \end{pgfonlayer}
        \end{tikzpicture}           
    }             
}

\newcommand{\Tcounit}[1]{
    \begin{tikzpicture}
        \begin{pgfonlayer}{nodelayer}
            \node [style=none] (0) at (1.25, 0) {};
            \ifthenelse{\isempty{#1}}{}{
                \node [style=label] (3) at (-0.75, 0) {$#1$};
            }
            \node [style=none] (6) at (-0.25, 0.5) {};
            \node [style=none] (7) at (-0.25, -0.5) {};
            \node [style=none] (13) at (0.75, -0.5) {};
            \node [style=none] (14) at (0.75, 0.5) {};
            \node [style=none] (15) at (-0.25, 0) {};
        \end{pgfonlayer}
        \begin{pgfonlayer}{edgelayer}
            \path [style=tape] (13.center)
                 to [bend right=45] (0.center)
                 to [bend right=45] (14.center)
                 to (6.center)
                 to (7.center)
                 to cycle;
            \draw (0.center) to (15.center);
        \end{pgfonlayer}
    \end{tikzpicture}      
}

\newcommand{\Tmonoid}[1]{
    \ifthenelse{\equal{#1}{\uno}}{
        \begin{tikzpicture}
            \begin{pgfonlayer}{nodelayer}
                \node [style=none] (1) at (-1, 1) {};
                \node [style=none] (2) at (-1, -1) {};
                \node [style=none] (6) at (1.5, 0.5) {};
                \node [style=none] (7) at (1.5, -0.5) {};
                \node [style=none] (8) at (-1, -1.5) {};
                \node [style=none] (9) at (-1, 1.5) {};
                \node [style=none] (10) at (-1, 0.5) {};
                \node [style=none] (11) at (-0.5, 0) {};
                \node [style=none] (12) at (-1, -0.5) {};
                \node [style=none] (13) at (0.5, -0.5) {};
                \node [style=none] (14) at (0.5, 0.5) {};
                \node [style=none] (15) at (1.5, 0) {};
            \end{pgfonlayer}
            \begin{pgfonlayer}{edgelayer}
                \path [style=tape] (9.center)
                     to [bend left] (14.center)
                     to (6.center)
                     to (7.center)
                     to (13.center)
                     to [bend left] (8.center)
                     to (12.center)
                     to [bend right=45] (11.center)
                     to [bend right=45] (10.center)
                     to cycle;
            \end{pgfonlayer}
        \end{tikzpicture}         
    }
    {
        \begin{tikzpicture}
            \begin{pgfonlayer}{nodelayer}
                \node [style=none] (0) at (0, 0) {};
                \node [style=none] (1) at (-1, 1) {};
                \node [style=none] (2) at (-1, -1) {};
                \ifthenelse{\isempty{#1}}{}{
                    \node [style=label] (3) at (2, 0) {$#1$};
                    \node [style=label] (4) at (-1.5, 1) {$#1$};
                    \node [style=label] (5) at (-1.5, -1) {$#1$};
                }
                \node [style=none] (6) at (1.5, 0.5) {};
                \node [style=none] (7) at (1.5, -0.5) {};
                \node [style=none] (8) at (-1, -1.5) {};
                \node [style=none] (9) at (-1, 1.5) {};
                \node [style=none] (10) at (-1, 0.5) {};
                \node [style=none] (11) at (-0.5, 0) {};
                \node [style=none] (12) at (-1, -0.5) {};
                \node [style=none] (13) at (0.5, -0.5) {};
                \node [style=none] (14) at (0.5, 0.5) {};
                \node [style=none] (15) at (1.5, 0) {};
            \end{pgfonlayer}
            \begin{pgfonlayer}{edgelayer}
                \path [style=tape] (9.center)
                     to [bend left] (14.center)
                     to (6.center)
                     to (7.center)
                     to (13.center)
                     to [bend left] (8.center)
                     to (12.center)
                     to [bend right=45] (11.center)
                     to [bend right=45] (10.center)
                     to cycle;
                \draw [bend left=45] (0.center) to (2.center);
                \draw [bend right=45] (0.center) to (1.center);
                \draw (0.center) to (15.center);
            \end{pgfonlayer}
        \end{tikzpicture}          
    }  
}

\newcommand{\Tunit}[1]{
    \begin{tikzpicture}
        \begin{pgfonlayer}{nodelayer}
            \node [style=none] (0) at (-1.25, 0) {};
            \ifthenelse{\isempty{#1}}{}{
                \node [style=label] (3) at (0.75, 0) {$#1$};
            }
            \node [style=none] (6) at (0.25, 0.5) {};
            \node [style=none] (7) at (0.25, -0.5) {};
            \node [style=none] (13) at (-0.75, -0.5) {};
            \node [style=none] (14) at (-0.75, 0.5) {};
            \node [style=none] (15) at (0.25, 0) {};
        \end{pgfonlayer}
        \begin{pgfonlayer}{edgelayer}
            \path [style=tape] (13.center)
                 to [bend left=45] (0.center)
                 to [bend left=45] (14.center)
                 to (6.center)
                 to (7.center)
                 to cycle;
            \draw (0.center) to (15.center);
        \end{pgfonlayer}
    \end{tikzpicture}      
}

\newcommand{\Tsymmp}[2]{
    \begin{tikzpicture}
        \begin{pgfonlayer}{nodelayer}
            \node [style=label] (3) at (-1.5, 1) {$#1$};
            \node [style=none] (6) at (-1, -0.5) {};
            \node [style=none] (7) at (-1, -1.5) {};
            \node [style=none] (15) at (-1, -1) {};
            \node [style=none] (17) at (-1, 1.5) {};
            \node [style=none] (18) at (-1, 0.5) {};
            \node [style=none] (21) at (-1, 1) {};
            \node [style=none] (22) at (2.5, 1) {};
            \node [style=none] (25) at (2.5, 0.5) {};
            \node [style=none] (26) at (2.5, 1.5) {};
            \node [style=none] (28) at (2.5, -1) {};
            \node [style=none] (31) at (2.5, -1.5) {};
            \node [style=none] (32) at (2.5, -0.5) {};
            \node [style=label] (34) at (3, -1) {$#1$};
            \node [style=label] (35) at (3, 1) {$#2$};
            \node [style=label] (36) at (-1.5, -1) {$#2$};
        \end{pgfonlayer}
        \begin{pgfonlayer}{edgelayer}
            \path [style=tapeFill] (26.center)
                 to [in=0, out=-180, looseness=0.75] (6.center)
                 to (7.center)
                 to [in=-180, out=0, looseness=0.75] (25.center)
                 to cycle;
            \path [style=tapeFill] (18.center)
                 to (17.center)
                 to [in=180, out=0, looseness=0.75] (32.center)
                 to (31.center)
                 to [in=0, out=-180, looseness=0.75] cycle;
            \path [style=tapeNoFill] (26.center)
                 to [in=0, out=-180, looseness=0.75] (6.center)
                 to (7.center)
                 to [in=-180, out=0, looseness=0.75] (25.center)
                 to cycle;
            \path [style=tapeNoFill] (18.center)
                 to (17.center)
                 to [in=180, out=0, looseness=0.75] (32.center)
                 to (31.center)
                 to [in=0, out=-180, looseness=0.75] cycle;
            \draw [in=-180, out=0, looseness=0.75] (15.center) to (22.center);
            \draw [in=180, out=0, looseness=0.75] (21.center) to (28.center);
        \end{pgfonlayer}
    \end{tikzpicture}    
}

\newcommand{\Twire}[1]{
    \ifthenelse{\equal{#1}{\uno}}{
        \begin{tikzpicture}
            \begin{pgfonlayer}{nodelayer}
                \node [style=none] (0) at (-1.5, 0.5) {};
                \node [style=none] (3) at (-1.5, -0.5) {};
                \node [style=none] (4) at (1.5, -0.5) {};
                \node [style=none] (5) at (1.5, 0.5) {};
                \node [style=none] (6) at (-1.5, 0) {};
                \node [style=none] (7) at (1.5, 0) {};
            \end{pgfonlayer}
            \begin{pgfonlayer}{edgelayer}
                \path [style=tape] (5.center)
                     to (4.center)
                     to (3.center)
                     to (0.center)
                     to cycle;
            \end{pgfonlayer}
        \end{tikzpicture}        
    }
    {
        \begin{tikzpicture}
            \begin{pgfonlayer}{nodelayer}
                \node [style=none] (0) at (-1.5, 0.5) {};
                \node [style=none] (3) at (-1.5, -0.5) {};
                \node [style=none] (4) at (1.5, -0.5) {};
                \node [style=none] (5) at (1.5, 0.5) {};
                \node [style=none] (6) at (-1.5, 0) {};
                \node [style=none] (7) at (1.5, 0) {};
                \node [style=label] (8) at (-2, 0) {$#1$};
                \node [style=label] (9) at (2, 0) {$#1$};
            \end{pgfonlayer}
            \begin{pgfonlayer}{edgelayer}
                \path [style=tape] (5.center)
                        to (4.center)
                        to (3.center)
                        to (0.center)
                        to cycle;
                \draw (6.center) to (7.center);
            \end{pgfonlayer}
        \end{tikzpicture}        
    }   
}

\newcommand{\TLwire}[2]{
    \begin{tikzpicture}
        \begin{pgfonlayer}{nodelayer}
            \node [style=label] (3) at (-2, -0.25) {$#1$};
            \node [style=none] (17) at (-1.5, 0.75) {};
            \node [style=none] (18) at (-1.5, -0.75) {};
            \node [style=none] (21) at (-1.5, -0.25) {};
            \node [style=none] (40) at (1.5, 0.75) {};
            \node [style=none] (41) at (1.5, -0.75) {};
            \node [style=none] (42) at (1.5, -0.25) {};
            \node [style=none] (46) at (0, -0.25) {};
            \node [style=none] (47) at (0, -0.75) {};
            \node [style=none] (48) at (0, 0.75) {};
            \node [style=label] (49) at (2, -0.25) {$#1$};
            \node [style=label] (50) at (-2, 0.5) {$#2$};
            \node [style=none] (51) at (-1.5, 0.25) {};
            \node [style=none] (52) at (1.5, 0.25) {};
            \node [style=none] (53) at (0, 0.25) {};
            \node [style=label] (54) at (2, 0.5) {$#2$};
        \end{pgfonlayer}
        \begin{pgfonlayer}{edgelayer}
            \path [style=tape] (41.center)
                 to [in=0, out=180, looseness=0.75] (47.center)
                 to [in=0, out=-180, looseness=0.75] (18.center)
                 to (17.center)
                 to [in=180, out=0, looseness=0.75] (48.center)
                 to [in=180, out=0, looseness=0.75] (40.center)
                 to cycle;
            \draw (42.center) to (46.center);
            \draw (21.center) to (46.center);
            \draw (52.center) to (53.center);
            \draw (51.center) to (53.center);
        \end{pgfonlayer}
    \end{tikzpicture}    
}

\newcommand{\TRwire}[2]{
    \TLwire{#2}{#1}
}

\newcommand{\Tcirc}[3]{
    \begin{tikzpicture}
        \begin{pgfonlayer}{nodelayer}
            \node [style=label] (35) at (2, 0) {$#3$};
            \node [style=none] (40) at (1.5, 0.75) {};
            \node [style=none] (41) at (-1.5, 0.75) {};
            \node [style=none] (42) at (-1.5, -0.75) {};
            \node [style=none] (43) at (1.5, -0.75) {};
            \node [style=none] (44) at (-1.5, 0) {};
            \node [style=label] (45) at (-2, 0) {$#2$};
            \node [style=bbox, scale=0.9] (47) at (0, 0) {$#1$};
            \node [style=none] (48) at (1.5, 0) {};
        \end{pgfonlayer}
        \begin{pgfonlayer}{edgelayer}
            \draw [style=tape] (41.center)
                 to (40.center)
                 to (43.center)
                 to (42.center)
                 to cycle;
            \draw (47) to (44.center);
            \draw (48.center) to (47);
        \end{pgfonlayer}
    \end{tikzpicture}          
}

\newcommand{\Ttapesymm}[2]{
    \begin{tikzpicture}
        \begin{pgfonlayer}{nodelayer}
            \node [style=none] (0) at (-1.5, 1) {};
            \node [style=none] (3) at (-1.5, -0.75) {};
            \node [style=none] (4) at (1.5, -0.75) {};
            \node [style=none] (5) at (1.5, 1) {};
            \node [style=none] (6) at (-1.5, -0.25) {};
            \node [style=none] (7) at (1.5, -0.25) {};
            \node [style=label] (8) at (-2.5, 0.5) {$#1$};
            \node [style=label] (9) at (2.5, 0.5) {$#2$};
            \node [style=none] (10) at (-1.5, 0.5) {};
            \node [style=none] (11) at (1.5, 0.5) {};
            \node [style=label] (12) at (-2.5, -0.25) {$#2$};
            \node [style=label] (13) at (2.5, -0.25) {$#1$};
        \end{pgfonlayer}
        \begin{pgfonlayer}{edgelayer}
            \path [style=tape] (5.center)
                 to (4.center)
                 to (3.center)
                 to (0.center)
                 to cycle;
            \draw [in=-180, out=0, looseness=1.25] (6.center) to (11.center);
            \draw [in=180, out=0, looseness=1.25] (10.center) to (7.center);
        \end{pgfonlayer}
    \end{tikzpicture}      
}

\newcommand{\Cgen}[3]{
    \begin{tikzpicture}
        \begin{pgfonlayer}{nodelayer}
            \node [style=label] (35) at (2, 0) {$#3$};
            \node [style=none] (44) at (-1.5, 0) {};
            \node [style=label] (45) at (-2, 0) {$#2$};
            \node [style=bbox, scale=0.9] (47) at (0, 0) {$#1$};
            \node [style=none] (48) at (1.5, 0) {};
        \end{pgfonlayer}
        \begin{pgfonlayer}{edgelayer}
            \draw (47) to (44.center);
            \draw (48.center) to (47);
        \end{pgfonlayer}
    \end{tikzpicture}         
}

\newcommand{\Csymm}[2]{
    \begin{tikzpicture}
        \begin{pgfonlayer}{nodelayer}
            \node [style=label] (8) at (1.75, 0.5) {$#2$};
            \node [style=label] (11) at (-1.75, 0.5) {$#1$};
            \node [style=label] (12) at (-1.75, -0.5) {$#2$};
            \node [style=label] (13) at (1.75, -0.5) {$#1$};
        \end{pgfonlayer}
        \begin{pgfonlayer}{edgelayer}
            \draw [in=-180, out=0, looseness=1.25] (11) to (13);
            \draw [in=-180, out=0, looseness=1.25] (12) to (8);
        \end{pgfonlayer}
    \end{tikzpicture}             
}

\newcommand{\Creg}[2]{
    \begin{tikzpicture}
        \begin{pgfonlayer}{nodelayer}
            \node [style=reg] (0) at (0, 0) {$#1$};
            \node [style=none] (1) at (1.5, 0) {};
            \node [style=none] (2) at (-1.5, 0) {};
            \ifthenelse{\equal{#2}{\uno}}{}{
            \node [style=label] (3) at (-2.5, 0) {$#2$};
            \node [style=label] (4) at (2.5, 0) {$#2$};
            }
        \end{pgfonlayer}
        \begin{pgfonlayer}{edgelayer}
            \draw (2.center) to (0);
            \draw (0) to (1.center);
        \end{pgfonlayer}
    \end{tikzpicture}       
}

\newcommand{\Ccoreg}[2]{
    \begin{tikzpicture}
        \begin{pgfonlayer}{nodelayer}
            \node [style=coreg] (0) at (0, 0) {$#1$};
            \node [style=none] (1) at (1.5, 0) {};
            \node [style=none] (2) at (-1.5, 0) {};
            \ifthenelse{\equal{#2}{\uno}}{}{
            \node [style=label] (3) at (-2.5, 0) {$#2$};
            \node [style=label] (4) at (2.5, 0) {$#2$};
            }
        \end{pgfonlayer}
        \begin{pgfonlayer}{edgelayer}
            \draw (2.center) to (0);
            \draw (0) to (1.center);
        \end{pgfonlayer}
    \end{tikzpicture}
}

\newcommand{\Cunit}[1]{
    \begin{tikzpicture}
        \begin{pgfonlayer}{nodelayer}
            \ifthenelse{\equal{#1}{\uno}}{}{
            \node [style=label] (4) at (1, 0) {$#1$};
            }
            \node [style=none] (11) at (0, 0) {};
            \node [style=black] (12) at (-1, 0) {};
        \end{pgfonlayer}
        \begin{pgfonlayer}{edgelayer}
            \draw (12) to (11.center);
        \end{pgfonlayer}
    \end{tikzpicture}    
}

\newcommand{\Ccounit}[1]{
    \begin{tikzpicture}
        \begin{pgfonlayer}{nodelayer}
            \ifthenelse{\equal{#1}{\uno}}{}{
            \node [style=label] (4) at (-1, 0) {$#1$};
            }
            \node [style=none] (11) at (0, 0) {};
            \node [style=black] (12) at (1, 0) {};
        \end{pgfonlayer}
        \begin{pgfonlayer}{edgelayer}
            \draw (12) to (11.center);
        \end{pgfonlayer}
    \end{tikzpicture}    
}

\newcommand{\Tcodiagdiag}[1]{
    \ifthenelse{\equal{#1}{\uno}}{
        \begin{tikzpicture}
            \begin{pgfonlayer}{nodelayer}
                \node [style=none] (1) at (-0.75, 1.25) {};
                \node [style=none] (8) at (-0.75, -1.75) {};
                \node [style=none] (9) at (-0.75, 1.75) {};
                \node [style=none] (10) at (-0.75, 0.75) {};
                \node [style=none] (11) at (-1.5, 0) {};
                \node [style=none] (12) at (-0.75, -0.75) {};
                \node [style=none] (13) at (-2.5, -0.5) {};
                \node [style=none] (14) at (-2.5, 0.5) {};
                \node [style=none] (35) at (-3.75, -0.5) {};
                \node [style=none] (36) at (-3.75, 0.5) {};
                \node [style=none] (81) at (0.75, 1.25) {};
                \node [style=none] (83) at (0.75, -1.75) {};
                \node [style=none] (84) at (0.75, 1.75) {};
                \node [style=none] (85) at (0.75, 0.75) {};
                \node [style=none] (86) at (1.5, 0) {};
                \node [style=none] (87) at (0.75, -0.75) {};
                \node [style=none] (88) at (2.5, -0.5) {};
                \node [style=none] (89) at (2.5, 0.5) {};
                \node [style=none] (91) at (3.75, -0.5) {};
                \node [style=none] (92) at (3.75, 0.5) {};
            \end{pgfonlayer}
            \begin{pgfonlayer}{edgelayer}
                \draw [style=tape] (83.center)
                     to (8.center)
                     to [bend left] (13.center)
                     to (35.center)
                     to (36.center)
                     to (14.center)
                     to [bend left] (9.center)
                     to (84.center)
                     to [bend left] (89.center)
                     to (92.center)
                     to (91.center)
                     to (88.center)
                     to [bend left] cycle;
                \path [style=tapeNoFill, fill=white] (87.center)
                     to (12.center)
                     to [bend left=45] (11.center)
                     to [bend left=45] (10.center)
                     to (85.center)
                     to [bend left=45] (86.center)
                     to [bend left=45] cycle;
            \end{pgfonlayer}
        \end{tikzpicture}                 
    }
    {
        \begin{tikzpicture}
            \begin{pgfonlayer}{nodelayer}
                \node [style=none] (0) at (-2, 0) {};
                \node [style=none] (1) at (-0.75, 1.25) {};
                \node [style=none] (2) at (-0.75, -1.25) {};
                \node [style=none] (8) at (-0.75, -1.75) {};
                \node [style=none] (9) at (-0.75, 1.75) {};
                \node [style=none] (10) at (-0.75, 0.75) {};
                \node [style=none] (11) at (-1.5, 0) {};
                \node [style=none] (12) at (-0.75, -0.75) {};
                \node [style=none] (13) at (-2.5, -0.5) {};
                \node [style=none] (14) at (-2.5, 0.5) {};
                \node [style=none] (24) at (-3.75, 0) {};
                \node [style=none] (35) at (-3.75, -0.5) {};
                \node [style=none] (36) at (-3.75, 0.5) {};
                \node [style=none] (50) at (-0.75, 1.25) {};
                \node [style=none] (80) at (2, 0) {};
                \node [style=none] (81) at (0.75, 1.25) {};
                \node [style=none] (82) at (0.75, -1.25) {};
                \node [style=none] (83) at (0.75, -1.75) {};
                \node [style=none] (84) at (0.75, 1.75) {};
                \node [style=none] (85) at (0.75, 0.75) {};
                \node [style=none] (86) at (1.5, 0) {};
                \node [style=none] (87) at (0.75, -0.75) {};
                \node [style=none] (88) at (2.5, -0.5) {};
                \node [style=none] (89) at (2.5, 0.5) {};
                \node [style=none] (90) at (3.75, 0) {};
                \node [style=none] (91) at (3.75, -0.5) {};
                \node [style=none] (92) at (3.75, 0.5) {};
                \node [style=none] (93) at (0.75, 1.25) {};
                \node [style=label] (94) at (-4.25, 0) {$#1$};
                \node [style=label] (95) at (4.25, 0) {$#1$};
            \end{pgfonlayer}
            \begin{pgfonlayer}{edgelayer}
                \path [style=tape] (83.center)
                     to (8.center)
                     to [bend left] (13.center)
                     to (35.center)
                     to (36.center)
                     to (14.center)
                     to [bend left] (9.center)
                     to (84.center)
                     to [bend left] (89.center)
                     to (92.center)
                     to (91.center)
                     to (88.center)
                     to [bend left] cycle;
                \path [style=tapeNoFill, fill=white] (87.center)
                     to (12.center)
                     to [bend left=45] (11.center)
                     to [bend left=45] (10.center)
                     to (85.center)
                     to [bend left=45] (86.center)
                     to [bend left=45] cycle;
                \draw [bend right=45] (0.center) to (2.center);
                \draw [bend left=45] (0.center) to (1.center);
                \draw (24.center) to (0.center);
                \draw [bend left=45] (80.center) to (82.center);
                \draw [bend right=45] (80.center) to (81.center);
                \draw (90.center) to (80.center);
                \draw (50.center) to (93.center);
                \draw (82.center) to (2.center);
            \end{pgfonlayer}
        \end{tikzpicture}                   
    }     
}

\newcommand{\Tbox}[3]{
    \begin{tikzpicture}
        \begin{pgfonlayer}{nodelayer}
            \node [style=none] (35) at (-2, -0.5) {};
            \node [style=none] (36) at (-2, 0.5) {};
            \node [style=dots] (90) at (1.5, 0.15) {$\vdots$};
            \node [style=none] (91) at (2, -0.5) {};
            \node [style=none] (92) at (2, 0.5) {};
            \node [style=label] (94) at (-1.5, 1) {$#2$};
            \node [style=label] (95) at (1.5, 1) {$#3$};
            \node [style=none] (96) at (-0.75, -1) {};
            \node [style=none] (97) at (0.75, -1) {};
            \node [style=none] (98) at (0.75, 1) {};
            \node [style=none] (99) at (-0.75, 1) {};
            \node [style=none] (104) at (0, 0) {$#1$};
            \node [style=dots] (105) at (-1.5, 0.15) {$\vdots$};
            \node [style=none] (106) at (-0.75, -0.5) {};
            \node [style=none] (107) at (0.75, -0.5) {};
            \node [style=none] (108) at (0.75, 0.5) {};
            \node [style=none] (109) at (-0.75, 0.5) {};
        \end{pgfonlayer}
        \begin{pgfonlayer}{edgelayer}
            \draw [style=tape] (97.center)
                 to (96.center)
                 to (106.center)
                 to (35.center)
                 to (36.center)
                 to (109.center)
                 to (99.center)
                 to (98.center)
                 to (108.center)
                 to (92.center)
                 to (91.center)
                 to (107.center)
                 to cycle;
        \end{pgfonlayer}
    \end{tikzpicture}        
}

\newcommand{\Czero}[1]{
    \begin{tikzpicture}
        \begin{pgfonlayer}{nodelayer}
            \ifthenelse{\equal{#1}{\uno}}{}{
            \node [style=label] (4) at (1, 0) {$#1$};
            }
            \node [style=none] (11) at (0, 0) {};
            \node [style=white] (12) at (-1, 0) {};
        \end{pgfonlayer}
        \begin{pgfonlayer}{edgelayer}
            \draw (12) to (11.center);
        \end{pgfonlayer}
    \end{tikzpicture}    
}

\newcommand{\Cuno}[1]{
    \begin{tikzpicture}
        \begin{pgfonlayer}{nodelayer}
            \ifthenelse{\equal{#1}{\uno}}{}{
            \node [style=label] (4) at (-1, 0) {$#1$};
            }
            \node [style=none] (11) at (0, 0) {};
            \node [style=none] (12) at (1, 0) {};
            \node [style=none] (13) at (1, 0.25) {};
            \node [style=none] (14) at (1, -0.25) {};
        \end{pgfonlayer}
        \begin{pgfonlayer}{edgelayer}
            \draw (12.center) to (11.center);
            \draw (14.center) to (13.center);
        \end{pgfonlayer}
    \end{tikzpicture}    
}

\newcommand{\Tcopier}[1]{\copier{#1}
}

\newcommand{\Tcocopier}[1]{
    \cocopier{#1}
}

\newcommand{\Tdischarger}[1]{\discharger{#1}
}

\newcommand{\Tcodischarger}[1]{\codischarger{#1}
}

\newcommand{\Copier}[1]{
    \begin{tikzpicture}
	\begin{pgfonlayer}{nodelayer}
		\node [style=none] (0) at (-1, 1) {};
		\node [style=none] (3) at (-1, -1) {};
		\node [style=none] (4) at (1.25, -1) {};
		\node [style=none] (5) at (1.25, 1) {};
		\node [style=none] (13) at (-1, 0) {};
		\node [style=black] (14) at (0, 0) {};
		\node [style=label] (15) at (-1.5, 0) {$#1$};
		\node [style=label] (16) at (1.75, 0.575) {$#1$};
		\node [style=none] (17) at (0.75, 0.575) {};
		\node [style=none] (18) at (0.75, -0.575) {};
		\node [style=label] (19) at (1.75, -0.575) {$#1$};
		\node [style=none] (20) at (1.25, 0.575) {};
		\node [style=none] (21) at (1.25, -0.575) {};
	\end{pgfonlayer}
	\begin{pgfonlayer}{edgelayer}
		\draw [style=tape] (5.center)
			 to (4.center)
			 to (3.center)
			 to (0.center)
			 to cycle;
		\draw (13.center) to (14);
		\draw [bend left] (14) to (17.center);
		\draw [bend right] (14) to (18.center);
		\draw (17.center) to (20.center);
		\draw (21.center) to (18.center);
	\end{pgfonlayer}
\end{tikzpicture}
}

\newcommand{\Cocopier}[1]{
    \begin{tikzpicture}
	\begin{pgfonlayer}{nodelayer}
		\node [style=none] (0) at (1.25, 1) {};
		\node [style=none] (3) at (1.25, -1) {};
		\node [style=none] (4) at (-1, -1) {};
		\node [style=none] (5) at (-1, 1) {};
		\node [style=none] (13) at (1.25, 0) {};
		\node [style=black] (14) at (0.25, 0) {};
		\node [style=label] (15) at (1.75, 0) {$#1$};
		\node [style=label] (16) at (-1.5, 0.575) {$#1$};
		\node [style=none] (17) at (-0.5, 0.575) {};
		\node [style=none] (18) at (-0.5, -0.575) {};
		\node [style=label] (19) at (-1.5, -0.575) {$#1$};
		\node [style=none] (20) at (-1, 0.575) {};
		\node [style=none] (21) at (-1, -0.575) {};
	\end{pgfonlayer}
	\begin{pgfonlayer}{edgelayer}
		\draw [style=tape] (5.center)
			 to (4.center)
			 to (3.center)
			 to (0.center)
			 to cycle;
		\draw (13.center) to (14);
		\draw [bend right] (14) to (17.center);
		\draw [bend left] (14) to (18.center);
		\draw (17.center) to (20.center);
		\draw (21.center) to (18.center);
	\end{pgfonlayer}
\end{tikzpicture}  
}

\newcommand{\Discharger}[1]{
    \begin{tikzpicture}
	\begin{pgfonlayer}{nodelayer}
		\node [style=none] (0) at (-1, 1) {};
		\node [style=none] (3) at (-1, -1) {};
		\node [style=none] (4) at (1.25, -1) {};
		\node [style=none] (5) at (1.25, 1) {};
		\node [style=none] (13) at (-1, 0) {};
		\node [style=black] (14) at (0.25, 0) {};
		\node [style=label] (15) at (-1.5, 0) {$#1$};
		\node [style=label] (16) at (1.75, 0) {$\vphantom{#1}$};
	\end{pgfonlayer}
	\begin{pgfonlayer}{edgelayer}
		\draw [style=tape] (5.center)
			 to (4.center)
			 to (3.center)
			 to (0.center)
			 to cycle;
		\draw (13.center) to (14);
	\end{pgfonlayer}
\end{tikzpicture}
}

\newcommand{\Codischarger}[1]{
    \begin{tikzpicture}
	\begin{pgfonlayer}{nodelayer}
		\node [style=none] (0) at (1.25, 1) {};
		\node [style=none] (3) at (1.25, -1) {};
		\node [style=none] (4) at (-1, -1) {};
		\node [style=none] (5) at (-1, 1) {};
		\node [style=none] (13) at (1.25, 0) {};
		\node [style=black] (14) at (0, 0) {};
		\node [style=label] (15) at (1.75, 0) {$#1$};
		\node [style=label] (16) at (-1.5, 0) {$\vphantom{#1}$};
	\end{pgfonlayer}
	\begin{pgfonlayer}{edgelayer}
		\draw [style=tape] (5.center)
			 to (4.center)
			 to (3.center)
			 to (0.center)
			 to cycle;
		\draw (13.center) to (14);
	\end{pgfonlayer}
\end{tikzpicture}
}

\newcommand{\CBcopier}[1]{
    \begin{tikzpicture}
        \begin{pgfonlayer}{nodelayer}
            \node [style=none] (6) at (-1.25, 0) {};
            \node [style=black] (7) at (-0.25, 0) {};
            \node [style=label] (8) at (-1.75, 0) {$#1$};
            \node [style=label] (9) at (1.5, 0.575) {$#1$};
            \node [style=none] (10) at (0.5, 0.575) {};
            \node [style=none] (11) at (0.5, -0.575) {};
            \node [style=label] (12) at (1.5, -0.575) {$#1$};
            \node [style=none] (13) at (1, 0.575) {};
            \node [style=none] (14) at (1, -0.575) {};
        \end{pgfonlayer}
        \begin{pgfonlayer}{edgelayer}
            \draw (6.center) to (7);
            \draw [bend left] (7) to (10.center);
            \draw [bend right] (7) to (11.center);
            \draw (10.center) to (13.center);
            \draw (14.center) to (11.center);
        \end{pgfonlayer}
    \end{tikzpicture}      
}

\newcommand{\CBcocopier}[1]{
    \begin{tikzpicture}
        \begin{pgfonlayer}{nodelayer}
            \node [style=none] (6) at (1, 0) {};
            \node [style=black] (7) at (0, 0) {};
            \node [style=label] (8) at (1.5, 0) {$#1$};
            \node [style=label] (9) at (-1.75, 0.575) {$#1$};
            \node [style=none] (10) at (-0.75, 0.575) {};
            \node [style=none] (11) at (-0.75, -0.575) {};
            \node [style=label] (12) at (-1.75, -0.575) {$#1$};
            \node [style=none] (13) at (-1.25, 0.575) {};
            \node [style=none] (14) at (-1.25, -0.575) {};
        \end{pgfonlayer}
        \begin{pgfonlayer}{edgelayer}
            \draw (6.center) to (7);
            \draw [bend right] (7) to (10.center);
            \draw [bend left] (7) to (11.center);
            \draw (10.center) to (13.center);
            \draw (14.center) to (11.center);
        \end{pgfonlayer}
    \end{tikzpicture}           
}

\newcommand{\CBdischarger}[1]{
    \begin{tikzpicture}
        \begin{pgfonlayer}{nodelayer}
            \node [style=none] (6) at (-1, 0) {};
            \node [style=black] (7) at (0.5, 0) {};
            \node [style=label] (8) at (-1.5, 0) {$#1$};
        \end{pgfonlayer}
        \begin{pgfonlayer}{edgelayer}
            \draw (6.center) to (7);
        \end{pgfonlayer}
    \end{tikzpicture}      
}

\newcommand{\CBcodischarger}[1]{
    \begin{tikzpicture}
        \begin{pgfonlayer}{nodelayer}
            \node [style=none] (6) at (0, 0) {};
            \node [style=black] (7) at (-1.5, 0) {};
            \node [style=label] (8) at (0.5, 0) {$#1$};
        \end{pgfonlayer}
        \begin{pgfonlayer}{edgelayer}
            \draw (6.center) to (7);
        \end{pgfonlayer}
    \end{tikzpicture}         
}

\newcommand{\TRelCodiagDiag}[1]{\begin{tikzpicture}
	\begin{pgfonlayer}{nodelayer}
		\node [style=none] (0) at (-1.5, 0) {};
		\node [style=none] (1) at (-0.25, 1.25) {};
		\node [style=none] (2) at (-0.25, -1.25) {};
		\node [style=none] (8) at (-0.25, -1.75) {};
		\node [style=none] (9) at (-0.25, 1.75) {};
		\node [style=none] (10) at (-0.25, 0.75) {};
		\node [style=none] (11) at (-1, 0) {};
		\node [style=none] (12) at (-0.25, -0.75) {};
		\node [style=none] (13) at (-2, -0.5) {};
		\node [style=none] (14) at (-2, 0.5) {};
		\node [style=none] (24) at (-3, 0) {};
		\node [style=none] (35) at (-3, -0.5) {};
		\node [style=none] (36) at (-3, 0.5) {};
		\node [style=none] (50) at (-0.25, 1.25) {};
		\node [style=none] (80) at (1.5, 0) {};
		\node [style=none] (81) at (0.25, 1.25) {};
		\node [style=none] (82) at (0.25, -1.25) {};
		\node [style=none] (83) at (0.25, -1.75) {};
		\node [style=none] (84) at (0.25, 1.75) {};
		\node [style=none] (85) at (0.25, 0.75) {};
		\node [style=none] (86) at (1, 0) {};
		\node [style=none] (87) at (0.25, -0.75) {};
		\node [style=none] (88) at (2, -0.5) {};
		\node [style=none] (89) at (2, 0.5) {};
		\node [style=none] (90) at (3, 0) {};
		\node [style=none] (91) at (3, -0.5) {};
		\node [style=none] (92) at (3, 0.5) {};
		\node [style=none] (93) at (0.25, 1.25) {};
		\node [style=label] (94) at (-3.5, 0) {$#1$};
		\node [style=label] (95) at (3.5, 0) {$#1$};
	\end{pgfonlayer}
	\begin{pgfonlayer}{edgelayer}
		\draw [style=tape] (83.center)
			 to (8.center)
			 to [bend left] (13.center)
			 to (35.center)
			 to (36.center)
			 to (14.center)
			 to [bend left] (9.center)
			 to (84.center)
			 to [bend left] (89.center)
			 to (92.center)
			 to (91.center)
			 to (88.center)
			 to [bend left] cycle;
		\draw [style=tapeNoFill, fill=white] (87.center)
			 to (12.center)
			 to [bend left=45] (11.center)
			 to [bend left=45] (10.center)
			 to (85.center)
			 to [bend left=45] (86.center)
			 to [bend left=45] cycle;
		\draw [bend right=45] (0.center) to (2.center);
		\draw [bend left=45] (0.center) to (1.center);
		\draw (24.center) to (0.center);
		\draw [bend left=45] (80.center) to (82.center);
		\draw [bend right=45] (80.center) to (81.center);
		\draw (90.center) to (80.center);
		\draw (50.center) to (93.center);
		\draw (82.center) to (2.center);
	\end{pgfonlayer}
\end{tikzpicture}
}

\newcommand{\st}[2]{(#1 \mid #2)}

\NewDocumentCommand{\TTrace}{gooo}{
        
}

\newcommand{\TstringConvolution}[4]{\begin{tikzpicture}
	\begin{pgfonlayer}{nodelayer}
		\node [style=none] (0) at (7.75, 1.75) {};
		\node [style=none] (3) at (7.75, -1.75) {};
		\node [style=none] (4) at (2.75, -1.75) {};
		\node [style=none] (5) at (2.75, 1.75) {};
		\node [style=black] (16) at (3.75, 0) {};
		\node [style=none] (17) at (4.5, 0.75) {};
		\node [style=none] (18) at (4.5, -0.775) {};
		\node [style=label] (19) at (2.25, 0) {$#3$};
		\node [style=none] (20) at (2.75, 0) {};
		\node [style=stringbox] (21) at (5.25, 0.75) {$#1\vphantom{#2}$};
		\node [style=stringbox] (22) at (5.25, -0.775) {$#2\vphantom{#1}$};
		\node [style=black] (23) at (6.75, 0) {};
		\node [style=none] (24) at (6, 0.75) {};
		\node [style=none] (25) at (6, -0.775) {};
		\node [style=label] (26) at (8.25, 0) {$#4$};
		\node [style=none] (27) at (7.75, 0) {};
	\end{pgfonlayer}
	\begin{pgfonlayer}{edgelayer}
		\draw [style=tape] (5.center)
			 to (4.center)
			 to (3.center)
			 to (0.center)
			 to cycle;
		\draw [bend left=45] (18.center) to (16);
		\draw [bend left=45] (16) to (17.center);
		\draw (20.center) to (16);
		\draw [bend right=45] (25.center) to (23);
		\draw [bend right=45] (23) to (24.center);
		\draw (27.center) to (23);
		\draw (17.center) to (21);
		\draw (21) to (24.center);
		\draw (25.center) to (22);
		\draw (22) to (18.center);
	\end{pgfonlayer}
\end{tikzpicture}
}

\NewDocumentCommand{\stringConvolution}{mmmm}{
    \begin{tikzpicture}
        \begin{pgfonlayer}{nodelayer}
            \node [style=black] (121) at (-1.5, 0) {};
            \node [style=none] (122) at (-0.75, 0.75) {};
            \node [style=none] (123) at (-0.75, -0.775) {};
            \node [style=label] (124) at (-3, 0) {$#3$};
            \node [style=none] (141) at (-2.5, 0) {};
            \node [style=stringbox] (142) at (0, 0.75) {$#1\vphantom{#2}$};
            \node [style=stringbox] (143) at (0, -0.775) {$#2\vphantom{#1}$};
            \node [style=black] (144) at (1.5, 0) {};
            \node [style=none] (145) at (0.75, 0.75) {};
            \node [style=none] (146) at (0.75, -0.775) {};
            \node [style=label] (147) at (3, 0) {$#4$};
            \node [style=none] (148) at (2.5, 0) {};
        \end{pgfonlayer}
        \begin{pgfonlayer}{edgelayer}
            \draw [bend left=45] (123.center) to (121);
            \draw [bend left=45] (121) to (122.center);
            \draw (141.center) to (121);
            \draw [bend right=45] (146.center) to (144);
            \draw [bend right=45] (144) to (145.center);
            \draw (148.center) to (144);
            \draw (122.center) to (142);
            \draw (142) to (145.center);
            \draw (146.center) to (143);
            \draw (143) to (123.center);
        \end{pgfonlayer}
    \end{tikzpicture}        
}

\newcommand{\TstringBottom}[2]{\begin{tikzpicture}
	\begin{pgfonlayer}{nodelayer}
		\node [style=none] (0) at (7.25, 1) {};
		\node [style=none] (3) at (7.25, -1) {};
		\node [style=none] (4) at (3.75, -1) {};
		\node [style=none] (5) at (3.75, 1) {};
		\node [style=black] (28) at (5, 0) {};
		\node [style=label] (29) at (3.25, 0) {$#1$};
		\node [style=none] (30) at (3.75, 0) {};
		\node [style=black] (31) at (6, 0) {};
		\node [style=label] (32) at (7.75, 0) {$#2$};
		\node [style=none] (33) at (7.25, 0) {};
	\end{pgfonlayer}
	\begin{pgfonlayer}{edgelayer}
		\draw [style=tape] (5.center)
			 to (4.center)
			 to (3.center)
			 to (0.center)
			 to cycle;
		\draw (30.center) to (28);
		\draw (33.center) to (31);
	\end{pgfonlayer}
\end{tikzpicture}
}

\NewDocumentCommand{\stringBottom}{mm}{
    
}

\newcommand{\TstringOp}[3]{\begin{tikzpicture}
	\begin{pgfonlayer}{nodelayer}
		\node [style=none] (0) at (8, 1.5) {};
		\node [style=none] (3) at (8, -1.5) {};
		\node [style=none] (4) at (2.5, -1.5) {};
		\node [style=none] (5) at (2.5, 1.5) {};
		\node [style=stringbox] (34) at (5.25, 0) {$#1$};
		\node [style=label] (35) at (8.5, 1) {$#2$};
		\node [style=label] (36) at (2, -1) {$#3$};
		\node [style=none] (37) at (6, 0) {};
		\node [style=none] (38) at (4.5, 0) {};
		\node [style=none] (39) at (6, -1) {};
		\node [style=none] (40) at (4.5, 1) {};
		\node [style=black] (41) at (6.75, -0.5) {};
		\node [style=black] (42) at (7.5, -0.5) {};
		\node [style=black] (43) at (3.75, 0.5) {};
		\node [style=black] (44) at (3, 0.5) {};
		\node [style=none] (45) at (8, 1) {};
		\node [style=none] (46) at (2.5, -1) {};
	\end{pgfonlayer}
	\begin{pgfonlayer}{edgelayer}
		\draw [style=tape] (5.center)
			 to (4.center)
			 to (3.center)
			 to (0.center)
			 to cycle;
		\draw (38.center) to (34);
		\draw (34) to (37.center);
		\draw (44) to (43);
		\draw [bend left] (43) to (40.center);
		\draw [bend right] (43) to (38.center);
		\draw [bend right] (39.center) to (41);
		\draw [bend left] (37.center) to (41);
		\draw (41) to (42);
		\draw (40.center) to (45.center);
		\draw (39.center) to (46.center);
	\end{pgfonlayer}
\end{tikzpicture}
}

\NewDocumentCommand{\stringOp}{mmm}{
    
}

\newcommand{\Tunion}[4]{
    \begin{tikzpicture}
	\begin{pgfonlayer}{nodelayer}
		\node [style=none] (6) at (-0.25, 1.9) {};
		\node [style=none] (7) at (-0.5, 0.4) {};
		\node [style=none] (15) at (-0.5, 1.15) {};
		\node [style=none] (22) at (0.5, 1.15) {};
		\node [style=none] (25) at (0.5, 0.4) {};
		\node [style=none] (26) at (0.25, 1.9) {};
		\node [style=label] (36) at (-3.5, 0) {$#3$};
		\node [style=bbox, scale=0.9] (43) at (0, 1.15) {$#1$};
		\node [style=none] (44) at (-0.5, -0.4) {};
		\node [style=none] (45) at (-0.25, -1.9) {};
		\node [style=none] (46) at (-0.5, -1.15) {};
		\node [style=none] (47) at (0.5, -1.15) {};
		\node [style=none] (48) at (0.25, -1.9) {};
		\node [style=none] (49) at (0.5, -0.4) {};
		\node [style=label] (51) at (3.5, 0) {$#4$};
		\node [style=bbox, scale=0.9] (52) at (0, -1.15) {$#2$};
		\node [style=none] (53) at (3, 0) {};
		\node [style=none] (54) at (3, -0.9) {};
		\node [style=none] (55) at (3, 0.9) {};
		\node [style=none] (56) at (-3, 0) {};
		\node [style=none] (57) at (-3, -0.9) {};
		\node [style=none] (58) at (-3, 0.9) {};
		\node [style=none] (59) at (1.75, 0) {};
		\node [style=none] (60) at (2, -0.9) {};
		\node [style=none] (61) at (2, 0.9) {};
		\node [style=none] (62) at (-1.75, 0) {};
		\node [style=none] (63) at (-2, -0.9) {};
		\node [style=none] (64) at (-2, 0.9) {};
	\end{pgfonlayer}
	\begin{pgfonlayer}{edgelayer}
		\draw [tape] (54.center) to (60.center) to[bend left] (48.center) to (45.center) to[bend left] (63.center) to (57.center) to (58.center) to (64.center) to[bend left] (6.center) to (26.center) to[bend left] (61.center) to (55.center) to cycle;
		\draw [tape, fill=white] (7.center) to (25.center) to[bend left=90, looseness=2.00] (49.center) to (44.center) to[bend left=90, looseness=2.00] cycle;
		\draw (15.center) to (22.center);
		\draw (46.center) to (47.center);
		\draw [bend left=45] (22.center) to (59.center);
		\draw (59.center) to (53.center);
		\draw [bend left=45] (59.center) to (47.center);
		\draw [bend left=45] (46.center) to (62.center);
		\draw (62.center) to (56.center);
		\draw [bend left=45] (62.center) to (15.center);
	\end{pgfonlayer}
\end{tikzpicture}
}

\newcommand{\Tkstar}[2]{\begin{tikzpicture}
	\begin{pgfonlayer}{nodelayer}
		\node [style=none] (71) at (-2.5, 1.25) {};
		\node [style=none] (77) at (-2.5, 2) {};
		\node [style=none] (80) at (-2.25, 0.5) {};
		\node [style=none] (82) at (2, 1.25) {};
		\node [style=none] (85) at (-2.5, -1.25) {};
		\node [style=none] (86) at (-2.25, -0.5) {};
		\node [style=none] (89) at (-2.5, -2) {};
		\node [style=none] (90) at (2.25, -1.25) {};
		\node [style=label] (91) at (4.5, -1.25) {$#2$};
		\node [style=none] (95) at (-0.5, 1) {};
		\node [style=none] (96) at (-0.5, -0.75) {};
		\node [style=none] (99) at (2, 2) {};
		\node [style=none] (100) at (2, 0.5) {};
		\node [style=none] (102) at (2, -0.5) {};
		\node [style=none] (103) at (2.25, -2) {};
		\node [style=none] (106) at (0.25, 1) {};
		\node [style=none] (107) at (0.25, -0.75) {};
		\node [style=none] (109) at (-1, 0) {};
		\node [style=none] (110) at (0.75, 0) {};
		\node [style=bbox, scale=0.9] (111) at (-2.5, 1.25) {$#1$};
		\node [style=none] (112) at (-5, -1.25) {};
		\node [style=none] (113) at (4, -0.5) {};
		\node [style=none] (114) at (4, -2) {};
		\node [style=none] (115) at (2, 2) {};
		\node [style=none] (116) at (2, 0.5) {};
		\node [style=none] (117) at (-5, -0.5) {};
		\node [style=none] (118) at (-5, -2) {};
		\node [style=none] (119) at (-3, 2) {};
		\node [style=none] (120) at (-3, 0.5) {};
		\node [style=none] (121) at (-3, 1.25) {};
		\node [style=none] (122) at (2, 1.25) {};
		\node [style=none] (123) at (4, -1.25) {};
		\node [style=none] (124) at (2, 4.5) {};
		\node [style=none] (125) at (2, 3) {};
		\node [style=none] (126) at (2, 3.75) {};
		\node [style=none] (127) at (-3, 4.5) {};
		\node [style=none] (128) at (-3, 3) {};
		\node [style=none] (129) at (-3, 3.75) {};
		\node [style=label] (130) at (-5.5, -1.25) {$#2$};
	\end{pgfonlayer}
	\begin{pgfonlayer}{edgelayer}
		\draw [tape] (116.center)
			 to [bend right=90, looseness=1.75] (124.center)
			 to (127.center)
			 to [bend left=270, looseness=1.75] (120.center)
			 to (80.center)
			 to [bend left=90, looseness=1.50] (86.center)
			 to (117.center)
			 to (118.center)
			 to (89.center)
			 to [bend right] (96.center)
			 to (107.center)
			 to [bend right] (103.center)
			 to (114.center)
			 to (113.center)
			 to (102.center)
			 to [bend left=90, looseness=1.50] (100.center)
			 to cycle;
		\draw [tape, fill=white] (99.center)
			 to [bend right] (106.center)
			 to [bend left=360] (95.center)
			 to [bend right] (77.center)
			 to [bend left=360] (119.center)
			 to [bend right=270, looseness=2.25] (128.center)
			 to (125.center)
			 to [bend left=90, looseness=2.00] (115.center)
			 to cycle;
		\draw [bend left=45] (110.center) to (82.center);
		\draw [bend right=45] (110.center) to (90.center);
		\draw [bend right=45] (85.center) to (109.center);
		\draw [bend right=45] (109.center) to (71.center);
		\draw (85.center) to (112.center);
		\draw (71.center) to (111);
		\draw (121.center) to (111);
		\draw (109.center) to (110.center);
		\draw (82.center) to (122.center);
		\draw (90.center) to (123.center);
		\draw [bend right=90, looseness=1.75] (129.center) to (121.center);
		\draw [bend right=90, looseness=1.75] (122.center) to (126.center);
		\draw (129.center) to (126.center);
	\end{pgfonlayer}
\end{tikzpicture}
}

\crefname{defi}{Definition}{Definitions}
\crefname{thm}{Theorem}{Theorems}
\crefname{lem}{Lemma}{Lemmas}
\crefname{prop}{Proposition}{Propositions}
\crefname{cor}{Corollary}{Corollaries}
\crefname{exa}{Example}{Examples}
\crefname{rem}{Remark}{Remarks}

\usepackage{subcaption}

\makeatletter
\renewcommand{\p@subtable}{\thetable.}
\makeatother

\makeatletter
\renewcommand{\p@subfigure}{\thefigure.}
\makeatother

\allowdisplaybreaks
\begin{document}

\title{A Diagrammatic Basis for Computer Programming}

\specialissue{Selected Papers of the 28th International Conference on Foundations of Software Science and Computation Structures (FoSSaCS 2025)}

\author[F.~Bonchi]{Filippo Bonchi\lmcsorcid{0000-0002-3433-723X}}[a]
\author[A.~Di Giorgio]{Alessandro Di Giorgio\lmcsorcid{0000-0002-6428-6461}}[b]
\author[E.~Di Lavore]{Elena Di Lavore\lmcsorcid{0000-0002-7783-5079}}[b,c]

\address{University of Pisa, Italy}	%

\address{Tallinn University of Technology, Estonia}	%
\email{alessd@taltech.ee}  %

\address{University of Oxford, United Kingdom}	%

\begin{abstract}
    \noindent Tape diagrams provide a convenient graphical notation for arrows of rig categories, i.e., categories equipped with two monoidal products, $\piu$ and $\per$.
    In this work, we introduce Kleene-Cartesian rig categories, namely rig categories where $\per$ provides a Cartesian bicategory, while $\piu$ a Kleene bicategory.
    We show that the associated tape diagrams can conveniently deal with imperative programs and various program logics.        
\end{abstract}

\bibliographystyle{alphaurl}
\maketitle

\tikzset{baseline=-0.5ex,>=stealth'}
\tikzset{every picture/.append style={scale=0.25}}

\section{Introduction}\label{sec:intro}
The calculus of relations, originally introduced by  De Morgan and Peirce in the late 19th century, is an ancestor of first-order logic that was revitalised by Tarski in 1941.
With the dawn of program logics, the calculus of relations -- extended with transitive closure -- was early recognised~\cite{pratt1976semantical} to provide them with solid algebraic foundations.

Around the same time, Lawvere was pioneering categorical logic by introducing the concept of \emph{functorial semantics}~\cite{LawvereOriginalPaper}. Given an algebraic theory $T$ (in the sense of universal algebra, i.e., a signature $\Sigma$ and a set of equations $E$), one can freely generate a Cartesian category $\mathcal{L}_T$. Models, in the standard algebraic sense, are in one-to-one correspondence with Cartesian functors $F$ from $\mathcal{L}_T$ to $\Cat{Set}$, the category of sets and functions. More generally, models of the theory in any Cartesian category $\Cat{C}$ are represented by Cartesian functors $F: \mathcal{L}_T \to \Cat{C}$. However this approach fails if one tries to apply it to relational theories by choosing $\Cat{C}$ as $\Rel$, the category of sets and relations, because the Cartesian product of sets is not the categorical product in $\Rel$.

A refinement of Lawvere's method for relational structures has been recently proposed in~\cite{Bonchi2017c,GCQ,DBLP:journals/corr/abs-2009-06836}. Starting from a \emph{monoidal signature}, one can freely generate a \emph{Cartesian bicategory}~\cite{katis1997bicategories} and define models as morphisms into $(\Rel, \per, \uno)$, the monoidal category of relations, where the monoidal product $\per$ is the Cartesian product of sets. This framework captures regular theories, i.e., those involving the $\{\exists, \land, \top\}$-fragment of first-order logic. More recently~\cite{DBLP:conf/lics/Bonchi0H024}, this approach was extended to full first-order logic by deriving negation from the interaction of Cartesian and linear bicategories~\cite{cockett2000introduction}.

In this paper, we extend the Cartesian bicategory framework in a different direction: program logics.
In the last decades, there has been an explosion of program logics and many researchers felt the need for more systematic approaches.
Our proposal is based on relational and categorical algebra.
We propose tape diagrams as an ``assembly language'' for interpreting various program logics (\Cref{rem:other-program-logics}).
While the inference rules for each of these logics are usually defined by the ingenuity of the researchers, in our approach such rules follow from the laws of Kleene-Cartesian rig categories.
These laws arise from the interaction of canonical categorical structures on the category of sets and relations.
Crucially, the same approach has led to the identification of various categorical structures corresponding to various well known logics (\Cref{fig:logics-categories}).
\begin{figure}[h!]%
  \centering
  \begin{tabular}{| c | l l |}
    \hline
    & Logic & Categorical structure \\
    \hline
    \cite{LawvereOriginalPaper} & Equational logic & Cartesian category \\
    \cite{GCQ} & Regular logic & Cartesian bicategory \\
    \cite{bonchi2023deconstructing} & Coherent logic & Finite-biproduct Cartesian bicategory \\
    \cite{DBLP:conf/lics/Bonchi0H024} & First-order logic & First-order bicategory \\
    This work & Program logic & Kleene-Cartesian rig category\\
    \hline
  \end{tabular}
  \caption{Categorical structures correspond to logics.}\label{fig:logics-categories}
\end{figure}

An idea, originating at least from Bainbridge~\cite{bainbridge1976feedback}, is to model data flow using the Cartesian product of relations, $(\Rel, \per, \uno)$, and control flow using a different monoidal structure on relations: $(\Rel, \piu, \zero)$.
In this second structure, the monoidal product $\piu$ is the disjoint union of sets, which acts both as a categorical product and a coproduct, hence, a \emph{biproduct}.
Both monoidal categories are \emph{traced}~\cite{joyal1996traced}: the trace in $(\Rel, \per, \uno)$ represents feedback, while in $(\Rel, \piu, \zero)$ --the focus of our work-- it provides \emph{iteration}~\cite{selinger1998note}.

Our first step is to extract from $(\Rel, \piu, \zero)$ the categorical structures essential for modelling control flow, which we term \emph{Kleene bicategories}. Essentially, a Kleene bicategory is a poset enriched traced monoidal category where the monoidal product $\piu$ is a biproduct, and the induced natural comonoid~\cite{fox1976coalgebras} is \emph{right adjoint} to the natural monoid. The trace must satisfy a posetal variant of the so called \emph{uniformity} condition~\cite{cuazuanescu1994feedback,hasegawa2003uniformity}. The term ``Kleene'' is justified because every Kleene bicategory forms a (typed) Kleene algebra in Kozen's sense~\cite{Kozen94acompleteness,kozen98typedkleene} (Corollary \ref{cor:kleeneareka}), while any Kleene algebra canonically gives rise, through the biproduct completion~\cite{mac_lane_categories_1978}, to a Kleene bicategory.

To model control and data flow within a unified categorical structure, we employ \emph{rig categories}~\cite{laplaza_coherence_1972}, categories equipped with two monoidal products, $\piu$ and $\per$, where $\per$ distributes over $\piu$. We define \emph{Kleene-Cartesian rig} (kc-rig) \emph{categories}, where $\piu$ and $\per$ exhibit the structures of Kleene and Cartesian bicategories, respectively. To construct the freely generated kc-rig category (Theorem \ref{thm:KleeneCartesiantapesfree}), we extend \emph{tape diagrams}~\cite{bonchi2023deconstructing}, a diagrammatic notation recently introduced for rig categories. Intuitively, tape diagrams are \emph{string diagrams}~\cite{joyal1991geometry} in which other string diagrams are nested: the inner diagrams model data flow, and the outer ones model control flow. On one hand, this offers an intuitive unified picture of Bainbridge's idea; on the other it allows for visualising the laws of kc-rig categories (Figures \ref{fig:tapesax}, \ref{fig:cb axioms} and \ref{fif:poset unif tape axioms}) in a way that sheds light on several monoidal algebras occurring in different types of systems \cite{stein2023probabilistic,coecke2011interacting,Backens-ZXcompleteness1,Fritz_stochasticmatrices,bruni2011connector,BaezCoya-propsnetworktheory,bonchi2015full,DBLP:journals/pacmpl/BonchiHPSZ19,BonchiPSZ19,lmcs:10963,Ghica2016}.

We then introduce \emph{Kleene-Cartesian theories} and their models that, like in Lawvere's approach, coincide with functors (Proposition \ref{funct:sem}). We illustrate an example of a Kleene-Cartesian theory which is not first-order: Peano's axiomatisation of natural numbers. We demonstrate how imperative programs and their logics~\cite{hoare1969axiomatic,kozen00hoarekleene,DBLP:conf/vmcai/CousotCFL13,o2019incorrectness,DBLP:journals/corr/abs-2310-18156} -- even more sophisticated ones, like \cite{benton2004simple}, where the interaction of data and control flow plays a key role -- can be encoded within Kleene-Cartesian tape diagrams. In particular, we show that the rules of Hoare logic follow from the laws of kc-rig categories (Proposition \ref{prop:Hoare}). %
Finally, the framework is expressive enough to capture the positive fragment of the calculus of relations with transitive closure, which is the departure of our journey. %

\paragraph{Synopsis}
In the next section we recall the calculus of relations equipped with reflexive 
and transitive closure.  
Its allegorical fragment can be expressed using Cartesian bicategories, reviewed 
in \Cref{sc:cb:background}, while its Kleene fragment is captured by Kleene 
bicategories, introduced in \Cref{sec:kleene}.  
To combine these two structures, we recall rig categories in 
\Cref{sec:rigcategories}.  
The central notion of Kleene-Cartesian rig categories is then developed in 
\Cref{sec:cb}.  
Tape diagrams for such categories, together with the notion of Kleene-Cartesian 
theory, are presented in \Cref{sec:kctapes}.  
In \Cref{sec:peano} we describe the Kleene-Cartesian theory of Peano's natural 
numbers, while in \Cref{sec:hoare} we introduce a theory for imperative programs 
and program logics.  
All remaining proofs appear in the appendix, which also contains the coherence conditions for rig categories, and further 
auxiliary material.  In particular, Appendix~\ref{app:dictionary} contains a dictionary for the structure of Kleene-Cartesian rig categories and its representation as string diagrams and tape diagrams. 
This paper extends the conference version~\cite{DBLP:conf/fossacs/BonchiGL25} by 
including full proofs, additional examples, and several minor results.

\section{The calculus of relations}\label{sec:CR} 
We commence our exposition by recalling the positive fragment of  the calculus of relations with
reflexive and transitive closure ($\CR$).  See \cite{DBLP:conf/stacs/Pous18} for a more detailed overview. Its syntax is given by the grammar below, where $R$ is taken from a given set $\Sigma$ of generating symbols.
 \begin{align}
     E  ::= \ &  R\phantom{\op{}}       \mid \id{}  \mid E;E      \;\;\mid  \label{eq:calculusofrelation1}\\
            &  \op{E}    \mid  \top  \mid E \cap E \mid  \label{eq:calculusofrelation2} \\
            &  \kstar{E} \mid  \bot  \mid E \cup E       \label{eq:calculusofrelation3}
   \end{align}
Beyond the usual relational composition $;$, union $\cup$, intersection $\cap$ and their units $\id{}$, $\bot$ and $\top$, the calculus features two unary operations:
the opposite $\op{(\cdot)}$ and the reflexive and transitive closure $\kstar{(\cdot)}$.  For all sets $X,Y,Z$ and relations $R\subseteq X \times Y$, $S\subseteq Y\times Z$, composition and identities are defined as
\begin{equation*}
R;S\defeq \{(x,z) \mid \exists y \in Y. \, (x,y)\in R \wedge (y,z)\in S\}  \; \text{ and } \; \id{X} \defeq \{(x,x)\mid x\in X\} \text{,} 
\end{equation*}
the opposite as $\op{R}\defeq \{(y,x)\mid (x,y)\in R\}$, while for $R\subseteq X \times X$, its reflexive and transitive closure is $\kstar{R}\defeq\bigcup_{n\in \N}R^n$ where $R^0\defeq\id{X}$ and $R^{n+1}\defeq R;R^n$.

Its semantics, illustrated below, is defined w.r.t. a  \emph{relational interpretation} $\interpretation$, that is, a set $X$ together with a binary relation $\rho(R)\subseteq X \times X$ for each $R\in \Sigma$. %
\begin{equation*}\label{eq:sematicsExpr}
   \begin{array}{ccc}
      \begin{array}{r@{\;}c@{\;}l}
         \dsemRel{R} & \defeq & \rho(R) \\
         \dsemRel{\op{E}} & \defeq & \op{\dsemRel{E}} \\
         \dsemRel{\kstar{E}} & \defeq & \dsemRel{E}^* 
      \end{array}
      &
      \begin{array}{r@{\;}c@{\;}l}
          \dsemRel{\id{}} & \defeq & \id{X} \\
          \dsemRel{\bot} & \defeq & \emptyset \\
          \dsemRel{\top} & \defeq & X \times X
      \end{array}
      &
      \begin{array}{r@{\;}c@{\;}l}
          \dsemRel{E_1 ; E_2} & \defeq & \dsemRel{E_1} ; \dsemRel{E_2} \\
          \dsemRel{E_1 \cup E_2} & \defeq & \dsemRel{E_1} \cup \dsemRel{E_2} \\
          \dsemRel{E_1 \cap E_2} & \defeq & \dsemRel{E_1} \cap \dsemRel{E_2}
      \end{array}
  \end{array}
\end{equation*}
Two expressions $E_1$, $E_2$ are said to be \emph{equivalent}, written $E_1 \equiv_{\CR} E_2$, iff $\dsemRel{E_1}=\dsemRel{E_2}$ for all interpretations $\interpretation$. For instance, $\op{(\kstar{R})} \equiv_{\CR} \kstar{(\op{R})}$. Inclusion, denoted by $\leq_{\CR}$, is defined analogously by replacing $=$ with $\subseteq$. Axiomatisations and decidability of $\equiv_{\CR}$ have been studied focusing on several different fragments: see \cite{DBLP:conf/stacs/Pous18} and the references therein. Particularly interesting are the \emph{allegorical fragment}, consisting of \eqref{eq:calculusofrelation1} and \eqref{eq:calculusofrelation2}, and the \emph{Kleene fragment} consisting of \eqref{eq:calculusofrelation1} and \eqref{eq:calculusofrelation3}. %

\medskip

Our starting observation is that these two fragments arise from two different traced monoidal structures on $\Rel$, the category  of sets and relations: $(\Rel, \per, \uno)$ and $(\Rel, \piu, \zero)$.
In the former, the monoidal product $\per$ is given by the cartesian product of sets and, for relations $R\colon X_1\to Y_1 $, $S \colon X_2 \to Y_2$, $R\per S \colon X_1 \per X_2 \to Y_1 \per Y_2$ is defined as  
\begin{equation*}
R \per S \defeq \{(\,(x_1,x_2),\, (y_1,y_2)\,) \mid (x_1,y_1)\in R \text{ and } (x_2,y_2)\in S  \} \text{ with unit }1\defeq \{\bullet\}\text{.} 
\end{equation*} 
In $(\Rel, \piu, \zero)$, $0\defeq \emptyset$,  $\piu$ on sets is their disjoint union and $R \piu S \colon X_1\piu X_2 \to Y_1\piu Y_2$ is
\begin{equation*} R \piu S \defeq \{(\,(x_1,1),\,(y_1,1)\,) \mid (x_1,y_1)\in R \} \cup \{(\,(x_2,2),\,(y_2,2)\,) \mid (x_2,y_2)\in S \}\text{.} \end{equation*}
Here, we tag with $1$ and $2$ the elements of the disjoint union of two arbitrary sets.

For all sets $X$, the unique function $\discharger{X}\colon X \to 1$ and the pairing $\langle \id{X},\id{X}\rangle \defeq \copier{X} \colon X \to X \per X$ form a comonoid in $(\Rel, \per, \uno)$. Similarly
the unique function $\cobang{X}\colon 0 \to X$ and the copairing $[\id{X},\id{X}]\defeq \codiag{X} \colon X\piu X \to X$ form a monoid in $(\Rel, \piu, \zero)$.
By taking their opposite relations, we obtain in total the two (co)monoid structures illustrated below. 
\begin{equation}\label{eq:comonoidsREL}
  \def\arraystretch{2}
  \begin{array}{r@{\,}c@{\,}l@{\quad\;\;}r@{\,}c@{\,}l}
    \copier{X}  &\defeq&  \{(x, \; (x,x)) \mid x\in X\} & \diag{X}  &\defeq&  \op{\codiag{X}}
    \\
    \discharger{X}  &\defeq&  \{(x, \bullet) \mid x\in X\} & \bang{X}  &\defeq&  \op{\cobang{X}} 
    \\
    \cocopier{X}  &\defeq&  \op{\copier{X}} & \codiag{X}  &\defeq&  \{((x, 1), \; x) \mid x\in X\} \cup \{((x, 2), \; x) \mid x\in X\}
    \\
    \codischarger{X}  &\defeq&  \op{\discharger{X}} & \cobang{X}  &\defeq&  \emptyset 
  \end{array}
\end{equation}
The black (co)monoids give to $(\Rel, \per, \uno)$ the structure of a \emph{Cartesian bicategory} \cite{Carboni1987}, while the white ones give to $(\Rel, \piu, \zero)$ the structure of, what we named, a \emph{Kleene bicategory}.
These are illustrated in the next two sections.

\section{Cartesian Bicategories}\label{sc:cb:background}
All bicategories considered in this paper are \emph{poset enriched symmetric monoidal categories}:  every homset carries a partial order $\leq$, and composition $;$ and monoidal product $\perG$ are monotone. That is, if $f_1\leq f_2$ and $g_1\leq g_2$ then $f_1; g_1 \leq f_2; g_2$ and $f_1\perG g_1 \leq f_2\perG g_2$. 
A \emph{poset enriched symmetric monoidal functor} is a symmetric monoidal functor that preserves the order $\leq$.
The notion of \emph{adjoint arrows}, which will play a key role, amounts to the following: for $f \colon X \to Y$ and $g \colon Y \to X$,  $f$ is  \emph{left adjoint} to $g$, or $g$ is \emph{right adjoint} to $f$, written $f \dashv g$, if $\id{X} \leq f ; g$  and $g ; f \leq \id{Y}$. We extend such terminology to pairs of arrows: $(a,b)$ is left adjoint to $(c,d)$ iff $a \dashv c$ and $b \dashv d$.

All monoidal categories and functors considered throughout this paper are tacitly assumed to be strict \cite{mac_lane_categories_1978}, i.e.\ $(X\perG Y)\perG Z = X \perG (Y \perG Z)$ and $\unoG \perG X = X =X \perG \unoG$ for all objects $X,Y,Z$. This is harmless: strictification~\cite{mac_lane_categories_1978,johnson2024bimonoidal} allows to transform any monoidal category into a strict one, dispensing with the administrative burden of structural isomorphisms while ensuring the rigorous application of \emph{string diagrams}.
In this and in the next section we will use the string diagrammatic notation for traced monoidal categories from~\cite{selinger2010survey}. The unfamiliar reader may have a look at  \Cref{app:stringDiagrams} or check~\cite[Sec. 2]{bonchi2024diagrammaticalgebraprogramlogics}. In particular multiplication, unit, comultiplication and counit of the various (co)monoids, always tacitly assumed to be (co)commutative, will be drawn hereafter respectively, as 
\[
    \CBcocopier{X} \colon X \perG X \to X 
    \quad\;\; 
    \CBcodischarger{X} \colon I \to X
    \quad\;\; 
    \CBcopier{X} \colon X \to X \perG X
    \quad\;\; 
    \CBdischarger{X} \colon X \to I.
\]
\begin{figure}
    \mylabel{ax:cb:comonoid:assoc}{$\copier{}$-as}
    \mylabel{ax:cb:comonoid:unit}{$\copier{}$-un}
    \mylabel{ax:cb:comonoid:comm}{$\copier{}$-co}
    \mylabel{ax:cb:monoid:assoc}{$\cocopier{}$-as}
    \mylabel{ax:cb:monoid:unit}{$\cocopier{}$-un}
    \mylabel{ax:cb:monoid:comm}{$\cocopier{}$-co}
    \mylabel{ax:comonoid:assoc}{$\diag{}$-as}
    \mylabel{ax:comonoid:unit}{$\diag{}$-un}
    \mylabel{ax:comonoid:comm}{$\diag{}$-co}
    \mylabel{ax:monoid:assoc}{$\codiag{}$-as}
    \mylabel{ax:monoid:unit}{$\codiag{}$-un}
    \mylabel{ax:monoid:comm}{$\codiag{}$-co}
    \[

    \end{array}
    \]
        \caption{String diagrams for the axioms of coherent (co)commutative (co)monoids in \Cref{def:cartesian bicategory}.(1),(2),(3) and \Cref{def:fb}.(1),(2),(3). The gray colouring on the labels should be instantiated to black for \Cref{def:cartesian bicategory} and to white for \Cref{def:fb}. The coherence axioms -- in last two rows -- are unlabeled as they will be implicit in the graphical representation by means of tapes diagrams.}
    \label{fig:(co)monoids string}
\end{figure}
Hereafter, we briefly recall  Cartesian bicategories from \cite{Carboni1987}, 
and refer the reader to \cite{AlessandroThesis} for a more detailed exposition.
\begin{figure}[t]
    \centering
    \mylabel{ax:cb:comonoid:nat:copy}{$\copier{}$-nat}
    \mylabel{ax:cb:comonoid:nat:discard}{$\discharger{}$-nat}
    \mylabel{ax:cb:monoid:nat:copy}{$\cocopier{}$-nat}
    \mylabel{ax:cb:monoid:nat:discard}{$\codischarger{}$-nat}
    \mylabel{ax:cb:specfrob}{S}
    \mylabel{ax:cb:frob}{F}
    \mylabel{ax:cb:dischargeradj1}{$\codischarger{}\discharger{}$}
    \mylabel{ax:cb:dischargeradj2}{$\discharger{}\codischarger{}$}
    \mylabel{ax:cb:copieradj1}{$\cocopier{}\copier{}$}
    \mylabel{ax:cb:copieradj2}{$\copier{}\cocopier{}$}
    \[

    \]
    \caption{String diagrams for the axioms of Cartesian bicategories: lax naturality, special Frobenius algebras and adjointness in \Cref{def:cartesian bicategory}.(4),(5),(6).}
    \label{fig:ax-cb-string}
  \end{figure}
\begin{defi}\label{def:cartesian bicategory}
A \emph{Cartesian bicategory} is a poset enriched symmetric monoidal category $(\Cat{C}, \per, \uno)$ where, for every object $X$, there are morphisms $\cocopier{X} \colon X \otimes X \to X$, $\codischarger{X} \colon \uno \to X$, $\copier{X} \colon X \to X \otimes X$ and $\discharger{X} \colon X \to \uno$, such that%
\begin{enumerate}
    \item $(\cocopier{X}, \codischarger{X})$ is a commutative monoid:
    \[ (\cocopier{X} \otimes \id{X}) ; \cocopier{X} = (\id{X} \otimes \cocopier{X}) ; \cocopier{X}, \qquad (\codischarger{X} \otimes \id{X}) ; \cocopier{X} = \id{X} \quad\text{and}\quad \symm{X}{X} ; \cocopier{X} = \cocopier{X}; \]
    \item $(\copier{X}, \discharger{X})$ is a cocommutative comonoid:
    \[ \copier{X} ; (\copier{X} \otimes \id{X}) = \copier{X} ; (\id{X} \otimes \copier{X}), \qquad \copier{X} ; (\discharger{X} \otimes \id{X}) = \id{X} \quad\text{and}\quad  \copier{X} ; \symm{X}{X} = \copier{X}; \]
    \item $(\cocopier{X}, \codischarger{X})$ and $(\copier{X}, \discharger{X})$ are coherent with the monoidal structure:
    \[
    \begin{array}{l@{\quad}l@{\quad}l@{\quad}l}
        \cocopier{\uno} = \id{\uno} & \cocopier{X \otimes Y} = (\id{X} \otimes \symm{Y}{X} \otimes \id{Y}) ; (\cocopier{X} \otimes \cocopier{Y}) & \codischarger{\uno} = \id{\uno} & \codischarger{X \otimes Y} = \codischarger{X} \otimes \codischarger{Y}
        \\
        \copier{\uno} = \id{\uno} & \copier{X \otimes Y} =  (\copier{X} \otimes \copier{Y}) ; (\id{X} \otimes \symm{X}{Y} \otimes \id{Y}) & \discharger{\uno} = \id{\uno} & \discharger{X \otimes Y} = \discharger{X} \otimes \discharger{Y};
    \end{array}
    \]
    \item arrows $f \colon X \to Y$ are comonoid lax morphisms: 
    \[f;\copier{Y}\leq \copier{X};(f \per f) \quad \text{and} \quad f;\discharger{Y}\leq \discharger{X};\]
    \item $(\copier{X}, \discharger{X})$ and  $(\cocopier{X}, \codischarger{X})$ form special Frobenius algebras:
    \[ \copier{X};\cocopier{X} = \id{X} \quad\text{and}\quad (\id{X} \otimes \copier{X});(\cocopier{X} \otimes \id{X}) = (\copier{X} \otimes \id{X});(\id{X} \otimes \cocopier{X}); \]
    \item $(\copier{X}, \discharger{X})$ is left adjoint to $(\cocopier{X}, \codischarger{X})$: 
    \[ \id{X}\leq \copier{X};\cocopier{X}, \qquad \cocopier{X};\copier{X} \leq \id{X\per X}, \qquad \id{X}\leq \discharger{X};\codischarger{X} \quad \text{and} \quad \codischarger{X} ; \discharger{X} \leq \id{\uno}. \]
\end{enumerate}
A \emph{morphism of Cartesian bicategories} is a poset enriched symmetric monoidal functor preserving monoids and comonoids. We denote by $\Cat{CB}$ the category of Cartesian bicategories and their morphisms.
\end{defi}
The archetypal example of a Cartesian bicategory is $(\Rel,\per,\uno)$ with $\copier{X}$, $\discharger{X}$, $\cocopier{X}$,  $\codischarger{X}$ defined as in \eqref{eq:comonoidsREL}. Simple computations confirm that all the laws of Definition \ref{def:cartesian bicategory} -- illustrated by means of string diagrams in Figures~\ref{fig:(co)monoids string} and~\ref{fig:ax-cb-string} --  are satisfied. %

The operations of $\CR$ in \eqref{eq:calculusofrelation2} can be defined in  any Cartesian bicategory, as
\begin{equation}\label{eq:cb:covolution}%
    \op{f} \defeq \stringOp{f}{X}{Y} \quad\qquad \top \defeq \stringBottom{X}{Y}  \quad\qquad  f \sqcap g \defeq \stringConvolution{f}{g}{X}{Y}
\end{equation}
for all objects $X,Y$ and arrows $f,g\colon X \to Y$. The reader can easily check that, in $\Rel$, these correspond to the opposite relation $\op{f}$, the top relation $X\times Y$ and the intersection $f\sqcap g$, respectively.

As in $\Rel$, in any Cartesian bicategory $\Cat{C}$ the operations $\sqcap$ and $\top$ make each hom-set $\Cat{C}[X,Y]$ into a meet-semilattice (see the top row of Table~\ref{table:re:daggerproperties}), and the order $\leq$ coincides with the one induced by $\sqcap$ and $\top$, i.e. $f \leq g$ iff $f \sqcap g = f$~\cite[Lemma 4.13]{Bonchi2017c}.
Yet $\Cat{C}$ is \emph{not} enriched over meet-semilattices, but only \emph{laxly} so (see the middle row).
Finally, the assignment $f \mapsto \op{f}$ defines an identity-on-objects monoidal isomorphism $\op{(\cdot)} \colon \Cat{C} \to \opcat{\Cat{C}}$ (see the bottom row).
\begin{prop}\label{prop:dagger}
    In any Cartesian bicategory, the laws in Table~\ref{table:re:daggerproperties} hold. 
    \end{prop}
\begin{proof}%
    See Theorem 2.4 in~\cite{Carboni1987}.
\end{proof}
\begin{rem}\label{rem:dagger diagram}
	Since $\op{(\cdot)}$ is an identity-on-objects monoidal isomorphism, its action on a diagram $f$ can be immediately visualised as the mirror reflection of $f$. For instance, $\op{(\,\CBcopier{X}\,)} = \CBcocopier{X}$. For this reason, from now on, we will depict a morphism $f \colon X \to Y$ as $\stringBox{f}{X}{Y}$, and use $\stringBoxOp{f}{X}{Y}$ as syntactic sugar for $\op{f}$.
\end{rem}
\begin{table}[t]
    \centering
    \[

    \]
\caption{Derived laws in Cartesian bicategories.}\label{table:re:daggerproperties}
\end{table}
\begin{figure}
	\[
	\def\arraystretch{3}
	\begin{array}{r@{}c@{}l @{\qquad} r@{}c@{}l}
		
    \InputIfFileExists{cblaws/assoc/lhs.tikz}{}{\input{./tikz/cblaws/assoc/lhs.tikz}}
 &=& 
    \InputIfFileExists{cblaws/assoc/rhs.tikz}{}{\input{./tikz/cblaws/assoc/rhs.tikz}}

		&
		\stringConvolution{f}{g}{X}{Y} &=& \stringConvolution{g}{f}{X}{Y}
		\\
		
    \InputIfFileExists{cblaws/unit/lhs.tikz}{}{\input{./tikz/cblaws/unit/lhs.tikz}}
 &=& \stringBox{f}{X}{Y}
		&
		\stringConvolution{f}{f}{X}{Y} &=& \stringBox{f}{X}{Y}
		\\
		
    \InputIfFileExists{cblaws/Rdist/lhs.tikz}{}{\input{./tikz/cblaws/Rdist/lhs.tikz}}
 &\leq& 
    \InputIfFileExists{cblaws/Rdist/rhs.tikz}{}{\input{./tikz/cblaws/Rdist/rhs.tikz}}

		&
		
    \InputIfFileExists{cblaws/Ldist/lhs.tikz}{}{\input{./tikz/cblaws/Ldist/lhs.tikz}}
 &\leq& 
    \InputIfFileExists{cblaws/Ldist/rhs.tikz}{}{\input{./tikz/cblaws/Ldist/rhs.tikz}}

		\\
		\multicolumn{6}{c}{
			
    \InputIfFileExists{cblaws/top/lhs.tikz}{}{\input{./tikz/cblaws/top/lhs.tikz}}
 \leq \stringBottom{X}{Y} \geq 
    \InputIfFileExists{cblaws/top/rhs.tikz}{}{\input{./tikz/cblaws/top/rhs.tikz}}

		}
		\\
		
    \InputIfFileExists{cblaws/opcomp/lhs.tikz}{}{\input{./tikz/cblaws/opcomp/lhs.tikz}}
 &=& 
    \InputIfFileExists{cblaws/opcomp/rhs.tikz}{}{\input{./tikz/cblaws/opcomp/rhs.tikz}}

		&
		
    \InputIfFileExists{cblaws/opper/lhs.tikz}{}{\input{./tikz/cblaws/opper/lhs.tikz}}
 &=& 
    \InputIfFileExists{cblaws/opper/rhs.tikz}{}{\input{./tikz/cblaws/opper/rhs.tikz}}

		\\
		
    \InputIfFileExists{cblaws/opid/lhs.tikz}{}{\input{./tikz/cblaws/opid/lhs.tikz}}
 &=& \stringId{X}
		&
		
    \InputIfFileExists{cblaws/opop/lhs.tikz}{}{\input{./tikz/cblaws/opop/lhs.tikz}}
 &=& \stringBox{f}{X}{Y}
	\end{array}
	\]
	\caption{String diagrams for the derived laws in Cartesian bicategories in Table~\ref{table:re:daggerproperties}.}\label{fig:string derived laws cb}
\end{figure}
An arrow $f \colon X \to Y$ is said to be \emph{single-valued} iff it satisfies~\eqref{eq:cb:sv} below on the left, \emph{total} iff it satisfies~\eqref{eq:cb:tot}, \emph{injective} iff it satisfies~\eqref{eq:cb:inj}, and \emph{surjective} iff it satisfies~\eqref{eq:cb:sur}. 
        {
    \noindent\begin{minipage}{0.50\textwidth}
        \begin{equation}\label{eq:cb:sv}\tag{SV}
            \hspace*{-1.7cm}
    %
}

        \end{equation}
    \end{minipage}
    \begin{minipage}{0.46\textwidth}
        \begin{equation}\label{eq:cb:adj:sur}
            \stringId{Y} \leq \stringSpan{f}{Y}
        \end{equation}
    \end{minipage}
    }
    
\noindent The reader can easily check that in $\Rel$ these four notions coincide with the expected ones. For instance, a relation $f\colon X \to Y$ satisfies \eqref{eq:cb:tot} iff the following inclusion holds.
\[ \{(x,\bullet) \mid x\in X\} \subseteq \{ (x,y) \mid (x,y)\in f \} ; \{(y,\bullet) \mid y \in Y\} = \{(x,\bullet) \mid \exists y\in Y. \, (x,y)\in f\}  \]

In any Cartesian bicategory, a \emph{map} is an arrow that is both single-valued and total. Similarly, a \emph{comap} is an arrow that is both injective and surjective. In $\Rel$, maps coincide with functions and comaps are opposites of functions. The following result generalises the well known fact that a relation is a function iff it has a right adjoint.

\begin{lem}\label{lemma:cb:adjoints}
    In a Cartesian bicategory, an arrow $f \colon X \to Y$ 
    \begin{itemize}
    \item is single-valued iff it satisfies~\eqref{eq:cb:adj:sv}, \\
    \item it is total iff it satisfies~\eqref{eq:cb:adj:tot},  \\
    \item it is injective iff it satisfies~\eqref{eq:cb:adj:inj}, and \\
    \item it is surjective iff it satisfies~\eqref{eq:cb:adj:sur}. 
    \end{itemize}
    In particular, an arrow is a map iff it has a right adjoint, namely $f \dashv \op{f}$; and it is a comap iff it has a left adjoint, namely $\op{f} \dashv f$.
\end{lem}
\begin{proof}%
    See Lemma 4.4 in~\cite{Bonchi2017c}.
\end{proof}
Although it will not be used in what follows, it is worth recalling that every Cartesian bicategory is self-dual compact closed, and hence \emph{traced}. 
For any arrow \( f \colon S \per X \to S \per Y \), its trace is defined as
\[
\trace_{S}(f) \defeq \begin{tikzpicture}
	\begin{pgfonlayer}{nodelayer}
		\node [style=label] (105) at (-3.5, -0.3) {$X$};
		\node [style=none] (117) at (-3, -0.3) {};
		\node [style=label] (120) at (3.5, -0.3) {$Y$};
		\node [style=none] (125) at (3, -0.3) {};
		\node [style=stringbox, scale=1.2] (129) at (0, 0.025) {$f$};
		\node [style=black] (130) at (2.25, 1.15) {};
		\node [style=none] (131) at (1, 0.275) {};
		\node [style=none] (132) at (1, 2.025) {};
		\node [style=black] (133) at (-2.25, 1.15) {};
		\node [style=none] (134) at (-1, 0.275) {};
		\node [style=none] (135) at (-1, 2.025) {};
		\node [style=black] (136) at (3, 1.15) {};
		\node [style=black] (137) at (-3, 1.15) {};
	\end{pgfonlayer}
	\begin{pgfonlayer}{edgelayer}
		\draw [bend left] (132.center) to (130);
		\draw [bend left] (130) to (131.center);
		\draw [bend right] (135.center) to (133);
		\draw [bend right] (133) to (134.center);
		\draw (137) to (133);
		\draw (136) to (130);
		\draw (117.center) to (125.center);
		\draw (134.center) to (131.center);
		\draw (132.center) to (135.center);
	\end{pgfonlayer}
\end{tikzpicture}
\]
In the case of \( (\Rel, \per, \uno) \), this specialises to
\[
\trace_{S}(f) = \{(x,y) \mid \exists s \in S.\, ((s,x),(s,y)) \in f\}.
\]
In Section \ref{sec:kleene}, we shall describe another trace on the monoidal category \( (\Rel, \piu, \zero) \), 
which provides a categorical account of the operation \( \kstar{(\cdot)} \) from~\eqref{eq:calculusofrelation3}.

\subsection{Coreflexives in Cartesian Bicategories}
Before addressing $(\Rel,\piu,\zero)$ and Kleene bicategories, we recall the theory of coreflexives in Cartesian bicategories, which will subsequently play a role in the treatment of \emph{guards} (or \emph{tests}) in imperative programs.
The results presented in this section are largely part of the folklore, although we are not aware of explicit references.

At the beginning of this section, the notions of single-valued, total, injective, and surjective relations have been extended to the arrows of an arbitrary Cartesian bicategory.
In the same spirit, we now generalise the usual relational properties of reflexivity, transitivity, and symmetry: an arrow $f \colon X \to X$ is said to be \emph{reflexive} iff it satisfies~\eqref{eq:cb:ref} below, \emph{transitive} iff it satisfies~\eqref{eq:cb:trans} and \emph{symmetric} iff it satisfies~\eqref{eq:cb:sym}.
\[\text{\begin{minipage}{0.33\textwidth}
        \begin{equation}\label{eq:cb:ref}\tag{REF}
        \id{X}\leq f
        \end{equation}
    \end{minipage}
    \begin{minipage}{0.33\textwidth}
        \begin{equation}\label{eq:cb:trans}\tag{TRN}
        f;f\leq f
        \end{equation}
    \end{minipage}
    \begin{minipage}{0.33\textwidth}
        \begin{equation}\label{eq:cb:sym}\tag{SYM}
        \op{f}\leq f
        \end{equation}
    \end{minipage}}\]
Note that, since $\op{(\cdot)}$ is involutive,~\eqref{eq:cb:sym} can be strengthened to an equality: $f = \op{(\op{f})} \leq \op{f}$.

 As suggested by the name, \emph{co}reflexives are dual to reflexives, i.e., \emph{coreflexives} are arrows $f\colon X \to X$ such that 
\begin{equation}\label{def:coreflexive}\tag{COR}
    f \leq \id{X}.
\end{equation}

\begin{lem}\label{lemma:coreflexive properties} 
    In any Cartesian bicategory, the following hold for all coreflexives $f,g \colon X \to X$:
    \begin{multicols}{2}
    \begin{enumerate}
        \item $
    \InputIfFileExists{coreflexive/copyR.tikz}{}{\input{./tikz/coreflexive/copyR.tikz}}
 = 
    \InputIfFileExists{coreflexive/Rcopy.tikz}{}{\input{./tikz/coreflexive/Rcopy.tikz}}
$
        \item $

}
 \leq 
    }
$%
    \end{enumerate}
    \end{multicols}
\end{lem}
\begin{proof} %
    \begin{enumerate}
        \item We prove the two inclusions separately: 
        \[ 
    \InputIfFileExists{coreflexive/Rcopy.tikz}{}{\input{./tikz/coreflexive/Rcopy.tikz}}
 \stackrel{\eqref{ax:cb:comonoid:nat:copy}}{\leq} 
    %
}
 \!\!\stackrel{\eqref{def:coreflexive}}{\leq}\!\! 
    }
$. Thus, by \Cref{lemma:cb:adjoints},  $f$ is single-valued.
        \item $
    }
 \stackrel{(4)}{\leq} 
    }
 \!\!\stackrel{\eqref{def:coreflexive}}{\leq} \!\! 
    }
$. Thus, by \Cref{lemma:cb:adjoints}, $f$ is injective. \qedhere
\end{enumerate}
\end{proof}
\begin{prop}\label{prop:coriff}
An arrow is coreflexive iff it is transitive, symmetric, and single-valued.
\end{prop}
\begin{proof}
One direction follows immediately from points (3), (4) and (5) of Lemma \ref{lemma:coreflexive properties}. For the other direction we first prove
\begin{equation}\label{eq:op in the middle}
    f \leq f ; \op{f} ; f
\end{equation}
as follows
\[

    }

\!\!\stackrel{\substack{\eqref{ax:cb:monoid:unit}\\ \eqref{ax:cb:comonoid:unit}}}{=}\!\!

\]
and then we conclude by showing that 
\[

    }

\!\!\stackrel{\eqref{eq:op in the middle}}{\leq}\!\!
%
\!\!\stackrel{\eqref{eq:cb:trans}}{\leq}\!\!

    }

\!\!\stackrel{\eqref{eq:cb:sym}}{\leq}\!\!

    }

\!\!\stackrel{\eqref{eq:cb:adj:sv}}{\leq}\!\!

    }
.
\]
\end{proof}
Coreflexives admit a further, and more useful, characterisation in terms of \emph{predicates}, that is, arrows of type $X \to \uno$.
(Recall that in $\Rel$, such arrows are in bijective correspondence with subsets of $X$). As is customary in string diagrammatic notation, an arrow $g \colon X \to \uno$ is drawn as $
    %
}
$, where the target $\uno$ is represented by the fact that the string diagram has no dangling wire on the right. 

Consider the functions $i \colon \Cat{C}[X, X] \to \Cat{C}[X, \uno]$ and $c \colon \Cat{C}[X, \uno] \to \Cat{C}[X, X]$ defined as %
\begin{equation}\label{eq:defc}%
            i(
    \InputIfFileExists{coreflexive/7/f.tikz}{}{\input{./tikz/coreflexive/7/f.tikz}}
) \defeq 
    %
}
\,) \defeq 
    \InputIfFileExists{coreflexive/7/Cg.tikz}{}{\input{./tikz/coreflexive/7/Cg.tikz}}
.  
\end{equation}
In $\Rel$, $i(
    \InputIfFileExists{coreflexive/7/f.tikz}{}{\input{./tikz/coreflexive/7/f.tikz}}
)$ is the relation $\{(y, \bullet) \mid \exists x\in X. \, (x,y)\in f\}$, representing the image of $f$;  while $c(\!
    %
}
\;)$ is the relation $\{(x,x)\mid (x,\bullet) \in g\}$. Observe that $c(\! 
    %
}
 \,) = 
    \InputIfFileExists{coreflexive/7/Cg.tikz}{}{\input{./tikz/coreflexive/7/Cg.tikz}}
 \stackrel{\eqref{ax:cb:comonoid:nat:discard}}{\leq} 
    \InputIfFileExists{coreflexive/7/step1.tikz}{}{\input{./tikz/coreflexive/7/step1.tikz}}
 \stackrel{\eqref{ax:cb:comonoid:unit}}{=} 
    }
.  \]
When $i$ is restricted to coreflexives, $i$ and $c$ are inverse to each other.

\begin{prop}\label{proop:bijcor}
Coreflexives $X \to X$ are in bijective correspondence with morphisms $X \to \uno$. 
Moreover, for all arrows $f \colon X \to \uno$ and coreflexives $g \colon X \to X$, $g ; f = f \sqcap i(g)$. %
\end{prop}
\begin{proof}
Observe that  $i$ and $c$ are inverse to each other: for all arrows $g\colon X \to \uno$
        \[ i(c(\! 

}
 \]
        where the third step comes from the definition of $i$ that reflects the whole diagram via $\op{(\cdot)}$ (see Remark~\ref{rem:dagger diagram}) and post-composes with $\discharger{X}$.
		For all coreflexives $f\colon X \to X$,
        \[ c(i(\! 
    \InputIfFileExists{coreflexive/7/f.tikz}{}{\input{./tikz/coreflexive/7/f.tikz}}
 \!)) = c(\! 
    %
\,.\] %
Observe that in the last two derivations we used \eqref{eq:cb:sym} as an equality. This is justified: by the properties of $\op{(\cdot)}$ listed in Table~\ref{table:re:daggerproperties}, it follows that if $f$ is symmetric, then $f \leq \op{f}$.
\end{proof}

\begin{rem}\label{rem:notation-coreflexives}
    From now on we will use $\stringCorefl{f}{X}$ to depict coreflexive arrows. This graphical representation is, in some sense, orientation agnostic, and it reflects the fact that coreflexives are symmetric, as stated by \Cref{lemma:coreflexive properties}.(4).
\end{rem}

\section{Kleene Bicategories}\label{sec:kleene}

Having recalled Cartesian bicategories, we now introduce our next ingredient: 
Kleene bicategories. The name reflects the fact that their axioms capture the 
complete axiomatisation of Kleene algebras given in \cite{Kozen94acompleteness}. We commence with the standard notion of a category with finite biproducts.

\begin{defi}\label{def:fb}
A \emph{finite biproduct} (shortly, fb) \emph{category} %
is a %
symmetric monoidal category $(\Cat{C}, \piu, \zero)$ where, for every object $X$,
there are morphisms $\codiag{X} \colon X \piu X \to X, \cobang{X} \colon \zero \to X, \diag{X} \colon X \to X \piu X$ and $\bang{X} \colon X \to \zero$, such that:
\begin{enumerate}
  \item $(\codiag{X}, \cobang{X})$ is a commutative monoid:
    \[ (\codiag{X} \piu \id{X}) ; \codiag{X} = (\id{X} \piu \codiag{X}) ; \codiag{X}, \qquad (\cobang{X} \piu \id{X}) ; \codiag{X} = \id{X} \quad\text{and}\quad \symm{X}{X} ; \codiag{X} = \codiag{X}; \]
    \item $(\diag{X}, \bang{X})$ is a cocommutative comonoid:
    \[ \diag{X} ; (\diag{X} \piu \id{X}) = \diag{X} ; (\id{X} \piu \diag{X}), \qquad \diag{X} ; (\bang{X} \piu \id{X}) = \id{X} \quad\text{and}\quad  \diag{X} ; \symm{X}{X} = \diag{X}; \]
    \item $(\codiag{X}, \cobang{X})$ and $(\diag{X}, \bang{X})$ are coherent with the monoidal structure:
    \[
    \begin{array}{l@{\quad}l@{\quad}l@{\quad}l}
        \codiag{\zero} = \id{\zero} & \codiag{X \piu Y} = (\id{X} \piu \symm{Y}{X} \piu \id{Y}) ; (\codiag{X} \piu \codiag{Y}) & \cobang{\zero} = \id{\zero} & \cobang{X \piu Y} = \cobang{X} \piu \cobang{Y}
        \\
        \diag{\zero} = \id{\zero} & \diag{X \piu Y} =  (\diag{X} \piu \diag{Y}) ; (\id{X} \piu \symm{X}{Y} \piu \id{Y}) & \bang{\zero} = \id{\zero} & \bang{X \piu Y} = \bang{X} \piu \bang{Y};
    \end{array}
    \]
  \item arrows $f \colon X \to Y$ are  both monoid and comonoid morphisms: 
  \[f;\diag{Y}= \diag{X};(f \piu f), \qquad f;\bang{Y}= \bang{X}, \qquad \codiag{Y}; f= (f \piu f); \codiag{X} \quad\text{and}\quad \cobang{X} ; f= \cobang{Y}. \]
\end{enumerate}
A \emph{morphism of finite biproduct categories} is a symmetric monoidal functor preserving monoids and comonoids. We denote by $\Cat{FBCat}$ the category of finite biproduct categories and their morphisms.
\end{defi}
Note that the first three conditions coincide with those in the definition of a Cartesian bicategory: they ensure that every object carries coherent monoid and comonoid structures. The fourth condition requires these (co)monoids to be natural. By Fox's theorem~\cite{fox1976coalgebras}, it then follows that $\zero$ is both an initial and a final object, while $\piu$ serves simultaneously as the categorical product and coproduct---that is, as a \emph{biproduct}. For further details, see~\cite{bonchi2023deconstructing}.

For any objects $X,Y$ and arrows $f,g \colon X \to Y$, the \emph{convolution monoid} on the hom-set $\Cat{C}[X,Y]$ is given below. %
\begin{equation}\label{eq:covolution}
f \sqcup g \defeq \stringConvolution{f}{g}{X}{Y}
\qquad
\bot \defeq \stringBottom{X}{Y}
\end{equation}
It is straightforward to verify that, under these operations, every fb category
becomes enriched over $\Cat{CMon}$, the category of commutative monoids: all
equations in \eqref{eq:cmon laws} and \eqref{eq:cmon enrichment} hold, with the
sole exception of idempotency, i.e.,\ $f \sqcup f = f$.

When $\sqcup$ \emph{is} idempotent, an fb category is not only
$\Cat{CMon}$-enriched but also \emph{poset}-enriched. In this case, the pair
$(\diag{X},\bang{X})$ forms a right adjoint to $(\codiag{X},\cobang{X})$:
\begin{equation}\label{eq:adjointnesfib}
\diag{X}\,;\codiag{X} \leq \id{X},
\qquad
\id{X \per X} \leq \codiag{X}\,;\diag{X},
\qquad
\bang{X}\,;\cobang{X} \leq \id{X},
\qquad
\id{\uno} \leq \cobang{X}\,;\bang{X}.
\end{equation}
Conversely, if an fb category is poset enriched and the inequalities in
\eqref{eq:adjointnesfib} hold, then convolution is necessarily idempotent. This
is established by the following result.

\begin{lem}\label{lem:idempfib}
Let $(\Cat{C}, \piu, \zero)$ be an fb category. The following are equivalent:
\begin{enumerate}
\item $(\Cat{C}, \piu, \zero)$ is a poset enriched symmetric monoidal category and the laws in \eqref{eq:adjointnesfib} hold;
\item $\sqcup$ is idempotent.
\end{enumerate}
\end{lem} 
\begin{proof}
  See Appendix~\ref{app:Kleene}.
\end{proof}
The above lemma justifies the following definition.

\begin{figure}[t]
    \centering
    \mylabel{ax:comonoid:nat:copy}{$\diag{}$-nat}
    \mylabel{ax:comonoid:nat:discard}{$\bang{}$-nat}
    \mylabel{ax:monoid:nat:copy}{$\codiag{}$-nat}
    \mylabel{ax:monoid:nat:discard}{$\cobang{}$-nat}
    \mylabel{ax:adjbiprod:1}{$\codiag{}\diag{}$}
    \mylabel{ax:adjbiprod:2}{$\cobang{}\bang{}$}
    \mylabel{ax:adjbiprod:3}{$\diag{}\codiag{}$}
    \mylabel{ax:adjbiprod:4}{$\bang{}\cobang{}$}
    \[

      }
    \end{array}
    \]
    \caption{String diagrams for the axioms of fb categories with idempotent convolution: naturality in \Cref{def:fb}.(4) and adjointness in \eqref{eq:adjointnesfib}.}
    \label{fig:ax-fb-string}
  \end{figure}

\begin{defi}\label{def:biproduct category}
A \emph{finite biproduct category with idempotent convolution} is a poset enriched symmetric monoidal category $(\Cat{C}, \piu, \zero)$ such that:
\begin{enumerate}
  \item $(\Cat{C}, \piu, \zero)$ is a finite biproduct category;
  \item $(\diag{X}, \bang{X})$ is right adjoint to $(\codiag{X}, \cobang{X})$, i.e., the laws in \eqref{eq:adjointnesfib} hold. 
\end{enumerate}
A \emph{morphism of finite biproduct categories with idempotent convolution} is a poset enriched symmetric monoidal functor preserving monoids and comonoids. We denote by $\Cat{FIBCat}$ the category of finite biproduct categories with idempotent convolution and their morphisms.
\end{defi}
The reader can easily check that $(\Rel,\piu,\zero)$ with $\diag{X}$, $\bang{X}$, $\codiag{X}$,  $\cobang{X}$ defined as in \eqref{eq:comonoidsREL} is a finite biproduct category with idempotent convolution and that $\sqcup$ and $\bot$ in \eqref{eq:covolution} coincide with union and empty relation.

\begin{rem}
The axioms of fb categories with idempotent convolution are depicted by means of 
string diagrams in Figures~\ref{fig:(co)monoids string} and~\ref{fig:ax-fb-string}. 
It is worth noting that, although the string diagrams in 
Figure~\ref{fig:(co)monoids string} serve both for Cartesian bicategories and for 
fb categories, their interpretations in $(\Rel,\per,\uno)$ and in 
$(\Rel,\piu,\zero)$ differ substantially. For example, both $\copier{X}$ and 
$\diag{X}$ are represented as
\[
  \CBcopier{X}
\]
yet the former is understood as a \emph{copier}, receiving some input on the left 
and producing two identical outputs on the right, whereas the latter should be 
understood as a \emph{split}, receiving an input on the left and emitting it on 
either the upper or the lower branch on the right. More generally, when a string 
diagram is interpreted in $(\Rel,\per,\uno)$, information flows through it 
\emph{as a wave}, whereas in $(\Rel,\piu,\zero)$ it behaves \emph{as a particle}. 

It is also worth observing that the diagrams in the last two rows of 
\Cref{fig:ax-cb-string} and~\Cref{fig:ax-fb-string} are identical except for the 
direction of the inequalities. This reflects the fact that 
$(\copier{X}, \discharger{X})$ is left adjoint to $(\cocopier{X}, \codischarger{X})$, 
while $(\diag{X}, \bang{X})$ is right adjoint to $(\codiag{X}, \cobang{X})$.  
This phenomenon is again explained by the wave--versus--particle analogy.  
Consider the following two diagrams and their interpretations in 
$(\Rel,\per,\uno)$ and $(\Rel,\piu,\zero)$:
\[
\small{

}
\]
When interpreted in $(\Rel,\per,\uno)$, in the leftmost diagram potentially 
different pieces of information travel simultaneously along the two wires, whereas 
in the rightmost diagram all four ports are forced to carry the same information.  
In contrast, when interpreted in $(\Rel,\piu,\zero)$, in the leftmost diagram the 
information travels either along the upper wire or along the lower one; in the 
rightmost diagram the information may enter either the upper left or the lower 
left port, and in both cases may exit either the upper right or lower right port. In \Cref{sec:kctapes} we introduce a diagrammatic notation that allows us to 
represent, simultaneously and within the same formalism, both 
$(\Rel,\per,\uno)$ and $(\Rel,\piu,\zero)$.
\end{rem}

We can now illustrate two key properties of fb categories with idempotent convolution. First observe that
by \Cref{lem:idempfib}, any fb category with idempotent convolution is enriched over $\Cat{Jsl}$, the category of join-semilattices (i.e., idempotent commutative monoids).
\begin{prop}
In a finite biproduct category with idempotent convolution, the laws \eqref{eq:cmon laws} and \eqref{eq:cmon enrichment} in \Cref{tab:TKA} hold.
\end{prop}
\begin{proof}
  The first three laws of~\eqref{eq:cmon laws} hold in any fb category; the last one, idempotency, holds by definition of fb category with idempotent convolution. By Lemma~\ref{lem:idempfib}, we have the enrichment in join-semilattices, hence the laws in~\eqref{eq:cmon enrichment} hold.
\end{proof}
The next  interesting property is the so called \emph{matrix normal form}.
\begin{prop}\label{prop:matrixform}
In any fb category $\Cat{C}$, for all arrows $f\colon S \piu X \to T \piu Y$, it holds that %
\[f = \stringMatrix{f}{S}{T}{X}{Y} \]
where $f_{ST}\colon S \to T$,  $f_{SY}\colon S \to Y$, $f_{XT}\colon X \to T$ and  $f_{XY}\colon X \to Y$ are defined as follows.
\begin{equation}\label{eq:components}
\begin{array}{cc}
f_{ST} \defeq (\id{S} \piu \cobang{X}); f ; (\id{T} \piu \cobang{Y}) &  f_{SY} \defeq (\id{S} \piu \cobang{X}); f ; (\cobang{T} \piu \id{Y} )\\
f_{XT} \defeq (\cobang{S} \piu \id{X} ); f ; (\id{T} \piu \cobang{Y}) &  f_{XY} \defeq(\cobang{S} \piu \id{X} ); f ; (\cobang{T} \piu \id{Y} )
\end{array}
\end{equation}
Moreover, if $\Cat{C}$ has idempotent convolution, for all $f,g\colon S \piu X \to T \piu Y$, 
\[f \leq g \text{ iff }
\begin{array}{cccc}
f_{ST} \leq g_{ST}, &  f_{SY} \leq g_{SY}, &
f_{XT} \leq g_{XT} \text{ and }&  f_{XY} \leq g_{XY}.
\end{array}
\]
\end{prop}
\begin{proof}
  See Appendix~\ref{app:Kleene}.
\end{proof}
In $(\Rel, \piu, \zero)$, the above result amounts to the well known fact that any relation $f\colon S \piu X \to T \piu Y$ can be decomposed as $f=f_{S,S} \cup f_{S,Y} \cup f_{X,S} \cup f_{X,Y}$ where 
\begin{equation*}
\begin{array}{rclrcl}
f_{S,T} &\defeq&\{(s,t)\mid (\,(s,0), (t,0) \,)\in f \}  &
f_{S,Y} &\defeq&\{(s,y)\mid (\,(s,0), (y,1) \,)\in f \}  \\ 
f_{X,T} &\defeq&\{(x,t)\mid (\,(x,1), (t,0) \,)\in f \}  &
f_{X,Y} &\defeq&\{(x,y)\mid (\,(x,1), (y,1) \,)\in f \}  
\end{array}
\end{equation*}
The above decomposition is also useful to define the monoidal trace in $(\Rel, \piu, \zero)$: given a relation $f\colon S \piu X \to S \piu Y$, its trace $\trace_{S}(f)\colon X \to Y$ is the relation defined as 
\begin{equation}\label{eq:traceREL}
\trace_{S}(f) \defeq ( f_{X,S} ; \kstar{f_{S,S}} ; f_{S,Y} ) \cup f_{X,Y}
\end{equation}
where, like in Section \ref{sec:CR}, $\kstar{(\cdot)}$ provides the reflexive and transitive closure \cite{joyal1996traced,selinger1998note}. Unfortunately, finite biproduct categories with idempotent convolution do not have in general terms enough structure to deal with $\kstar{(\cdot)}$: differently from Cartesian bicategories they are not necessarily traced. Such structure is explicitly added in the next section.

\subsection{Kleene Bicategories are Typed Kleene Algebras}
\begin{table}[t]
  \[
\begin{tabular}{c}
    \toprule
    \vbox{\begin{equation}\label{eq:cmon laws}
      (f \sqcup g) \sqcup h = f \sqcup(g \sqcup h)
      \qquad
      f \sqcup g=g \sqcup f
      \qquad
      f \sqcup \bot=f
      \qquad 
      f \sqcup f=f
    \end{equation}}
    \\[-8pt]
    \vbox{\begin{equation}\label{eq:cmon enrichment}
      (f \sqcup g);h = (f;h)  \sqcup  (g;h)
      \qquad
      h;(f \sqcup g) = (h;f)  \sqcup  (h;g )
      \qquad
      f;\bot =\bot= \bot;f
    \end{equation}}
    \\[-8pt]
    \vbox{\begin{equation}\label{eq:Kllenelaw}
      \id{X}  \sqcup f \dcomp \kstar{f} \leq \kstar{f} 
      \quad 
      \id{X}  \sqcup \kstar{f} \dcomp f \leq \kstar{f} 
      \quad 
      f \dcomp r \leq r \implies  \kstar{f} \dcomp r \leq r  
      \quad
      l \dcomp f  \leq l \implies   l \dcomp \kstar{f} \leq l
    \end{equation}}
    \\[-5pt]
    \bottomrule
  \end{tabular}
  \]
    \caption{Axioms of (typed) Kleene algebras.}\label{tab:TKA}
\end{table}
\begin{figure}
  \[
	\def\arraystretch{3}
	%
	\]
  \caption{String diagrams for the axioms of (typed) Kleene algebras in Table~\ref{tab:TKA}.}\label{fig:string kleene axioms}
\end{figure}

We can now introduce the main structure of this section: Kleene bicategories. These are fb categories with idempotent convolution equipped with a monoidal trace that, intuitively, behaves well w.r.t. the poset enrichment.  

\begin{defi}\label{def:kleenebicategory}
    A \emph{Kleene bicategory} is both a finite biproduct category with idempotent convolution and a poset enriched traced monoidal category such that 
    \begin{enumerate}
    \item for all objects $X$, the trace $\trace_{X}$ satisfies the axiom  $\trace_{X}(\codiag{X};\diag{X}) \leq \id{X}$; %
    \item the trace is posetal uniform: for all $f\colon S\piu X \to S \piu Y$ and $g \colon T\piu X \to T \piu Y$,
    \begin{itemize}
    \item[(AU1)] if $\exists r\colon S \to T$ such that  $f ; (r \piu \id{Y}) \leq (r \piu \id{X}) ; g$, then $\trace_{S}f \leq \trace_{S}g$;
    \item[(AU2)] if $\exists r\colon T \to S$ such that  $(r \piu \id{X})  ; f \leq  g ; (r \piu \id{Y})$, then $\trace_{S}f \leq \trace_{S}g$.
    \end{itemize}
    \end{enumerate}
A \emph{morphism of Kleene bicategories} is a poset enriched symmetric monoidal functor preserving monoids, comonoids and traces. Kleene bicategories and their morphisms form a category, hereafter referred as \(\KBicat\).
\end{defi}
We have already seen that $(\Rel, \piu, \zero)$ is a fb category with idempotent convolution. The reader can use \eqref{eq:traceREL} to check that the three laws above hold.
The laws in (AU1) and (AU2) are the posetal extension of the uniformity condition \cite{hasegawa2003uniformity} for traces (see also \Cref{def:utr-category} in \Cref{app:stringDiagrams}). To the best of our knowledge, they have never been studied. In contrast, the axiom in (1) already appeared in the literature (see e.g. \cite{lmcs:10963}).
They are illustrated in diagrammatic form in \Cref{fig:ineq-uniformity}.

  \begin{figure}[t]
    \centering
        \mylabel{ax:kb:traceid}{AT1}
    \mylabel{ax:posetunif:1}{AU1}
    \mylabel{ax:posetunif:2}{AU2}
    \[
    \begin{array}{c}
    
    \InputIfFileExists{kb/traceid_lhs.tikz}{}{\input{./tikz/kb/traceid_lhs.tikz}}
 \axsubeq{\ref{ax:kb:traceid}} 
    \InputIfFileExists{kb/traceid_rhs.tikz}{}{\input{./tikz/kb/traceid_rhs.tikz}}
\\[20pt]
      
    \InputIfFileExists{posetunif/AU1_lhs.tikz}{}{\input{./tikz/posetunif/AU1_lhs.tikz}}
 \leq 
    \InputIfFileExists{posetunif/AU1_rhs.tikz}{}{\input{./tikz/posetunif/AU1_rhs.tikz}}
 \stackrel{(\ref{ax:posetunif:1})}{\implies} 
    \InputIfFileExists{posetunif/AU1TR_lhs.tikz}{}{\input{./tikz/posetunif/AU1TR_lhs.tikz}}
 \leq 
    \InputIfFileExists{posetunif/AU1TR_rhs.tikz}{}{\input{./tikz/posetunif/AU1TR_rhs.tikz}}

      \\[20pt]
      
    \InputIfFileExists{posetunif/AU2_lhs.tikz}{}{\input{./tikz/posetunif/AU2_lhs.tikz}}
 \leq 
    \InputIfFileExists{posetunif/AU2_rhs.tikz}{}{\input{./tikz/posetunif/AU2_rhs.tikz}}
 \stackrel{(\ref{ax:posetunif:2})}{\implies} 
    \InputIfFileExists{posetunif/AU2TR_lhs.tikz}{}{\input{./tikz/posetunif/AU2TR_lhs.tikz}}
 \leq 
    \InputIfFileExists{posetunif/AU2TR_rhs.tikz}{}{\input{./tikz/posetunif/AU2TR_rhs.tikz}}

    \end{array}
    \]
    \caption{String diagrams for the axioms of traces in Kleene bicategories.\label{fig:ineq-uniformity}}
  \end{figure}

Like in any finite biproduct category with trace (see e.g. \cite{cuazuanescu1994feedback}), in a Kleene bicategory one can define for each endomorphism $f\colon X \to X$, a morphism $\kstar{f}\colon X \to X$ as below.
\begin{equation}\label{eq:star}
  f^\ast \defeq \Crepetition{f}{X}{X}
\end{equation}
The distinguishing property of Kleene bicategories is that $\kstar{(\cdot)}$ satisfies the laws of Kleene star as axiomatised by Kozen in \cite{Kozen94acompleteness}.

\begin{defi}
A \emph{Kleene star operator} on a $\Cat{Jsl}$-enriched category $\Cat{C}$ consists of a family of operations \(\kstar{(\cdot)} \colon \Cat{C}(X,X) \to  \Cat{C}(X,X)\)  such that for all $f\colon X\to X$, $r\colon X \to Y$ and $l\colon Y \to X$ the four laws in \eqref{eq:Kllenelaw} hold.
\end{defi}
\begin{defi}
A \emph{typed Kleene algebra} is a $\Cat{Jsl}$-enriched category equipped with a Kleene star operator. A \emph{morphism of typed Kleene algebras} is a $\Cat{Jsl}$-enriched functor preserving Kleene star. Typed Kleene algebras and their morphisms form a category referred to as $\TKAlg$.
\end{defi}
\begin{rem}
The notion of typed Kleene algebra has been introduced by Kozen in \cite{kozen98typedkleene} in order to deal with Kleene algebras \cite{Kozen94acompleteness} with multiple sorts. In other words, a Kleene algebra is a typed Kleene algebra with a single object.
\end{rem}
On the one hand, the laws of Kleene bicategories are sufficient for defining a Kleene star operation. On the other, any Kleene star operation over a fb category with idempotent convolution gives rise to a trace satisfying the laws of Kleene bicategories.
\begin{thm}\label{prop:trace-star}
  Let $\Cat{C}$ be a fb category with idempotent convolution. $\Cat{C}$ is a Kleene bicategory iff $\Cat{C}$ has a Kleene-star operator.
\end{thm}
\begin{proof}
  See Appendix~\ref{app:star theorem}.
\end{proof}
\begin{cor}\label{cor:kleeneareka}
  All Kleene bicategories are typed Kleene algebras.
\end{cor}
The converse is false: not all typed Kleene algebras are monoidal categories.
Nevertheless, from any Kleene algebra one can canonically construct a Kleene 
bicategory via the so-called \emph{matrix construction} (also known as 
\emph{biproduct completion}) \cite{coecke2017two,mac_lane_categories_1978}. 
Since this construction plays no role in our development, we include it for 
interested readers in \Cref{ssec:matrix}.

We conclude this section with two derived laws that will be useful later on.
\begin{lem}\label{lemma:derivedlawsKleene}
Let $f,g\colon X \to X$ be arrows of a Kleene bicategory. Then:
\begin{equation}\label{eq:starequality}\id{X}  \sqcup f \dcomp \kstar{f} = \kstar{f} = \id{X}  \sqcup \kstar{f} \dcomp f \end{equation}
\begin{equation}\label{eq:starsum}\kstar{f}; \kstar{g} \leq \kstar{(f\sqcup g)}\end{equation}
\end{lem}
\begin{proof}
  See e.g., \cite{Kozen94acompleteness}.
\end{proof}

\section{Rig Categories}\label{sec:rigcategories}
We have seen that $\Rel$ carries two monoidal categories $(\Rel, \per, \uno)$ and $(\Rel, \piu, \zero)$. The appropriate setting for studying their interaction is given by rig categories~\cite{laplaza_coherence_1972,johnson2021bimonoidal}.
\begin{defi}\label{def:rig}
    A \emph{rig category} is a category $\Cat{C}$ with 
    two symmetric monoidal structures $(\Cat{C}, \per, \uno)$ and 
    $(\Cat{C}, \piu, \zero)$ and natural isomorphisms 
    \[ \dl{X}{Y}{Z} \colon X \per (Y \piu Z) \to (X \per Y) \piu (X \per Z) \qquad  \annl{X} \colon \zero \per X \to \zero \]
    \[ \dr{X}{Y}{Z} \colon (X \piu Y) \per Z \to (X \per Z) \piu (Y \per Z) \qquad \annr{X} \colon X \per \zero \to \zero \]
satisfying the coherence axioms in \Cref{fig:rigax}. 

A rig category is said to be \emph{right} (resp. \emph{left}) \emph{strict} when both its monoidal structures are 
strict and $\lambda^\bullet, \rho^\bullet$ and $\delta^r$ (resp. $\delta^l$) are all identity 
natural isomorphisms. A \emph{right strict rig functor} is a strict symmetric monoidal functor for both $\per$ and $\piu$ preserving $\delta^l$.
\end{defi}
Note that only one of the two distributors is forced to be the identity within a strict rig category. If both  $\delta^r$ and $\delta^l$ would be  identities, one would obtain several undesired equalities. For instance,  for all objects $A,B,C,D$, it would hold that 
\begin{align*}
 (A\piu B) \per (C\piu D) = ((A \per C)\piu (B\per C)) \piu ((A \per D)\piu (B \per D)) 
\end{align*}
and
\begin{align*}
(A\piu B) \per (C\piu D) = ((A \per C)\piu (A\per D)) \piu ((B \per C)\piu (B \per D))
\end{align*}
forcing $\piu$ to be commutative.  We refer the curious reader to \cite[Section 4]{bonchi2023deconstructing} for a more detailed explanation. In loc. cit., it is also explained that the above notion of strictness is somewhat inconvenient when studying  freely generated categories: consider a right strict rig category freely generated by a monoidal signature $\sign$ with sorts $\sort$. The objects of this category are terms generated by the grammar in Table~\ref{table:eq objects fsr} modulo the equations in the first three rows of the same table. These equivalence classes of terms do not come with a very handy form, unlike, for instance, the objects of a strict monoidal category, which are words. For this reason several authors, like~\cite{comfort2020sheet,johnson2021bimonoidal}, prefer to take as objects polynomials in $\sort$ at the price of working with a category that is not freely generated but only equivalent to a freely generated one. This fact forces one to consider functors that are not necessarily strict,  thus most of the constructions need to properly deal with the tedious natural isomorphisms.

\begin{table}
	\begin{center}
		\begin{subtable}{0.52\textwidth}
            \scalebox{0.7}{
            \begin{tabular}{c}
				\toprule
				$X \; ::=\; \; A \; \mid \; \uno \; \mid \; \zero \; \mid \;  X \per X \; \mid \;  X \piu X \qquad (A\in \sort)$\\
				\midrule
				\begin{tabular}{ccc}
					$ (X \per Y) \per Z = X \per (Y \per Z)$ & $\uno \per X = X$ &  $X \per \uno = X $\\
					$(X \piu Y) \piu Z = X \piu (Y \piu Z)$ & $\zero \piu X = X$  & $X \piu \zero =X  $\\
					$(X \piu Y) \per Z = (X \per Z) \piu (Y \per Z)$ &  $\!\zero \per X = \zero$ & $X \per \zero= \zero $\\
				\end{tabular}\\
				$A \per (Y \piu Z) = (A \per Y) \piu (A \per Z)$\\
				\bottomrule
			\end{tabular}
            }
            \caption{}
            \label{table:eq objects fsr}
        \end{subtable}
		\begin{subtable}{0.45\textwidth}
            \scalebox{0.7}{
            \begin{tabular}{c}
				\toprule
				$n$-ary sums and products $\vphantom{\mid}$\\
				\midrule
				\makecell{
					\\[-1pt]
					$\Piu[i=1][0]{X_i} \!=\! \zero \;\; \Piu[i=1][1]{X_i}\!=\! X_1 \;\; \Piu[i=1][n+1]{X_i} \!=\! X_1 \piu (\Piu[i=1][n]{X_{i+1}} )$ \\[1em]
					$\Per[i=1][0]{X_i} \!=\! \uno \;\; \Per[i=1][1]{X_i}\!=\! X_1 \;\; \Per[i=1][n+1]{X_i} \!=\! X_1 \per (\Per[i=1][n]{X_{i+1}} )$ \\[3pt]
				} \\
				\bottomrule
			\end{tabular}
            }
            \caption{}
            \label{table:n-ary sums and prod}
        \end{subtable}
	\end{center}
	\caption{Equations for the objects of a sesquistrict rig category freely generated by a rig signature $(\sort,\sign)$.}\label{tab:equationsonobject}
\end{table}

An alternative solution is proposed in \cite{bonchi2023deconstructing}: the focus is  on  freely generated rig categories that are \emph{sesquistrict}, i.e.\ right strict but only partially left strict: namely the left distributor  $\dl{X}{Y}{Z} \colon X \per (Y \piu Z) \to (X \per Y) \piu (X \per Z)$ is the identity only when $X$ is a basic sort $A\in \sort$. In terms of the equations to impose on objects, this amounts to the one in the fourth row in Table~\ref{table:eq objects fsr} for each $A\in \sort$.  By orienting from left to right \emph{all} the equations in Table~\ref{table:eq objects fsr}, one obtains a rewriting system that is confluent and terminating and, most importantly, the unique normal forms are exactly polynomials: a term $X$  is in \emph{polynomial} form if there exist $n$, $m_i$ and $A_{i,j}\in \sort$ for $i=1 \dots n$ and $j=1 \dots m_i$ such that $X=\Piu[i=1][n]{\Per[j=1][m_i]{A_{i,j}}}$ (for $n$-ary sums and products as in Table~\ref{table:n-ary sums and prod}).
We will always refer to terms in polynomial form as \emph{polynomials} and, for a polynomial like the aforementioned  $X$, we will call \emph{monomials} of $X$ the $n$ terms $\Per[j=1][m_i]{A_{i,j}}$. For instance the monomials of $(A \per B) \piu \uno$ are $A \per B$ and $1$. Note that, differently from the polynomials we are used to dealing with, here neither $\piu$ nor $\per$ is commutative so, for instance, $(A \per B) \piu \uno$ is different from both $\uno \piu (A \per B)$ and $(B \per A) \piu \uno$. Note that non-commutative polynomials are in one-to-one correspondence with \emph{words of words} over $\sort$, while monomials are words over $\sort$. 
\begin{notation}
Hereafter, we will denote by $A,B,C\dots$ the sorts in $\sort$, by $U,V,W \dots$ the words in $\sort^\star$ and by $P,Q,R,S \dots$ the words of words in $(\sort^\star)^\star$. Given two words $U,V\in \sort^\star$, we will write $UV$ for their concatenation and $1$ for the empty word. Given two words of words $P,Q\in (\sort^\star)^\star$, we will write $P\piu Q$ for their concatenation and $\zero$ for the empty word of words. Given a word of words $P$, we will write $\pi P$ for the corresponding term in polynomial form, for instance $\pi(A \piu BCD\piu 1 )$ is the term $A \piu ((B \per (C \per D)) \piu \uno)$. Throughout this paper  we  often omit $\pi$, thus we implicitly identify words of words with polynomials.
\end{notation}
Beyond concatenation ($\piu$), one can define a product operation $\per$ on $(\sort^\star)^\star$ by taking the unique normal form of $\pi(P) \per \pi(Q)$ for any $P,Q\in (\sort^\star)^\star$. More explicitly for
$P = \Piu[i]{U_i}$ and $Q = \Piu[j]{V_j}$, 
\begin{equation}\label{def:productPolynomials} P \per Q \defeq \Piu[i]{\Piu[j]{U_iV_j}}.
\end{equation}
For instance, $(A\piu B) \per (C \piu D)$ is $(A \per C) \piu (A \per D) \piu (B \per C) \piu (B \per D)$ and not $(A \per C) \piu (B \per C) \piu (A \per D)  \piu (B \per D)$. 
Observe that, if both $P$ and $Q$ are monomials, namely, $P=U$ and $Q=V$ for some $U,V\in \sort^\star$, then $P\per Q = UV$. We can thus safely write $PQ$ in place of $P\per Q$ without the risk of any confusion.
\begin{defi}
A \emph{sesquistrict rig category} is a functor $H \colon \Cat S \to \Cat C$, where $\Cat S$ is a discrete category and $\Cat C$ is a right strict rig category, such that for all $A \in \Cat S$
	\[
	\dl{H(A)}{X}{Y} \colon H(A) \per (X \piu Y) \to (H(A) \per X) \piu (H(A) \per Y)
	\]
	is an identity morphism. %
	Given $H \!\colon\! \Cat S \!\to\! \Cat C$ and $H' \!\colon\! \Cat S' \!\to\! \Cat C'$ two sesquistrict rig categories, a \emph{sesquistrict rig functor} from $H$ to $H'$ is a pair $(\alpha \!\colon\! \Cat S \!\to\! \Cat S', \beta \!\colon\! \Cat C \!\to\! \Cat C')$, with $\alpha$ a functor and $\beta$ a strict rig functor, such that $\alpha; H' = H; \beta$.
\end{defi}
From any rig category $\Cat{C}$, one can construct its (right) strictification $\overline{\Cat{C}}$ \cite{johnson2021bimonoidal} and then embed $ob(\Cat{C})$, the discrete category of the objects of $\Cat{C}$, into $\overline{\Cat{C}}$. The embedding $ob(\Cat{C}) \to \overline{\Cat{C}}$ forms a sesquistrict category and it is equivalent (as a rig category) to the original $\Cat{C}$ \cite[Corollary 4.5]{bonchi2023deconstructing}. Throughout this paper, when dealing with a rig category $\Cat{C}$, we will often implicitly refer to the equivalent sesquistrict $ob(\Cat{C}) \to \overline{\Cat{C}}$.

Given a set of sorts $\sort$, a \emph{monoidal signature} is  a tuple $(\sort,\sign,\ar,\coar)$ where $\ar$ and $\coar$ assign to each symbol $\gen \in \sign$ an arity and a coarity in $\sort^\star$. 
A \emph{rig signature} is the same but with arity and coarity in $(\sort^\star)^\star$. %
An \emph{interpretation} $\interpretation$ of a rig signature $(\sort,\sign,\ar,\coar)$ in a sesquistrict  rig category $H \colon \Cat M \to \Cat D$ is a pair of functions $(\alpha_{\sort} \colon \sort \to Ob(\Cat M), \alpha_\sign \colon \sign \to Ar(\Cat D))$ such that, for all $\gen \in \sign$, $\alpha_{\sign}(s)$ is an arrow having as domain and codomain $(\alpha_{\sort};H)^\sharp(\ar(s))$ and  $(\alpha_{\sort};H)^\sharp(\coar(s))$. Here, $(\alpha_{\sort};H)^\sharp$ stands for inductive extension of $\alpha_{\sort};H \colon \sort \to Ob(\Cat{D})$ to $(\sort^\star)^\star$.

\begin{defi}\label{def:freesesqui}
Let $(\sort,\sign,\ar,\coar)$ (simply $\sign$ for short) be a rig signature. A sesquistrict  rig category $H \colon \Cat M \to \Cat D$ is said to be \emph{freely generated} by $\sign$ if there is an interpretation $(\alpha_S,\alpha_\sign)$ of $\sign$ in $H$ such that for every sesquistrict rig category $H' \colon \Cat M' \to \Cat D'$ and every interpretation $(\alpha_\sort' \colon \sort \to Ob(\Cat M'), \alpha_\sign' \colon \sign \to Ar(\Cat D'))$ there exists a unique sesquistrict rig functor $(\alpha \colon \Cat M \to \Cat M', \beta \colon \Cat D \to \Cat D')$ such that $\alpha_\sort ; \alpha = \alpha_\sort'$ and $\alpha_\sign ; \beta = \alpha_{\sign}'$. 
\end{defi}
This is the definition of free object on a generating one instantiated in the category of sesquistrict rig categories and the category of rig signatures.
Thus, sesquistrict rig categories generated by a given signature are isomorphic to each other and we may refer to ``the'' free sesquistrict rig category generated by a signature. 
To simplify notation, we will  denote the free sesquistrict rig category generated by 
$(\sort,\sign)$, written formally as  $\mathcal{S}\to \Cat{C}$, simply by $\Cat{C}$.

\begin{rem}\label{rem:rigsig}
Theorem 4.9 in \cite{bonchi2023deconstructing} guarantees that, whenever $\piu$ is forced to be a biproduct, every rig signature can be reduced to a monoidal one. 
Since in the rig categories relevant for this paper, namely those introduced in the next section, $\piu$ is always a biproduct, we can restrict without loss of generality to consider just monoidal signatures rather than arbitrary rig signatures.
\end{rem}

\section{Kleene-Cartesian Bicategories}\label{sec:cb}
We have seen that Cartesian bicategories provide sufficient structure  to capture the \emph{allegorical fragment} of~$\CR$,  while Kleene bicategories account for its \emph{Kleene fragment}. 
To encompass the entire calculus~$\CR$, we let the Cartesian and Kleene bicategory structures interact as rig categories.

\begin{defi}\label{def:kcbrig}
A \emph{Kleene-Cartesian rig category} (shortly kc-rig) is a poset enriched rig category $(\Cat{C}, \piu, \per, \uno, \zero)$ such that 
\begin{enumerate}
\item $(\Cat{C}, \piu, \zero)$ is a Kleene bicategory;
\item $(\Cat{C}, \per, \uno)$ is a Cartesian bicategory;
\item the trace in $(\Cat{C}, \piu, \zero)$ satisfies the following coherence condition 
\begin{equation}\label{eq:missing}
\trace_{S}(f) \per \id{Z} = \trace_{S\per Z}(f \per \id{Z})
\end{equation}
for all objects $Z$ and arrows $f\colon S\piu X \to S \piu  Y$;
\item the (co)monoids satisfy the following coherence conditions: %
\end{enumerate}
    \begin{equation}\label{eq:fbcbcoherence}
        \scalebox{0.9}{$
        \begin{array}{r@{}c@{}l@{\quad}r@{}c@{}l}
            \copier{X \piu Y} &=& (\Tcopier{X} \piu  \copier{Y}) ; (\id{XX} \piu \cobang{XY} \piu \cobang{YX} \piu \id{YY}) ; (\Idl{X}{X}{Y} \piu \Idl{Y}{X}{Y}) & \discharger{X \piu Y} &=& (\Tdischarger{X} \piu \discharger{Y}) ; \codiag{\uno} \\
            \cocopier{X \piu Y} &=& (\Tcocopier{X} \piu  \cocopier{Y}) ; (\id{XX} \piu \bang{XY} \piu \bang{YX} \piu \id{YY}) ; (\dl{X}{X}{Y} \piu \dl{Y}{X}{Y}) & \codischarger{X \piu Y} &=& \diag{\uno} ; (\Tcodischarger{X} \piu \codischarger{Y}).
        \end{array}$}
\end{equation}
A \emph{morphism of kc-rig categories} is a poset enriched rig functor that is a morphism of both Kleene and Cartesian bicategories. We write $\Cat{KCB}$ for the category of kc-rig categories and their morphisms.
\end{defi}
The law in \eqref{eq:missing} rules the interaction of the  monoidal trace for $\piu$ with the product $\per$. This law appears in several works and, as expected, $\Rel$ satisfies it: see e.g., \cite{goncharov2021metalanguage}.

The axioms in~\eqref{eq:fbcbcoherence} govern the interaction between the black and white (co)monoids.
Observe that the black (co)multiplication interacts only with the white (co)unit, and vice versa, the black (co)unit interacts only with the white (co)multiplication.
The significance of these coherence laws will become more intuitive in the next section, where we will represent the corresponding arrows as tape diagrams. %
The reader may verify, using the definitions of comonoids in~\eqref{eq:comonoidsREL}, that these laws indeed hold in~$\Rel$.
Since we have already shown that $(\Rel,\piu,\zero)$ forms a Kleene bicategory and $(\Rel,\per,\uno)$ a Cartesian bicategory, we can conclude that~$\Rel$ is a kc-rig category.
\begin{rem}
Interestingly, the two laws at the top of~\eqref{eq:fbcbcoherence} also hold in Kleisli categories for arbitrary monoidal monads (see~\cite{DBLP:conf/calco/BonchiC0L25}).
However, such categories are, in general, neither Kleene nor Cartesian bicategories: the monoidal structure given by~$\piu$ yields a category with finite coproducts (having just natural and coherent monoids), while the one given by~$\per$ yields a \emph{copy-discard category} \cite{corradini1999algebraic,cho2019disintegration} (coherent comonoids). The curious reader may consult~\cite{DBLP:conf/calco/BonchiC0L25} for further details.
\end{rem}
\begin{rem}[Strictification]
  The equivalence between a rig category and its sesquistrictified rig category is strong monoidal for both monoidal structures~\cite[Volume I, Section 5.2]{johnson2024bimonoidal}.
  As a consequence, the equivalence preserves monoids and comonoids, and the equations they satisfy;
  in particular, it preserves naturality of the monoids and comonoids of \(\piu\).
  The poset structure is a consequence of the monoids and comonoids and, therefore, it is automatically preserved by sesquistrictification.
  To see that the equivalence preserves the structure of kc-rig category, we are only left to check that it preserves the uniform trace.
  This is a consequence of the strong monoidal equivalence, of coherence and of the axioms of trace.
\end{rem}
Definition~7.1 of~\cite{bonchi2023deconstructing} introduces \emph{fb-cb rig categories}.  An fb-cb rig category~$\Cat{C}$ is the same as a kc-rig category, 
except that the additive structure~$(\Cat{C}, \oplus, \zero)$  forms merely an fb category with idempotent convolution (\Cref{def:biproduct category}), rather than a full Kleene bicategory. 
In rough terms, one may regard a kc-rig category as an fb-cb rig category equipped with uniform traces. While fb-cb rig categories correspond, as shown in~\cite{bonchi2023deconstructing}, to coherent logic, we shall see in Section~\ref{sec:hoare} that the additional trace structure enables the treatment of program logics.

Since every kc-rig category is, in particular, an fb-cb rig category, several of the results established in~\cite{bonchi2023deconstructing}  can be directly reused.

\begin{lem}\label{lemma:fb-cbsummary}
Let $\Cat{C}$ be a kc-rig category. The following hold:
\begin{enumerate}
\item $\piu$-(co)monoids are coherent with respect to $\per$:
\[\begin{array}{c@{\qquad\qquad}c}
\diag{X \per Y} =  \diag{X} \per \id{Y} = (\id{X} \per \diag{Y} ); \dl{X}{Y}{Y}  & \bang{X \per Y} = \id{X} \per \bang{Y}\quad \\
\codiag{X \per Y} = \codiag{X} \per \id{Y} = \Idl{X}{Y}{Y} ; (\id{X} \per \codiag{Y}) & \cobang{X\per Y} = \id{X} \per \cobang{Y}
\end{array}
\]
\item $(\Cat{C}, \per, \uno)$ is \emph{monoidally} enriched over join-semilattices:
\[g \per (f_1\sqcup f_2) = (g \per f_1 \sqcup  g \per f_2) \qquad (f_1\sqcup f_2)\per g =(f_1\per g \sqcup f_2\per g)  \qquad \bot \per g = \bot = g \per \bot\] 
\item for all objects $X,Y$, $\Cat{C}[X,Y]$ has the structure of a distributive lattice: %
   \[f \sqcap (g \sqcup  h) = (\, f \sqcap g \, ) \sqcup  (\, f \sqcap h  \,)
        \qquad
        f \sqcup  \top = \top
        \qquad
        f \sqcup  (g \sqcap h) = (\, f \sqcup  g \, ) \sqcap (\, f \sqcup  h  \,)
        \qquad
        f \sqcap \bot = \bot\]
\item $\op{(\cdot)}$ distributes over $\piu$, $\sqcup$ and $\bot$:
\[ \op{(f \piu g)} = \op{f} \piu \op{g} \qquad \op{(f \sqcup g)} = \op{f} \sqcup \op{g} \qquad \op{\bot}=\bot\]
\end{enumerate}
\end{lem}
\begin{proof}
The first two points are Proposition D.1 and Proposition 6.1   in \cite{bonchi2023deconstructing}. The remaining points follow directly from Theorem 7.5 in \cite{bonchi2023deconstructing}.
\end{proof}
The last item above ensures that  $\op{(\cdot)}$ distributes over $\piu$ and $\sqcup$; the following two results state that $\kstar{(\cdot)}$ only distributes \emph{laxly}  over $\per$ and $\sqcap$.
\begin{lem}\label{prop:star-per}
Let $\Cat{C}$ be a kc-rig category. For all $f \colon X \to X$ and $g \colon Y \to Y$,  \[\kstar{(f \per g)} \leq \kstar{f} \per \kstar{g}\text{.}\]
\end{lem}
\begin{proof}
	First observe that the following inequality holds:
	\begin{align*}
		(f \per g) ; (\kstar{f} \per \kstar{g})
		&= (f ; \kstar{f}) \per (g ; \kstar{g}) \tag{Functoriality of $\per$} \\ 
		&\leq (\kstar{f} \per \kstar{g}) \tag{\ref{eq:Kllenelaw}} 
	\end{align*}
Thus, by the first implication in~\eqref{eq:Kllenelaw}, it follows that
		\begin{equation}\label{eq:A}\tag{$\spadesuit$}
		\kstar{(f \per g)} ; (\kstar{f} \per \kstar{g}) \leq \kstar{f} \per \kstar{g}.
	\end{equation}
	We conclude with the following derivation.
	\begin{align*}
		\kstar{(f \per g)} 
		&= \kstar{(f \per g)} ; (\id{X} \per \id{Y}) \tag{Functoriality of $\per$} \\ 
		&\leq \kstar{(f \per g)} ; (\kstar{f} \per \kstar{g}) \tag{\ref{eq:Kllenelaw}} \\
		&\leq \kstar{f} \per \kstar{g} \tag{\ref{eq:A}}
	\end{align*}
\end{proof}

\begin{lem}\label{lemma:starcap}
Let $\Cat{C}$ be a kc-rig category. For all $f,g \colon X \to X$,
\[ \kstar{(f \sqcap g)} \leq \kstar{f} \sqcap \kstar{g} \qquad \kstar{\top}=\top\]
\end{lem}
\begin{proof}
	\begin{align*}
		\kstar{f} \sqcap \kstar{g} &= \;\; \copier{X} ; (\kstar{f} \per \kstar{g}) ; \cocopier{X} \tag{\ref{eq:cb:covolution}} \\
		&\geq \;\; \copier{X} ; \kstar{(f \per g)} ; \cocopier{X} \tag{\Cref{prop:star-per}} \\
		&= 

}
 \tag{\ref{ax:specfrob}}\\
		&= \kstar{(f \sqcap g)}.  \tag{\ref{eq:cb:covolution}}
	\end{align*}
	
The above derivation proves the leftmost inequality. In order to prove that $\kstar{\top} = \top$, we first show that $\top ; \kstar{\top} = \top$: the inclusion $\top ; \kstar{\top} \leq \top$  trivially holds.  For the other inclusion, we have the following derivation.
	\begin{align*}
		\top ; \kstar{\top} &= 
    %
}
 \tag{\ref{ax:kb:traceid}} \\
		&= \top. \tag{\ref{eq:cb:covolution}}
	\end{align*}
	To conclude, observe that:
	\[ \kstar{\top} \stackrel{\eqref{eq:starequality}}{=} \id{X} \sqcup \top ; \kstar{\top} = \id{X} \sqcup \top \stackrel{(\text{\Cref{lemma:fb-cbsummary}.(3)})}{=} \top. \qedhere\]
\end{proof}

As expected, $\kstar{(\cdot)}$ commutes with $\op{(\cdot)}$.
\begin{lem}\label{lemma:opstar}
Let $\Cat{C}$ be a kc-rig category. For all $f\colon X \to X$,
\[\kstar{(\op{f})} = \op{(\kstar{f})}\]
\end{lem}
\begin{proof}
 First, note that the following law holds in a Kleene algebra (see e.g. Equation (14) in \cite{Kozen94acompleteness}), thus in particular it holds in any Kleene bicategory by~\Cref{cor:kleeneareka}: %
	\begin{equation}\label{eq:kozen-equiv-axiom}\tag{$\heartsuit$}
		g \sqcup (f;r) \leq r \implies \kstar{f};g \leq r.
	\end{equation}
	Then observe that the following holds for all $f \colon X \to X$:
	\[ \id{X} \sqcup \op{f} ; \op{(\kstar{f})} \stackrel{(\text{\Cref{table:re:daggerproperties}})}{=} \id{X} \sqcup \op{(\kstar{f} ; f)} = \op{(\id{X} \sqcup \kstar{f} ; f)} \stackrel{\eqref{eq:Kllenelaw}}{\leq} \op{(\kstar{f})}. \]
	Thus, by \eqref{eq:kozen-equiv-axiom} the inequality below holds:
	\begin{equation}\label{eq:star-dagger-inclusion}\tag{$\clubsuit$}
		\kstar{(\op{f})} = \kstar{(\op{f})} ; \id{X} \leq \op{(\kstar{f})}.
	\end{equation}
	For the other inequality we exploit~\eqref{eq:star-dagger-inclusion} and the fact that $\op{(\cdot)}$ is involutive:
	\[ \op{(\kstar{f})} \stackrel{(\text{\Cref{table:re:daggerproperties}})}{=} \op{(\kstar{(f^{\dag \dag})})} \stackrel{\eqref{eq:star-dagger-inclusion}}{\leq} (\kstar{(\op{f})})^{\dag \dag} \stackrel{(\text{\Cref{table:re:daggerproperties}})}{=} \kstar{(\op{f})}. \]

\end{proof}
\begin{table}[t]
    \[ 
    \begin{array}{cccc}
        \toprule
        f \sqcap (g \sqcup  h) = (\, f \sqcap g \, ) \sqcup  (\, f \sqcap h  \,)
        &
        f \sqcup  \top = \top
        &
        \kstar{(f \per g)} \leq \kstar{f} \per \kstar{g}
        &
        \op{(f \piu g)} = \op{f} \piu \op{g} 
        \\
        f \sqcup  (g \sqcap h) = (\, f \sqcup  g \, ) \sqcap (\, f \sqcup  h  \,)
        &
        f \sqcap \bot = \bot
        &
        \kstar{(f \sqcap g)} \leq \kstar{f} \sqcap \kstar{g}       
        &
        \op{(f \sqcup g)} = \op{f} \sqcup \op{g}
        \\
        &  \kstar{(\op{f})} = \op{(\kstar{f})} & \kstar{\top}=\top & \op{\bot}=\bot
        \\
        \bottomrule
    \end{array}
    \]
	\caption{Derived laws in kc-rig categories.} 
	\label{table:kc-rig derived laws}
\end{table}
We found it convenient to compactly summarise the lemmas above as follows.
\begin{prop}\label{prop:kc-rig laws}%
	The laws in \Cref{table:kc-rig derived laws} hold in any kc-rig category. 
\end{prop}
\begin{proof}
By \Cref{lemma:fb-cbsummary,prop:star-per,lemma:starcap,lemma:opstar}.
\end{proof}
\begin{cor}
Any kc-rig category is a typed Kleene algebra with converse~\cite{brunet2014kleene,bloom1995notes}.
\end{cor}
The lax distributivity of $\kstar{(\cdot)}$ over $\per$ (\Cref{prop:star-per}) is somehow unsatisfactory when one is interested in taking products of imperative programs (see Section~\ref{ssec:relationalHoare}). We conclude this section with a result (\Cref{prop:perstarcup}) that characterises the products of Kleene stars. First we need the following.
\begin{lem}\label{lemma:starperid}
Let $\Cat{C}$ be a kc-rig category. For all arrows $f\colon X \to X$ and objects $Z$
\[\kstar{f} \per \id{Z} = \kstar{(f\per \id{Z})} \qquad  \id{Z} \per \kstar{f}  = \kstar{(\id{Z} \per f)}  \]
\end{lem}
\begin{proof}
The following derivation proves the leftmost equality. 
\begin{align*}
\kstar{(f \per \id{Z})} &= \trace_{X\per Z}\big( (f \per \id{Z}) \piu \id{X \per Z} ; \codiag{X \per Z} ; \diag{X \per Z} \big) \tag{\ref{eq:star}}\\
& = \trace_{X \per Z}\big( (f \per \id{Z}) \piu (\id{X} \per \id{Z}) ; (\codiag{X} \per \id{Z}) ; (\diag{X} \per \id{Z}) \big) \tag{\Cref{lemma:fb-cbsummary}}\\
& = \trace_{X \per Z}\big( (f \piu\id{X} ) \per \id{Z} ; (\codiag{X} \per \id{Z}) ; (\diag{X} \per \id{Z}) \big) \tag{\Cref{def:rig}}\\
& = \trace_{X \per Z}\big( ( \,(f \piu\id{X} ) ; \codiag{X} ; \diag{X} \,) \per \id{Z} \big) \tag{Functoriality}\\
& = \trace_{X}\big( ( \,(f \piu\id{X} ) ; \codiag{X} ; \diag{X} \,) \big) \per \id{Z}  \tag{\ref{eq:missing}}\\
& = \kstar{f} \per \id{Z} \tag{\ref{eq:star}}
\end{align*}
For the rightmost statement, one first proves \Cref{lemma:leftwisktrace} and then proceeds as above.
\end{proof}
\begin{rem}
While checking the previous derivation, the reader may have noticed that term-based proofs are considerably less intuitive than their diagrammatic counterparts. Unfortunately, standard string diagrams are not well suited to represent the non-trivial interactions between $\piu$ and $\per$ arising in the derivation above. In the next section, we introduce a diagrammatic notation that captures these interactions naturally.
\end{rem}
\begin{prop}\label{prop:perstarcup}
Let $\Cat{C}$ be a kc-rig category. For all arrows $f\colon X \to X$, $g\colon Y \to Y$,
\[\kstar{f}\per \kstar{g} = \kstar{((f\per \id{Y})\sqcup (\id{X}\per g))}\]
\end{prop}
\begin{proof}
For $\kstar{f}\per \kstar{g} \leq \kstar{((f\per \id{Y})\sqcup (\id{X}\per g))}$, we have the following derivation.
\begin{align*}
\kstar{f}\per \kstar{g} & = (\kstar{f}\per\id{Y}); (\id{X} \otimes \kstar{g}) \tag{Functoriality}\\
&=  \kstar{(f\per\id{Y})} ; \kstar{(\id{X} \per g)} \tag{\Cref{lemma:starperid}}\\
& \leq   \kstar{((f\per\id{Y}) \sqcup (\id{X} \otimes g))} \tag{\ref{eq:starsum}}
\end{align*}
To prove the opposite inequality, observe that by  \Cref{lemma:fb-cbsummary}, the following equality holds.
\begin{equation}\label{eq:instarprod}
(f\per \id{Y})\sqcup(\id{X} \per g) \leq (\id{X} \per \id{Y}) \sqcup (f\per \id{Y})\sqcup(\id{X} \per g) \sqcup (f \per g) = (\id{X}\sqcup f) \per (\id{Y} \sqcup g)
\end{equation}
Thus
\begin{align*}
((f\per \id{Y})\sqcup(\id{X} \per g)) ; (\kstar{f} \per \kstar{g}) &\leq ((\id{X}\sqcup f) \per (\id{Y} \sqcup g)) ;(\kstar{f} \per \kstar{g}) \tag{\ref{eq:instarprod}}\\
&=  ((\id{X}\sqcup f) ; \kstar{f} ) \per ((\id{Y} \sqcup g)  ; \kstar{g})) \tag{Functoriality}\\
&\leq  \kstar{f}  \per \kstar{g} \tag{\ref{eq:Kllenelaw}}
\end{align*}
We can now use the implication in \eqref{eq:Kllenelaw} to obtain that 
\begin{equation}\label{eq:instraprod2}
\kstar{((f\per \id{Y})\sqcup(\id{X} \per g))} ; (\kstar{f} \per \kstar{g}) \leq  \kstar{f}  \per \kstar{g}
\end{equation}
We can now conclude as
\begin{align*}
\kstar{((f\per \id{Y})\sqcup(\id{X} \per g))} & = \kstar{((f\per \id{Y})\sqcup(\id{X} \per g))} ; (\id{X} \per \id{Y}) \tag{Functoriality} \\
& \leq \kstar{((f\per \id{Y})\sqcup(\id{X} \per g))} ; (\kstar{f} \per \kstar{g}) \tag{\ref{eq:Kllenelaw}}\\
& \leq  \kstar{f}  \per \kstar{g} \tag{\ref{eq:instraprod2}}
\end{align*}
\end{proof}

\section{Kleene-Cartesian Tape Diagrams}\label{sec:kctapes}
In the proof of \Cref{lemma:starcap}, we manipulated several string diagrams representing arrows of a kc-rig category. However, in such a representation, vertical composition of diagrams amounts to $\piu$, while the other monoidal product, $\per$, is represented in textual form, as in the following diagram.
\[
    }
\]
To properly visualise arrows of rig categories, one would need  three dimensions (see e.g. \cite{comfort2020sheet}): one for $;$, one for $\piu$ and one for $\per$.  An alternative, which remains within two dimensions, is offered by \emph{tape diagrams}~\cite{bonchi2023deconstructing}. Intuitively, tape diagrams are \emph{string diagrams of string diagrams}: the vertical composition of inner diagrams represents $\per$, whereas the vertical composition of outer diagrams represents $\piu$. For instance the diagram above is drawn as
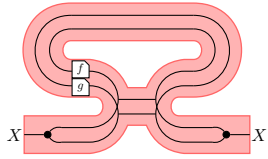
\[
    \InputIfFileExists{traccia.tikz}{}{\input{./tikz/traccia.tikz}}
\]
The key intuition is that, like objects in a free sesquistrict rig category can be written in the form of polynomials $P=\Piu[i]{\Per[j]{A_{i,j}}}$, similarly arrows of a free sesquistrict kc-rig category can always be written as sums of products of certain basic arrows.

In this section we introduce \emph{Kleene-Cartesian tape diagrams}, a graphical notation for kc-rig categories, and we show that they provide the kc-rig category freely generated by a rig signature $(\sort,\sign)$. Our work extends \cite[Section 7]{bonchi2023deconstructing} that identifies tape diagrams as  freely generated \emph{fb-cb rig categories}. %

Thanks to Remark \ref{rem:rigsig} we can restrict, without loss of  generality, to the simpler case where $(\sort,\sign)$ is just a monoidal signature.
As explained in Section \ref{sec:rigcategories}, we can consider the sesquistrict rig categories having as sets of objects $(\sort^\star)^\star$. %
For arrows, consider the following two-layer grammar where $s \in \sign$, $A,B \in \sort$ and $U,V \in \sort^\star$.
\begin{equation}\label{kleeneCartesianGrammar}
      \begin{array}{rcccccccccccc}
          c    & ::= &   \id{A}   &      \mid      &    \id{\uno}    &      \mid      &    \gen    &      \mid      &    \sigma_{A,B}    &      \mid      &      c ; c      &      \mid      &     c \per c   \\
               &     &   \discharger{A}   &      \mid      &   \copier{A}   &      \mid      &   \codischarger{A}   &      \mid      &   \cocopier{A} \\[8pt]
          \t   & ::= &   \id{U}   &      \mid      &    \id{\zero}    &      \mid      &    \tapeFunct{c}    &      \mid      &    \sigma_{U,V}^{\piu}    &      \mid      &      \t ; \t      &      \mid      &     \t \piu \t    \\
               &     &  \bang{U}    &      \mid      &   \diag{U}   &      \mid      &   \cobang{U}   &      \mid      &   \codiag{U}    &       \mid      &   \trace_{U}\t      
    \end{array}
\end{equation}
 The terms of the first layer,  called \emph{circuits}, intuitively represent arrows of a Cartesian bicategory. The terms of the second layer, called \emph{tapes}, represent arrows of a Kleene bicategory. Crucially, a circuit $c$ can occur within a tape as the term $\tapeFunct{c}$.

We only consider those terms to which it is possible to associate source and target objects according to the simple type system in~\Cref{tab:typingTAPE}. In particular, circuits have type $U\to V$ for $U,V \in \sort^\star$ while tapes $P\to Q$ for $P,Q\in (\sort^\star)^\star$.

Constants and operations in \eqref{kleeneCartesianGrammar} can be extended to arbitrary polynomials in $(\sort^\star)^\star$ via the inductive definitions in Table \ref{tab:SyntacticSUGAR}. %

Particularly interesting is that one can define $\per$ on tapes: for $\t_1 \colon P \to Q$, $\t_2 \colon R \to S$,  
\begin{equation}\label{eq:pertapes}
\t_1 \per \t_2 \defeq \LW{P}{\t_2} ; \RW{S}{\t_1} \end{equation} 
where $\LW{P}{\cdot}$, $\RW{S}{\cdot}$ are the left and right whiskerings inductively defined in Table \ref{tab:SyntacticSUGAR}. The same table illustrates the inductive definitions of the left distributors
$\dl{P}{Q}{R} \colon P \per (Q\piu R)  \to (P \per Q) \piu (P\per R)$ and $\per$-symmetries $\symmt{P}{Q}\colon P \per Q \to Q \per P$. %

Next, we impose the laws of kc-rig categories on tapes. However, this should be done carefully, in order to properly tackle the two uniformity laws (AU1) and (AU2) which are implications and not (in)equalities.
Let $\basicR$ be a set of pairs $(\t_1, \t_2)$ of tapes with the same domain and codomain. We define $\precongB$ to be the set generated by the following inference system (where $\t \precongB \s$ is a shorthand for $(\t,\s)\in \precongB$).%
\begin{equation}\label{eq:uniformprecong}
    \scalebox{0.9}{$\begin{array}{@{\qquad}c@{\qquad\qquad}c@{\qquad}c@{\qquad}}
            \inferrule*[right=($\basicR$)]{\t_1 \mathbin{\basicR} \t_2}{\t_1 \mathrel{\precongB} \t_2}
            &
            \inferrule*[right=($r$)]{-}{\t \mathrel{\precongB} \t}
            &    
            \inferrule*[right=($t$)]{\t_1 \mathrel{\precongB} \t_2 \quad \t_2 \mathrel{\precongB} \t_3}{\t_1 \mathrel{\precongB} \t_3}
            \\[8pt]
            \inferrule*[right=($;$)]{\t_1 \mathrel{\precongB} \t_2 \quad \s_1 \mathrel{\precongB} \s_2}{\t_1;\s_1 \mathrel{\precongB} \t_2;\s_2}
            &
            \inferrule*[right=($\piu$)]{\t_1 \mathrel{\precongB} \t_2 \quad \s_1 \mathrel{\precongB} \s_2}{\t_1\piu\s_1 \mathrel{\precongB} \t_2 \piu \s_2}
            &
            \inferrule*[right=($\per $)]{\t_1 \mathrel{\precongB} \t_2 \quad \s_1 \mathrel{\precongB} \s_2}{\t_1\per \s_1 \mathrel{\precongB} \t_2 \per \s_2}
            \\[8pt]
            \multicolumn{3}{c}{
            \inferrule*[right=($u_r$)]{\s_2 \mathrel{\precongB} \s_1 \quad \t_1 ; (\s_1 \piu \id{}) \mathrel{\precongB} (\s_2 \piu \id{}) ; \t_2}{\trace_{S_1}\t_1 \mathrel{\precongB} \trace_{S_2}\t_2}
            \qquad
            \inferrule*[right=($u_l$)]{ \s_2 \mathrel{\precongB} \s_1 \quad (\s_1 \piu \id{}) ; \t_1 \mathrel{\precongB} \t_2 ; (\s_2 \piu \id{})}{\trace_{S_1}\t_1 \mathrel{\precongB} \trace_{S_2}\t_2}
            }
        \end{array}$}
\end{equation}
The first six laws ensure that $\precongB$ is a precongruence (w.r.t. $;$, $\piu$ and $\per$) containing $\basicR$. The last two rules force the uniformity laws: observe that, while in (AU1) and (AU2) the same arrow $r$ occurs in both the left and the right-hand side of the premises, here $r$ is replaced by two different but related tapes $\s_1$ and $\s_2$. This technicality is needed to guarantee uniformity in the category resulting from the following construction.

We take $\basicKC$ to be the set of all pairs of tapes containing the axioms in Table~\ref{tab:axiomTAPE} %
and define $\leq_{\basicKC}$ according to \eqref{eq:uniformprecong}. We fix $=_{\basicKC} \defeq \leq_{\basicKC} \cap \geq_{\basicKC}$. With these definitions we can construct the category of Kleene-Cartesian tapes $\KTCB$: objects are polynomials in $(\sort^\star)^\star$ with $\piu$ and $\per$ defined as in \eqref{def:productPolynomials}; arrows are $=_\basicKC$-equivalence classes of tapes; every homset $\KTCB[P,Q]$ is ordered by $\precongKC$.  The construction of $\KTCB$  gives rise to a sesquistrict kc-rig category. More importantly, $\KTCB$ is the freely generated one.

\begin{thm}\label{thm:KleeneCartesiantapesfree}
$\KTCB$ is the free sesquistrict kc-rig category generated by $(\sort, \sign)$.
\end{thm}
The proof of the above result relies on several adjunctions between categories 
that we have not had the opportunity to introduce.  
For this reason, the proof is presented in detail in 
\Cref{sec:free-kleene-cartesian}.

\begin{table}[H]
    \centering
        \[
        \def\arraystretch{1.4}
\begin{array}{c@{\qquad\;}c@{\qquad\;}c@{\qquad\;}c}
            \toprule
            \id{\uno} \colon \uno \to \uno & \id{A} \colon A \to A  & \sigma_{A,B} \colon A \per B \to B \per A & s \colon \ar(s) \to \coar(s)
            \\[5pt]
            \discharger{A} \colon A \to \uno & \copier{A} \colon A \to A \per A & \codischarger{A} \colon \uno \to A & \cocopier{A} \colon A \per A \to A
            \\[5pt]
            \multicolumn{4}{c}{
            \inferrule{c \colon U_1 \to V_1 \and d \colon U_2 \to V_2}{c \per d \colon U_1 \per  U_2 \to V_1 \per  V_2}
            \qquad\qquad
            \inferrule{c \colon U \to V \and d \colon V \to W}{c ; d \colon U \to W}
            }
            \\[10pt]
            {id_\zero \colon \zero \to \zero} & {id_U \colon U \to U} & {\sigma_{U, V}^{\piu} \colon U \piu V \to V \piu U} &  \inferrule{c \colon U \to V}{\tapeFunct{c}\colon U \to V} 
            \\[5pt]
            {\bang{U}\colon U \to \zero}   &    {\diag{U}\colon U \to U\piu U}   &  { \cobang{U} \colon \zero \to U}   &   {\codiag{U}\colon U \piu U \to U}
            \\[5pt]
            \multicolumn{4}{c}{
            \inferrule{\t \colon P_1 \to Q_1 \and \s \colon P_2 \to Q_2}{\t \piu \s \colon P_1 \piu P_2 \to Q_1 \piu Q_2}
            \qquad
            \inferrule{\t \colon P \to Q \and \s \colon Q \to R}{\t ; \s \colon P \to R}
            \qquad
            \inferrule{\t \colon U\piu P \to U \piu Q}{ \trace_{U}\t   \colon P \to Q}
            }
            \\[10pt]
            \bottomrule
        \end{array}
        \]
    \caption{Typing rules for Kleene-Cartesian tapes.}\label{tab:typingTAPE}
\end{table}
\par\medskip
\begin{table}[H]
    \centering
    \scalebox{0.78}{
         $
%
  \end{array}
  \\
  \bottomrule
 \end{array}$
                }
    \caption{Axioms for Kleene-Cartesian tapes. For each axiom $l=r$ in top-left corner, the set $\basicKC$ contains the pairs $(l,r)$ and $(r,l)$  where  $\perG$ and $\unoG$ are replaced by $\piu$ and $\zero$ and  the pairs $(\tape{l},\tape{r})$ and $(\tape{r},\tape{l})$ where $\perG$ and $\unoG$ are replaced by $\per$ and $\uno$. In the rest, for each $l\leq r$ $\basicKC$ contains a pair $(l,r)$ and, additionally, the pair $(r,l)$ in case of an axiom $l=r$. }
    \label{tab:axiomTAPE}
\end{table}
\begin{table}[H]
    \centering
    \scriptsize
    \begin{subtable}{\textwidth}
        \centering
        \scalebox{0.95}{$
\begin{array}{l@{\,}c@{\,}l l@{\,}c@{\,}l | l@{\,}c@{\,}l l@{\,}c@{\,}l}
            \toprule
            \LW U {\id\zero} &\defeq& \id\zero & \LW U {\id\uno} &\defeq& \id{U} 
            &
            \RW U {\id\zero} &\defeq& \id\zero & \RW U {\id\uno} &\defeq& \id{U} 
            \\
            \LW U {\tape{c}} &\defeq& \tape{\id U \per c} & \LW U {\symmp{V}{W}} &\defeq& \symmp{UV}{UW}
            &
            \RW U {\tape{c}} &\defeq& \tape{c \per \id U} & \RW U {\symmp{V}{W}} &\defeq& \symmp{VU}{WU}
            \\
            \LW U {\diag V} &\defeq& \diag{UV} & \LW U {\bang V} &\defeq& \bang{UV}
            &
            \RW U {\diag V} &\defeq& \diag{VU} & \RW U {\bang V} &\defeq& \bang{VU} 
            \\
            \LW U {\codiag V} &\defeq& \codiag{UV} & \LW U {\cobang V} &\defeq& \cobang{UV}
            &
            \RW U {\codiag V} &\defeq& \codiag{VU} & \RW U {\cobang V} &\defeq& \cobang{VU} 
            \\
            \LW U {\t_1 ; \t_2} &\defeq& \LW U {\t_1} ; \LW U {\t_2}
            & 
            \LW{U}{\trace_{V}\t}&\defeq&\trace_{UV}\LW{U}{\t}
            &
            \RW U {\t_1 ; \t_2} &\defeq& \RW U {\t_1} ; \RW U {\t_2} 
            &
            \RW{U}{\trace_{V}\t}&\defeq&\trace_{VU}\RW{U}{\t}
            \\
            \LW U {\t_1 \piu \t_2} &\defeq& \LW U {\t_1} \piu \LW U {\t_2}
            & &&&
            \RW U {\t_1 \piu \t_2} &\defeq& \RW U {\t_1} \piu \RW U {\t_2}
            \\
            \midrule
            \LW{\zero}{\t} &\defeq& \id{\zero} & &&&  \RW{\zero}{\t} &\defeq& \id{\zero} \\
            \LW{W\piu S'}{\t} &\defeq& \LW{W}{\t} \piu \LW{S'}{\t} &  &&& \RW{W \piu S'}{\t} &\defeq& \multicolumn{4}{l}{\dl{P}{W}{S'} ; (\RW{W}{\t} \piu \RW{S'}{\t}) ; \Idl{Q}{W}{S'}}
            \\
\bottomrule
        \end{array}
        $}
        \caption{Inductive definition of left and right monomial whiskerings (top); inductive definition of polynomial whiskerings (center); definition of $\per$ (bottom).}\label{tab:producttape}
    \end{subtable}

    \par\medskip

    \begin{subtable}{0.58\linewidth}
        \centering
        \scalebox{0.95}{\begin{tabular}{@{}l}
            \toprule
            $\dl{P}{Q}{R} \colon P \per (Q\piu R)  \to (P \per Q) \piu (P\per R) \vphantom{\symmt{P}{Q}}$ \\
            \midrule
            $\dl{\zero}{Q}{R} \defeq \id{\zero} \vphantom{\symmt{P}{\zero} \defeq \id{\zero}}$ \\
            $\dl{U \piu P'}{Q}{R} \defeq (\id{U\per (Q \piu R)} \piu \dl{P'}{Q}{R});(\id{U\per Q} \piu \symmp{U\per R}{P'\per Q} \piu \id{P'\per R}) \vphantom{\symmt{P}{V \piu Q'} \defeq \dl{P}{V}{Q'} ; (\Piu[i]{\tapesymm{U_i}{V}} \piu \symmt{P}{Q'})}$ \\
            \bottomrule
        \end{tabular}}
        \caption{Inductive definition of $\delta^l$.}
        \label{table:def dl}
    \end{subtable}\hfill
    \begin{subtable}{0.4\linewidth}
        \centering
        \scalebox{0.95}{\begin{tabular}{@{}l@{}}
            \toprule
            $\symmt{P}{Q} \colon P\per Q \to Q \per P$, with $P = \Piu[i]{U_i}$ \\
            \midrule
            $\symmt{P}{\zero} \defeq \id{\zero}$ \\
            $\symmt{P}{V \piu Q'} \defeq \dl{P}{V}{Q'} ; (\Piu[i]{\tapesymm{U_i}{V}} \piu \symmt{P}{Q'})$ \\
            \bottomrule
        \end{tabular}}
        \caption{Inductive definition of $\symmt$.}
        \label{table:def symmt}
    \end{subtable}

    \par\medskip

    \begin{subtable}{\textwidth}
        \centering
        {$
            %
        $}   
        \caption{Inductive definitions of $\copier{U}, \discharger{U}$ and $\cocopier{U}, \codischarger{U}$ for all monomials $U \in \sort^\star$.}
        \label{table:co-monoids coherence}
    \end{subtable}

    \par\medskip

    \begin{subtable}{\textwidth}
        \centering
        {$
            %
        $}
        \caption{Inductive definitions of $\copier{P}, \discharger{P}$ and $\cocopier{P}, \codischarger{P}$ for all polynomials $P \in (\sort^\star)^\star$.}
        \label{table:defmonoidper}
    \end{subtable}

    \par\medskip

    \begin{subtable}{\textwidth}
        \centering
        {
                $%
                    \\
                    \bottomrule
                    \addlinespace[4pt]
                    \multicolumn{2}{l@{}}{
                        \trace_{P} \colon
                        \Cat{C}[P \piu Q] \to \Cat{C}[P \piu S]
                    }
                    \\
                    \midrule
                    \multicolumn{2}{l@{}}{
                        \trace_{\zero}(\t) \defeq \t
                        \qquad
                        \trace_{U \piu P'}(\t) \defeq
                        \trace_{P'} \trace_U(\t)
                    }
                \\
                \bottomrule
                \end{array}
        $}
        \caption{Inductive definitions of $\diag{P}, \bang{P}, \codiag{P}, \cobang{P}$ and $\trace_{P}$ for all polynomials $P \in (\sort^\star)^\star$.}
        \label{tab:inddefutfb}
    \end{subtable}
    \caption{Syntactic sugar for Kleene-Cartesian Tapes.}\label{tab:SyntacticSUGAR}
\end{table}

\newpage

\subsection{Diagrammatic Syntax}
As mentioned earlier, the key feature of tapes is that they can be drawn nicely in 2 dimensions despite representing arrows of rig categories.
Indeed, both circuits and tapes can be drawn as string diagrams. Note however that \emph{inside} tapes, there are string diagrams. %
Thus, the grammar in~\eqref{kleeneCartesianGrammar} can be graphically rendered as follows. %
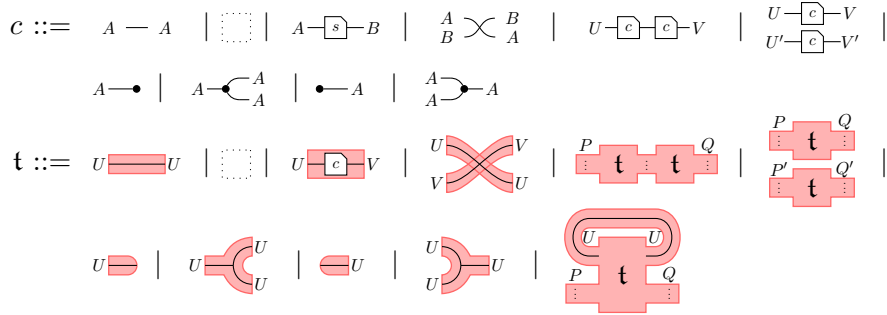
\begin{equation*}\label{tracedTapesDiagGrammar}
    \setlength{\tabcolsep}{2pt}
    \begin{tabular}{rc cccccccccccc}
        $c$  & ::= &  $\wire{A}$ & $\mid$ & $ 
    \InputIfFileExists{empty.tikz}{}{\input{./tikz/empty.tikz}}
 $ & $\mid$ & $ \Cgen{\gen}{A}{B}  $ & $\mid$ & $ \Csymm{A}{B} $ & $\mid$ & $ 
    \InputIfFileExists{seq_compC.tikz}{}{\input{./tikz/seq_compC.tikz}}
   $ & $\mid$ & $  
    \InputIfFileExists{par_compC.tikz}{}{\input{./tikz/par_compC.tikz}}
$ & $\mid$ \\
             &     &  \multicolumn{12}{l}{$\CBdischarger{A} \; \mid \; \CBcopier{A} \; \mid \; \CBcodischarger{A} \; \mid \; \CBcocopier{A}$} \\[-5pt]
        $\t$ & ::= & $\Twire{U}$ & $\mid$ & $ 
    \InputIfFileExists{empty.tikz}{}{\input{./tikz/empty.tikz}}
 $ & $\mid$ & $ \Tcirc{c}{U}{V}  $ & $\mid$ & $ \Tsymmp{U}{V} $ & $\mid$ & $ 
    \InputIfFileExists{tapes/seq_comp.tikz}{}{\input{./tikz/tapes/seq_comp.tikz}}
  $ & $\mid$ & $  
    \InputIfFileExists{tapes/par_comp.tikz}{}{\input{./tikz/tapes/par_comp.tikz}}
$ & $\mid$ \\
             &     & \multicolumn{12}{l}{$\Tcounit{U} \; \mid \; \Tcomonoid{U} \; \mid \; \Tunit{U} \; \mid \; \Tmonoid{U} \; \mid \; \TTraceMon{\t}[U][P][Q]$}
    \end{tabular}
\end{equation*} 
The identity $\id\zero$ is rendered as the empty tape $
    \InputIfFileExists{empty.tikz}{}{\input{./tikz/empty.tikz}}
$, while $\id\uno$ is $
    \InputIfFileExists{tapes/empty.tikz}{}{\input{./tikz/tapes/empty.tikz}}
$: a tape filled with the empty circuit. %
For a monomial $U \!=\! A_1\dots A_n$, $\id U$ is depicted as a tape containing  $n$ wires labelled by $A_i$. For instance, $\id{AB}$ is rendered as $\TRwire{A}{B}$. When clear from the context, we will simply represent it as a single wire  $\Twire{U}$ with the appropriate label.
Similarly, for a polynomial $P = \Piu[i=1][n]{U_i}$, $\id{P}$ is obtained as a vertical composition of tapes, as illustrated below on the left.

\noindent\scalebox{0.82}{$ 
    \id{AB \piu \uno \piu C} = \!\!\!\begin{aligned}\begin{gathered} \TRwire{A}{B} \\[-1.8mm] \Twire{\uno} \\[-1.8mm] \Twire{C} \end{gathered}\end{aligned}
    \quad
    \codiag{A\piu B \piu C} = \!\!
    \InputIfFileExists{tapes/examples/codiagApBpC.tikz}{}{\input{./tikz/tapes/examples/codiagApBpC.tikz}}

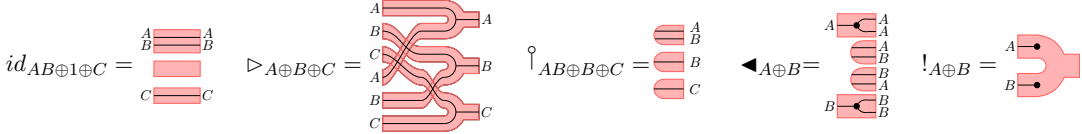
    \quad
    \cobang{AB \piu B \piu C} = \!
    \InputIfFileExists{tapes/examples/cobangABpBpC.tikz}{}{\input{./tikz/tapes/examples/cobangABpBpC.tikz}}

    \quad
    \copier{A \piu B} = \!\!\!

}

$}

The diagonal $\diag{U} \colon U \!\to\! U \piu U$ is represented as a splitting of tapes, while the bang $\bang{U} \colon U \!\to\! \zero$ is a tape closed on its right boundary. 
Codiagonals and cobangs are represented in the same way but mirrored along the y-axis. Exploiting the usual coherence conditions (Definition \ref{def:biproduct category}.(3)), we can construct (co)diagonals and (co)bangs for arbitrary polynomials.
For example, $\codiag{A\piu B \piu C}$ and $\cobang{AB \piu B \piu C}$ are depicted as the second and third diagrams above.

The copier $\copier{A} \colon A \!\to\! A \per A$ is represented as a splitting of wires, while the discharger $\discharger{A} \colon A \to \uno$ is a wire closed on the right. As expected, the cocopier and codischarger are obtained via mirroring. 
From the coherence laws in~\eqref{eq:fbcbcoherence}, one can build (co)copiers and (co)dischargers for arbitrary polynomials.
For instance, $\copier{A \piu B} \colon A \piu B \to (A \piu B)\per(A \piu B) = AA \piu AB \piu BA \piu BB$ and $\discharger{A \piu B} \colon A \piu B \to \uno$ are drawn as the last two diagrams above.
For an arbitrary tape diagram $\t \colon P \to Q$ we write $\Tbox{\t}{P}{Q}$.

The graphical representation embodies several axioms such as those of monoidal categories and several axioms for traces. Those axioms which are not implicit in the graphical representation are illustrated in Figures \ref{fig:tapesax} and \ref{fig:cb axioms}.
Figure \ref{fif:poset unif tape axioms} illustrates the uniformity laws in the form of tape diagrams.

\begin{rem}
The diagrammatic representation provided by tapes does not allow for a direct 
visualisation of the $\per$--composition of tapes. Indeed, by 
\eqref{eq:pertapes} we have
\[
  \t_1 \per \t_2 \;=\; \LW{P}{\t_2} \,;\, \RW{S}{\t_1},
\]
where $\LW{P}{\cdot}$ and $\RW{S}{\cdot}$ must be computed using the 
definitions given in Table~\ref{tab:SyntacticSUGAR}.  
When the polynomials $P$ and $S$ are monomials, the corresponding diagrams are, 
however, easy to draw: $\LW{P}{\t_2}$ is obtained from $\t_2$ by adding extra 
upper wires for $P$, and analogously $\RW{S}{\t_1}$ is obtained by adding 
extra lower wires for $S$. For instance, consider the tapes $\t_1 \colon A \otimes A \to A$ and $\t_2 \colon B \to B$, illustrated below on the left, and their product $\t_1 \per \t_2$ on the right:
\[

    \InputIfFileExists{tMonProd.tikz}{}{\input{./tikz/tMonProd.tikz}}
 \quad \per \quad 
    \InputIfFileExists{sMonProd.tikz}{}{\input{./tikz/sMonProd.tikz}}
 \;=\; 
    \InputIfFileExists{tXsMonProd.tikz}{}{\input{./tikz/tXsMonProd.tikz}}

\]

When $P$ and $S$ are not monomials, the situation is more subtle.  
In particular, the right whiskering $\RW{S}{\t_1}$ requires the use of left 
distributors.  
For instance, consider the tapes 
$\t_1\colon U \piu V \to W \piu Z$ and 
$\t_2\colon U' \piu V' \to W' \piu Z'$, illustrated below on the left:
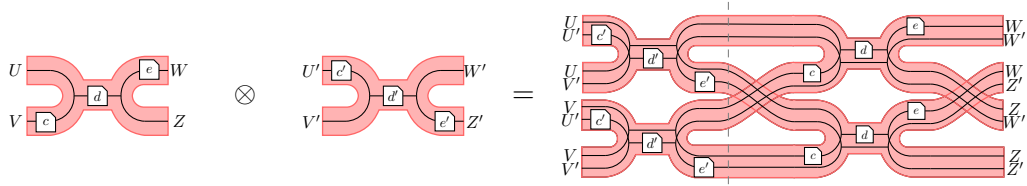
\[
  
    \InputIfFileExists{tapes/examples/t1.tikz}{}{\input{./tikz/tapes/examples/t1.tikz}}

  \quad \per \quad
  
    \InputIfFileExists{tapes/examples/t2.tikz}{}{\input{./tikz/tapes/examples/t2.tikz}}

  \;=\;
  
    \InputIfFileExists{tapes/examples/t1xt2_line.tikz}{}{\input{./tikz/tapes/examples/t1xt2_line.tikz}}

\]
Then $\t_1 \per \t_2$ is the sequential composition of 
$\LW{U \piu V}{\t_2}$ and $\RW{W' \piu Z'}{\t_1}$.  
The dashed vertical line highlights the boundary between the left and the right 
polynomial whiskerings.  
On the left-hand side, $\LW{U \piu V}{\t_2}$ is computed as the vertical 
composition of the monomial whiskerings $\LW{U}{\t_2}$ and $\LW{V}{\t_2}$.  
On the right-hand side, $\RW{W' \piu Z'}{\t_1}$ is given by the vertical 
composition of $\RW{W'}{\t_1}$ and $\RW{Z'}{\t_1}$, suitably pre- and post-composed 
with the appropriate left distributors.
\end{rem}

\begin{figure}[ht!]
    \mylabel{ax:symmpinv}{$\symmp$-inv}
    \mylabel{ax:symmpnat}{$\symmp$-nat}
    \mylabel{ax:diagas}{$\diag{}$-as}
    \mylabel{ax:diagun}{$\diag{}$-un}
    \mylabel{ax:diagco}{$\diag{}$-co}
    \mylabel{ax:codiagas}{$\codiag{}$-as}
    \mylabel{ax:codiagun}{$\codiag{}$-un}
    \mylabel{ax:codiagco}{$\codiag{}$-co}
    \mylabel{ax:bi}{bi}
    \mylabel{ax:bo}{bo}
    \mylabel{ax:diagbi}{$\diag{}$-bi}
    \mylabel{ax:codiagbi}{$\codiag{}$-bi}
    \mylabel{ax:diagnat}{$\diag{}$-nat}
    \mylabel{ax:bangnat}{$\bang{}$-nat}
    \mylabel{ax:codiagnat}{$\codiag{}$-nat}
    \mylabel{ax:cobangnat}{$\cobang{}$-nat}
       \mylabel{ax:relCobangBang}{$\cobang{}\bang{}$}
        \mylabel{ax:relBangCobang}{$\bang{}\cobang{}$}
        \mylabel{ax:relCodiagDiag}{$\codiag{}\diag{}$}
        \mylabel{ax:relDiagCodiag}{$\diag{}\codiag{}$}
        \mylabel{ax:trace:tape:kstar-id}{AT1}
    \begin{center}
        \scalebox{0.9}{

\end{tabular}
        }
    \end{center}
    \caption{Tape axioms for Kleene bicategories.} 
    \label{fig:tapesax}
\end{figure}

\begin{figure}[ht!]
    \mylabel{ax:symmtinv}{$\sigma^{\per}$-inv}
    \mylabel{ax:symmtnat}{$\sigma^{\per}$-nat}
    \mylabel{ax:copieras}{$\copier{}$-as}
    \mylabel{ax:copierun}{$\copier{}$-un}
    \mylabel{ax:copierco}{$\copier{}$-co}
    \mylabel{ax:cocopieras}{$\cocopier{}$-as}
    \mylabel{ax:cocopierun}{$\cocopier{}$-un}
    \mylabel{ax:cocopierco}{$\cocopier{}$-co}
    \mylabel{ax:specfrob}{S}
    \mylabel{ax:frob}{F}
    \mylabel{ax:copiernat}{$\copier{}$-nat}
    \mylabel{ax:dischargernat}{$\discharger{}$-nat}
    \mylabel{ax:dischargeradj1}{$\codischarger{}\discharger{}$}
    \mylabel{ax:dischargeradj2}{$\discharger{}\codischarger{}$}
    \mylabel{ax:copieradj1}{$\cocopier{}\copier{}$}
    \mylabel{ax:copieradj2}{$\copier{}\cocopier{}$}
	\centering

    \caption{Axioms of Cartesian bicategories.}
    \label{fig:cb axioms}
\end{figure}
\begin{figure}
        \mylabel{ax:trace:tape:AU1}{AU1}
        \mylabel{ax:trace:tape:AU2}{AU2}
        \[ 
    \InputIfFileExists{traceTapeAx/posetUnif/1/lhs.tikz}{}{\input{./tikz/traceTapeAx/posetUnif/1/lhs.tikz}}
 \leq 
    \InputIfFileExists{traceTapeAx/posetUnif/1/rhs.tikz}{}{\input{./tikz/traceTapeAx/posetUnif/1/rhs.tikz}}
 \stackrel{(\text{AU1})}{\implies} 
    \InputIfFileExists{traceTapeAx/posetUnif/1/TR_lhs.tikz}{}{\input{./tikz/traceTapeAx/posetUnif/1/TR_lhs.tikz}}
 \leq 
    \InputIfFileExists{traceTapeAx/posetUnif/1/TR_rhs.tikz}{}{\input{./tikz/traceTapeAx/posetUnif/1/TR_rhs.tikz}}
 \]
        \[ 
    \InputIfFileExists{traceTapeAx/posetUnif/2/lhs.tikz}{}{\input{./tikz/traceTapeAx/posetUnif/2/lhs.tikz}}
 \leq 
    \InputIfFileExists{traceTapeAx/posetUnif/2/rhs.tikz}{}{\input{./tikz/traceTapeAx/posetUnif/2/rhs.tikz}}
 \stackrel{(\text{AU2})}{\implies} 
    \InputIfFileExists{traceTapeAx/posetUnif/2/TR_lhs.tikz}{}{\input{./tikz/traceTapeAx/posetUnif/2/TR_lhs.tikz}}
 \leq 
    \InputIfFileExists{traceTapeAx/posetUnif/2/TR_rhs.tikz}{}{\input{./tikz/traceTapeAx/posetUnif/2/TR_rhs.tikz}}
 \]
        \caption{Posetal uniformity axioms in tape diagrams.}
        \label{fif:poset unif tape axioms}
\end{figure}

\subsection{Semantics of Tape Diagrams}
Recall from Section \ref{sec:rigcategories} that an interpretation $\mathcal{I}=(\alpha_{\sort}, \alpha_{\sign})$ of a monoidal signature $(\sort,\sign)$ in a sesquistrict rig category  $H \colon \Cat{S} \to \Cat{C}$ consists of $\alpha_{\sort}\colon \sort \to Ob(\Cat{S})$ and $\alpha_{\sign}\colon \sign \to Ar(\Cat{C})$  preserving (co)arities of symbols $s\in \sign$. 
Whenever $\Cat{C}$ is a kc-rig category, $\mathcal{I}$ gives rise uniquely, by freeness of $\KTCB$, to a morphism of kc-rig categories\footnote{To be completely formal, the unique sesquistrict rig functor from the sesquistrict rig category $\sort \to \KTCB$ to $H\colon \Cat{S} \to \Cat{C}$ is the pair $(\alpha_\sort \colon \sort \to \Cat{S}, \TCBdsem{\cdot}_{\interpretation} \colon \KTCB \to \Cat{C} )$.} $\TCBdsem{\cdot}_{\interpretation} \colon \KTCB \to \Cat{C}$
 defined on polynomials $P$ as 
\[\TCBdsem{A}_{\interpretation}=H(\alpha_{\sort}(A))  \quad \TCBdsem{\zero}_{\interpretation}=\zero \quad \TCBdsem{\uno}_{\interpretation}=\uno \quad  \TCBdsem{P \piu Q}_{\interpretation}= \TCBdsem{P}_{\interpretation} \piu \TCBdsem{Q}_{\interpretation} \quad  \TCBdsem{P \per Q}_{\interpretation}= \TCBdsem{P}_{\interpretation} \per \TCBdsem{Q}_{\interpretation} \]
and on tapes as %
\noindent\begin{equation}\label{eq:semantic}\scalebox{0.85}{
\renewcommand{\arraystretch}{1.5}
\begin{tabular}{@{\!\!\!\!\!\!\!}lllll}
    $\CBdsem{s}_{\interpretation} = \alpha_{\sign}(s)$ & $\CBdsem{\id{A}}_{\interpretation}= \id{\TCBdsem{A}_{\interpretation} } $ & $\CBdsem{\copier{A}}_{\interpretation} = \copier{\TCBdsem{A}_{\interpretation}}$ & $\CBdsem{\discharger{A}}_{\interpretation} = \discharger{\TCBdsem{A}_{\interpretation}}$ & $\CBdsem{c;d}_{\interpretation} = \CBdsem{c}_{\interpretation}; \CBdsem{d}_{\interpretation}$ \\
    $\CBdsem{\symmt{A}{B}}_{\interpretation} = \symmt{\TCBdsem{A}_{\interpretation}}{\TCBdsem{B}_{\interpretation}}$ & $\CBdsem{\id{1}}_{\interpretation}= \id{1}$ & $ \CBdsem{\cocopier{A}}_{\interpretation} = \cocopier{\TCBdsem{A}_{\interpretation}}  $&$ \CBdsem{\codischarger{A}}_{\interpretation} = \codischarger{\TCBdsem{A}_{\interpretation}}$ &  $ \CBdsem{c\per d}_{\interpretation} = \CBdsem{c}_{\interpretation} \per \CBdsem{d}_{\interpretation}$ \\
    $\CBdsem{\, \tapeFunct{c} \,}_{\interpretation}= \CBdsem{c}_{\interpretation} $ & $\CBdsem{\id{U}}_{\interpretation}= \id{\TCBdsem{U}_{\interpretation}} $ & $ \CBdsem{\diag{U}}_{\interpretation}= \diag{\TCBdsem{U}_{\interpretation}} $&$ \CBdsem{\bang{U}}_{\interpretation} = \bang{\TCBdsem{U}_{\interpretation}}$ & $\CBdsem{\s;\t}_{\interpretation} = \CBdsem{\s}_{\interpretation} ; \CBdsem{\t}_{\interpretation}$ \\
    $\CBdsem{\symmp{U}{V}}_{\interpretation}= \symmp{\TCBdsem{U}_{\interpretation}}{\TCBdsem{V}_{\interpretation}}$ & $ \CBdsem{\id{\zero}}_{\interpretation} = \id{\zero}$ & $\CBdsem{\codiag{U}}_{\interpretation}= \codiag{\TCBdsem{U}_{\interpretation}} $&$ \CBdsem{\cobang{U}}_{\interpretation} = \cobang{\TCBdsem{U}_{\interpretation}}$ & $ \CBdsem{\s \piu \t}_{\interpretation} = \CBdsem{\s}_{\interpretation} \piu \CBdsem{\t}_{\interpretation} $\\
    \multicolumn{5}{c}{
        $\CBdsem{\trace_{U}\t}_{\interpretation} = \trace_{\TCBdsem{U}_{\interpretation}} \CBdsem{\t}_{\interpretation}$
    }
\end{tabular}
}\end{equation}
The functor $\TCBdsem{\cdot}_{\interpretation} \colon \KTCB \to \Cat{C}$ serves as semantics for tape diagrams. Hereafter, we will typically take as $\Cat{C}$ the   kc-rig category $\Rel$. 
\begin{exa}[Back to the calculus of relations]
Recall the calculus of relations $\CR$ from Section \ref{sec:CR}. The set $\Sigma$ of generating symbols in $\CR$ corresponds to a monoidal signature $(\sort, \sign)$ where $\sort$ contains a single sort $A$ and each symbol in $\sign$ has both arity and coarity $A$. 
An interpretation of $\CR$ is exactly an interpretation $\interpretation$ of the monoidal signature $(\sort,\sign)$ into $\Rel$: a set $\alpha_{\sort}(A)$ and a relation $\alpha_\sign(R)\subseteq \alpha_{\sort}(A) \times \alpha_{\sort}(A)$ for each symbol $R\in \sign$.

One can encode $\CR$ expressions into tapes as follows 
\[
\begin{array}{rcl@{\qquad}rcl@{\qquad}rcl}
\encoding{R}&\defeq& \alpha_{\sign}(R)&
\encoding{\id{}}&\defeq& \id{\alpha_{\sort(A)}} &
\encoding{E_1;E_2}& \defeq & \encoding{E}; \encoding{E_2}\\
\encoding{\op{E}} & \defeq &\op{\encoding{E}}&
\encoding{\top} &  \defeq & \top &
\encoding{E_1 \cap E_2}&  \defeq & \encoding{E_1} \sqcap \encoding{E_2}\\
\encoding{\kstar{E}} & \defeq &\kstar{\encoding{E}}&
\encoding{\bot} &  \defeq & \bot &
\encoding{E_1 \cup E_2}&  \defeq & \encoding{E_1} \sqcup \encoding{E_2}\\
\end{array}
 \]
where $\op{(\cdot)},\kstar{(\cdot)},\top,\bot, \sqcap,\sqcup$ on the right-hand side are defined as in \eqref{eq:cb:covolution}, \eqref{eq:covolution}, and~\eqref{eq:star}. 
 
 A simple inductive argument confirms that the encoding preserves the semantics: for all interpretations $\interpretation$ and expressions $E\in \CR$, \[\dsemRel{E} = \TCBdsem{\encoding{E}}_{\interpretation}\text{.}\] 
 Thus, one can safely check $\leq_{\basicKC}$ by first encoding $\CR$ expressions into tapes and then use the axioms of $\basicKC$. Indeed, if $\encoding{E_1}\leq_{\basicKC}\encoding{E_2}$, then $E_1 \leq_{\CR} E_2$. The converse implication, completeness, is an open problem.

To conclude, note that $\KTCB$ is strictly more expressive than $\CR$: thanks to the monoidal product $\otimes$, tapes can deal e.g. with $n$-ary functions, as we will see in Example \ref{ex:functions}. 
\end{exa}
\subsection{Kleene-Cartesian Theories and Functorial Semantics}\label{ssec:funsem}
A \emph{Kleene-Cartesian theory}, shortly kc theory, is a pair $(\sign, \basicR)$ where $\sign$ is a monoidal signature and $\basicR$ is a set of pairs $(\t_1, \t_2)$ of tapes with same domain and codomain. We think of each pair $(\t_1, \t_2)$ as an inequality $\t_1 \leq \t_2$, but the results that we develop in this section trivially hold  also for equations: it is enough to add in $\basicR$ a pair $(\t_2, \t_1)$ for each $(\t_1, \t_2) \in \basicR$.
Hereafter we always keep implicit $\basicKC$ and we write $\precongB$ for $\precongR{\basicKC \cup \basicR}$. We fix $=_{\basicR}\defeq \leq_{\basicR}\cap \geq_{\basicR}$.

We say that an interpretation $\interpretation$ of $\sign$ in a kc-rig category $\Cat{C}$ is \emph{a model of the theory} $(\sign, \basicR)$ whenever $\TCBdsem{\cdot}_{\interpretation} \colon \KTCB \to \Cat{C}$ preserves $\precongB$: if $\t_1 \precongB \t_2$, then $\TCBdsem{\t_1}_{\interpretation} $ is below $\TCBdsem{\t_2}_{\interpretation}$ in $\Cat{C}$. 

\begin{exa}[Linear Orders]\label{ex:totalorder}
	Consider the signature $(\sort, \sign)$ where $\sort$ contains a single sort $A$ and $\sign = \{ R \colon A \to A \}$. Take as $\basicR$ the set consisting of the following inequalities:
	\mylabel{R:refl}{refl}
	\mylabel{R:tr}{tr}
	\mylabel{R:anti}{anti}
	\mylabel{R:lin}{lin}
	\[
		\begin{array}{c@{}c@{}c c@{}c@{}c}
			
    \InputIfFileExists{linOrd/id.tikz}{}{\input{./tikz/linOrd/id.tikz}}
 &\stackrel{}{\leq}& 
    \InputIfFileExists{linOrd/R.tikz}{}{\input{./tikz/linOrd/R.tikz}}
 & 
    \InputIfFileExists{linOrd/RR.tikz}{}{\input{./tikz/linOrd/RR.tikz}}
 &\stackrel{}{\leq}& 
    \InputIfFileExists{linOrd/R.tikz}{}{\input{./tikz/linOrd/R.tikz}}
  
			\\[10pt]
			
    \InputIfFileExists{linOrd/and.tikz}{}{\input{./tikz/linOrd/and.tikz}}
 &\stackrel{}{\leq}& 
    \InputIfFileExists{linOrd/id.tikz}{}{\input{./tikz/linOrd/id.tikz}}
 & 
    \InputIfFileExists{linOrd/top.tikz}{}{\input{./tikz/linOrd/top.tikz}}
 &\stackrel{}{\leq}& 
    \InputIfFileExists{linOrd/or.tikz}{}{\input{./tikz/linOrd/or.tikz}}
		
		\end{array}
	\]
	An interpretation $\interpretation$ of $(\sort,\sign)$ in $\Rel$, consists of a set $\alpha_\sort(A)$ and a relation $\alpha_\sign(R) \subseteq \alpha_{\sort}(A) \times \alpha_{\sort}(A)$. It is a model iff $R$ is a linear order, i.e. it is reflexive~\eqref{eq:cb:ref}, transitive~\eqref{eq:cb:trans}, antisymmetric ($R\cap\op{R}\subseteq \id{}$) and linear ($\top \subseteq R\cup \op{R}$).
\end{exa}

\begin{exa}[Functions]\label{ex:functions}
Let $\sort$ be a set of sorts and $\sign \defeq \{f\colon U \to A\}$ for some $A\in \sort$ and $U\in S^\star$. Let $\basicR$ be the set of the two equalities in \eqref{eq:deterministic-total}. An interpretation $\interpretation$ of $(\sort,\sign)$ in $\Rel$, consists of a set $\alpha_\sort(A_i)$ for each $A_i\in \sort$ and a relation $\alpha_\sign(f) \subseteq \alpha_{\sort}^\sharp(U) \times \alpha_{\sort}(A)$. $\interpretation$ is a model of $(\sign,\basicR)$ iff $\alpha_\sign(f)$ is single-valued \eqref{eq:cb:sv} and total \eqref{eq:cb:tot}, i.e., a function.
\begin{equation}\label{eq:deterministic-total} 


}

\end{equation}
\end{exa}

\begin{exa}[KAT]\label{ex:kat}
  Let $\mathcal{P}$ be a set of predicate symbols $R\colon U \to \uno$ for $U\in \sort^\star$. Take $\bar{\mathcal{P}}\defeq\{\bar{R}\colon U \to \uno \mid R\in \mathcal{P}\}$ and $\sign\defeq \mathcal{P}\cup \bar{\mathcal{P}}$.
Let \(\basicR\) be the set of equalities below.
  \begin{equation}\label{eq:axioms-predicates}
    \!\!\!\!\!\!\!\!\! \scalebox{0.8}{
    %
}
}
  \end{equation}
  An interpretation $\interpretation$ of \((\sort,\sign)\) in \(\Rel\) is a set \(\alpha_\sort (A_i)\) for each $A_i\in \sort$ together with predicates \(R \subseteq X\times \uno \) and \(\bar{R} \subseteq X\times \uno\) for $X\defeq \alpha_\sort^\sharp(U)$.
  $\interpretation$ is a model of $(\sign, \basicR)$ iff, for all \(R\), \(\bar{R}\) is its set-theoretic complement: by \eqref{eq:covolution} and \eqref{eq:cb:covolution}, the equalities above assert that $R \cup \bar{R}=X$ and $R \cap \bar{R}=\emptyset$. 
  By Proposition \ref{prop:kc-rig laws}, $\KTCBI[U,\uno]$ carries a distributive lattice and even a Boolean algebra when defining $\lnot P$, for all $P\colon U \to \uno$ as follows.
    \begin{equation}\label{eq:neg}
        \scalebox{0.9}{$
     \neg R \defeq \bar{R} 
     \quad   
     \neg \top \defeq \bot 
     \quad 
     \neg(P \sqcup Q) \defeq \neg P \sqcap \neg Q 
     \quad
     \neg \bar{R} \defeq R
     \quad
     \neg \bot \defeq \top 
     \quad
     \neg(P \sqcap Q) \defeq \neg P \sqcup \neg Q
    $}
    \end{equation}
     Consider the set \(C\) that contains, for each \(P \colon U \to \uno \), the associated coreflexive $c(P)$ defined as in \eqref{eq:defc}. 
 Again coreflexives form a Boolean algebra  \((C,\sqcup, \dcomp, \lnot, \bot, \id{})\) where composition \(\dcomp\) acts as $\sqcap$: see Lemma \ref{lemma:coreflexive properties}.3. Moreover, by \Cref{cor:kleeneareka}, $\KTCBI[U,U]$ is a Kleene algebra.
  Thus \((\KTCBI[U,U], C, \sqcup, \dcomp, (\cdot)^{*}, \bot, \id{}, \lnot)\) is a \emph{Kleene algebra with tests}~\cite{kozen1997kleene}.
\end{exa}
Models enjoy a beautiful characterisation provided by Proposition \ref{funct:sem} below. Let $\KTCBI$ be the category having the same objects as $\KTCB$ and arrows $=_{\basicR}$-equivalence classes of arrows of $\KTCB$ ordered by $\precongB$. Since $\KTCB$ is a kc-rig category, then also $\KTCBI$ is so.
\begin{prop}\label{funct:sem}
	Let $(\sign, \basicR)$ be a kc tape theory and $H \colon \Cat{S} \to \Cat{C}$ a sesquistrict kc-rig category.  Models of $(\sign, \basicR)$ are in bijective correspondence with  morphisms of sesquistrict kc-rig categories from $\KTCBI$ to $H\colon \Cat{S} \to \Cat{C}$.
\end{prop} 
\begin{proof} %
First, we observe that there exists a kc-rig morphism $\eta \colon \KTCB \to \KTCBI$ defined as the identity on objects and mapping tapes $\t\colon P \to Q$ into $=_{\basicR}$-equivalence classes $[\t]\colon P \to Q$.
    
Let $\mathcal{I}=(\alpha_\sort, \alpha_\sign)$ be a model of $(\sign, \basicR)$ in $\Cat{S} \to \Cat{C}$ and $\TCBdsem{\cdot}_{\interpretation} \colon \KTCB \to \Cat{C}$ be the morphism induced by freeness of $\KTCB$. Define $\tilde{\alpha}_\sign^\sharp \colon \KTCBI \to \Cat{C}$ for all objects $P$ and $=_{\basicR}$-equivalence classes $[\t]\colon P \to Q$ as 
\[\tilde{\alpha}_\sign^\sharp(P)\defeq \TCBdsem{P}_{\interpretation} \qquad \tilde{\alpha}_\sign^\sharp([\t])\defeq \TCBdsem{\t}_{\interpretation}\text{.}\]
Since $(\alpha_\sort, \alpha_\sign)$ is a model, then $\TCBdsem{\cdot}_{\interpretation}$ preserves $\precongB$ and thus $\tilde{\alpha}_\sign^\sharp$ is well defined. Checking that $\tilde{\alpha}_\sign^\sharp$ is a kc-rig morphism is immediate from the fact that $\TCBdsem{\cdot}_{\interpretation}$ is a kc-rig morphism.
    
    Vice versa, from a morphism $\beta \colon \KTCBI \to \Cat{C}$ one can construct an interpretation  $\mathcal{I}$ of $(\sort, \sign)$ in $\Cat{S} \to \Cat{C}$ by precomposing first with $\eta$ and then with the trivial interpretation of  $(\sort, \sign)$ in $\sort \to \KTCB$. The unique sesquistrict kc-rig morphism induced by $\mathcal{I}$ is exactly $\eta; \beta$. Since $\eta;\beta$ factors through $\KTCBI$, it obviously preserves $\precongB$ and thus $\mathcal{I}$ is a model of $(\sign, \basicR)$. 
    To conclude that the correspondence is bijective, observe that $\TCBdsem{\cdot}_{\interpretation} = \eta ; \tilde{\alpha}_\sign^\sharp$.
\end{proof}

\section{The Kleene-Cartesian  Theory of Peano}\label{sec:peano} 
In Section \ref{ssec:funsem}, we introduced kc theories and presented several illustrative examples.
We now provide an additional example of a kc theory: Peano's axiomatisation of the natural numbers.
Although this is not a first-order theory \cite{harsanyi1983mathematics}, it can be concisely formulated as a kc tape theory.

As expected, the single-sorted signature contains a constant $

}
 $. Formally, we have $\sort\defeq \{A\}$ and $\sign \defeq \{ \; 
    %
}
 \colon A \to A  \}$.
An interpretation of $\sign$ in $\Rel$ consists of  a set $X$ (i.e., $\alpha_\sort(A)$), a relation $0 \subseteq 1 \times X$  (i.e., $\alpha_\sign( \, 
    %
$
    \caption{The Kleene-Cartesian theory of Peano.}
    \label{fig:tape-theory-natural-numbers}
\end{figure}

The set of axioms $\mathbb{P}$ consists of those in Figure \ref{fig:tape-theory-natural-numbers}.
From a universal-algebraic standpoint, the natural numbers form the smallest set $X$ such that $X$ is isomorphic to $X \piu \uno$. The two leftmost axioms in Figure~\ref{fig:tape-theory-natural-numbers} force $[s,0]  \defeq (s \piu 0) ; \codiag{X}$ to be an isomorphism of type $X \piu \uno \to X$: \eqref{ax:peanoT:iso:1} states that $\op{[s,0]}\dcomp [s,0] = \id{X}$, while axiom \eqref{ax:peanoT:iso:2} asserts the converse identity  $[s,0]\dcomp \op{[s,0]}= \id{X \piu \uno}$; the rightmost axiom expresses minimality: $X \subseteq 0 \dcomp s^* \colon 1\to X$. Consequently, an interpretation in $\Rel$ is a model of the theory $(\sign, \mathbb{P})$ precisely when $X$ is isomorphic to $X\piu 1$ and contained in $0;s^*$.
As expected, a model of $(\sign, \mathbb{P})$ in $\Rel$ is the usual set of natural numbers $\N$, equipped with the element zero $0\colon 1 \to \N$ and the successor function $s\colon \N \to \N$. We will shortly see that this is the unique model up to isomorphism.

First, we illustrate that $(\sign, \mathbb{P})$ is equivalent to Peano's axiomatisation of natural numbers. Possibly, the most interesting axiom is the principle of induction:~\eqref{ax:peano:ind} in Figure~\ref{fig:peano-theory-natural-numbers}. 
This follows easily from posetal uniformity and \eqref{ax:peanoT:ind}.

\begin{thm}[Principle of Induction]\label{thm:induction} Let $(\Pi, \mathbb{Q})$ be a kc theory, such that $\sign \subseteq \Pi$ and $\mathbb{P}\subseteq \mathbb{Q}$. For all $P\colon\uno \to A$ in $\Cat{KCT}_{\Pi,\mathbb{Q}}$, \eqref{ax:peano:ind} in Figure \ref{fig:peano-theory-natural-numbers} holds 
\end{thm}

\begin{proof} %
    Observe that the following holds:
    \[
        \scalebox{0.8}{

}
}.
    \]
    Thus, by~\eqref{ax:trace:tape:AU2} the inclusion below holds and the derivation concludes the proof.
    \[
        
    %
}
. \qedhere
    \]
\end{proof}
The other Peano's axioms state that $0$ is a natural number, $s$ is an injective function and that $0$ is \emph{not} the successor of any natural number. These are illustrated by means of tapes in Figure \ref{fig:peano-theory-natural-numbers}, where we use the characterisation of total, single-valued and injective relations provided by \Cref{lemma:cb:adjoints}. Observe that  \eqref{ax:peano:bottom} states that $\{x\in X \mid (x,0)\in s \} \subseteq \emptyset$.
\begin{prop}\label{lemma:Peano}
    The laws in Figures~\ref{fig:tape-theory-natural-numbers} and \ref{fig:peano-theory-natural-numbers} are equivalent. 
     In particular, the following equivalences hold:
     \begin{enumerate}
         \item $\eqref{ax:peanoT:iso:1} \iff \eqref{ax:peano:succ:sv}, \eqref{ax:peano:zero:sv}$
         \item $\eqref{ax:peanoT:iso:2} \iff \eqref{ax:peano:succ:tot}, \eqref{ax:peano:zero:tot}, \eqref{ax:peano:succ:inj}, \eqref{ax:peano:bottom}$
         \item $\eqref{ax:peanoT:ind} \iff \eqref{ax:peano:ind}$
     \end{enumerate}
\end{prop}
\begin{proof}
    See Appendix~\ref{app:Peano}.
\end{proof}

\begin{cor}  Any model of $(\sign, \mathbb{P})$ is isomorphic to the one on natural numbers.
\end{cor}
\begin{proof}
By \Cref{lemma:Peano} and the result by Dedekind in \cite{dedekind1888nature} that shows that any two models of Peano axioms are isomorphic. %
\end{proof}
\begin{figure}[t]
    \mylabel{ax:peano:zero:sv}{$0$-sv}
    \mylabel{ax:peano:zero:tot}{$0$-tot}
    \mylabel{ax:peano:succ:sv}{$s$-sv}
    \mylabel{ax:peano:succ:tot}{$s$-tot}
    \mylabel{ax:peano:succ:inj}{$s$-inj}
    \mylabel{ax:peano:bottom}{$\bot$}
    \mylabel{ax:peano:ind}{ind-princ}
    \[

            }
        \end{array}
    \]
    \caption{Peano's theory of the natural numbers.}
    \label{fig:peano-theory-natural-numbers}
\end{figure}

\subsection{First steps with tape arithmetic}
To give the reader a taste of how one can program with tapes, we now illustrate how to start to encode arithmetic within  $(\sign, \mathbb{P})$.
The tape for addition is illustrated below. %
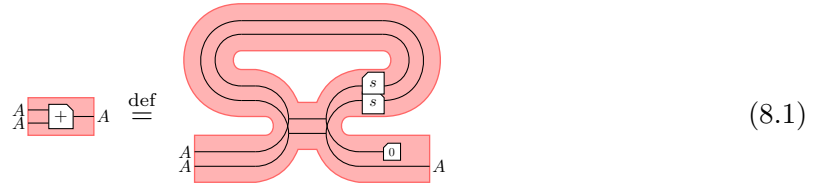
\begin{equation}\label{eq:add}
    
    \InputIfFileExists{addition/addCirc.tikz}{}{\input{./tikz/addition/addCirc.tikz}}
 \defeq 
    \InputIfFileExists{addition/additionDef.tikz}{}{\input{./tikz/addition/additionDef.tikz}}

\end{equation}
This can be thought of as a simple imperative program:
\begin{center}
\texttt{add(x,y) = while (x>0) \{ x:=x-1; y:=y+1 \}; return y}
\end{center}
The variable \texttt{x} corresponds to the top wire in \eqref{eq:add}, while \texttt{y} to the bottom one. At any iteration, the program checks whether 
 \texttt{x} is $0$, in which case it returns \texttt{y}, or the successor of some number, in which case \texttt{x} takes such number, while \texttt{y} takes its own successor.
 Note that \eqref{eq:add} exploits both $
    \begin{tikzpicture}
	\begin{pgfonlayer}{nodelayer}
		\node [style=bbox] (69) at (0, 0) {$s$};
		\node [style=none] (70) at (-1.5, 0) {};
		\node [style=none] (71) at (1.5, 0) {};
		\node [style=label] (72) at (2, 0) {$A$};
		\node [style=label] (73) at (-2, 0) {$A$};
	\end{pgfonlayer}
	\begin{pgfonlayer}{edgelayer}
		\draw (70.center) to (69);
		\draw (69) to (71.center);
	\end{pgfonlayer}
\end{tikzpicture}
}
$ and $\op{(
    \begin{tikzpicture}
	\begin{pgfonlayer}{nodelayer}
		\node [style=bbox] (69) at (0, 0) {$s$};
		\node [style=none] (70) at (-1.5, 0) {};
		\node [style=none] (71) at (1.5, 0) {};
		\node [style=label] (72) at (2, 0) {$A$};
		\node [style=label] (73) at (-2, 0) {$A$};
	\end{pgfonlayer}
	\begin{pgfonlayer}{edgelayer}
		\draw (70.center) to (69);
		\draw (69) to (71.center);
	\end{pgfonlayer}
\end{tikzpicture}
}
)}$. The latter acts at the same time as a test ($\texttt{x>0}$) and as an assignment (\texttt{x:=x-1}), since its unique (modulo iso) interpretation is the relation $\{(x+1,x) \mid x\in \N\}$.

One can easily prove that~\eqref{eq:add} satisfies the usual inductive definition of addition in Peano's arithmetics.
\begin{lem}\label{lemma:additionPeano}
    The following hold in $\KTCBP$:
    \begin{enumerate}
        \item $
    \InputIfFileExists{addition/addZero.tikz}{}{\input{./tikz/addition/addZero.tikz}}
 = 
    \InputIfFileExists{addition/addZeroRhs.tikz}{}{\input{./tikz/addition/addZeroRhs.tikz}}
 \qquad \textnormal{(\texttt{ add(0,y) = y })}$ 
        \item $
    \InputIfFileExists{addition/addSuccLhs.tikz}{}{\input{./tikz/addition/addSuccLhs.tikz}}
 = 
    \InputIfFileExists{addition/addSuccRhs.tikz}{}{\input{./tikz/addition/addSuccRhs.tikz}}
 \qquad \textnormal{(\texttt{ add(succ(x),y) = succ(add(x,y)) })}$
    \end{enumerate}
\end{lem}
\begin{proof}
    First, observe that by \eqref{eq:starequality} 
    the following holds in $\KTCBP$:
    \begin{equation}\label{eq:add-fixpoint}
        
    \InputIfFileExists{addition/additionDef.tikz}{}{\input{./tikz/addition/additionDef.tikz}}
 =_{\mathbb{P}} 
    \InputIfFileExists{addition/addStar.tikz}{}{\input{./tikz/addition/addStar.tikz}}
. 
    \end{equation}
    Then, for $(1)$ the following holds:
    \[
\begin{array}{@{}c@{}c@{}c@{}c@{}c@{}c@{}c}
      
    \InputIfFileExists{addition/addZero.tikz}{}{\input{./tikz/addition/addZero.tikz}}
 &\stackrel{\substack{\eqref{eq:add}\\\eqref{eq:add-fixpoint} }}{=_{\mathbb{P}}}& 
    \InputIfFileExists{addition/partialEval/zero/step1.tikz}{}{\input{./tikz/addition/partialEval/zero/step1.tikz}}
 & \stackrel{\substack{\eqref{ax:diagnat} \\ \eqref{ax:codiagnat}}}{=_{\mathbb{P}}} & 
    \InputIfFileExists{addition/partialEval/zero/step2.tikz}{}{\input{./tikz/addition/partialEval/zero/step2.tikz}}
 & \stackrel{\eqref{ax:peano:bottom}}{=_{\mathbb{P}}} & 
    \InputIfFileExists{addition/partialEval/zero/step3.tikz}{}{\input{./tikz/addition/partialEval/zero/step3.tikz}}
 \\
                                 &\stackrel{\eqref{ax:cobangnat}}{=_{\mathbb{P}}} & 
    \InputIfFileExists{addition/partialEval/zero/step4.tikz}{}{\input{./tikz/addition/partialEval/zero/step4.tikz}}
 &\stackrel{\substack{\eqref{ax:diagun} \\ \eqref{ax:codiagun}}}{=_{\mathbb{P}}} & 
    \InputIfFileExists{addition/partialEval/zero/step5.tikz}{}{\input{./tikz/addition/partialEval/zero/step5.tikz}}
 & \stackrel{\eqref{ax:peano:zero:tot}}{=_{\mathbb{P}}} & 
    \InputIfFileExists{addition/partialEval/zero/step6.tikz}{}{\input{./tikz/addition/partialEval/zero/step6.tikz}}
.
    \end{array}
    \]
    And for $(2)$ the following holds:
    \[
    \begin{array}{@{}c@{}c@{}c@{}c@{}c@{}c@{}c}
        
    \InputIfFileExists{addition/addSuccLhs.tikz}{}{\input{./tikz/addition/addSuccLhs.tikz}}
 &\stackrel{\substack{\eqref{eq:add}\\\eqref{eq:add-fixpoint} }}{=_{\mathbb{P}}} & 
    \InputIfFileExists{addition/partialEval/succ/step1.tikz}{}{\input{./tikz/addition/partialEval/succ/step1.tikz}}
 &\stackrel{\substack{\eqref{ax:diagnat} \\ \eqref{ax:codiagnat}}}{=_{\mathbb{P}}} & 
    \InputIfFileExists{addition/partialEval/succ/step2.tikz}{}{\input{./tikz/addition/partialEval/succ/step2.tikz}}
 & \stackrel{\substack{\eqref{ax:peano:succ:inj} \\ \eqref{ax:peano:succ:tot}}}{=_{\mathbb{P}}} & 
    \InputIfFileExists{addition/partialEval/succ/step3.tikz}{}{\input{./tikz/addition/partialEval/succ/step3.tikz}}
 \\[12pt]
                                      &\stackrel{\eqref{ax:peano:bottom}}{=_{\mathbb{P}}} & 
    \InputIfFileExists{addition/partialEval/succ/step4.tikz}{}{\input{./tikz/addition/partialEval/succ/step4.tikz}}
 &\stackrel{\substack{\eqref{ax:diagun} \\ \eqref{ax:codiagun}}}{=_{\mathbb{P}}}   & 
    \InputIfFileExists{addition/partialEval/succ/step5.tikz}{}{\input{./tikz/addition/partialEval/succ/step5.tikz}}
 & \stackrel{\text{(sliding)}}{=_{\mathbb{P}}}  & 
    \InputIfFileExists{addition/partialEval/succ/step6.tikz}{}{\input{./tikz/addition/partialEval/succ/step6.tikz}}
 \\[12pt]
                                      & \stackrel{\substack{\eqref{ax:codiagnat} \\ \eqref{ax:diagnat}}}{=_{\mathbb{P}}} & 
    \InputIfFileExists{addition/partialEval/succ/step7.tikz}{}{\input{./tikz/addition/partialEval/succ/step7.tikz}}
 & \stackrel{\eqref{eq:add}}{=_{\mathbb{P}}} & 
    \InputIfFileExists{addition/addSuccRhs.tikz}{}{\input{./tikz/addition/addSuccRhs.tikz}}
. \qedhere
    \end{array}
    \]
\end{proof}
While it is straightforward that $
    \InputIfFileExists{addition/addCirc.tikz}{}{\input{./tikz/addition/addCirc.tikz}}
$ terminates on all possible inputs, it is interesting to see how this can be proved within the kc theory $(\Sigma,\mathbb{P})$.

\begin{lem}
    The tape $
    \InputIfFileExists{addition/addCirc.tikz}{}{\input{./tikz/addition/addCirc.tikz}}
$ is total, i.e. $
}
$. %
\end{lem}
\begin{proof}
    First observe that the following holds:
    \begin{align*}
        \!\!
    \InputIfFileExists{addition/terminates/derivation/step1.tikz}{}{\input{./tikz/addition/terminates/derivation/step1.tikz}}
 
        \!\!\stackrel{\eqref{ax:dischargeradj2}}{\leq_{\mathbb{P}}}\!\!
        
    \InputIfFileExists{addition/terminates/derivation/step2.tikz}{}{\input{./tikz/addition/terminates/derivation/step2.tikz}}
 
        \!\!\stackrel{\substack{\eqref{ax:codiagnat} \\ \eqref{ax:diagnat}}}{=_{\mathbb{P}}}\!\!
        
    \InputIfFileExists{addition/terminates/derivation/step3.tikz}{}{\input{./tikz/addition/terminates/derivation/step3.tikz}}
 
        \!\!\stackrel{\eqref{ax:dischargeradj1}}{\leq_{\mathbb{P}}}\!\! 
        
    \InputIfFileExists{addition/terminates/derivation/step4.tikz}{}{\input{./tikz/addition/terminates/derivation/step4.tikz}}
 
        \!\!\stackrel{ \eqref{ax:peano:succ:tot}}{\leq_{\mathbb{P}}}\!\!
        
    \InputIfFileExists{addition/terminates/derivation/step5.tikz}{}{\input{./tikz/addition/terminates/derivation/step5.tikz}}
.
    \end{align*}
    Then, by~\eqref{ax:trace:tape:AU1}, the inequality below holds and the derivation concludes the proof. %
    \[
        
    }
  \stackrel{\eqref{ax:peanoT:ind}}{\leq_{\mathbb{P}}} 
    \InputIfFileExists{addition/terminates/derivation/unif_lhs.tikz}{}{\input{./tikz/addition/terminates/derivation/unif_lhs.tikz}}
 \leq_{\mathbb{P}} 
    \InputIfFileExists{addition/terminates/derivation/unif_rhs.tikz}{}{\input{./tikz/addition/terminates/derivation/unif_rhs.tikz}}
 \stackrel{\eqref{eq:add}}{=_{\mathbb{P}}} 

}
        \qedhere
    \]
\end{proof} 

\begin{rem}
Beyond natural numbers, one can analogously define kc-rig theories for other algebraic data types, such as lists. Interestingly, the combined use of the opposite operation \( \op{(\cdot)} \) and the diagonal \( \diag{X} \colon X \to X \piu X \) makes it possible to express \emph{pattern matching}. For instance, consider the arrow
\[
\diag{A\otimes A} ; (\op{(
    %
}
\otimes \id{A})}) \colon A \otimes A \to (A \otimes A) \oplus A
\]
occurring in \eqref{eq:add}. At this stage, the reader might wonder whether tape diagrams could provide a suitable setting for functional programming. We argue that this is not the case in the current formulation, as it lacks explicit linguistic constructs for handling functional types. Rather, tapes provide a particularly well-suited setting for imperative programs and program logics, as illustrated in the next section.
\end{rem}

\section{A Diagrammatic View of Imperative Programming}\label{sec:hoare} %

In this section, we begin by showing how imperative programs can be represented as tape diagrams. This encoding establishes the foundation for our main objective: demonstrating that the structure of kc-rig categories offers a natural and expressive framework for reasoning about imperative programs. In particular, we show that the axioms of kc-rig categories (1) make it possible to derive sophisticated program equivalences involving nontrivial interactions between data and control flow; (2) yield a proof system that is at least as powerful as Hoare logic; and (3) support relational reasoning about pairs of programs, in the style of relational Hoare logic.

\subsection{Programs as tape diagrams} %
For the sake of generality, we avoid fixing basic types and operations and, rather, we work parametrically with respect to a triple $(\sort,\mathcal{F},\mathcal{P})$: $\sort$ is a set of sorts, representing basic types; $\mathcal{F}$ is a set of function symbols, equipped with an arity in $\sort^*$ and a coarity in $\sort$;  $\mathcal{P}$ is a set of predicate symbols equipped just with an arity in $\sort^*$. The coarity of predicates is fixed to be $1$.

We consider the monoidal signature $\sign \defeq \mathcal{F} \cup \mathcal{P} \cup \bar{\mathcal{P}}$ where $\bar{\mathcal{P}}$ is as in Example \ref{ex:kat}. %
The set of axioms $\basicR$ contains, for all $f\colon U \to A$ in $\mathcal{F}$, the axioms in \eqref{eq:deterministic-total} and, for each $R\colon U \to \uno$ in $\mathcal{P}$, those in \eqref{eq:axioms-predicates}. These force any model of the kc-theory $(\sign, \basicR)$ in $\Rel$ to interpret symbols in $\mathcal{F}$ as functions and $\bar{P}$ as the complement of $P$.  One may add to $\basicR$ other axioms, e.g., those of $\mathbb{P}$ in Section \ref{sec:peano}.

We consider terms generated by the following grammar  
\[
\begin{array}{rcl}
e &\Coloneqq& x \mid f(e_{1}, \ldots, e_{n})\\
P &\Coloneqq& R(e_1, \dots, e_n) \mid \bar{R}(e_1, \dots, e_n) \mid \top \mid \bot  \mid P \lor P \mid P \land P  \\
C &\Coloneqq& \mathsf{abort} \mid \mathsf{skip} \mid \mathop{\mathsf{if}} P \mathop{\mathsf{then}} C \mathop{\mathsf{else}} C \mid \mathop{\mathsf{while}} P \mathop{\mathsf{do}} C \mid C ; C \mid x \coloneqq e
\end{array}
\]
where $f\in \mathcal{F}$, $R\in \mathcal{P}$ and $x$ is taken from a fixed set of variables.  As usual, $e$ are expressions, $P$ are predicates, and $C$ are commands. Negation of predicates can be expressed as in \eqref{eq:neg}.
In order to encode terms into diagrams, we need to make copying and discarding of variables explicit; we thus define a simple type system with judgement of the form
\[\Gamma \vdash e \colon A \qquad \Gamma \vdash P \colon 1 \qquad \Gamma \vdash C\]
where $A$ is a sort in $\sort$ and $\Gamma$ is a \emph{typing context}, i.e., an ordered sequence $x_1\colon A_1, \dots x_n \colon A_n$, where all the $x_i$ are distinct variables and $A_i \in \sort$.
The type system is  in Table \ref{tab:Hoaretypes} where $\Gamma$, $\Gamma'$ and $\Delta'$ stand for arbitrary typing contexts. In particular, the notation $\Gamma = \Gamma', x \colon A, \Delta'$, appearing in the premises of the rules \textsc{(VAR)} and \textsc{(ASSN)}, means that there exist contexts $\Gamma'$ and $\Delta'$ so that $\Gamma$ is the concatenation of $\Gamma'$, $x \colon A$ and $\Delta'$.

\begin{table}[t]
{\scriptsize	
  \begin{equation*}
\begin{array}{c@{\qquad}c@{\qquad}c@{\qquad}}
    \multicolumn{3}{c}{
      \inferrule*[right=(var)]{\Gamma = \Gamma', x : A, \Delta'}{\Gamma \vdash x \colon A} \qquad \qquad \inferrule*[right=(op)]{\Gamma \vdash e_{i} \colon A_{i} \quad f \colon A_{1} \per \cdots \per A_{n} \to A}{\Gamma \vdash f(e_{1}, \ldots, e_{n}) \colon A}
    }\\
    \hline
    \inferrule*[right=(top)]{ }{\Gamma \vdash \top \colon 1} &  \inferrule*[right=(and)]{\Gamma \vdash P \colon 1 \quad \Gamma \vdash Q \colon 1}{\Gamma \vdash (P \land Q) \colon 1} & \inferrule*[right=($\bar{\textnormal{R}}$)]{\Gamma\vdash e_i\colon A_i \quad R \colon A_1 \per \dots \per A_n \to 1}{\Gamma \vdash \bar{R}(e_1, \dots e_n) \colon 1} \\ 
    \inferrule*[right=(bot)]{ }{\Gamma \vdash \bot \colon 1} & \inferrule*[right=(or)]{\Gamma \vdash P \colon 1 \quad \Gamma \vdash Q \colon 1}{\Gamma \vdash (P \lor Q) \colon 1} & \inferrule*[right=(R)]{\Gamma\vdash e_i\colon A_i \quad R \colon A_1 \per \dots \per A_n \to 1}{\Gamma \vdash R(e_1, \dots e_n) \colon 1}\\ 
    \hline
    \inferrule*[right=(abort)]{ }{\Gamma \vdash \mathsf{abort} }&
    \inferrule*[right=(skip)]{ }{\Gamma \vdash \mathsf{skip} }&
    \inferrule*[right=(assn)]{\Gamma = \Gamma',x \colon A,\Delta' \qquad \Gamma \vdash e \colon A}{\Gamma \vdash x \coloneqq e }\\ 
    \inferrule*[right=(;)]{\Gamma \vdash C   \quad \Gamma \vdash D }{\Gamma \vdash C ; D }&
    \inferrule*[right=(while)]{\Gamma \vdash P \colon 1 \quad \Gamma \vdash C  }{ \Gamma \vdash \mathop{\mathsf{while}} P \mathop{\mathsf{do}} C } &
    \inferrule*[right=(if)]{\Gamma \vdash P \colon 1 \quad \Gamma \vdash C   \quad \Gamma \vdash D  }{\Gamma \vdash \mathop{\mathsf{if}} P \mathop{\mathsf{then}} C \mathop{\mathsf{else}} D }
\end{array}
\end{equation*}}
\caption{Type system for expressions, predicates and commands.}\label{tab:Hoaretypes}
\end{table}

\begin{table}[t]
{\scriptsize	
$$\begin{array}{rcl}
\encoding{\Gamma \vdash x \colon A} &\defeq& \discharger{\encoding{\Gamma'}} \per \id{A} \per \discharger{\encoding{\Delta'}}\\
\encoding{\Gamma \vdash f(e_{1}, \ldots, e_{n}) \colon A} &\defeq& \copier{\encoding{\Gamma}}^n ; (\encoding{\Gamma \vdash e_1} \per \dots \per \encoding{\Gamma \vdash e_n} ) ; f\\
\hline
\encoding{\Gamma \vdash R(e_1, \dots e_n)\colon 1} &\defeq& \copier{\encoding{\Gamma}}^n ; (\encoding{\Gamma \vdash e_1 \colon 1} \per \dots \per \encoding{\Gamma \vdash e_n \colon 1}) ; R \\
\encoding{\Gamma \vdash \bar{R}(e_1, \dots e_n)\colon 1} &\defeq& \copier{\encoding{\Gamma}}^n ; (\encoding{\Gamma \vdash e_1 \colon 1} \per \dots \per \encoding{\Gamma \vdash e_n \colon 1}) ; \bar{R} \\
\encoding{\Gamma \vdash \top\colon 1} &\defeq& \discharger{\encoding{\Gamma}} \\
\encoding{\Gamma \vdash \bot \colon 1} &\defeq& \bang{\encoding{\Gamma}} \dcomp \cobang{1}\\
\encoding{\Gamma \vdash P \lor  Q \colon 1} &\defeq& \diag{\encoding{\Gamma}} ; (\encoding{\Gamma \vdash P\colon 1} \piu \encoding{\Gamma \vdash Q\colon 1} ) ; \codiag{1}\\ %
\encoding{\Gamma \vdash P \land Q \colon 1} &\defeq& \copier{\encoding{\Gamma}} ;  (\encoding{\Gamma \vdash P\colon 1} \per \encoding{\Gamma \vdash Q\colon 1} )  \\
\end{array}$$}
\caption{Encoding of expressions, predicates and commands into diagrams. }\label{tab:encoding}
\end{table}

The encoding $\encoding{\cdot}$, defined inductively on the typing rules, maps well-typed expressions $\Gamma \vdash e \colon A$,  predicates $\Gamma \vdash P \colon 1$ and commands $\Gamma \vdash C $  into, respectively,  diagrams of the following types
\[\encoding{\Gamma} \to A \qquad \encoding{\Gamma} \to 1 \qquad \encoding{\Gamma} \to \encoding{\Gamma}\] 
where for $\Gamma =x_1\colon A_1, \dots , x_n \colon A_n$, we fix $\encoding{\Gamma} \defeq A_1 \per \dots  \per A_n$. 

The definition of $\encoding{\cdot}$, compactly summarised in Table \ref{tab:encoding}, is illustrated thoroughly below in diagrammatic form. For the sake of readability, we label wires directly with the typing contexts rather than their encodings.

\paragraph{Expressions} For the case of variables $\Gamma \vdash x \colon A$, the context is $\Gamma = \Gamma', x \colon A, \Delta'$, according to the typing rule \textsc{(VAR)}. In the encoding, the part of the context that is not $x \colon A$ is discarded using $\discharger{\encoding{\Gamma'}}$ and $\discharger{\encoding{\Delta'}}$, as shown below.
\[
\encoding{\Gamma \vdash x \colon A} 
\defeq


}

\]
For operations $\Gamma \vdash f(e_1,\ldots,e_n) \colon A$, the context $\Gamma$ is shared amongst the arguments $e_i$, as specified by the typing rule \textsc{(OP)}. In the encoding we use $\copier{\encoding{\Gamma}}^n \colon \encoding{\Gamma} \to \encoding{\Gamma}^n$, to copy $n$ times the content of the variables in $\Gamma$. Formally, $\copier{U}^n \colon U \to U^n$ is defined as $\copier{U}^0 \defeq \discharger{U}$ and $\copier{U}^{n+1} \defeq \copier{U} ; (\copier{U}^n \per \id{U})$.
\[
\encoding{\Gamma \vdash f(e_{1}, \ldots, e_{n}) \colon A}
\defeq

    %
}

\]

\paragraph{Predicates} Analogously to the case of operations, predicate symbols $\Gamma \vdash R(e_1, \ldots, e_n) \colon \uno$ and their complements $\Gamma \vdash \bar{R}(e_1, \ldots, e_n) \colon \uno$ share the context across the arguments $e_i$, in accordance with the rules \textsc{(R)} and \textsc{$(\bar{\text{R}})$}.
\[
  \encoding{\Gamma \vdash R(e_1, \dots e_n)\colon 1} \defeq 
    \InputIfFileExists{hoare/enc/pred/R.tikz}{}{\input{./tikz/hoare/enc/pred/R.tikz}}

  \qquad
  \encoding{\Gamma \vdash \bar{R}(e_1, \dots e_n)\colon 1} \defeq 
    \InputIfFileExists{hoare/enc/pred/notR.tikz}{}{\input{./tikz/hoare/enc/pred/notR.tikz}}

\]
The remaining Boolean operations and constants are encoded as in Example~\ref{ex:kat}.
\[
\begin{array}{r@{\,}c@{\,}l@{\qquad\qquad}r@{\,}c@{\,}l}
  \encoding{\Gamma \vdash \top\colon 1} &\defeq& 
    \InputIfFileExists{hoare/enc/pred/top.tikz}{}{\input{./tikz/hoare/enc/pred/top.tikz}}

  &
  \encoding{\Gamma \vdash \bot \colon 1} &\defeq& 
    \InputIfFileExists{hoare/enc/pred/bot.tikz}{}{\input{./tikz/hoare/enc/pred/bot.tikz}}

  \\[10pt]
  \encoding{\Gamma \vdash P \land  Q \colon 1} &\defeq& 
    \InputIfFileExists{hoare/enc/pred/and.tikz}{}{\input{./tikz/hoare/enc/pred/and.tikz}}

  &
  \encoding{\Gamma \vdash P \lor  Q \colon 1} &\defeq& 
    \InputIfFileExists{hoare/enc/pred/or.tikz}{}{\input{./tikz/hoare/enc/pred/or.tikz}}

\end{array}
\]

\paragraph{Commands} 
The command $\Gamma \vdash \mathsf{skip}$ does not perform any computation, leaving the program state unchanged. Therefore it is represented as the identity of the context $\Gamma$.
\[
\encoding{\Gamma \vdash  \mathsf{skip} } \defeq 
    \InputIfFileExists{hoare/enc/cmd/skip.tikz}{}{\input{./tikz/hoare/enc/cmd/skip.tikz}}

\]
The command $\Gamma \vdash \mathsf{abort}$ terminates abruptly, preventing any further computation. Therefore it is represented as $\bot$, expressing that no initial state leads to any final state.
\[
\encoding{\Gamma \vdash \mathsf{abort}}  \defeq 
    \InputIfFileExists{hoare/enc/cmd/abort.tikz}{}{\input{./tikz/hoare/enc/cmd/abort.tikz}}

\]
As expected, the encoding of composition $\Gamma \vdash C ; D$ is given by the sequential composition of the encodings of $\Gamma \vdash C$ and $\Gamma \vdash D$.
\[
\encoding{\Gamma \vdash  C;D }  \defeq 
    \InputIfFileExists{hoare/enc/cmd/seq.tikz}{}{\input{./tikz/hoare/enc/cmd/seq.tikz}}

\]
For conditional branching and iteration, we encode guards $\Gamma \vdash P \colon 1$ as their corresponding coreflexives $c(\encoding{P}) \colon \encoding{\Gamma} \to \encoding{\Gamma}$. 
The encodings of $\Gamma \vdash \mathop{\mathsf{if}} P \mathop{\mathsf{then}} C \mathop{\mathsf{else}} D$ and $\Gamma \vdash \mathop{\mathsf{while}} P \mathop{\mathsf{do}} C$ are pretty standard (see e.g.~\cite{kozen1997kleene}) and exploit, respectively, the join $\sqcup$ and the Kleene star $\kstar{(\cdot)}$, as shown below.
\[
\begin{array}{r@{\,}c@{}l}
  \encoding{\Gamma \vdash \mathop{\mathsf{if}} P \mathop{\mathsf{then}} C \mathop{\mathsf{else}} D}  & \defeq & 
    \InputIfFileExists{hoare/enc/cmd/if.tikz}{}{\input{./tikz/hoare/enc/cmd/if.tikz}}

  \\[20pt]
  \encoding{\Gamma \vdash\mathop{\mathsf{while}} P \mathop{\mathsf{do}} C}  & \defeq & 
    \InputIfFileExists{hoare/enc/cmd/while.tikz}{}{\input{./tikz/hoare/enc/cmd/while.tikz}}

\end{array}
\]
Finally, to encode the assignment $\Gamma \vdash x \coloneqq e$, we exploit the structure of Cartesian bicategories to correctly model data flow. The typing rule \textsc{(ASSN)} requires the context to be of the form $\Gamma = \Gamma', x \colon A, \Delta'$. Then, $\copier{\encoding{\Gamma'}}$ and $\copier{\encoding{\Delta'}}$ are used to copy the parts of the context that are not affected by the assignment, as illustrated below.
\[
\encoding{\Gamma \vdash x \coloneqq e}  \defeq 
    \InputIfFileExists{hoare/enc/cmd/asgn.tikz}{}{\input{./tikz/hoare/enc/cmd/asgn.tikz}}

\]

\begin{rem}
Recall that any model $\interpretation$ of $(\sign, \basicR)$ in $\Rel$ assigns to each basic sort $A \in \sort$ a set $\dsem{A}_\interpretation$, and interprets a tensor $A_1 \otimes \dots \otimes A_n$ as the Cartesian product $\dsem{A_1}_\interpretation \times \dots \times \dsem{A_n}_\interpretation$. In particular, this coincides with $\dsem{\encoding{\Gamma}}_\interpretation$ for a context $\Gamma = x_1\colon A_1, \dots, x_n \colon A_n$.

Operationally, $\dsem{\encoding{\Gamma}}_\interpretation$ can be viewed as the space of all possible \emph{states}, namely tuples $(a_1, \dots, a_n)$ such that $a_i \in \dsem{A_i}_\interpretation$. Since $\encoding{\Gamma \vdash C}$ is a tape of type $\encoding{\Gamma} \to \encoding{\Gamma}$, its semantics
\(
\dsem{\encoding{\Gamma \vdash C}}_\interpretation \colon \dsem{\encoding{\Gamma}}_\interpretation \to \dsem{\encoding{\Gamma}}_\interpretation
\)
 is a relation on the state space $\dsem{\encoding{\Gamma}}_\interpretation$. This relation provides the \emph{extensional semantics} of the program $C$, namely the set of all pairs of states $(i,o)$ such that executing $C$ from state $i$ may lead to state $o$.

Now, for all commands $C$ and $D$, by the definition of models in Section~\ref{ssec:funsem}, we have that
\begin{center}
if $\encoding{\Gamma \vdash C} =_{\basicR} \encoding{\Gamma \vdash D}$, then $\dsem{\encoding{\Gamma \vdash C}}_\interpretation = \dsem{\encoding{\Gamma \vdash D}}_\interpretation$
\end{center}
for any model $\interpretation$. Therefore, if two commands can be proved equal via their tape encodings, then they have the same extensional semantics, i.e., they are \emph{extensionally equivalent}. This establishes the soundness of our approach. The converse implication---completeness---does not hold in general.
\end{rem}

\begin{rem}
Our encoding is reminiscent of the standard translation of imperative programs into Kleene Algebra with Tests (KAT) \cite{kozen1997kleene}. In that setting, assignments are typically interpreted as atomic symbols drawn from a fixed set of generators, whereas tape diagrams allow us to represent them explicitly. As a consequence, the axioms of kc-rig categories make it possible to establish program equivalences that are \emph{not} directly provable in KAT without introducing additional laws (see, e.g., \cite{angus2001kleene}). We present several examples below.
\end{rem}

\begin{exa}\label{ex:program equivalence 1}
  Let $\Gamma=x:A,y:A,z:A$, then $\encoding{\cdot}$ maps the programs $\Gamma \vdash x \coloneqq z ; y\coloneqq z$ and $\Gamma \vdash y \coloneqq z ; x \coloneqq z$ into the same tape.
  \[ \encoding{\Gamma \vdash x \coloneqq z ; y\coloneqq z} =\!\! \scalebox{0.8}{

}
} \!\!= \encoding{\Gamma \vdash y \coloneqq z ; x \coloneqq z}.\]
\end{exa}
\begin{exa} Consider a signature containing a unary function symbol $(-)+1$ and a unary predicate symbol $(-)=0$, written in infix notation. 
Let $\Gamma = x \colon A, y \colon A$ and consider the following program 
  \[
  \Gamma \vdash x \coloneqq x+1 ; \mathop{\mathsf{if}} y = 0 \mathop{\mathsf{then}} y \coloneqq y+1 \mathop{\mathsf{else}} \mathsf{skip}.
  \]
  Observe that the assignment  on $x$ commutes with the conditional branching, since they operate on disjoint sets of variables. We show this via the encoding, as follows:
  \[

  \]
\end{exa}

\begin{exa}
Let $\Gamma = x \colon A, y \colon A$ and consider the following program:
  \[
  \Gamma \vdash \mathop{\mathsf{while}} x = 0 \mathop{\mathsf{do}} y \coloneqq y + 1.
  \]
  If the guard holds initially, the loop does not terminate, since the assignment to $y$ never affects the test on $x$. Conversely, if the guard is false from the beginning, the program terminates immediately without modifying the state. In other words, the program is equivalent to:
  \[
    \Gamma \vdash \mathop{\mathsf{if}} x = 0 \mathop{\mathsf{then}} \mathsf{abort} \mathop{\mathsf{else}} \mathsf{skip}.
  \]
  We establish this equivalence using the encoding, relying on the monoidal product $\per$, which makes it possible to reason about data flow--and in particular about predicates and commands acting on disjoint sets of variables.

  First, to simplify calculations, it is convenient to note the following equivalences:
  \[
  c(\encoding{\Gamma \vdash x = 0}) 
  = 
  c(\!

}
.
  \end{split}
  \end{equation}
  Then, observe that the following holds:
  \[
  \begin{array}{cccc}
                                                & \encoding{\Gamma \vdash \mathop{\mathsf{while}} x = 0 \mathop{\mathsf{do}} y \coloneqq y + 1 } & = & 
    \InputIfFileExists{while/step1.tikz}{}{\input{./tikz/while/step1.tikz}}
 \\[20pt]
  \stackrel{\eqref{eq:starequality}}{=_{\basicR}}    & 
    \InputIfFileExists{while/step2.tikz}{}{\input{./tikz/while/step2.tikz}}
 & \stackrel{\eqref{ax:monoid:nat:copy}}{=_{\basicR}} & 
    \InputIfFileExists{while/step3.tikz}{}{\input{./tikz/while/step3.tikz}}
 \\[30pt]
  \stackrel{\eqref{eq:while-ex-bot}}{=_{\basicR}} & 
    \InputIfFileExists{while/step4.tikz}{}{\input{./tikz/while/step4.tikz}}
 & \stackrel{\eqref{eq:cmon enrichment}}{=_{\basicR}} & 
    \InputIfFileExists{while/step5.tikz}{}{\input{./tikz/while/step5.tikz}}
 \\[28pt]
  \stackrel{\eqref{ax:comonoid:nat:discard}}{=_{\basicR}} & 
    \InputIfFileExists{while/step6.tikz}{}{\input{./tikz/while/step6.tikz}}
 & =_{\phantom{{\basicR}}} & \encoding{\Gamma \vdash \mathop{\mathsf{if}} x = 0 \mathop{\mathsf{then}} \mathsf{abort} \mathop{\mathsf{else}} \mathsf{skip}}.
  \end{array}
  \]
\end{exa}

\subsection{Diagrammatic Hoare logic} %

Hoare logic \cite{hoare1969axiomatic} is one of the most influential languages to reason about imperative programs. Its rules --in the version appearing in \cite{winskel1993formal}-- are in \Cref{fig:hoare-rules}. 
In partial correctness, a triple \(\{P\}C\{Q\}\) asserts that if, starting from a state satisfying the precondition $P$, the execution of $C$ terminates, then the resulting state satisfies the postcondition $Q$. We can express Hoare triples as inequalities \(\op{\encoding{P}} ; \encoding{C} \leq_{\basicR} \op{\encoding{Q}}\), that, in form of diagrams appear as follows:
\[
    \begin{tikzpicture}
	\begin{pgfonlayer}{nodelayer}
		\node [style=none] (138) at (6.25, 1.25) {};
		\node [style=none] (139) at (6.25, -1.25) {};
		\node [style=none] (144) at (7.75, 0) {};
		\node [style=none] (148) at (12.25, 1.25) {};
		\node [style=none] (149) at (12.25, -1.25) {};
		\node [style=none] (159) at (12.25, 0) {};
		\node [style=label] (160) at (13, 0) {$\Gamma$};
		\node [style=bboxOp] (165) at (7.75, 0) {$\encoding{P}$};
		\node [style=bbox] (166) at (10.5, 0) {$\encoding{C}$};
		\node [style=label] (167) at (5.5, 0) {$\phantom{\Gamma}$};
	\end{pgfonlayer}
	\begin{pgfonlayer}{edgelayer}
		\draw [tape] (149.center)
			 to (139.center)
			 to (138.center)
			 to (148.center)
			 to cycle;
		\draw (165) to (166);
		\draw (166) to (159.center);
	\end{pgfonlayer}
\end{tikzpicture}
}
 \leq_{\basicR} 
    \begin{tikzpicture}
	\begin{pgfonlayer}{nodelayer}
		\node [style=none] (167) at (6, 1.25) {};
		\node [style=none] (168) at (6, -1.25) {};
		\node [style=none] (169) at (10, 0) {};
		\node [style=none] (170) at (12, 1.25) {};
		\node [style=none] (171) at (12, -1.25) {};
		\node [style=none] (172) at (12, 0) {};
		\node [style=label] (173) at (12.75, 0) {$\Gamma$};
		\node [style=bboxOp] (174) at (10, 0) {$\encoding{Q}$};
		\node [style=label] (175) at (5.25, 0) {$\phantom{\Gamma}$};
	\end{pgfonlayer}
	\begin{pgfonlayer}{edgelayer}
		\draw [tape] (171.center)
			 to (168.center)
			 to (167.center)
			 to (170.center)
			 to cycle;
		\draw (174) to (172.center);
	\end{pgfonlayer}
\end{tikzpicture}
}
.\]%
Here we prove that  if a triple \(\{P\}C\{Q\}\) is derivable within the Hoare logic, then it follows from the rules of kc-rig categories. We first need the following lemma illustrating how substitutions are handled by the encoding.
\begin{lem}\label{lemma:encodingsubstPed}
Let $\Gamma ' = \Gamma, x\colon A, \Delta$. If $\Gamma' \vdash P\colon 1$ and $\Gamma ' \vdash t\colon A$, then
\[\encoding{ \Gamma' \vdash P[t /x] \colon 1} = 
    \InputIfFileExists{hoare/pred/step1.tikz}{}{\input{./tikz/hoare/pred/step1.tikz}}
. \]
\end{lem}

\begin{table}[t]
{\scriptsize
  \begin{equation*}
  \begin{array}{c@{\quad}c@{\quad}c}
    \inferrule*[right=($\mathsf{skip}$)]{ }{\{P\}\mathsf{skip}\{P\}} & \inferrule*[right=($\mathsf{assn}$)]{ }{\{P[e/x]\}x \coloneqq e\{P\}} & \inferrule*[right=(\(\subseteq\))]{P_1 \subseteq P_2 \quad \{P_2\}C\{Q_2\} \quad Q_2 \subseteq Q_1}{\{P_1\}C\{Q_1\}} \\[12pt]
    \inferrule*[right=($\mathsf{seq}$)]{\{P\}C\{Q\} \quad \{Q\}D\{R\}}{\{P\}C ; D\{R\}} & \inferrule*[right=(\(\mathsf{if}\))]{\{P \land B\}C\{Q\} \quad \{P \land \lnot B \}D\{Q\}}{\{P\}\mathop{\mathsf{if}} B \mathop{\mathsf{then}} C \mathop{\mathsf{else}} D\{Q\}} & \inferrule*[right=(\(\mathsf{while}\))]{\{P \land B\}C\{P\}}{\{P\}\mathop{\mathsf{while}} B \mathop{\mathsf{do}} C\{P \land \lnot B\}} %
  \end{array}
  \end{equation*}
  }
  \caption{Rules of Hoare logic.}\label{fig:hoare-rules}
\end{table}

\begin{prop}\label{prop:Hoare}
  If  \(\{P\}C\{Q\}\) is derivable as in \Cref{fig:hoare-rules}, then \(\op{\encoding{P}} \dcomp \encoding{C} \leq_{\basicR} \op{\encoding{Q}}\).
\end{prop}
\begin{proof} %
By induction on the rules in \Cref{fig:hoare-rules}.

   \noindent\((\mathsf{skip})\). \; \(\op{\encoding{P}} \dcomp \encoding{\mathsf{skip}} \stackrel{(\text{Table}~\ref{tab:encoding})}{=_{\basicR}} \op{\encoding{P}} \dcomp \id{} =_{\basicR} \op{\encoding{P}}\).\\

  \medskip

  \noindent\((\mathsf{assn})\). \; By Lemma \ref{lemma:encodingsubstPed}, $\op{\encoding{P[e/x]}} \dcomp \encoding{x \coloneqq e}$ is the leftmost diagram below. Thus:
  \[ {

}
}. \] 
In the last inequality we are using the fact that, for all expressions $e$, $\encoding{e}$ is single-valued thanks to the axioms \eqref{eq:deterministic-total} in $\basicR$.
  \medskip

  \noindent\((\subseteq)\). \;\;\; \(\op{\encoding{P_{1}}} \dcomp \encoding{C} \stackrel{(P_1 \subseteq P_2)}{\leq_{\basicR}} \op{\encoding{P_{2}}} \dcomp  \encoding{C} \stackrel{\text{(Ind. hyp.)}}{\leq_{\basicR}} \op{\encoding{Q_{2}}} \stackrel{(Q_1 \subseteq Q_2)}{\leq_{\basicR}} \op{\encoding{Q_{1}}}\).\\

  \medskip

  \noindent\((\mathsf{seq})\). \; $\op{\encoding{P}} \dcomp \encoding{C ; D} \stackrel{(\text{Table}~\ref{tab:encoding})}{=_{\basicR}} \op{\encoding{P}} \dcomp \encoding{C} \dcomp \encoding{D} \stackrel{\text{(Ind. hyp.)}}{\leq_{\basicR}} \op{\encoding{Q}} \dcomp \encoding{D} \stackrel{\text{(Ind. hyp.)}}{\leq_{\basicR}} \op{\encoding{R}}.$

  \medskip

  \noindent\((\mathsf{if})\). \; By \Cref{tab:encoding}, $\op{\encoding{P}}; \encoding{\mathop{\mathsf{if}} B \mathop{\mathsf{then}} C \mathop{\mathsf{else}} D }$ is the leftmost diagram below.
  \begin{align*}
    

}
.
  \end{align*}
In the central inequality, we are using the induction hypothesis and the fact that
\begin{align*}
\op{\encoding{P}}; c(\encoding{B}) &= \op{\encoding{P}} ; \op{c(\encoding{B})} \tag{\Cref{prop:coriff}} \\
&= \op{(c(\encoding{B}) ; \encoding{P}  )} \tag{\Cref{table:re:daggerproperties}}\\
&= \op{(\encoding{P} \sqcap \encoding{B} )} \tag{\Cref{proop:bijcor}}\\
&= \op{\encoding{P  \wedge B} } \tag{\Cref{tab:encoding}}
\end{align*}

  \medskip

  \noindent\((\mathsf{while})\). \; First, observe that the following holds:
    \[ 
    %
}
. \]
    Then, by~\eqref{ax:trace:tape:AU2}, the inequality below holds.
    \[ \scalebox{0.9}{
    %
}
. \]
    Observe that the last diagram is $\op{\encoding{P  \wedge \lnot B} }$ by the derivation used for the $(\mathsf{if})$ case above.
\end{proof}

\begin{rem}[Other Program Logics]\label{rem:other-program-logics}
The above result proves a syntactic correspondence amongst the deduction systems in Table \ref{fig:hoare-rules} and  $\leq_{\basicR}$. %
However, by recalling that, for a fixed interpretation $\interpretation$,  $\CBdsem{\encoding{\Gamma \vdash C}}_{\interpretation} $ is the intended extensional semantics of a command $C$, one can immediately see that the correspondence at the semantic level --illustrated below on the top-left corner-- holds.
\[
\begin{array}{r@{\;\;}c@{\;\;}l r@{\;\;}c@{\;\;}l}
\{P\}C\{Q\} &\text{iff}& \CBdsem{\op{\encoding{P}} \dcomp \encoding{C}}_{\interpretation} \subseteq \CBdsem{\op{\encoding{Q}}}_{\interpretation}{} &  \lAngle P \rAngle C \lAngle Q \rAngle & \text{iff} & \CBdsem{\encoding{P} }_{\interpretation} \subseteq \CBdsem{ \encoding{C}\dcomp  \encoding{Q}}_{\interpretation}{}  \\
 \text{[}P\text{]} C \text{[}Q\text{]}  &\text{iff}& \CBdsem{\op{\encoding{P}} \dcomp \encoding{C}}_{\interpretation} \supseteq \CBdsem{\op{ \encoding{Q}}}_{\interpretation}{} &  (P)C(Q) & \text{iff} & \CBdsem{\encoding{P} }_{\interpretation} \supseteq \CBdsem{ \encoding{C}\dcomp  \encoding{Q}}_{\interpretation}{}  \\
\end{array}
\]
The other correspondences concern other logics: $\text{[}P\text{]} C \text{[}Q\text{]}$ are triples of incorrectness logic \cite{o2019incorrectness}, \(\lAngle P \rAngle C \lAngle Q \rAngle\) of sufficient incorrectness logic~\cite{DBLP:journals/corr/abs-2310-18156} and  \((P)C(Q)\) of necessary logic~\cite{DBLP:conf/vmcai/CousotCFL13}. These correspondences are displayed by means of tape diagrams in \Cref{tab:triples-inequalities}. 
\end{rem}

\begin{table}[H]
  \centering
  \renewcommand{\arraystretch}{1.7}
  \begin{tabular}{ c @{\qquad} c @{\qquad} c }
    \toprule
    Logic & Triple & Inequality \\
    \midrule
    Hoare & \(\{P\}C\{Q\}\) & $
    \InputIfFileExists{hoare/inclusions/hoare/lhs.tikz}{}{\input{./tikz/hoare/inclusions/hoare/lhs.tikz}}
 \leq_{\basicR} 
    \InputIfFileExists{hoare/inclusions/hoare/rhs.tikz}{}{\input{./tikz/hoare/inclusions/hoare/rhs.tikz}}
$\\ %
    Incorrectness & \([P]C[Q]\) & $
    \InputIfFileExists{hoare/inclusions/hoare/lhs.tikz}{}{\input{./tikz/hoare/inclusions/hoare/lhs.tikz}}
 \geq_{\basicR} 
    \InputIfFileExists{hoare/inclusions/hoare/rhs.tikz}{}{\input{./tikz/hoare/inclusions/hoare/rhs.tikz}}
$\\ %
    Sufficient incorrectness & \(\lAngle P \rAngle C \lAngle Q \rAngle\) & $
    \InputIfFileExists{hoare/inclusions/suffInc/lhs.tikz}{}{\input{./tikz/hoare/inclusions/suffInc/lhs.tikz}}
 \leq_{\basicR} 
    \InputIfFileExists{hoare/inclusions/suffInc/rhs.tikz}{}{\input{./tikz/hoare/inclusions/suffInc/rhs.tikz}}
$\\ %
    Necessary & \((P)C(Q)\) & $
    \InputIfFileExists{hoare/inclusions/suffInc/lhs.tikz}{}{\input{./tikz/hoare/inclusions/suffInc/lhs.tikz}}
 \geq_{\basicR} 
    \InputIfFileExists{hoare/inclusions/suffInc/rhs.tikz}{}{\input{./tikz/hoare/inclusions/suffInc/rhs.tikz}}
$\\[4pt] %
    \bottomrule
  \end{tabular}
  \caption{Correspondence between triples and inequalities.\label{tab:triples-inequalities}}
\end{table}

\subsection{Relational Hoare Logic}\label{ssec:relationalHoare}
Relational Hoare Logic (RHL), also known as Benton logic~\cite{benton2004simple}, differs from traditional Hoare logic in that preconditions and postconditions are regarded as relations between two states, rather than predicates on a single state. This enables reasoning about how a program behaves across two different executions or how two programs run in a related manner. Such relational reasoning has found various applications in program optimization~\cite{barthe2011relational} and in verification of security properties and cryptographic protocols~\cite{barthe2012probabilistic,sousa2016cartesian,unruh2019quantum}.

Formally, a \emph{quadruple} $c_1 \sim c_2 \colon P \Rightarrow Q$ of RHL asserts that for any pair of initial states $s_1$ and $s_2$ related by the precondition, $(s_1, s_2) \in P$, if the executions of $c_1$ in $s_1$ and $c_2$ in $s_2$ terminate in the final states $s_1'$ and $s_2'$ respectively, then the final states will be related by the postcondition $(s_1', s_2') \in Q$.

The relational nature of Kleene-Cartesian tape diagrams makes them particularly well-suited for handling this setting. In particular, RHL quadruples can be characterised as
\begin{equation}\label{eq:rhl-validity}
  \CBdsem{\op{\encoding{\Gamma_1, \Gamma_2 \vdash P\colon 1}} \dcomp \encoding{\Gamma_1 \vdash C_1} \per \encoding{\Gamma_2 \vdash C_2}}_{\interpretation}  \subseteq  \CBdsem{\op{\encoding{\Gamma_1, \Gamma_2 \vdash Q\colon 1}}}
\end{equation}
Note that in the inclusion above we use the monoidal product $\per$ to combine the two programs. This is no accident: in fact, RHL quadruples implicitly construct a \emph{product program}~\cite{barthe2011relational}, namely a single program that simulates two separate programs in lockstep.
\begin{exa}[Product program]\label{ex:productprogram} Let $\Gamma_1$ and $\Gamma_2$ be two contexts, and consider the following programs:
  \[ 
  \Gamma_1 \vdash\mathop{\mathsf{while}} P_1 \mathop{\mathsf{do}} C_1
  \qquad\qquad\text{and}\qquad\qquad
  \Gamma_2 \vdash\mathop{\mathsf{while}} P_2 \mathop{\mathsf{do}} C_2.
  \]
\end{exa}
The product program of the two programs above can be understood as the monoidal product $\per$ of their encodings  
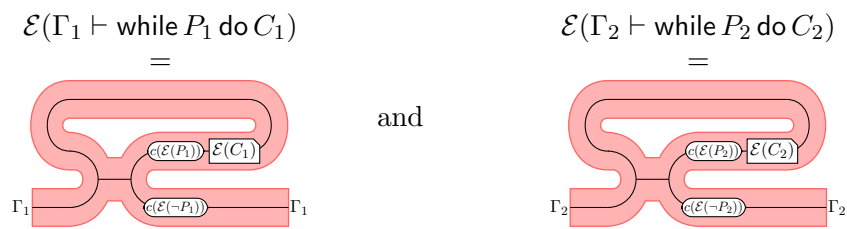
\[
\begin{array}{@{}c@{}c@{}c}
  \begin{array}{@{}c@{}}
    \encoding{\Gamma_1 \vdash\mathop{\mathsf{while}}  P_1 \mathop{\mathsf{do}} C_1}
    \\
    =
    \\
    
    \InputIfFileExists{prodotto/while1.tikz}{}{\input{./tikz/prodotto/while1.tikz}}

  \end{array}
  &
  \qquad\text{and}\qquad\qquad
  &
  \begin{array}{@{}c@{}}
    \encoding{\Gamma_2 \vdash\mathop{\mathsf{while}} P_2 \mathop{\mathsf{do}} C_2}
    \\
    =
    \\
    
    \InputIfFileExists{prodotto/while2.tikz}{}{\input{./tikz/prodotto/while2.tikz}}

  \end{array}
\end{array}
\]
which is given by the following tape diagram.
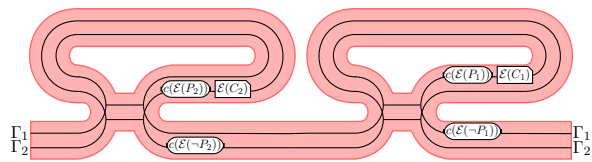
\[
 
    \InputIfFileExists{prodotto/step1.tikz}{}{\input{./tikz/prodotto/step1.tikz}}

\]
Note that, by definition of $\per$ in \eqref{eq:pertapes}, the diagram above is the sequential composition of the left and right whiskerings of the second and first program, respectively. Each whiskering simply amounts to extending the first (resp. second) program with the additional variables of the second (resp. first) program.

It is interesting to see that, thanks to Proposition~\ref{prop:perstarcup}, the diagram above can be rewritten into an equivalent diagram having a single trace, as shown below:
\begin{align*}
  \encoding{\Gamma_1 \vdash\mathop{\mathsf{while}} &P_1 \mathop{\mathsf{do}} C_1}
  \per
  \encoding{\Gamma_2 \vdash\mathop{\mathsf{while}} P_2 \mathop{\mathsf{do}} C_2}
  \\
  = & \quad
    \InputIfFileExists{prodotto/while1.tikz}{}{\input{./tikz/prodotto/while1.tikz}}
 \!\!\per\!\! 
    \InputIfFileExists{prodotto/while2.tikz}{}{\input{./tikz/prodotto/while2.tikz}}
 \tag{Table~\ref{tab:encoding}} \\
  = & \quad(
    \InputIfFileExists{prodotto/star1.tikz}{}{\input{./tikz/prodotto/star1.tikz}}
 \!\!\per\!\! 
    \InputIfFileExists{prodotto/star2.tikz}{}{\input{./tikz/prodotto/star2.tikz}}
) ; 
    \InputIfFileExists{prodotto/negs.tikz}{}{\input{./tikz/prodotto/negs.tikz}}
 \tag{Functoriality of $\per$} \\
  = & \quad
    \InputIfFileExists{prodotto/finale.tikz}{}{\input{./tikz/prodotto/finale.tikz}}
 \tag{Proposition~\ref{prop:perstarcup}}
\end{align*}
Returning to RHL, consider the following rule, found e.g. in \cite{naumann2020thirty},
\[
  \inferrule*[right=($\mathsf{frame}$)]{ c_1 \sim c_2 \colon P \Rightarrow Q   \qquad \textsf{vars}(S) \cap (\mathsf{mod}(c_1) \cup \mathsf{mod}(c_2)) = \emptyset  }{ c_1 \sim c_2 \colon P \land S \Rightarrow Q \land S}
\]
asserting the fact that the pre- and postconditions can be both strengthened with a relation $S$, provided that $S$ has no variables in common with those modified by the two commands.

In our diagrammatic algebra we can prove a stronger version of the $(\mathsf{frame})$ rule. In particular, we require that in the final state, some of the variables of $c_1$ and $c_2$ retain the same value they had in the initial state. Pictorially this is represented by the following equations:
\begin{equation}\label{eq:rhl:c1-c2}
  \quad \modCirc{c_1}{\Gamma_1}{\Delta_1} =_{\basicR} \modCircCB{c_1}{\Gamma_1}{\Delta_1} \qquad \modCirc{c_2}{\Gamma_2}{\Delta_2} =_{\basicR} \modCircCB{c_2}{\Gamma_2}{\Delta_2}.
\end{equation}
Moreover, we require that the additional predicate $S$ only operates on certain variables, specifically those that in $c_1$ and $c_2$ do not change. Graphically, this corresponds to the following equation, holding for some predicate $S'$:
\begin{equation}\label{eq:rhl:s}
  
    \InputIfFileExists{hoare/frame/lhs2.tikz}{}{\input{./tikz/hoare/frame/lhs2.tikz}}
 =_{\basicR} 
    \InputIfFileExists{hoare/frame/rhs2.tikz}{}{\input{./tikz/hoare/frame/rhs2.tikz}}
.
\end{equation}
As in the original $(\mathsf{frame})$ rule, we require that $c_1 \sim c_2 \colon P \Rightarrow Q$. According to~\eqref{eq:rhl-validity} this amounts to say that the following inequality holds:
\begin{equation}\label{eq:rhl:hyp}
  
    \InputIfFileExists{hoare/frame/lhs.tikz}{}{\input{./tikz/hoare/frame/lhs.tikz}}
 \leq_{\basicR} 
    \InputIfFileExists{hoare/frame/rhs.tikz}{}{\input{./tikz/hoare/frame/rhs.tikz}}

\end{equation}
Then we can conclude that $c_1 \sim c_2 \colon P \land S \Rightarrow Q \land S$, as witnessed by the following derivation:
\begin{align*}
  
    \InputIfFileExists{hoare/frame/step0.tikz}{}{\input{./tikz/hoare/frame/step0.tikz}}
 &\stackrel{\eqref{eq:rhl:s}}{=_{\basicR}} 
    \InputIfFileExists{hoare/frame/step1.tikz}{}{\input{./tikz/hoare/frame/step1.tikz}}
 \stackrel{\eqref{eq:rhl:c1-c2}}{=_{\basicR}} 
    \InputIfFileExists{hoare/frame/step2.tikz}{}{\input{./tikz/hoare/frame/step2.tikz}}
 \stackrel{\eqref{ax:cocopierun}}{=_{\basicR}} 
    \InputIfFileExists{hoare/frame/step3.tikz}{}{\input{./tikz/hoare/frame/step3.tikz}}
 \\
  &\stackrel{\eqref{ax:copieradj1}}{\leq_{\basicR}} 
    \InputIfFileExists{hoare/frame/step4.tikz}{}{\input{./tikz/hoare/frame/step4.tikz}}
 \stackrel{\eqref{eq:rhl:hyp}}{\leq_{\basicR}} 
    \InputIfFileExists{hoare/frame/step5.tikz}{}{\input{./tikz/hoare/frame/step5.tikz}}
  \stackrel{\eqref{ax:cocopierun}}{=_{\basicR}} 
    \InputIfFileExists{hoare/frame/step6.tikz}{}{\input{./tikz/hoare/frame/step6.tikz}}
 \stackrel{\eqref{eq:rhl:s}}{=_{\basicR}}  
    \InputIfFileExists{hoare/frame/step7.tikz}{}{\input{./tikz/hoare/frame/step7.tikz}}
.
\end{align*}

\begin{rem}
    In the derivation above, we make crucial use of the structure of Cartesian bicategories. As noted at the beginning of the section, this structure models data flow, which in traditional Hoare logic is somehow less visible, only becoming apparent in the case of assignment.
\end{rem}

\section{Concluding remarks}

We introduced Kleene bicategories and proved that they form Kleene algebras in 
Kozen's sense (Corollary~\ref{cor:kleeneareka}).  
By examining their interaction with Cartesian bicategories, we developed 
Kleene-Cartesian rig categories and characterised the free such structure in 
terms of tape diagrams (Theorem~\ref{thm:KleeneCartesiantapesfree}).  
Following Lawvere's approach to functorial semantics, we showed that morphisms of 
kc-rig categories out of the freely generated one provide models of theories 
(\Cref{funct:sem}).  
We then exhibited a Kleene-Cartesian theory equivalent to Peano's axioms for the 
natural numbers (\Cref{lemma:Peano}), as well as theories corresponding to 
Kleene algebra with tests~\cite{kozen1997kleene} (\Cref{ex:kat}) and to 
imperative programming (\Cref{sec:hoare}).  
In the latter case, tape diagrams yield an assembly language for several program 
logics (\Cref{rem:other-program-logics}), including relational Hoare logic 
(\Cref{ssec:relationalHoare}), and we showed that the rules of Hoare logic follow 
directly from the structure of kc-rig categories 
(Proposition~\ref{prop:Hoare}).

Regarding Kleene bicategories, although uniform traces over biproduct categories 
have been previously studied (see, e.g.,~\cite{cuazuanescu1994feedback}), our use 
of \emph{posetal} uniformity and the resulting correspondence with Kozen's axioms 
are, to the best of our knowledge, new.

Control-flow and data-flow graphs are familiar and intuitive tools for computer 
scientists.  
Tape diagrams combine these two viewpoints into a single diagrammatic language 
equipped with a fully compositional semantics.  
A further distinguishing feature of tape diagrams, compared with other relational 
or categorical approaches to program logics 
(e.g.,~\cite{kozen00hoarekleene,aguirre2020weakest,goncharov2013relatively,martin2006hoare,hasuo2015generic,DBLP:conf/lics/BarrettCH24}), lies in their 
$\per$--monoidal and rig structure, which enables the direct representation of 
\emph{product programs}~\cite{barthe2011relational} (\Cref{ex:productprogram}).  
Via tape (in)equalities and the $\per$ operator, complex properties such as 
\emph{non--interference}~\cite{goguen1984unwinding} become straightforward to 
express.  
Investigating how such properties can be proved using only the laws of 
kc-rig categories remains an interesting and promising direction for future 
research.

\section*{Acknowledgments}
  \noindent The authors would like to acknowledge Alessio Santamaria, Chad Nester and the students of the ACT school 2022 for several useful discussions at an early stage of this project. Gheorghe Stefanescu and Dexter Kozen provided some wise feedback and offered some guidance through the rather wide literature. The authors would like to thank Mario Román for pointing out a simplification of some proofs. This research was partly funded by the Advanced Research + Invention Agency (ARIA) Safeguarded AI Programme and carried out within the National Centre on HPC, Big Data and Quantum Computing - SPOKE 10 (Quantum Computing) and by the EU Next-GenerationEU - National Recovery and Resilience Plan (NRRP) – MISSION 4 COMPONENT 2, INVESTMENT N. 1.4 – CUP N. I53C22000690001 and by the EPSRC grant No. EP/V002376/1. Bonchi is supported by the Ministero dell'Università e della Ricerca of Italy grant PRIN 2022 PNRR No. P2022HXNSC - RAP (Resource Awareness in Programming). Di Giorgio is supported by the EU grant No. 101087529.

\bibliography{main}

\newcommand{\etalchar}[1]{$^{#1}$}
\begin{thebibliography}{BDGDL24}

\bibitem[ABGL23]{DBLP:journals/corr/abs-2310-18156}
Flavio Ascari, Roberto Bruni, Roberta Gori, and Francesco Logozzo.
\newblock Sufficient incorrectness logic: {SIL} and separation {SIL}, 2023.
\newblock URL: \url{https://doi.org/10.48550/arXiv.2310.18156}, \href
  {http://arxiv.org/abs/2310.18156} {\path{arXiv:2310.18156}}, \href
  {https://doi.org/10.48550/ARXIV.2310.18156}
  {\path{doi:10.48550/ARXIV.2310.18156}}.

\bibitem[AK01]{angus2001kleene}
Allegra Angus and Dexter Kozen.
\newblock Kleene algebra with tests and program schematology.
\newblock Technical report, Cornell University, 2001.

\bibitem[AK20]{aguirre2020weakest}
Alejandro Aguirre and Shin{-}ya Katsumata.
\newblock Weakest preconditions in fibrations.
\newblock 352:5--27, 2020.
\newblock URL: \url{https://doi.org/10.1016/j.entcs.2020.09.002}, \href
  {https://doi.org/10.1016/J.ENTCS.2020.09.002}
  {\path{doi:10.1016/J.ENTCS.2020.09.002}}.

\bibitem[Bac14]{Backens-ZXcompleteness1}
Miriam Backens.
\newblock The zx-calculus is complete for stabilizer quantum mechanics.
\newblock {\em New Journal of Physics}, 16(9):093021, sep 2014.
\newblock \href {https://doi.org/10.1088/1367-2630/16/9/093021}
  {\path{doi:10.1088/1367-2630/16/9/093021}}.

\bibitem[Bai76]{bainbridge1976feedback}
Edwin~Stewart Bainbridge.
\newblock Feedback and generalized logic.
\newblock {\em Inf. Control.}, 31(1):75--96, 1976.
\newblock \href {https://doi.org/10.1016/S0019-9958(76)90390-9}
  {\path{doi:10.1016/S0019-9958(76)90390-9}}.

\bibitem[BCGL25]{DBLP:conf/calco/BonchiC0L25}
Filippo Bonchi, Cipriano~Junior Cioffo, Alessandro~Di Giorgio, and Elena~Di
  Lavore.
\newblock Tape diagrams for monoidal monads.
\newblock In Corina C{\^{\i}}rstea and Alexander Knapp, editors, {\em 11th
  Conference on Algebra and Coalgebra in Computer Science, {CALCO} 2025, June
  16-18, 2025, University of Strathclyde, {UK}}, volume 342 of {\em LIPIcs},
  pages 11:1--11:24. Schloss Dagstuhl - Leibniz-Zentrum f{\"{u}}r Informatik,
  2025.
\newblock \href {https://doi.org/10.4230/LIPICS.CALCO.2025.11}
  {\path{doi:10.4230/LIPICS.CALCO.2025.11}}.

\bibitem[BCH24]{DBLP:conf/lics/BarrettCH24}
Chris Barrett, Daniel Castle, and Willem Heijltjes.
\newblock The relational machine calculus.
\newblock In Pawel Sobocinski, Ugo~Dal Lago, and Javier Esparza, editors, {\em
  Proceedings of the 39th Annual {ACM/IEEE} Symposium on Logic in Computer
  Science, {LICS} 2024, Tallinn, Estonia, July 8-11, 2024}, pages 9:1--9:15.
  {ACM}, 2024.
\newblock \href {https://doi.org/10.1145/3661814.3662091}
  {\path{doi:10.1145/3661814.3662091}}.

\bibitem[BCK11]{barthe2011relational}
Gilles Barthe, Juan~Manuel Crespo, and C{\'{e}}sar Kunz.
\newblock Relational verification using product programs.
\newblock In Michael~J. Butler and Wolfram Schulte, editors, {\em {FM} 2011:
  Formal Methods - 17th International Symposium on Formal Methods, Limerick,
  Ireland, June 20-24, 2011. Proceedings}, volume 6664 of {\em Lecture Notes in
  Computer Science}, pages 200--214. Springer, 2011.
\newblock \href {https://doi.org/10.1007/978-3-642-21437-0\_17}
  {\path{doi:10.1007/978-3-642-21437-0\_17}}.

\bibitem[BCR17]{BaezCoya-propsnetworktheory}
John~C. Baez, Brandon Coya, and Franciscus Rebro.
\newblock Props in network theory.
\newblock {\em CoRR}, abs/1707.08321, 2017.
\newblock URL: \url{http://arxiv.org/abs/1707.08321}, \href
  {http://arxiv.org/abs/1707.08321} {\path{arXiv:1707.08321}}.

\bibitem[BDGDL24]{bonchi2024diagrammaticalgebraprogramlogics}
Filippo Bonchi, Alessandro Di~Giorgio, and Elena Di~Lavore.
\newblock A diagrammatic algebra for program logics, 2024.
\newblock URL: \url{https://arxiv.org/abs/2410.03561}, \href
  {http://arxiv.org/abs/2410.03561} {\path{arXiv:2410.03561}}.

\bibitem[BDGHS24]{DBLP:conf/lics/Bonchi0H024}
Filippo Bonchi, Alessandro Di~Giorgio, Nathan Haydon, and Pawel Sobocinski.
\newblock Diagrammatic algebra of first order logic.
\newblock In Pawel Sobocinski, Ugo~Dal Lago, and Javier Esparza, editors, {\em
  Proceedings of the 39th Annual {ACM/IEEE} Symposium on Logic in Computer
  Science, {LICS} 2024, Tallinn, Estonia, July 8-11, 2024}, pages 16:1--16:15.
  {ACM}, 2024.
\newblock \href {https://doi.org/10.1145/3661814.3662078}
  {\path{doi:10.1145/3661814.3662078}}.

\bibitem[Ben04]{benton2004simple}
Nick Benton.
\newblock Simple relational correctness proofs for static analyses and program
  transformations.
\newblock pages 14--25, 2004.
\newblock \href {https://doi.org/10.1145/964001.964003}
  {\path{doi:10.1145/964001.964003}}.

\bibitem[B{\'E}S95]{bloom1995notes}
S.~L. Bloom, Z.~{\'E}sik, and Gh. Stefanescu.
\newblock Notes on equational theories of relations.
\newblock {\em algebra universalis}, 33(1):98--126, 1995.
\newblock \href {https://doi.org/10.1007/BF01190768}
  {\path{doi:10.1007/BF01190768}}.

\bibitem[BGL25]{DBLP:conf/fossacs/BonchiGL25}
Filippo Bonchi, Alessandro~Di Giorgio, and Elena~Di Lavore.
\newblock A diagrammatic algebra for program logics.
\newblock In Parosh~Aziz Abdulla and Delia Kesner, editors, {\em Foundations of
  Software Science and Computation Structures - 28th International Conference,
  FoSSaCS 2025, Held as Part of the International Joint Conferences on Theory
  and Practice of Software, {ETAPS} 2025, Hamilton, ON, Canada, May 3-8, 2025,
  Proceedings}, volume 15691 of {\em Lecture Notes in Computer Science}, pages
  308--330. Springer, 2025.
\newblock \href {https://doi.org/10.1007/978-3-031-90897-2\_15}
  {\path{doi:10.1007/978-3-031-90897-2\_15}}.

\bibitem[BGS23]{bonchi2023deconstructing}
Filippo Bonchi, Alessandro~Di Giorgio, and Alessio Santamaria.
\newblock Deconstructing the calculus of relations with tape diagrams.
\newblock {\em Proc. {ACM} Program. Lang.}, 7({POPL}):1864--1894, 2023.
\newblock \href {https://doi.org/10.1145/3571257} {\path{doi:10.1145/3571257}}.

\bibitem[BHP{\etalchar{+}}19]{DBLP:journals/pacmpl/BonchiHPSZ19}
Filippo Bonchi, Joshua Holland, Robin Piedeleu, Pawe{\l} Soboci{\'n}ski, and
  Fabio Zanasi.
\newblock Diagrammatic algebra: From linear to concurrent systems.
\newblock {\em Proceedings of the ACM on Programming Languages},
  3(POPL):25:1--25:28, January 2019.
\newblock \href {https://doi.org/10.1145/3290338} {\path{doi:10.1145/3290338}}.

\bibitem[BKOZ12]{barthe2012probabilistic}
Gilles Barthe, Boris K{\"{o}}pf, Federico Olmedo, and Santiago
  Zanella{-}B{\'{e}}guelin.
\newblock Probabilistic relational reasoning for differential privacy.
\newblock In John Field and Michael Hicks, editors, {\em Proceedings of the
  39th {ACM} {SIGPLAN-SIGACT} Symposium on Principles of Programming Languages,
  {POPL} 2012, Philadelphia, Pennsylvania, USA, January 22-28, 2012}, pages
  97--110. {ACM}, 2012.
\newblock \href {https://doi.org/10.1145/2103656.2103670}
  {\path{doi:10.1145/2103656.2103670}}.

\bibitem[BMM11]{bruni2011connector}
Roberto Bruni, Hern{\'{a}}n~C. Melgratti, and Ugo Montanari.
\newblock Connector algebras, {P}etri nets, and {BIP}.
\newblock In Edmund~M. Clarke, Irina~B. Virbitskaite, and Andrei Voronkov,
  editors, {\em Perspectives of Systems Informatics - 8th International Andrei
  Ershov Memorial Conference, {PSI} 2011, Novosibirsk, Russia, June 27-July 1,
  2011, Revised Selected Papers}, volume 7162 of {\em Lecture Notes in Computer
  Science}, pages 19--38. Springer, 2011.
\newblock \href {https://doi.org/10.1007/978-3-642-29709-0\_2}
  {\path{doi:10.1007/978-3-642-29709-0\_2}}.

\bibitem[BP14]{brunet2014kleene}
Paul Brunet and Damien Pous.
\newblock Kleene algebra with converse.
\newblock In Peter H{\"{o}}fner, Peter Jipsen, Wolfram Kahl, and Martin~Eric
  M{\"{u}}ller, editors, {\em Relational and Algebraic Methods in Computer
  Science - 14th International Conference, RAMiCS 2014, Marienstatt, Germany,
  April 28-May 1, 2014. Proceedings}, volume 8428 of {\em Lecture Notes in
  Computer Science}, pages 101--118. Springer, 2014.
\newblock \href {https://doi.org/10.1007/978-3-319-06251-8\_7}
  {\path{doi:10.1007/978-3-319-06251-8\_7}}.

\bibitem[BPS25]{Bonchi2017c}
Filippo Bonchi, Dusko Pavlovic, and Paweł Sobociński.
\newblock Functorial semantics for relational theories.
\newblock {\em Polynesian Journal of Mathematics}, 2(6):1--56, 10 2025.
\newblock \href {https://doi.org/10.69763/polyjmath.2.6}
  {\path{doi:10.69763/polyjmath.2.6}}.

\bibitem[BPSZ19]{BonchiPSZ19}
Filippo Bonchi, Robin Piedeleu, Pawel Sobocinski, and Fabio Zanasi.
\newblock Graphical affine algebra.
\newblock In {\em 34th Annual {ACM/IEEE} Symposium on Logic in Computer
  Science, {LICS} 2019, Vancouver, BC, Canada, June 24-27, 2019}, pages 1--12.
  {IEEE}, 2019.
\newblock \href {https://doi.org/10.1109/LICS.2019.8785877}
  {\path{doi:10.1109/LICS.2019.8785877}}.

\bibitem[BSS18]{GCQ}
Filippo Bonchi, Jens Seeber, and Pawel Sobocinski.
\newblock {Graphical Conjunctive Queries}.
\newblock In Dan Ghica and Achim Jung, editors, {\em 27th EACSL Annual
  Conference on Computer Science Logic (CSL 2018)}, volume 119 of {\em Leibniz
  International Proceedings in Informatics (LIPIcs)}, pages 13:1--13:23,
  Dagstuhl, Germany, 2018. Schloss Dagstuhl--Leibniz-Zentrum fuer Informatik.
\newblock \href {https://doi.org/10.4230/LIPIcs.CSL.2018.13}
  {\path{doi:10.4230/LIPIcs.CSL.2018.13}}.

\bibitem[BSZ15]{bonchi2015full}
Filippo Bonchi, Pawel Sobocinski, and Fabio Zanasi.
\newblock Full abstraction for signal flow graphs.
\newblock pages 515--526, 2015.
\newblock \href {https://doi.org/10.1145/2676726.2676993}
  {\path{doi:10.1145/2676726.2676993}}.

\bibitem[CCFL13]{DBLP:conf/vmcai/CousotCFL13}
Patrick Cousot, Radhia Cousot, Manuel F{\"{a}}hndrich, and Francesco Logozzo.
\newblock Automatic inference of necessary preconditions.
\newblock In Roberto Giacobazzi, Josh Berdine, and Isabella Mastroeni, editors,
  {\em Verification, Model Checking, and Abstract Interpretation, 14th
  International Conference, {VMCAI} 2013, Rome, Italy, January 20-22, 2013.
  Proceedings}, volume 7737 of {\em Lecture Notes in Computer Science}, pages
  128--148. Springer, 2013.
\newblock \href {https://doi.org/10.1007/978-3-642-35873-9\_10}
  {\path{doi:10.1007/978-3-642-35873-9\_10}}.

\bibitem[CD11]{coecke2011interacting}
Bob Coecke and Ross Duncan.
\newblock Interacting quantum observables: categorical algebra and
  diagrammatics.
\newblock {\em New Journal of Physics}, 13(4):043016, apr 2011.
\newblock \href {https://doi.org/10.1088/1367-2630/13/4/043016}
  {\path{doi:10.1088/1367-2630/13/4/043016}}.

\bibitem[CDH20]{comfort2020sheet}
Cole Comfort, Antonin Delpeuch, and Jules Hedges.
\newblock Sheet diagrams for bimonoidal categories, 2020.
\newblock \href {http://arxiv.org/abs/2010.13361} {\path{arXiv:2010.13361}}.

\bibitem[CG99]{corradini1999algebraic}
Andrea Corradini and Fabio Gadducci.
\newblock An algebraic presentation of term graphs, via gs-monoidal categories.
\newblock {\em Applied Categorical Structures}, 7(4):299--331, 1999.
\newblock \href {https://doi.org/10.1023/A:1008647417502}
  {\path{doi:10.1023/A:1008647417502}}.

\bibitem[CJ19]{cho2019disintegration}
Kenta Cho and Bart Jacobs.
\newblock Disintegration and bayesian inversion via string diagrams.
\newblock {\em Math. Struct. Comput. Sci.}, 29(7):938--971, 2019.
\newblock \href {https://doi.org/10.1017/S0960129518000488}
  {\path{doi:10.1017/S0960129518000488}}.

\bibitem[CKS00]{cockett2000introduction}
J.~R.~B. Cockett, J.~Koslowski, and R.~A.~G. Seely.
\newblock Introduction to linear bicategories.
\newblock {\em Mathematical. Structures in Comp. Sci.}, 10(2), April 2000.
\newblock \href {https://doi.org/10.1017/S0960129520003047}
  {\path{doi:10.1017/S0960129520003047}}.

\bibitem[CS94]{cuazuanescu1994feedback}
Virgil~Emil Cazanescu and Gheorghe Stefanescu.
\newblock Feedback, iteration, and repetition.
\newblock In Gheorghe Paun, editor, {\em Mathematical Aspects of Natural and
  Formal Languages}, volume~43 of {\em World scientific series in computer
  science}, pages 43--61. World Scientific, 1994.
\newblock \href {https://doi.org/10.1142/9789814447133\_0003}
  {\path{doi:10.1142/9789814447133\_0003}}.

\bibitem[CST18]{coecke2017two}
Bob Coecke, John Selby, and Sean Tull.
\newblock Two {{Roads}} to {{Classicality}}.
\newblock 266:104--118, February 2018.
\newblock \href {http://arxiv.org/abs/1701.07400} {\path{arXiv:1701.07400}},
  \href {https://doi.org/10.4204/EPTCS.266.7} {\path{doi:10.4204/EPTCS.266.7}}.

\bibitem[CW87]{Carboni1987}
A.~Carboni and R.F.C. Walters.
\newblock Cartesian bicategories {I}.
\newblock {\em Journal of Pure and Applied Algebra}, 49(1):11--32, 1987.
\newblock URL:
  \url{https://www.sciencedirect.com/science/article/pii/0022404987901216},
  \href {https://doi.org/https://doi.org/10.1016/0022-4049(87)90121-6}
  {\path{doi:https://doi.org/10.1016/0022-4049(87)90121-6}}.

\bibitem[Ded88]{dedekind1888nature}
Richard Dedekind.
\newblock The nature and meaning of numbers, 1888.

\bibitem[DG24]{AlessandroThesis}
Alessandro Di~Giorgio.
\newblock {\em {Diagrammatic Algebras of Relations}}.
\newblock PhD thesis, University of Pisa, 2024.

\bibitem[Fox76]{fox1976coalgebras}
T.~Fox.
\newblock Coalgebras and cartesian categories.
\newblock {\em Communications in Algebra}, 4(7):665--667, 1976.
\newblock \href {https://doi.org/10.1080/00927877608822127}
  {\path{doi:10.1080/00927877608822127}}.

\bibitem[Fri09]{Fritz_stochasticmatrices}
Tobias Fritz.
\newblock A presentation of the category of stochastic matrices.
\newblock {\em CoRR}, abs/0902.2554, 2009.

\bibitem[FS20]{DBLP:journals/corr/abs-2009-06836}
Brendan Fong and David Spivak.
\newblock String diagrams for regular logic (extended abstract).
\newblock In John Baez and Bob Coecke, editors, {\em Applied Category Theory
  2019}, volume 323 of {\em Electronic Proceedings in Theoretical Computer
  Science}, page 196–229. Open Publishing Association, Sep 2020.
\newblock \href {https://doi.org/10.4204/eptcs.323.14}
  {\path{doi:10.4204/eptcs.323.14}}.

\bibitem[GJ16]{Ghica2016}
Dan~R. Ghica and Achim Jung.
\newblock Categorical semantics of digital circuits.
\newblock In {\em 2016 Formal Methods in Computer-Aided Design (FMCAD)}, pages
  41--48, 2016.
\newblock \href {https://doi.org/10.1109/FMCAD.2016.7886659}
  {\path{doi:10.1109/FMCAD.2016.7886659}}.

\bibitem[GM84]{goguen1984unwinding}
Joseph~A. Goguen and Jos{\'{e}} Meseguer.
\newblock Unwinding and inference control.
\newblock In {\em Proceedings of the 1984 {IEEE} Symposium on Security and
  Privacy, Oakland, California, USA, April 29 - May 2, 1984}, pages 75--87.
  {IEEE} Computer Society, 1984.
\newblock \href {https://doi.org/10.1109/SP.1984.10019}
  {\path{doi:10.1109/SP.1984.10019}}.

\bibitem[GRS21]{goncharov2021metalanguage}
Sergey Goncharov, Christoph Rauch, and Lutz Schr{\"{o}}der.
\newblock A metalanguage for guarded iteration.
\newblock {\em Theor. Comput. Sci.}, 880:111--137, 2021.
\newblock URL: \url{https://doi.org/10.1016/j.tcs.2021.04.005}, \href
  {https://doi.org/10.1016/J.TCS.2021.04.005}
  {\path{doi:10.1016/J.TCS.2021.04.005}}.

\bibitem[GS13]{goncharov2013relatively}
Sergey Goncharov and Lutz Schr{\"{o}}der.
\newblock A relatively complete generic hoare logic for order-enriched effects.
\newblock In {\em 28th Annual {ACM/IEEE} Symposium on Logic in Computer
  Science, {LICS} 2013, New Orleans, LA, USA, June 25-28, 2013}, pages
  273--282. {IEEE} Computer Society, 2013.
\newblock \href {https://doi.org/10.1109/LICS.2013.33}
  {\path{doi:10.1109/LICS.2013.33}}.

\bibitem[Har83]{harsanyi1983mathematics}
John~C. Harsanyi.
\newblock Mathematics, the empirical facts, and logical necessity.
\newblock {\em Erkenntnis}, 19(1):167--192, 1983.
\newblock \href {https://doi.org/10.1007/BF00174780}
  {\path{doi:10.1007/BF00174780}}.

\bibitem[Har08]{harding2008orthomodularity}
John Harding.
\newblock Orthomodularity in dagger biproduct categories.
\newblock {\em Preprint}, 2008.

\bibitem[Has02]{hasegawa2003uniformity}
Masahito Hasegawa.
\newblock The uniformity principle on traced monoidal categories.
\newblock 69:137--155, 2002.
\newblock \href {https://doi.org/10.1016/S1571-0661(04)80563-2}
  {\path{doi:10.1016/S1571-0661(04)80563-2}}.

\bibitem[Has15]{hasuo2015generic}
Ichiro Hasuo.
\newblock Generic weakest precondition semantics from monads enriched with
  order.
\newblock {\em Theor. Comput. Sci.}, 604:2--29, 2015.
\newblock URL: \url{https://doi.org/10.1016/j.tcs.2015.03.047}, \href
  {https://doi.org/10.1016/J.TCS.2015.03.047}
  {\path{doi:10.1016/J.TCS.2015.03.047}}.

\bibitem[Hoa69]{hoare1969axiomatic}
C.~A.~R. Hoare.
\newblock An axiomatic basis for computer programming.
\newblock {\em Commun. ACM}, 12(10), October 1969.
\newblock \href {https://doi.org/10.1145/363235.363259}
  {\path{doi:10.1145/363235.363259}}.

\bibitem[JS91]{joyal1991geometry}
Andr{\'e} Joyal and Ross Street.
\newblock The geometry of tensor calculus, {{I}}.
\newblock {\em Advances in Mathematics}, 88(1):55--112, July 1991.
\newblock \href {https://doi.org/10.1016/0001-8708(91)90003-P}
  {\path{doi:10.1016/0001-8708(91)90003-P}}.

\bibitem[JSV96]{joyal1996traced}
Andr{\'e} Joyal, Ross Street, and Dominic Verity.
\newblock Traced monoidal categories.
\newblock In {\em Mathematical proceedings of the cambridge philosophical
  society}, volume 119, pages 447--468. Cambridge University Press, 1996.

\bibitem[JY22]{johnson2021bimonoidal}
Niles Johnson and Donald Yau.
\newblock Bimonoidal categories, $ e_n $-monoidal categories, and algebraic $ k
  $-theory.
\newblock 2022.
\newblock URL: \url{https://nilesjohnson.net/En-monoidal.html}.

\bibitem[JY24]{johnson2024bimonoidal}
Niles Johnson and Donald Yau.
\newblock {\em Bimonoidal Categories, $ E_n $-Monoidal Categories, and
  Algebraic $ K $-Theory: Volume III: From Categories to Structured Ring
  Spectra}, volume 285.
\newblock American Mathematical Society, 2024.

\bibitem[Koz94]{Kozen94acompleteness}
Dexter Kozen.
\newblock A completeness theorem for kleene algebras and the algebra of regular
  events.
\newblock {\em Inf. Comput.}, 110(2):366--390, 1994.
\newblock URL: \url{https://doi.org/10.1006/inco.1994.1037}, \href
  {https://doi.org/10.1006/INCO.1994.1037} {\path{doi:10.1006/INCO.1994.1037}}.

\bibitem[Koz97]{kozen1997kleene}
Dexter Kozen.
\newblock Kleene algebra with tests.
\newblock {\em ACM Trans. Program. Lang. Syst.}, 19(3), May 1997.
\newblock \href {https://doi.org/10.1145/256167.256195}
  {\path{doi:10.1145/256167.256195}}.

\bibitem[Koz98]{kozen98typedkleene}
Dexter Kozen.
\newblock Typed {K}leene algebra.
\newblock Technical Report TR98-1669, Computer Science Department, Cornell
  University, March 1998.

\bibitem[Koz00]{kozen00hoarekleene}
Dexter Kozen.
\newblock On {H}oare logic and {K}leene algebra with tests.
\newblock {\em ACM Trans. Comput. Logic}, 1(1), July 2000.
\newblock \href {https://doi.org/10.1145/343369.343378}
  {\path{doi:10.1145/343369.343378}}.

\bibitem[KSW97]{katis1997bicategories}
P.~Katis, N.~Sabadini, and R.F.C. Walters.
\newblock Bicategories of processes.
\newblock {\em Journal of Pure and Applied Algebra}, 115(2):141--178, 1997.
\newblock URL:
  \url{https://www.sciencedirect.com/science/article/pii/S0022404996000126},
  \href {https://doi.org/https://doi.org/10.1016/S0022-4049(96)00012-6}
  {\path{doi:https://doi.org/10.1016/S0022-4049(96)00012-6}}.

\bibitem[KSW02]{katis2002feedback}
P.~Katis, Nicoletta Sabadini, and Robert~F.C. Walters.
\newblock Feedback, trace and fixed-point semantics.
\newblock {\em RAIRO - Theoretical Informatics and Applications},
  36(2):181--194, 2002.
\newblock \href {https://doi.org/10.1051/ita:2002009}
  {\path{doi:10.1051/ita:2002009}}.

\bibitem[Lap72]{laplaza_coherence_1972}
Miguel~L. Laplaza.
\newblock Coherence for distributivity.
\newblock In G.~M. Kelly, M.~Laplaza, G.~Lewis, and Saunders Mac~Lane, editors,
  {\em Coherence in {{Categories}}}, Lecture {{Notes}} in {{Mathematics}},
  pages 29--65, {Berlin, Heidelberg}, 1972. {Springer}.
\newblock \href {https://doi.org/10.1007/BFb0059555}
  {\path{doi:10.1007/BFb0059555}}.

\bibitem[Law63]{LawvereOriginalPaper}
F.~W. Lawvere.
\newblock {\em Functorial {{Semantics}} of {{Algebraic Theories}}}.
\newblock PhD thesis, Columbia University, {New York, NY, USA}, 1963.

\bibitem[ML78]{mac_lane_categories_1978}
S.~Mac~Lane.
\newblock {\em Categories for the {{Working Mathematician}}}, volume~5 of {\em
  Graduate {{Texts}} in {{Mathematics}}}.
\newblock {Springer-Verlag}, {New York}, second edition, 1978.

\bibitem[MMO06]{martin2006hoare}
Ursula Martin, Erik~A. Mathiesen, and Paulo Oliva.
\newblock Hoare logic in the abstract.
\newblock In Zolt{\'a}n {\'E}sik, editor, {\em Computer Science Logic}, pages
  501--515, Berlin, Heidelberg, 2006. Springer Berlin Heidelberg.

\bibitem[Nau20]{naumann2020thirty}
David~A. Naumann.
\newblock Thirty-seven years of relational {H}oare logic: Remarks on its
  principles and history.
\newblock In {\em Leveraging Applications of Formal Methods, Verification and
  Validation: Engineering Principles: 9th International Symposium on Leveraging
  Applications of Formal Methods, ISoLA 2020, Rhodes, Greece, October 20-30,
  2020, Proceedings, Part II}, Berlin, Heidelberg, 2020. Springer-Verlag.
\newblock \href {https://doi.org/10.1007/978-3-030-61470-6_7}
  {\path{doi:10.1007/978-3-030-61470-6_7}}.

\bibitem[O'H19]{o2019incorrectness}
Peter~W. O'Hearn.
\newblock Incorrectness logic.
\newblock {\em Proc. ACM Program. Lang.}, 4(POPL), December 2019.
\newblock \href {https://doi.org/10.1145/3371078} {\path{doi:10.1145/3371078}}.

\bibitem[Pou18]{DBLP:conf/stacs/Pous18}
Damien Pous.
\newblock On the positive calculus of relations with transitive closure.
\newblock In Rolf Niedermeier and Brigitte Vall{\'{e}}e, editors, {\em 35th
  Symposium on Theoretical Aspects of Computer Science, {STACS} 2018, February
  28 to March 3, 2018, Caen, France}, volume~96 of {\em LIPIcs}, pages
  3:1--3:16. Schloss Dagstuhl - Leibniz-Zentrum f{\"{u}}r Informatik, 2018.
\newblock \href {https://doi.org/10.4230/LIPIcs.STACS.2018.3}
  {\path{doi:10.4230/LIPIcs.STACS.2018.3}}.

\bibitem[PR97]{POWER_ROBINSON_1997}
John Power and Edmund Robinson.
\newblock Premonoidal categories and notions of computation.
\newblock {\em Mathematical Structures in Computer Science}, 7(5):453--468,
  1997.
\newblock \href {https://doi.org/10.1017/S0960129597002375}
  {\path{doi:10.1017/S0960129597002375}}.

\bibitem[Pra76]{pratt1976semantical}
Vaughan~R. Pratt.
\newblock Semantical considerations on {F}loyd-{H}oare logic.
\newblock In {\em 17th Annual Symposium on Foundations of Computer Science
  (sfcs 1976)}, pages 109--121, 1976.
\newblock \href {https://doi.org/10.1109/SFCS.1976.27}
  {\path{doi:10.1109/SFCS.1976.27}}.

\bibitem[PZ23]{lmcs:10963}
Robin Piedeleu and Fabio Zanasi.
\newblock {A Finite Axiomatisation of Finite-State Automata Using String
  Diagrams}.
\newblock {\em {Logical Methods in Computer Science}}, {Volume 19, Issue 1},
  February 2023.
\newblock \href {https://doi.org/10.46298/lmcs-19(1:13)2023}
  {\path{doi:10.46298/lmcs-19(1:13)2023}}.

\bibitem[SD16]{sousa2016cartesian}
Marcelo Sousa and Isil Dillig.
\newblock Cartesian hoare logic for verifying k-safety properties.
\newblock In Chandra Krintz and Emery~D. Berger, editors, {\em Proceedings of
  the 37th {ACM} {SIGPLAN} Conference on Programming Language Design and
  Implementation, {PLDI} 2016, Santa Barbara, CA, USA, June 13-17, 2016}, pages
  57--69. {ACM}, 2016.
\newblock \href {https://doi.org/10.1145/2908080.2908092}
  {\path{doi:10.1145/2908080.2908092}}.

\bibitem[Sel98]{selinger1998note}
Peter Selinger.
\newblock A note on {Bainbridge's} power set construction.
\newblock {\em Ann Arbor}, 1001:48109--1109, 1998.

\bibitem[Sel11]{selinger2010survey}
P.~Selinger.
\newblock {\em A Survey of Graphical Languages for Monoidal Categories}, pages
  289--355.
\newblock Springer Berlin Heidelberg, Berlin, Heidelberg, 2011.
\newblock \href {https://doi.org/10.1007/978-3-642-12821-9_4}
  {\path{doi:10.1007/978-3-642-12821-9_4}}.

\bibitem[SS24]{stein2023probabilistic}
Dario Stein and Sam Staton.
\newblock Probabilistic programming with exact conditions.
\newblock {\em J. {ACM}}, 71(1):2:1--2:53, 2024.
\newblock \href {https://doi.org/10.1145/3632170} {\path{doi:10.1145/3632170}}.

\bibitem[Unr19]{unruh2019quantum}
Dominique Unruh.
\newblock Quantum relational {H}oare logic.
\newblock {\em Proc. ACM Program. Lang.}, 3(POPL), January 2019.
\newblock \href {https://doi.org/10.1145/3290346} {\path{doi:10.1145/3290346}}.

\bibitem[Win93]{winskel1993formal}
Glynn Winskel.
\newblock {\em The Formal Semantics of Programming Languages: An Introduction}.
\newblock The MIT Press, 02 1993.
\newblock \href {https://doi.org/10.7551/mitpress/3054.001.0001}
  {\path{doi:10.7551/mitpress/3054.001.0001}}.

\end{thebibliography}

\appendix
\section{Dictionary}\label{app:dictionary}

\begin{table}[H]
    \[
    \def\arraystretch{1.8}

    \]
    \caption{Structure of kc-rig categories and its representation as string diagrams and tape diagrams.}\label{tab:dictionary}
\end{table}

\newpage

\section{String Diagrams for (Uniformly) Traced Monoidal Categories}\label{app:stringDiagrams}

String diagrams provide a convenient graphical representation for arrows of symmetric monoidal categories. Arrows are depicted as boxes with labeled wires, indicating the source and target objects.
For instance $\gen \colon A\perG B \to C$ is depicted as the leftmost diagram below. Moreover, $\id{A}$ is displayed as one wire,  $id_{\unoG} $ as the empty diagram and $\sigma_{A,B}^{\perG}$ as a crossing:
\[
    \InputIfFileExists{generator.tikz}{}{\input{./tikz/generator.tikz}}
 \qquad \qquad 
    \InputIfFileExists{id.tikz}{}{\input{./tikz/id.tikz}}
 \qquad  \qquad     
    \InputIfFileExists{empty.tikz}{}{\input{./tikz/empty.tikz}}
 \qquad  \qquad   
    \InputIfFileExists{symm.tikz}{}{\input{./tikz/symm.tikz}}
\]
Finally, composition $f;g$ is represented by connecting the right wires 
of $f$ with the left wires of $g$ when their labels match, 
while the monoidal product $f \perG g$ is depicted by stacking the corresponding 
diagrams on top of each other: \[
    \InputIfFileExists{seq_comp.tikz}{}{\input{./tikz/seq_comp.tikz}}
 \qquad \qquad \qquad  
    \InputIfFileExists{par_comp.tikz}{}{\input{./tikz/par_comp.tikz}}
 \]
Theorem 2.3 in~\cite{joyal1991geometry} states that the laws of monoidal categories are implicitly embodied in the 
graphical representation while the axioms for symmetries  are displayed as in Figure \ref{fig:symm-axioms}. 

\begin{figure}[h!]
\[ 
    \InputIfFileExists{stringdiag_ax1_left.tikz}{}{\input{./tikz/stringdiag_ax1_left.tikz}}
 = 
    \InputIfFileExists{stringdiag_ax1_right.tikz}{}{\input{./tikz/stringdiag_ax1_right.tikz}}
 \quad\qquad 
    \InputIfFileExists{stringdiag_ax2_left.tikz}{}{\input{./tikz/stringdiag_ax2_left.tikz}}
 = 
    \InputIfFileExists{stringdiag_ax2_right.tikz}{}{\input{./tikz/stringdiag_ax2_right.tikz}}
 \]
\caption{String diagram axioms for symmetries.}\label{fig:symm-axioms}
\end{figure}

\begin{defi}\label{def:traced-category}
A  symmetric monoidal category $(\Cat{C}, \perG, \unoG)$  is \emph{traced} if it is endowed with an operator \(\trace_{S} \colon \Cat{C}(S \perG X, S \perG Y) \to \Cat{C}(X,Y)\), for all objects \(S\), \(X\) and \(Y\) of \(\Cat{C}\), that satisfies the axioms in \Cref{tab:trace-axioms} for all suitably typed \(f\), \(g\), \(u\) and \(v\).
 A \emph{morphism of traced  monoidal categories} is a symmetric monoidal functor \(\fun{F} \colon \Cat{B} \to \Cat{C}\) 
  that preserves the trace, namely \(\fun{F}(\trace_{S}f) = \trace_{\fun{F}S}(\fun{F}f)\). We write $\TSMC$ for the category of traced monoidal categories and their morphisms.
\end{defi}

String diagrams can be extended to deal with traces~\cite{joyal1996traced} (see e.g., \cite{selinger2010survey} for a survey). For a morphism $f\colon S \perG X \to S \perG Y$, we draw its trace as %
  \[ 
    \begin{tikzpicture}
      \begin{pgfonlayer}{nodelayer}
        \node [style=label] (105) at (-2.75, -0.625) {$X$};
        \node [style=none] (117) at (-2.25, -0.625) {};
        \node [style=none] (118) at (-1.75, 0.625) {};
        \node [style=label] (120) at (2.75, -0.625) {$Y$};
        \node [style=none] (125) at (2.25, -0.625) {};
        \node [style=none] (126) at (1.75, 0.625) {};
        \node [style=stringlongbox] (128) at (0, 0) {$f$};
        \node [style=none] (129) at (-1.75, 2.125) {};
        \node [style=none] (130) at (1.75, 2.125) {};
        \node [style=label] (131) at (-1.25, 1.125) {$S$};
        \node [style=label] (132) at (1.25, 1.125) {$S$};
      \end{pgfonlayer}
      \begin{pgfonlayer}{edgelayer}
        \draw (125.center) to (117.center);
        \draw (118.center) to (126.center);
        \draw (130.center) to (129.center);
        \draw [bend right=90, looseness=1.50] (129.center) to (118.center);
        \draw [bend left=90, looseness=1.50] (130.center) to (126.center);
      \end{pgfonlayer}
    \end{tikzpicture}.     
  \]
Using this convention, the axioms in \Cref{tab:trace-axioms} acquire a more intuitive flavour: see the string diagram in \Cref{fig:trace-axioms}.  %

In the paper we will often need to require the trace to be uniform. This constraint, that arises from technical necessity, turns out to be the key to recover the axiomatisation of Kleene algebras in Section \ref{sec:kleene}, the induction proof principle in Section \ref{sec:peano} and the proof rules for while loops in Hoare logic in Section \ref{sec:hoare}.
\begin{defi}\label{def:utr-category}
  A traced monoidal category \(\Cat{C}\)
  is \emph{uniformly traced} if the trace operator satisfies the implication in \Cref{tab:uniformity} for all suitably typed \(f\), \(g\) and \(r\). A \emph{morphism of uniformly traced monoidal categories} is simply a morphism of traced monoidal categories. The category of uniformly traced monoidal categories and their morphisms is denoted by  \(\UTSMC\).
\end{defi}

With string diagrams, the uniformity axiom is drawn as in  \Cref{fig:uniformity}.

\begin{rem}[Uniformity and sliding]\label{rem:uniformity-sliding}
  The sliding axiom is redundant as it follows from uniformity: 
    \[ 
    \InputIfFileExists{ufu.tikz}{}{\input{./tikz/ufu.tikz}}
 = 
    \InputIfFileExists{ufu.tikz}{}{\input{./tikz/ufu.tikz}}
 \implies 

}
 \]
 This fact will be useful for constructing the uniformly traced monoidal category freely generated by a monoidal category \(\Cat{C}\).
\end{rem}

\begin{figure}[htbp]
  \centering

  \begin{subfigure}{\textwidth}
    \centering

    \mylabel{ax:trace:tightening}{tightening}
    \mylabel{ax:trace:strength}{strength}
    \mylabel{ax:trace:joining}{joining}
    \mylabel{ax:trace:vanishing}{vanishing}
    \mylabel{ax:trace:yanking}{yanking}
    \mylabel{ax:trace:sliding}{sliding}

    \begin{tabular}{l r}
      \toprule
      \(\trace_{S}((\id{} \perG u) \dcomp f \dcomp (\id{} \perG v)) = 
      u \dcomp \trace_{S}f \dcomp v\) & (\ref*{ax:trace:tightening}) \\
      \(\trace_{S}(f \perG g) = \trace_{S}f \perG g\) & (\ref*{ax:trace:strength})\\
      \(\trace_{T} \trace_{S} f = \trace_{S \perG T} f\) & (\ref*{ax:trace:joining})\\
      \(\trace_{\unoG} f = f\) & (\ref*{ax:trace:vanishing})\\
      \(\trace_{T}(f \dcomp (u \perG \id{})) = 
      \trace_{S}((u \perG \id{}) \dcomp f)\) & (\ref*{ax:trace:sliding})\\
      \(\trace_{S} \sigma^\perG_{S,S} = \id{S}\) & (\ref*{ax:trace:yanking})\\
      \bottomrule
    \end{tabular}

    \caption{Algebraic formulation of the trace axioms.}
    \label{tab:trace-axioms}
  \end{subfigure}

  \vspace{1.2em}

  \begin{subfigure}{\textwidth}
    \centering

    \mylabel{ax:trace:uniformity}{uniformity}

    \begin{tabular}{l r}
      \toprule
      if \(f \dcomp (r \perG \id{}) = (r \perG \id{}) \dcomp g\), then
      \(\trace_{S} f = \trace_{T} g\) & (\ref*{ax:trace:uniformity})\\
      \bottomrule
    \end{tabular}

    \caption{Algebraic formulation of the uniformity axiom.}
    \label{tab:uniformity}
  \end{subfigure}

  \vspace{1.2em}

  \begin{subfigure}{\textwidth}
    \centering

    \[
    \begin{array}{c@{}c@{}c @{\qquad} c@{}c@{}c}
      
    \InputIfFileExists{traceAx/tight/lhs.tikz}{}{\input{./tikz/traceAx/tight/lhs.tikz}}
 &\stackrel{\text{(tightening)}}{=}& 
      
    \InputIfFileExists{traceAx/tight/rhs.tikz}{}{\input{./tikz/traceAx/tight/rhs.tikz}}

      &
      
    \InputIfFileExists{traceAx/strength/lhs.tikz}{}{\input{./tikz/traceAx/strength/lhs.tikz}}
 &\stackrel{\text{(strength)}}{=}& 
      
    \InputIfFileExists{traceAx/strength/rhs.tikz}{}{\input{./tikz/traceAx/strength/rhs.tikz}}

      \\[20pt]
      
    \InputIfFileExists{traceAx/join/lhs.tikz}{}{\input{./tikz/traceAx/join/lhs.tikz}}
 &\stackrel{\text{(joining)}}{=}& 
      
    \InputIfFileExists{traceAx/join/rhs.tikz}{}{\input{./tikz/traceAx/join/rhs.tikz}}

      &
      
    \InputIfFileExists{traceAx/vanish/lhs.tikz}{}{\input{./tikz/traceAx/vanish/lhs.tikz}}
 &\stackrel{\text{(vanishing)}}{=}& 
      
    \InputIfFileExists{traceAx/vanish/rhs.tikz}{}{\input{./tikz/traceAx/vanish/rhs.tikz}}

      \\[20pt]
      
    \InputIfFileExists{traceAx/slide/lhs.tikz}{}{\input{./tikz/traceAx/slide/lhs.tikz}}
 &\stackrel{\text{(sliding)}}{=}& 
      
    \InputIfFileExists{traceAx/slide/rhs.tikz}{}{\input{./tikz/traceAx/slide/rhs.tikz}}

      &
      
    \InputIfFileExists{traceAx/yank/lhs.tikz}{}{\input{./tikz/traceAx/yank/lhs.tikz}}
 &\stackrel{\text{(yanking)}}{=}& 
      
    \InputIfFileExists{traceAx/yank/rhs.tikz}{}{\input{./tikz/traceAx/yank/rhs.tikz}}

    \end{array}
    \]

    \caption{Diagrammatic formulation of the trace axioms.}
    \label{fig:trace-axioms}
  \end{subfigure}

  \vspace{1.2em}

  \begin{subfigure}{\textwidth}
    \centering

    \[
      \begin{array}{c}
        
    \InputIfFileExists{traceAx/uniformity/lhs.tikz}{}{\input{./tikz/traceAx/uniformity/lhs.tikz}}

        =
        
    \InputIfFileExists{traceAx/uniformity/rhs.tikz}{}{\input{./tikz/traceAx/uniformity/rhs.tikz}}

        \stackrel{(\text{uniformity})}{\implies}
        
    \InputIfFileExists{traceAx/uniformity/TR_lhs.tikz}{}{\input{./tikz/traceAx/uniformity/TR_lhs.tikz}}

        =
        
    \InputIfFileExists{traceAx/uniformity/TR_rhs.tikz}{}{\input{./tikz/traceAx/uniformity/TR_rhs.tikz}}

      \end{array}
    \]

    \caption{Diagrammatic formulation of the uniformity axiom.}
    \label{fig:uniformity}
  \end{subfigure}

  \caption{Algebraic and string diagrammatic presentations of the trace axioms and the uniformity axiom.}
  \label{fig:trace_all_four}
\end{figure}

\newpage
\section{Coherence Axioms for Rig Categories}\label{app:coherence axioms}

\begin{figure}[H]
    \begin{minipage}[t]{0.45\textwidth}
        \begin{equation}
            \label{eq:rigax1}\tag{R1}
            \scalebox{0.9}{$\input{tikz-cd/rigax1.tikz}$}
        \end{equation}    
    \end{minipage}
    \hfill
    \begin{minipage}[t]{0.45\textwidth}
        \begin{equation}
            \label{eq:rigax2}\tag{R2}
            \scalebox{0.9}{$\input{tikz-cd/rigax2.tikz}$}
        \end{equation}
    \end{minipage}
    
    \begin{equation}
        \label{eq:rigax3}\tag{R3}
        \scalebox{0.9}{$\input{tikz-cd/rigax3.tikz}$}
    \end{equation}
    \begin{equation}
        \label{eq:rigax4}\tag{R4}
        \scalebox{0.9}{$\input{tikz-cd/rigax4.tikz}$}
    \end{equation}
    \begin{equation}
        \label{eq:rigax5}\tag{R5}
        \scalebox{0.9}{$\input{tikz-cd/rigax5.tikz}$}
    \end{equation}

    \begin{minipage}[t]{0.25\textwidth}
        \begin{equation}
            \label{eq:rigax6}\tag{R6}
            \scalebox{0.9}{$\input{tikz-cd/rigax6.tikz}$}
        \end{equation}
    \end{minipage}
    \hfill
    \begin{minipage}[t]{0.45\textwidth}
        \begin{equation}
            \label{eq:rigax7}\tag{R7}
            \scalebox{0.9}{$\input{tikz-cd/rigax7.tikz}$}
        \end{equation}
    \end{minipage}
    \hfill
    \begin{minipage}[t]{0.25\textwidth}
        \begin{equation}
            \label{eq:rigax8}\tag{R8}
            \scalebox{0.9}{$\input{tikz-cd/rigax8.tikz}$}
        \end{equation}
    \end{minipage}
    \\
    \begin{minipage}[t]{0.48\textwidth}
        \begin{equation}
            \label{eq:rigax9}\tag{R9}
            \scalebox{0.9}{$\input{tikz-cd/rigax9.tikz}$}
        \end{equation}    
    \end{minipage}
    \hfill
    \begin{minipage}[t]{0.48\textwidth}
        \begin{equation}
            \label{eq:rigax10}\tag{R10}
            \scalebox{0.9}{$\input{tikz-cd/rigax10.tikz}$}
        \end{equation}
    \end{minipage}
    \\
    \begin{minipage}[t]{0.50\textwidth}
        \begin{equation}
            \label{eq:rigax11}\tag{R11}
            \scalebox{0.9}{$\input{tikz-cd/rigax11.tikz}$}
        \end{equation}    
    \end{minipage}
    \hfill
    \begin{minipage}[t]{0.46\textwidth}
        \begin{equation}
            \label{eq:rigax12}\tag{R12}
            \scalebox{0.9}{$\input{tikz-cd/rigax12.tikz}$}
        \end{equation}
    \end{minipage}
    \caption{Coherence Axioms of symmetric rig categories}
    \label{fig:rigax}
\end{figure}

\newpage

\newpage

\section{Appendix to Section \ref{sec:kleene}}\label{app:Kleene} 

In this appendix we discuss some properties of Kleene bicategories, we prove Theorem \ref{prop:trace-star} and we illustrate the matrix construction that allows to transform any typed Kleene algebra into a Kleene bicategory.

\begin{lem}\label{lemma:order-adjointness}
Let $(\Cat{C}, \piu, \zero)$ be an fb category. If $(\Cat{C}, \piu, \zero)$ is a poset enriched symmetric monoidal category and the laws in \eqref{eq:adjointnesfib} hold then for all $f,g \colon X \to Y$,
\begin{center}
$f \leq g$ iff $f \sqcup g=g$. 
\end{center}
\end{lem}
\begin{proof}%
  For the $(\implies)$ direction we assume $f \leq g$ and prove separately the following two inclusions.

  \noindent\begin{minipage}{0.48\textwidth}
    \begin{align*}
      

}
 \tag{\ref{ax:adjbiprod:4}}
    \end{align*}
  \end{minipage}

  \medskip

  For the $(\impliedby)$ direction we assume $f \sqcup g = g$ and prove the following inclusion.
  \begin{equation*}
    
    %
}

  \end{equation*}
\end{proof}

\begin{proof}[Proof of Lemma \ref{lem:idempfib}]
To prove that $1 \Rightarrow 2$, first observe that by Lemma \ref{lemma:order-adjointness}, we know that the ordering is forced to be the one induced by $\sqcup$. Then $\bot \leq f$ for all arrows $f\colon X\to Y$ and thus $f=f\sqcup \bot \leq f \sqcup f$. For the opposite inequality we have the following derivation.
\begin{align*}
f\sqcup f &= (f; \id{Y}) \sqcup (f; \id{Y})\\
&= f; (\id{Y} \sqcup \id{Y}) \tag{\ref{eq:cmon enrichment}}\\
&\leq f; \id{Y} \tag{\ref{eq:adjointnesfib}}\\
&=f
\end{align*}
To prove that $2 \Rightarrow 1$, define $f\leq  g$ as $f\sqcup g=g$. Straightforward computations prove the four laws in \eqref{eq:adjointnesfib}. To prove that $(\Cat{C}, \piu, \zero)$ is poset enriched one can use the enrichment over commutative monoids. For instance assuming that $f_1\leq f_2$, one can conclude that $f_1;g \leq f_2;g$ as follows.
\begin{align*}
f_1; g \sqcup f_2 ; g &= f_1; g \sqcup ( \,(f_1\sqcup f_2) ; g \,) \tag{$f_1\leq f_2$} \\
&= f_1;g \sqcup f_1;g \sqcup f_2;g \tag{\ref{eq:cmon enrichment}}\\
&=  f_1;g \sqcup f_2;g \tag{Idempotency}\\
&=  (f_1 \sqcup f_2 );g \tag{\ref{eq:cmon enrichment}}\\
&=   f_2 ;g \tag{$f_1\leq f_2$}
\end{align*}
\end{proof}

\begin{proof}[Proof of Proposition \ref{prop:matrixform}]
The normal form of fb category is well known: see e.g. \cite[Proposition 2.7]{harding2008orthomodularity}.
For the ordering, observe that if $f \leq g$ then, since $\Cat{C}$ is poset enriched,  
\[(\id{S} \piu \cobang{X}); f ; (\id{T} \piu \cobang{Y}) \leq  (\id{S} \piu \cobang{X}); g ; (\id{T} \piu \cobang{Y})\]
that is $f_{ST} \leq g_{ST}$. Similarly for the others.

Vice versa from $f_{ST} \leq g_{ST}$ $f_{SY} \leq g_{SY}$
$f_{XT} \leq g_{XT}$ and $ f_{XY} \leq g_{XY}$, one can use the normal form to deduce immediately that $f\leq g$.
\end{proof}

The two posetal uniformity axioms enjoy  alternative characterisations that will be useful later on through this appendix: 
\begin{equation}\label{eq:equivalentuni1}\tag{AU1'}
\text{If }\exists r_1,r_2\colon S \to T\text{ s.t. } r_2 \leq r_1\text{ and }f ; (r_1 \piu \id{Y}) \leq (r_2 \piu \id{X}) ; g\text{, then }\trace_{S}f \leq \trace_{T}g\text{;}
\end{equation}
\begin{equation}\label{eq:equivalentuni2}\tag{AU2'}
\text{If }\exists r_1,r_2\colon T \to S\text{ s.t. } r_2 \leq r_1\text{ and }(r_1 \piu \id{X}) ; f   \leq   g; (r_2 \piu \id{Y})\text{, then }\trace_{S}f \leq \trace_{T}g\text{;}
\end{equation}
\begin{lem}\label{lemma:equivalentUnif}
The following hold:
\begin{enumerate}
\item $(AU1)$ iff $(AU1')$; 
\item   $(AU2)$ iff $(AU2')$.
\end{enumerate}
\end{lem}
\begin{proof}
We prove the first point. The second one is analogous.

Since the conclusions of  $(AU1)$ and $(AU1')$ are identical, it is enough to prove the equivalence of the premises of the two laws.
\begin{itemize}
\item We prove that the premises of $(AU1')$ entail $(AU1)$. Assume that $\exists r_1,r_2$ such that (a) $r_2 \leq r_1$ and (b) $f ; (r_1 \piu \id{}) \leq (r_2 \piu \id{}) ; g$. Thus:
\[ f ; (r_2 \piu \id{}) \stackrel{(a)}{\leq} f ; (r_1 \piu \id{})\stackrel{(b)}{\leq} (r_2 \piu \id{}) ; g\]
Observe that by replacing $r_2$ by $r$ in the above, one obtains exactly the premise of $(AU1)$.
\item We prove that $(AU1)$ entails $(AU1')$. Assume that $(AU1)$ holds. Then $(AU1)'$ holds by taking both $r_1$ and $r_2$ to be $r$.
\end{itemize}
\end{proof}

\subsection{Proof of Theorem \ref{prop:trace-star}}\label{app:star theorem}

In order to prove Theorem \ref{prop:trace-star}, we fix  the correspondence between Kleene star operator and trace: see \Cref{fig:star-trace}.

\begin{figure}[t]
\centering
\[
\begin{array}{l@{\qquad}l}
 f \colon X \to X & a \colon S \piu X \to S \piu Y \\[5pt]
 \kstar{f} \defeq \Crepetition{f}{X}{X} & \trace_S a \defeq \Ctracerep{a\vphantom{\kstar{a}_{SS}}}{S}{X}{Y}
 \end{array}
 \]
 \caption{Kleene star from trace and trace from Kleene star in finite biproduct categories.\label{fig:star-trace}}
 \end{figure}

Our argument relies on the following result from \cite{cuazuanescu1994feedback} (see also \cite{selinger2010survey}).
\begin{prop}[From \cite{cuazuanescu1994feedback}]\label{prop:stef} In a category $\Cat{C}$ with finite biproducts, giving a trace is equivalent to giving a repetition operation, i.e., a family of operators \(\kstar{(\cdot)} \colon \Cat{C}(S,S) \to  \Cat{C}(S,S)\) satisfying the following three laws.
\begin{equation}\label{eq:fromstefanescu}
\kstar{f}=\id{} \sqcup f;\kstar{f} \qquad \kstar{(f \sqcup g)}= \kstar{(\kstar{f};g)};\kstar{f} \qquad \kstar{(f;g)};f = f;\kstar{(g;f)}
\end{equation}
\end{prop}
For the sake of completeness, we illustrate also the following that is closely related to the leftmost in \eqref{eq:fromstefanescu}.
\begin{prop}\label{prop:star-fixpoint}
  Let \(\Cat{C}\) be a monoidal category with finite biproducts and trace. For each $f\colon X \to X$ define $\kstar{f}$ as in \Cref{fig:star-trace}. %
  Then, 
  \[\kstar{f}=\id{X} \sqcup \kstar{f};f\]
\end{prop}
\begin{proof} 
  \begingroup
  \allowdisplaybreaks
  \begin{align*}
   \id{} \sqcup \kstar{f};f
    &= 
    \InputIfFileExists{propStarFix/step1.tikz}{}{\input{./tikz/propStarFix/step1.tikz}}
 \tag{\ref{eq:covolution}} \\
    &= 
    \InputIfFileExists{propStarFix/step1slide.tikz}{}{\input{./tikz/propStarFix/step1slide.tikz}}
 \tag{\Cref{fig:star-trace}}  \\
    &= 
    \InputIfFileExists{propStarFix/step2.tikz}{}{\input{./tikz/propStarFix/step2.tikz}}
 \tag{\ref{ax:trace:sliding}}  \\
    &= 
    \InputIfFileExists{propStarFix/step3.tikz}{}{\input{./tikz/propStarFix/step3.tikz}}
 \tag{\ref{ax:monoid:nat:copy}} \\
    &= 
    \InputIfFileExists{propStarFix/step4.tikz}{}{\input{./tikz/propStarFix/step4.tikz}}
 \tag{\ref{ax:trace:sliding}} \\
    &= 
    \InputIfFileExists{propStarFix/step5.tikz}{}{\input{./tikz/propStarFix/step5.tikz}}
 \tag{\ref{ax:trace:yanking}} \\
    &= 
    \InputIfFileExists{propStarFix/step6.tikz}{}{\input{./tikz/propStarFix/step6.tikz}}
 \tag{\ref{ax:comonoid:nat:copy}} \\
    &= \kstar{f} \tag{\Cref{fig:star-trace}}
  \end{align*}
  \endgroup
\end{proof}

\begin{cor}\label{cor:bounds-star}
   Let \(\Cat{C}\) be a fb category with idempotent convolution and trace. %
  The following inequalities hold for all $f\colon X \to X$:
\[      \begin{array}{c}
      \id{X} \sqcup f \dcomp \kstar{f} \leq \kstar{f}  \\
      \id{X} \sqcup \kstar{f} \dcomp f \leq \kstar{f} 
    \end{array}
  \]
In particular,   $f \dcomp \kstar{f} \leq \kstar{f} $, $ \kstar{f}; f \leq \kstar{f} $ and $\id{X} \leq \kstar{f}$.
\end{cor}
\begin{proof}
  By \Cref{lemma:order-adjointness}, \Cref{prop:stef} and \Cref{prop:star-fixpoint}.
\end{proof}

\begin{lem}\label{lemma:uniform-star}
  Let $\Cat{C}$ be a poset enriched monoidal category with finite biproducts and trace. $\Cat{C}$ satisfies the axioms in \Cref{fig:ineq-uniformity} iff it satisfies those in \Cref{fig:uniform-star}.
  \begin{figure}[ht!]
    \centering
    \[
    \begin{array}{c}
      
    \InputIfFileExists{equivunif/EU1_lhs.tikz}{}{\input{./tikz/equivunif/EU1_lhs.tikz}}
 \leq 
    \InputIfFileExists{equivunif/EU1_rhs.tikz}{}{\input{./tikz/equivunif/EU1_rhs.tikz}}
 \implies 
    \InputIfFileExists{equivunif/EU1TR_lhs.tikz}{}{\input{./tikz/equivunif/EU1TR_lhs.tikz}}
 \leq 
    \InputIfFileExists{equivunif/EU1TR_rhs.tikz}{}{\input{./tikz/equivunif/EU1TR_rhs.tikz}}

      \\[20pt]
      
    \InputIfFileExists{equivunif/EU1_lhs.tikz}{}{\input{./tikz/equivunif/EU1_lhs.tikz}}
 \geq 
    \InputIfFileExists{equivunif/EU1_rhs.tikz}{}{\input{./tikz/equivunif/EU1_rhs.tikz}}
 \implies 
    \InputIfFileExists{equivunif/EU1TR_lhs.tikz}{}{\input{./tikz/equivunif/EU1TR_lhs.tikz}}
 \geq 
    \InputIfFileExists{equivunif/EU1TR_rhs.tikz}{}{\input{./tikz/equivunif/EU1TR_rhs.tikz}}

    \end{array}
    \]
    \caption{Equivalent uniformity axioms.\label{fig:uniform-star}}
  \end{figure}
\end{lem}
\begin{proof}
  The poset enriched monoidal category obtained by inverting the 2-cells also has biproducts and trace.
  Thus, we show the first of the implications in \Cref{fig:uniform-star} and \Cref{fig:ineq-uniformity}, while the other ones follow by this observation.

  For one direction, suppose that the trace satisfies the axioms in \Cref{fig:ineq-uniformity}, and consider \(f \colon X \to X\), \(g \colon Y \to Y\) and \(r \colon X \to Y\) in \(\Cat{C}\) such that \(f \dcomp r \leq r \dcomp g\).
Observe that %
  \begin{align*}
    

}
 \tag{Hypothesis}
  \end{align*}
  By the first implication in \Cref{fig:ineq-uniformity}, we obtain
  \[ 
    %
}
 .\]

  For the other direction, suppose that the axioms in \Cref{fig:uniform-star} hold, and consider \(f \colon S \piu X \to S \piu Y\), \(g \colon T \piu X \to T \piu Y\) and \(r \colon S \to T\) in \(\Cat{C}\) such that \(f \dcomp (r \piu \id{}) \leq (r \piu \id{}) \dcomp g\).
  Since \(\Cat{C}\) has biproducts, both \(f\) and \(g\) can be written in matrix normal form to obtain element-by-element inequalities.
  \begingroup
  \allowdisplaybreaks
  \begin{align*}
    
    %
 \right. \tag{Hypothesis} \\
  \end{align*}
  \endgroup
  With these inequalities, we show the inequality between the traces.
  \begingroup
  \allowdisplaybreaks
  \begin{align*}
    
    \InputIfFileExists{lemmaUnifEquiv/LTR/traceEquiv/step1.tikz}{}{\input{./tikz/lemmaUnifEquiv/LTR/traceEquiv/step1.tikz}}
 
    &= 
    \InputIfFileExists{lemmaUnifEquiv/LTR/traceEquiv/step2.tikz}{}{\input{./tikz/lemmaUnifEquiv/LTR/traceEquiv/step2.tikz}}
 \tag{Proposition \ref{prop:matrixform}} \\
    &= 
    \InputIfFileExists{lemmaUnifEquiv/LTR/traceEquiv/step3.tikz}{}{\input{./tikz/lemmaUnifEquiv/LTR/traceEquiv/step3.tikz}}
 \tag{trace axioms} \\
    &\leq 
    \InputIfFileExists{lemmaUnifEquiv/LTR/traceEquiv/step4.tikz}{}{\input{./tikz/lemmaUnifEquiv/LTR/traceEquiv/step4.tikz}}
 \tag{ii} \\
    &\leq 
    \InputIfFileExists{lemmaUnifEquiv/LTR/traceEquiv/step5.tikz}{}{\input{./tikz/lemmaUnifEquiv/LTR/traceEquiv/step5.tikz}}
 \tag{i} \\
    &\leq 
    \InputIfFileExists{lemmaUnifEquiv/LTR/traceEquiv/step6.tikz}{}{\input{./tikz/lemmaUnifEquiv/LTR/traceEquiv/step6.tikz}}
 \tag{iii} \\
    &\leq 
    \InputIfFileExists{lemmaUnifEquiv/LTR/traceEquiv/step6.tikz}{}{\input{./tikz/lemmaUnifEquiv/LTR/traceEquiv/step6.tikz}}
 \tag{iv} \\
    &= 
    \InputIfFileExists{lemmaUnifEquiv/LTR/traceEquiv/step8.tikz}{}{\input{./tikz/lemmaUnifEquiv/LTR/traceEquiv/step8.tikz}}
 \tag{trace axioms} \\
    &= 
    \InputIfFileExists{lemmaUnifEquiv/LTR/traceEquiv/step9.tikz}{}{\input{./tikz/lemmaUnifEquiv/LTR/traceEquiv/step9.tikz}}
 \tag{Proposition \ref{prop:matrixform}}
  \end{align*}
  \endgroup
\end{proof}

The above results can be rephrased in terms of $\kstar{(\cdot)}$ as defined in Figure \ref{fig:star-trace}: 
$\Cat{C}$ satisfies the axioms in \Cref{fig:ineq-uniformity} iff $\kstar{(\cdot)}$ satisfies
\begin{equation}\label{eq:unifstar}
\begin{array}{r@{}c@{}l }
     f \dcomp r \leq r \dcomp g & \implies & \kstar{f} \dcomp r \leq r \dcomp \kstar{g} \\
    f \dcomp r \geq r \dcomp g & \implies & \kstar{f} \dcomp r \geq r \dcomp \kstar{g}
\end{array}
\end{equation}

\begin{rem}
It is worth remarking that in \cite{cuazuanescu1994feedback}, it was proved that the implications obtained by replacing $\leq$ by $=$ in \eqref{eq:unifstar} are equivalent to the standard uniformity axioms.
\end{rem}

It is also immediate to see that the axiom \eqref{ax:kb:traceid} in \Cref{fig:ineq-uniformity} is equivalent to the following.
\begin{equation}\label{eq:happystar}
\kstar{\id{}} \leq \id{}
\end{equation}

\begin{lem}\label{lemma:uniform-kozen}
 Let $\Cat{C}$ be a fb category with idempotent convolution and trace. 
 $\Cat{C}$ satisfies the axioms in \Cref{fig:ineq-uniformity} iff $\kstar{(\cdot)}$ as defined in Figure \ref{fig:star-trace} satisfies the following:
\begin{equation}\label{eq:stroingstar}     
\begin{array}{r@{}c@{}l}
 f \dcomp r \leq r & \implies & \kstar{f} \dcomp r \leq r \\
 l \dcomp f \leq l & \implies & l \dcomp \kstar{f}  \leq l  
    \end{array}
\end{equation}

\end{lem}

\begin{proof}

We prove that \eqref{eq:unifstar} and \eqref{eq:happystar} hold iff \eqref{eq:stroingstar} holds.

  For one direction, assume that \eqref{eq:unifstar} and \eqref{eq:happystar} hold. To prove that the first implication in \eqref{eq:stroingstar} holds, 
 consider \(f \colon X \to X\) and \(r \colon X \to Y\) such that \(f \dcomp r \leq r\).
  Then, \(f \dcomp r \leq r \dcomp \id{Y}\) and,
  \begin{align*}
    \kstar{f} \dcomp r & \leq r \dcomp \kstar{\id{Y}} \tag{\eqref{eq:unifstar}}\\
    & \leq r \dcomp \id{Y} \tag{\eqref{eq:happystar}}\\
    & = r
  \end{align*}
  The second implication follows the symmetric argument. %
  
  For the other direction, assume that \eqref{eq:stroingstar} holds.  To prove \eqref{eq:happystar}, observe that \(\id{} \dcomp \id{} \leq \id{}\).   By \eqref{eq:stroingstar}, \(\kstar{\id{}} = \kstar{\id{}} \dcomp \id{} \leq \id{}\).
  
  To prove the first implication in  \eqref{eq:unifstar}, consider \(f \colon X \to X\), \(g \colon Y \to Y\) and \(r \colon X \to Y\) such that \(f \dcomp r \leq r \dcomp g\).
  Then \(f \dcomp r \dcomp \kstar{g} \leq r \dcomp g \dcomp \kstar{g} \leq r \dcomp \kstar{g}\), where the latter inequality holds by \Cref{cor:bounds-star}.
  By \eqref{eq:stroingstar}, \(\kstar{f} \dcomp r \dcomp \kstar{g} \leq r \dcomp \kstar{g}\).
  By \Cref{cor:bounds-star}, \(\kstar{f} \dcomp r \leq \kstar{f} \dcomp r \dcomp \kstar{g}\), which gives  \(\kstar{f} \dcomp r \leq r \dcomp \kstar{g}\).
  
  To prove the second implication in  \eqref{eq:unifstar}, we proceed similarly: assume that  \(  r \dcomp g \leq f \dcomp r \).
  Then \(\kstar{f} \dcomp r \dcomp g \leq   \kstar{f} \dcomp f \dcomp r \leq  \kstar{f} \dcomp r\), where the latter inequality holds by \Cref{cor:bounds-star}.
  By the second implication in \eqref{eq:stroingstar}, \(\kstar{f} \dcomp r \dcomp \kstar{g} \leq \kstar{f}\dcomp r\).
  By \Cref{cor:bounds-star}, \(  r \dcomp \kstar{g} \leq \kstar{f} \dcomp r \dcomp \kstar{g}\), which gives  \( r \dcomp \kstar{g} \leq \kstar{f} \dcomp r \).
  
\end{proof}
We have now all the ingredients to prove \Cref{prop:trace-star}.
\begin{proof}[Theorem \ref{prop:trace-star}]
  Suppose that \(\Cat{C}\) is a Kleene bicategory. Then one can define a $\kstar{(\cdot)}$ as in \Cref{fig:star-trace}. By Corollary \ref{cor:bounds-star} and
 \Cref{lemma:uniform-kozen}, $\kstar{(\cdot)}$ satisfies the laws in \eqref{eq:Kllenelaw}. Thus, it is a Kleene star.

  Conversely, suppose that \(\Cat{C}\) has a Kleene star operator $\kstar{(\cdot)}$. One can easily show (e.g., by using completeness of Kozen's axiomatisation in \cite{Kozen94acompleteness}) that the laws of Kleene star in \eqref{eq:Kllenelaw} entail those in \eqref{eq:fromstefanescu}. Thus, by Proposition \ref{prop:stef}, $\kstar{(\cdot)}$ gives us a trace as defined in the right of  \Cref{fig:star-trace}.
 By \Cref{lemma:uniform-kozen}, this trace satisfies the laws in \Cref{fig:ineq-uniformity}. Thus $\Cat{C}$ is a Kleene bicategory.
\end{proof}

\subsection{The Matrix Construction}\label{ssec:matrix}

Thanks to Corollary \ref{cor:kleeneareka}, one can easily construct a forgetful functor $U\colon \KBicat \to \TKAlg$: any Kleene bicategory is a typed Kleene algebra and any morphism of Kleene bicategories is a morphism of typed Kleene algebras. To see the latter, observe that preserving the join-semilattice and $\kstar{(\cdot)}$, as defined in \eqref{eq:covolution} and \eqref{eq:star}, is enough to preserve monoidal product, (co)monoids, and traces.

We now illustrate that $U\colon \KBicat \to \TKAlg$ has a left adjoint provided by the matrix construction. In \cite[Exercises VIII.2.5-6]{mac_lane_categories_1978}, it is shown that there exists an adjunction between $\CMonCat$, the category of $\Cat{CMon}$-enriched categories, and $\FBC$, the category of fb categories.
\begin{equation}\label{eq:matrixadj}
\begin{tikzcd}
        \CMonCat
        \arrow[r, "\MatFun"{name=F}, bend left] &
        \FBC
        \arrow[l, "\fun{U}"{name=U}, bend left]
        \arrow[phantom, from=F, to=U, "\vdash" rotate=90]
    \end{tikzcd}
\end{equation}
The functor $U$ is the obvious forgetful functor: every fb category is  $\Cat{CMon}$-enriched. 
Given a $\CMon$-enriched category $\Cat S$, one can form the biproduct completion of $\Cat S$, denoted as $\Mat{\Cat S}$. Its objects are formal $\piu$'s of objects of $\Cat S$, while a morphism $M \colon \Piu[k=1][n]{A_k} \to \Piu[k=1][m]{B_k}$ is a $m \times n$ matrix where $M_{ji} \in \Cat S[A_i,B_j]$. Composition is given by matrix multiplication, with the addition being the plus operation on the homsets (provided by the enrichment) and multiplication being composition. The identity morphism of $\Piu[k=1][n]{A_k}$ is given by the $n \times n$ matrix $(\delta_{ji})$, where $\delta_{ji} = \id{A_j}$ if $i=j$, while if $i \neq j$, then $\delta_{ji}$ is the zero morphism of $\Cat S[A_i,A_j]$.

\begin{prop}\label{prop:matrices-kleene-bicategory}
Let  \(\Cat{K}\) be a typed Kleene algebra. Then  \(\Mat{\Cat{K}}\) is a Kleene bicategory.
\end{prop}
\begin{proof} %
  By \eqref{eq:matrixadj}, \(\Mat{\Cat{K}}\) has finite biproducts.
  The posetal structure is defined element-wise from the posetal structure of \(\Cat{K}\). 
  We check that it gives adjoint biproducts. The following two derivations prove $\codiag{} \dashv \diag{}$.

  \noindent\begin{minipage}{0.44\textwidth}
    \begin{align*}
      
}
 \tag{$(\, \id{X} \,) \sqcup (\, \id{X} \,) = (\, \id{X} \,)$}
    \end{align*}
  \end{minipage}

  \bigskip

  \noindent The following two derivations prove $\cobang{} \dashv \bang{}$.

  \noindent\begin{minipage}{0.44\textwidth}
    \begin{align*}
      
    %
}
 \tag{$0_{X,X} \leq ( \, \id{X} \, )$}
    \end{align*}
  \end{minipage}

  \bigskip
  
  By Lemma 3.3 in \cite{Kozen94acompleteness}, \(\Mat{\Cat{K}}\) has a Kleene star operator. Thus, by \Cref{prop:trace-star}, $\Mat{\Cat{K}}$ is a Kleene bicategory.
\end{proof}

More generally, one can show that  the functor \(\MatFun \colon \CMonCat \to \fbCat\) restricts to typed Kleene algebras and Kleene bicategories
and that this gives rise to the left adjoint to $U\colon \KBicat \to \TKAlg$.

\begin{cor}\label{cor:adjKleene}
The adjunction in \eqref{eq:matrixadj} restricts to 
\[
\begin{tikzcd}
        \TKAlg
        \arrow[r, "\MatFun"{name=F}, bend left] &
        \KBicat
        \arrow[l, "\fun{U}"{name=U}, bend left]
        \arrow[phantom, from=F, to=U, "\vdash" rotate=90]
    \end{tikzcd}
\]
\end{cor}

\newpage

\section{Appendix to Section \ref{sec:cb}} %

The following result provides the left-version of the law in \eqref{eq:missing}.
\begin{lem}\label{lemma:leftwisktrace}
Let $\Cat{C}$ be a kc-rig category. For all arrows $f\colon S\piu X \to S \piu  Y$ and objects $Z$
\[\id{Z} \per \trace_{S}(f) = \trace_{Z \per S}(\Idl{Z}{S}{X} ; (\id{Z} \per f) ; \dl{Z}{S}{Y} )\]
\end{lem}
\begin{proof}
By the following derivation.
\begin{align*}
\id{Z} \per \trace_{S}(f) &= \symmt{Z}{X};  (\trace_{S}(f) \otimes \id{Z}) ; \symmt{Y}{Z} \tag{Symmetric Monoidal Categories}\\
&=  \symmt{Z}{X};  \trace_{S \otimes Z}(f \otimes \id{Z}) ; \symmt{Y}{Z}  \tag{\ref{eq:missing}}\\
&=   \trace_{S \otimes Z}\big( (\id{S\otimes Z} \per \symmt{Z}{X}) ;  (f \otimes \id{Z}) ; (\id{S\otimes Z} \piu \symmt{Y}{Z}) \big) \tag{Tightening}\\
&=   \trace_{S \otimes Z}\big( (\id{S\otimes Z} \per \symmt{Z}{X}) ;  (f \otimes \id{Z}) ; ((\symmt{S}{Z}; \symmt{Z}{S}) \piu \symmt{Y}{Z}) \big) \tag{Symmetry}\\
&=   \trace_{Z \otimes S}\big( (\symmt{Z}{S} \per \symmt{Z}{X}) ;  (f \otimes \id{Z}) ; ((\symmt{S}{Z}; ) \piu \symmt{Y}{Z}) \big) \tag{Sliding}\\
&= \trace_{Z \per S}(\Idl{Z}{S}{X} ; (\id{Z} \per f) ; \dl{Z}{S}{Y} ) \tag{\Cref{fig:rigax}}
\end{align*}
\end{proof}

The following  result illustrates the interaction of the product $\per$  with the $\oplus$-trace. Although we never used it, we hope that it can be of some interest,
\begin{prop}\label{prop:trace-per}
	For all $f \colon S \piu X \to S \piu Y$ and $f' \colon S' \piu X' \to S' \piu Y'$ in a Kleene rig category, it holds that 
	\[ 
		\trace_{S \per S'} \begin{psmallmatrix}
			f_{SS} \per f'_{S' S'} & f_{SY} \per f'_{S'Y'} \\
			f_{XS} \per f'_{X' S'} & f_{XY} \per f'_{X'Y'}
		\end{psmallmatrix}
		\; \leq \;
		\trace_{S} f \; \per \; \trace_{S'} f'.
	 \]
	 where $\begin{psmallmatrix} f_{SS} & f_{SY} \\ f_{XS} & f_{XY} \end{psmallmatrix}$ and $\begin{psmallmatrix} f'_{S'S'} & f'_{S'Y'} \\ f'_{X'S'} & f'_{X'Y'} \end{psmallmatrix}$ are, respectively, the matrix normal forms of $f$ and $f'$.
\end{prop}
\begin{proof}
	\begin{align*}
		\trace_{S} f \; \per \; \trace_{S'} f'
		&= (f_{XS} ; \kstar{f_{SS}} ; f_{SY} \sqcup  f_{XY}) \per (f'_{X'S'} ; \kstar{(f'_{S'S'})} ; f'_{S'Y'} \sqcup  f'_{X'Y'}) \tag{\Cref{fig:star-trace}} \\
		&= \begin{array}[t]{cl}
			\multicolumn{2}{c}{ ((f_{XS} ; \kstar{f_{SS}} ; f_{SY}) \per (f'_{X'S'} ; \kstar{(f'_{S'S'})} ; f'_{S'Y'})) } \\ 
			\sqcup  & ((f_{XS} ; \kstar{f_{SS}} ; f_{SY} \sqcup  f_{XY}) \per f'_{X'Y'}) \\
			\sqcup  & (f_{XY} \per (f'_{X'S'} ; \kstar{(f'_{S'S'})} ; f'_{S'Y'})) \\
			\sqcup  & (f_{XY} \per f'_{X'Y'})
		\end{array} \tag{\Cref{lemma:fb-cbsummary}.(2)}  \\
		&\geq { ((f_{XS} ; \kstar{f_{SS}} ; f_{SY}) \per (f'_{X'S'} ; \kstar{(f'_{S'S'})} ; f'_{S'Y'})) \; \sqcup  \; (f_{XY} \per f'_{X'Y'}) }  \tag{\Cref{lemma:order-adjointness}} \\
		&= (f_{XS} \per f'_{X'S'}) ; (\kstar{f_{SS}} \per \kstar{(f'_{S'S'})}) ; (f_{SY} \per f'_{X'Y'}) \sqcup  (f_{XY} \per f'_{X'Y'}) \tag{Functoriality of $\per$} \\
		&\geq (f_{XS} \per f'_{X'S'}) ; \kstar{(f_{SS} \per f'_{S'S'})} ; (f_{SY} \per f'_{X'Y'}) \sqcup  (f_{XY} \per f'_{X'Y'}) \tag{\Cref{prop:star-per}} \\
		&= \trace_{S \per S'} \begin{psmallmatrix}
			f_{SS} \per f'_{S' S'} & f_{SY} \per f'_{S'Y'} \\
			f_{XS} \per f'_{X' S'} & f_{XY} \per f'_{X'Y'}
		\end{psmallmatrix} \tag{\Cref{fig:star-trace}}
	\end{align*}
\end{proof}

\newpage
\section{Free Kleene-Cartesian bicategories}\label{sec:free-kleene-cartesian}
This section constructs the free Kleene bicategory on a signature and the free 
Kleene-Cartesian bicategory on a monoidal signature, thereby completing the 
proof of \Cref{thm:KleeneCartesiantapesfree}.  
Our approach is modular: we assemble a number of intermediate results concerning 
structures that have not yet been defined explicitly.  
The following definitions---where we tacitly assume that all morphisms are 
structure-preserving functors---serve to summarise these intermediate notions.

\begin{defi}
A \emph{uniformly traced fb category with idempotent convolution} is a monoidal category $(\Cat{C}, \piu, \zero)$ that is both a fb category with idempotent convolution (\Cref{def:biproduct category}) and a uniformly traced monoidal category (\Cref{def:utr-category}). We write $\Cat{UTFIBCat}$ for the category of uniformly traced fb categories with idempotent convolution and their morphisms. %
\end{defi}

\begin{defi} Let $(\Cat{C}, \piu, \per, \zero, \uno)$ be a rig category.
\begin{enumerate}
\item $\Cat{C}$ is a \emph{finite biproduct rig category} (shortly fb rig) if $(\Cat{C},\piu, \zero)$ is a fb category (\Cref{def:fb}); We write $\Cat{FBRig}$ for the category of finite biproduct rig categories and their morphisms.
\item $\Cat{C}$ is a \emph{uniformly traced finite biproduct rig category with idempotent convolution} if $(\Cat{C},\piu, \zero)$ is a uniformly traced fb category with idempotent convolution and \eqref{eq:missing} holds; We write $\Cat{UTFIBRig}$ for the category of such categories and their morphisms.
\item $\Cat{C}$ is a \emph{Kleene rig category} if $(\Cat{C},\piu, \zero)$ is a Kleene bicategory (\Cref{def:kleenebicategory}) and \eqref{eq:missing} holds. We write $\Cat{KRig}$ for the category of such categories and their morphisms.
\item $\Cat{C}$ is a \emph{finite biproduct cartesian rig category} (shortly fb-cb rig) if $(\Cat{C},\piu, \zero)$ is a fb category (\Cref{def:fb}), $(\Cat{C},\per, \uno)$ is a Cartesian bicategory (\Cref{def:cartesian bicategory}) and \eqref{eq:fbcbcoherence} hold. We write $\Cat{FBCB}$ for the category of such categories and their morphisms.
\item $\Cat{C}$ is a \emph{uniformly traced fb-cb rig category with idempotent convolution} if it is both a uniformly traced finite biproduct rig category with idempotent convolution and a fb-cb rig category. We write $\Cat{UTFIBCB}$ for the category of such categories and their morphisms.
\end{enumerate}
\end{defi}

Before delving into the proof, we illustrate our overall strategy:
\begin{equation}%
  \label{diag:proof-strategy-free-kc}
  \begin{tikzcd}
    {\Cat{Sig}} \arrow[bend left=15]{r}{\fun{M}} \arrow[phantom]{r}{\text{\cite{mac_lane_categories_1978}}} & {\Cat{FBCat}} \arrow[bend left=15]{r}{\fun{KB}} \arrow[phantom]{r}{\text{\ref{th:free-kleene-on-signature}}} \arrow[bend left=15]{l}{U} & {\Cat{KBicat}} \arrow[bend left=15]{l}{U}\\
    {\Cat{MSig}}  \arrow[bend left=15]{r}{\fun{FBCT}} \arrow[phantom]{r}{\text{\ref{prop:free-fb-cb-rig}}} & {\Cat{FBCB}} \arrow{u}{\fun{J}} \arrow[bend left=15]{r}{\fun{KC}} \arrow[phantom]{r}{\text{\ref{prop:free-kleene-on-fbcb}}} \arrow[bend left=15]{l}{\fun{U_L}} & {\Cat{KCB}} \arrow{u}[swap]{\fun{H}} \arrow[bend left=15]{l}{U}
  \end{tikzcd}
\end{equation}
In the top row, the leftmost adjunction \(\adjunction{\fun{M}}{\Cat{Sig}}{\Cat{FBCat}}{U}\) is known as the matrix construction~\cite{mac_lane_categories_1978}.
\Cref{ssec:free-kleene} proves the rightmost adjunction.
For this construction, we will need to quotient traced monoidal categories and apply a result about such quotients (\Cref{lem:traced-monoidal-congruence}).

\Cref{ssec:free-kleene-cartesian} proves the adjunctions on the bottom of~\eqref{diag:proof-strategy-free-kc} by restricting two adjunctions: the one of free finite-biproduct rig categories~\cite{bonchi2023deconstructing} and the one of \Cref{th:free-kleene-on-signature}.
For restricting adjunctions along faithful functors, we exploit another useful result (\Cref{lem:restricting-adjunctions}).

\subsection{Two convenient lemmas}

When quotienting the hom-sets of a category by an equivalence relation, we need to ensure that the equivalence is also a congruence with respect to the categorical structure~\cite[II.8]{mac_lane_categories_1978}.
Similarly, if we want to quotient the hom-sets of traced monoidal categories, we need to ensure that the equivalence is additionally a traced monoidal congruence.

\begin{defi}%
  \label{def:traced-monoidal-congruence}
  Let \((\Cat{C}, \perG, I, \trace)\) be a traced monoidal category.
  A \emph{traced monoidal congruence} \(R\) on \(\Cat{C}\) is a family of equivalence relations \(R(X,Y) \subseteq \Cat{C}[X,Y] \times \Cat{C}[X,Y]\) that respects compositions, monoidal products and traces: 
 \begin{itemize}
\item  if \((f,f') \in R(X,Y)\) and \((g,g') \in R(Y,Z)\), then \((f \dcomp g, f' \dcomp g') \in R(X,Z)\); 
\item if \((f,g) \in R(X,Y)\) and \((f',g') \in R(X',Y')\), then \((f \perG f', g \perG g') \in R(X \perG X', Y \perG Y')\); 
\item if \((f,g) \in R(S \perG X, S \perG Y)\), then \((\trace(f), \trace(g)) \in R(X,Y)\).
 \end{itemize}
\end{defi}

\begin{lem}\label{lem:traced-monoidal-congruence}
  Let \((\Cat{C}, \perG, I, \trace)\) be a traced monoidal category and \(R\) be a traced monoidal congruence on \(\Cat{C}\).
  Then, there is a traced monoidal category \(\Cat{C}\slash R\) with the same objects as \(\Cat{C}\) and where the hom-sets are quotients of the hom-sets of \(\Cat{C}\) by the equivalence relations \(R\).
  Moreover, there is a traced monoidal functor \(\eta_{R} \colon \Cat{C} \to \Cat{C}\slash R\) such that, for all traced monoidal functors \(H \colon \Cat{C} \to \Cat{D}\) with \(H(f) = H(g)\) whenever \((f,g) \in R(X,Y)\), there is a unique traced monoidal functor \(H^{\#} \colon \Cat{C}\slash R \to D\) with \(\eta_{R} \dcomp H^{\#} = H\).

\end{lem}
\begin{proof}
  Define the category \(\Cat{C} \slash R\) with the same objects as \(\Cat{C}\) but \(R\)-equivalence classes of morphisms \(X \to Y\) in \(\Cat{C}\) as morphisms \(X \to Y\).
  We need to show that this defines a traced monoidal category.
  Since \(R\) respects compositions in \(\Cat{C}\), the composition of equivalence classes is well-defined, associative and unital. This is exactly the usual construction~\cite[Section~II.8]{mac_lane_categories_1978}.

  The monoidal products of \(\Cat{C} \slash R\) are defined from those of \(\Cat{C}\): \(X \perG Y = X \perG Y\) and \([f] \perG [g] = [f \perG g]\), where \([f]\) denotes the equivalence class of \(f\) under \(R\).
  On morphisms, they are well-defined because \(R\) respects monoidal products: if \([f] = [g]\) and \([f'] = [g']\), then \([f \perG f'] = [g \perG g']\).
  Functoriality, and the pentagon, triangle and hexagon equations also follow from this assumption, e.g.\
  \begin{equation*}
    ([f] \dcomp [g]) \perG ([f'] \dcomp [g']) = [(f \dcomp g) \perG (f' \dcomp g')] = [(f \perG f') \dcomp (g \perG g')] = ([f] \perG [f']) \dcomp ([g] \perG [g']).
  \end{equation*}

  The trace on \(\Cat{C} \slash R\) is defined from the trace on \(\Cat{C}\): \(\trace([f]) = [\trace(f)]\).
  This is well-defined on equivalence classes by assumption: if \([f] = [g]\), then \([\trace(f)] = [\trace(g)]\).
  The trace axioms in \(\Cat{C} \slash R\) also follow from this assumption, e.g.\
  \begin{equation*}
    \trace([f] \dcomp ([u] \perG \id{})) = [\trace(f \dcomp (u \perG \id{}))] = [\trace((u \perG \id{}) \dcomp f)] = \trace(([u] \perG \id{}) \dcomp [f]).
  \end{equation*}
  Define \(\eta_{R}(X) = X\) and \(\eta_{R}(f) = [f]\), for all objects \(X\) and \(Y\) of \(\Cat{C}\), and all morphisms \(f \colon X \to Y\) in \(\Cat{C}\).
  Since \(R\) is a traced monoidal congruence, we obtain that \(\eta_{R} \colon \Cat{C} \to \Cat{C} \slash R\) is a traced monoidal functor.

  Consider a traced monoidal functor \(H \colon \Cat{C} \to \Cat{D}\) such that \(H(f) = H(g)\) whenever \((f,g) \in R(X,Y)\).
  Then, the assignment \(H^{\#}(X) = H(X)\) and \(H^{\#}([f]) = H(f)\) is well-defined.
  It is also a traced monoidal functor \(H^{\#} \colon \Cat{C} \slash R \to \Cat{D}\) because \(H\) is.
  By its definition, \(H^{\#}(\eta_{R}(X)) = H^{\#}(X) = H(X)\) and \(H^{\#}(\eta_{R}(f)) = H^{\#}([f]) = H(f)\) and any other traced monoidal functor satisfying this equation must coincide with \(H^{\#}\).

\end{proof}

In the following sections, we will repeatedly restrict adjunctions along faithful functors, applying the lemma below.

\begin{lem}\label{lem:restricting-adjunctions}
  Consider an adjunction \(\adjunction{F}{\Cat{C}}{\Cat{D}}{U}\) with unit \(\eta\) and counit \(\epsilon\).
  Suppose that 
  \begin{itemize}
\item  there are faithful functors \(J \colon \Cat{C}' \to \Cat{C}\) and \(K \colon \Cat{D}' \to \Cat{D}\) such that the components of the natural transformation \(\eta_{J}\) lie in the image of \(J\) and those of \(\epsilon_{K}\) lie in the image of \(K\).
  \item There are functions on objects \(F'_{o} \colon \ob{\Cat{C}'} \to \ob{\Cat{D}'}\) and \(U'_{o} \colon \ob{\Cat{D}'} \to \ob{\Cat{C}'}\), and functions on hom-sets \(F'_{C,C'} \colon \Cat{C}'[C,C'] \to \Cat{D}'[F'_{o}C,F'_{o}C']\) and \(U'_{D,D'} \colon \Cat{D}'[D,D'] \to \Cat{C}'[U'_{o}D,U'_{o}D']\) such that \[FJC = KF'_{o}C\text{, } UKD = JU'_{o}D\text{, }FJf = KF'_{C,C'}f \text{ and  } UKg = JU'_{D,D'}g\]
  for all objects $C,C'$ in $\Cat{C'}$, $D,D'$ in $\Cat{D'}$ and arrows $f\colon C \to C'$ and $g\colon D \to D'$.
  \end{itemize}
  Then, these definitions give functors \(F' \colon \Cat{C}' \to \Cat{D}'\) and \(U' \colon \Cat{D}' \to \Cat{C}'\) with \(F'\) left adjoint to \(U'\).
  \begin{center}
  \begin{tikzcd}
  {\Cat{C}} \arrow[bend left=20]{r}{F} \arrow[phantom]{r}{\bot} & {\Cat{D}} \arrow[bend left=20]{l}{U}\\
  {\Cat{C}'} \arrow{u}{J} \arrow[bend left=20, dashed]{r}{F'} & {\Cat{D}'} \arrow{u}[swap]{K} \arrow[bend left=20, dashed]{l}{U'}
\end{tikzcd}
\end{center}
\end{lem}
\begin{proof}
  We first prove that the functions on objects and hom-sets assemble into functors.
  We write it for \(F'\); the symmetric reasoning applies to \(U'\).
  \begin{align*}
    KF'_{C,C}(\id{C}) & = FJ(\id{C}) = \id{FJC} = \id{KF'_{o}C} = K(\id{F'_{o}C})\\
    KF'_{A,C}(f \dcomp g) & = FJ(f \dcomp g) = FJf \dcomp FJg = KF'_{A,B}f \dcomp KF'_{B,C}g = K(F'_{A,B}f \dcomp F'_{B,C}g)
  \end{align*}
  By faithfulness of \(K\), we obtain that \(F'_{C,C}(\id{C}) = \id{F'_{o}C}\) and \(F'_{A,C}(f \dcomp g) = F'_{A,B}f \dcomp F'_{B,C}g\), which means that \(F'\) is functorial.

  We now define the candidate unit \(\eta'_{C} \colon C \to U'(F'(C))\) and counit \(\epsilon'_{D} \colon F'(U'(D)) \to D\) for objects \(C\) of \(\Cat{C}'\) and \(D\) of \(\Cat{D}'\).
  Consider \(\eta_{JC} \colon J(C) \to UFJ(C) = UKF'(C) = JU'F'(C)\) and \(\epsilon_{KD} \colon KF'U'(D) = FJU'(D) = FUK(D) \to K(D)\).
  Since \(J\) and \(K\) are faithful, \(\eta_{J}\) lies in the image of \(J\) and \(\epsilon_{K}\) lies in the image of \(K\), there are unique morphisms \(\eta'_{C} \colon C \to U'(F'(C))\) and \(\epsilon'_{D} \colon F'(U'(D)) \to D\) such that \(J(\eta'_{C}) = \eta_{JC}\) and \(K(\epsilon'_{D}) = \epsilon_{KD}\).

  We check that \(\eta'\) and \(\epsilon'\) are natural.
  Since \(J\) and \(K\) are faithful, the naturality squares for \(\eta'\) and \(\epsilon'\) hold if and only if they hold after applying \(J\) and \(K\) as below.
  \begin{center}
  \begin{tikzcd}
    {J(C)} \arrow{r}{Jf} \arrow{d}[swap]{J\eta'_C} & {J(C')} \arrow{d}{J\eta'_{C'}} & & {KF'U'(D)} \arrow{r}{KF'U'g} \arrow{d}[swap]{K\epsilon'_D} & {KF'U'(D')} \arrow{d}{K\epsilon'_{D'}}\\
    {JU'F'(C)} \arrow{r}[swap]{JU'F'f} & {JU'F'(C')} & & {K(D)} \arrow{r}[swap]{Kg} & {K(D')}
  \end{tikzcd}
  \end{center}
  But we know that \(J\eta' = \eta_{J}\), \(JU'F' = UFJ\), \(K\epsilon' = \epsilon_{K}\) and \(KF'U' = FUK\), so the squares above coincide with the naturality squares of \(\eta_{J}\) and \(\epsilon_{K}\).
  This shows that \(\eta'\) and \(\epsilon'\) are natural.

  A similar reasoning shows that \(\eta'\) and \(\epsilon'\) satisfy the snake equations.
  Since \(J\) and \(K\) are faithful, the snake equations for \(\eta'\) and \(\epsilon'\) hold if and only if they hold after applying \(J\) and \(K\) as below.
  \begin{center}
    \begin{tikzcd}
      {JU'(D)} \arrow{r}{J\eta'_{U'D}} \arrow{dr}[swap]{\id{}} & {JU'F'U'(D)} \arrow{d}{JU'\epsilon'_D} & {KF'(C)} \arrow{r}{KF'\eta'_C} \arrow{dr}[swap]{\id{}} & {KF'U'F'(C)} \arrow{d}{K\epsilon'_{F'C}} \\
      {} & {JU'(D)} & & {KF'(C)}
    \end{tikzcd}
  \end{center}
  But we know that \(JU'(D) = UK(D)\), \(JU'F'U'(D) = UKF'U'(D) = UFJU'(D) = UFUK(D)\), \(JU'\epsilon'_{D} = UK\epsilon'_{D} = U\epsilon_{KD}\) and \(J\eta'_{U'D} = \eta_{JU'D} = \eta_{UKD}\); this means that the diagram on the left above corresponds to one of the snake equations for \(\eta\) and \(\epsilon\) on the object \(KD\), \(\eta_{UKD} \dcomp U\epsilon_{KD} = \id{UKD}\), and therefore it commutes.
  We also know that \(KF'(C) = FJ(C)\), \(KF'U'F'(C) = FJU'F'(C) = FUKF'(C) = FUFJ(C)\), \(KF'\eta'_{C} = FJ\eta'_{C} = F\eta_{JC}\) and \(K\epsilon'_{F'C} = \epsilon_{KF'C} = \epsilon_{FJC}\); this means that the diagram on the right above corresponds to the other snake equation for \(\eta\) and \(\epsilon\) on the object \(JC\), \(F\eta_{JC} \dcomp \epsilon_{FJC} = \id{FJC}\), and therefore it commutes.
\end{proof}

\subsection{Free Kleene Bicategories}%
\label{ssec:free-kleene}
The free Kleene bicategory on a signature is constructed by composing adjunctions.
\begin{equation}%
  \label{eq:chain-adjunctions-free-kleene}
\begin{tikzcd}
  {\Cat{Sig}} \arrow[bend left]{r}{\fun{M}} \arrow[phantom]{r}{\text{\cite{mac_lane_categories_1978}}} & {\Cat{FBCat}} \arrow[bend left]{r}{\fun{Q_{id}}} \arrow[phantom]{r}{\text{\cite{mac_lane_categories_1978}}} \arrow[bend left]{l}{\fun{U_M}} & {\Cat{FIBCat}} \arrow[bend left]{r}{\fun{UTr_B}} \arrow[phantom]{r}{\text{\ref{prop:free-fib-uniform-trace}}} \arrow[bend left]{l}{\fun{\iota}_{id}} & {\Cat{UTFIBCat}} \arrow[bend left]{r}{\fun{Q_{K}}} \arrow[phantom]{r}{\text{\ref{prop:free-kleene}}} \arrow[bend left]{l}{\fun{U_{TB}}} & {\Cat{KBicat}} \arrow[bend left]{l}{\fun{\iota_{K}}}
\end{tikzcd}
\end{equation}
The first adjunction is well known: it first constructs the free category on a signature~\cite[Section~II.7]{mac_lane_categories_1978} and then the free finite-biproduct category on a category via the matrix construction~\cite[Section~VIII.2, Exercises~5-6]{mac_lane_categories_1978}.
In the second adjunction $\fun{Q_{id}} \colon \Cat{FBCat} \to \Cat{FIBCat}$ simply quotients by idempotency of convolution~\cite{mac_lane_categories_1978}, as defined below. %

\begin{defi}%
  \label{def:idempotency-congruence}
  For a finite-biproduct category \((\Cat{C}, \piu, \zero)\), generate a monoidal congruence \((\idcong_{\Cat{C}})\) inductively by the rules below.
  \begin{gather}
    \inferrule*[right=(\(\id{}\))]{X \in \Cat{C}}{\diag{} \dcomp \codiag{} \idcong[X,X] \id{}}
    \qquad
    \inferrule*[right=(r)]{f \in \Cat{C}(X,Y)}{f \idcong[X,Y] f}
    \qquad
    \inferrule*[right=(t)]{f \idcong[X,Y] g \quad g \idcong[X,Y] h}{f \idcong[X,Y] h}
    \qquad
    \inferrule*[right=(s)]{f \idcong[X,Y] g}{g \idcong[X,Y] f}\nonumber
    \\[8pt]
    \inferrule*[right=($\dcomp$)]{f \idcong[X,Y] f' \quad g \idcong[Y,Z] g'}{f\dcomp g \idcong[X,Z] f'\dcomp g'}
    \qquad
    \inferrule*[right=($\piu$)]{f \idcong[X,Y] f' \quad g \idcong[X',Y'] g'}{f\piu g \idcong[X \piu X', Y \piu Y'] f' \piu g'}\label{eq:idempotency-cong}
  \end{gather}
\end{defi}

We prove the last two adjunctions in~\eqref{eq:chain-adjunctions-free-kleene}.
For the adjunction between fb categories with idempotent convolution and uniformly traced fb categories with idempotent convolution (in \Cref{prop:free-fib-uniform-trace}), we restrict the adjunction between symmetric monoidal categories and uniformly traced monoidal categories (\Cref{def:utr-category}) as illustrated below.

\begin{tikzcd}
  {\SMC} \arrow[bend left=20]{r}{\freeTr} \arrow[phantom]{r}{\text{\cite{katis2002feedback}}} & {\TSMC} \arrow[bend left=20]{r}{\fun{Q_U}} \arrow[phantom]{r}{\text{\ref{prop:free-uniform-trace}}} \arrow[bend left=20]{l}{U} & {\Cat{UTMonCat}} \arrow[bend left=20]{l}{\fun{\iota_{T}}}\\
  {\Cat{FIBCat}} \arrow{u}{\fun{J_1}} \arrow[bend left=10]{rr}{\fun{UTr_B}} \arrow[phantom]{rr}{\text{\ref{prop:free-fib-uniform-trace}}} && {\Cat{UTFIBCat}} \arrow{u}[swap]{\fun{J_2}} \arrow[bend left=10]{ll}{\fun{U_{TB}}}
\end{tikzcd}

We construct the adjunction between symmetric monoidal categories and uniformly traced monoidal categories by composing two adjunctions.
The first one gives the free traced monoidal category on a symmetric monoidal category~\cite[Sections~2.2-2.3]{katis2002feedback}; \Cref{prop:free-trace} recalls this construction.
In the second adjunction, illustrated in \Cref{prop:free-uniform-trace}, $\fun{Q_U}$ quotients by uniformity of the trace (\Cref{def:uniformity-congruence}).

\begin{prop}%
  [{\cite[Sections~2.2-2.3]{katis2002feedback}}]%
  \label{prop:free-trace}
  The free traced monoidal category \(\freeTr (\Cat{C})\) on a symmetric monoidal category \(\Cat{C}\) has  the same objects as \(\Cat{C}\), and morphisms \((f \mid S) \colon X \to Y\) are pairs of an object \(S\) and a morphism \(f \colon S \perG X \to S \perG Y\) of \(\Cat{C}\).
  Morphisms are quotiented by \emph{yanking} and \emph{sliding}.
  Compositions, monoidal products and trace in \(\freeTr (\Cat{C})\) are recalled below.
  \begin{align}
    \st{f}{S} \dcomp \st{g}{T} &\defeq \st{\CStSeq{f}{S}{X}{g}{T}{Z}}{S \perG T} \label{eq:seq in Utr} \\
    \st{f}{S} \perG \st{g}{T} &\defeq \st{\CStPar{f}{S}{X}{Y}{g}{T}{Z}{W}}{S \perG T} \label{eq:per in Utr} \\
    \trace_T\st{f}{S} &\defeq \st{f}{S \perG T} \label{eq:trace in Utr}
  \end{align}
  The assignment $\Cat{C} \mapsto \freeTr(\Cat{C})$ extends to a functor $\freeTr \colon \SMC \to \TSMC$ which is the left adjoint to the obvious forgetful
  $\fun{U}\colon \TSMC \to \SMC$.
  \begin{equation}\label{eq:nonuniform}
    \begin{tikzcd}
      \SMC
      \arrow[r, "\freeTr"{name=F}, bend left] &
      \TSMC
      \arrow[l, "\fun{U}"{name=U}, bend left]
      \arrow[phantom, from=F, to=U, "\vdash" rotate=90]
    \end{tikzcd}
  \end{equation}
\end{prop}

\begin{defi}%
  [Uniformity congruence]%
  \label{def:uniformity-congruence}
  For a traced monoidal category \((\Cat{C}, \perG, I, \trace)\), generate a traced monoidal congruence \((\Ucong_{\Cat{C}})\) inductively by the rules below.
  We denote with \((\Ucong[X,Y]_{\Cat{C}})\) the component of \((\Ucong_{\Cat{C}})\) on the hom-set \(\Cat{C}[X,Y]\).
  We omit the subscript \(\Cat{C}\) and use infix notation to improve readability.
  \begin{gather}
    \inferrule*[right=(r)]{f \in \Cat{C}[X,Y]}{f \Ucong[X,Y] f}
    \qquad
    \inferrule*[right=(t)]{f \Ucong[X,Y] g \quad g \Ucong[X,Y] h}{f \Ucong[X,Y] h}
    \qquad
    \inferrule*[right=(s)]{f \Ucong[X,Y] g}{g \Ucong[X,Y] f}\nonumber
    \\[8pt]
    \inferrule*[right=($\dcomp$)]{f \Ucong[X,Y] f' \quad g \Ucong[Y,Z] g'}{f\dcomp g \Ucong[X,Z] f'\dcomp g'}
    \qquad
    \inferrule*[right=($\perG$)]{f \Ucong[X,Y] f' \quad g \Ucong[X',Y'] g'}{f\perG g \Ucong[X \perG X', Y \perG Y'] f' \perG g'}\label{eq:uniformcong}
    \\[8pt]
    \inferrule*[right=(ut)]{u \Ucong[S,T] v \qquad f \dcomp (u \perG \id{}) \Ucong[S \perG X, T \perG Y] (v \perG \id{}) \dcomp g}{\trace_{S}f \Ucong[X,Y] \trace_{T}g}\nonumber
  \end{gather}
\end{defi}

\begin{prop}\label{prop:free-uniform-trace}
  The quotient by the uniformity congruence \((\Ucong)\) defines a functor \(\fun{Q_{U}}\) that is left adjoint of the inclusion: \(\adjunction{\fun{Q_{U}}}{\TSMC}{\Cat{UTMonCat}}{\fun{\iota_{T}}}\).
\end{prop}
\begin{proof}
  By \Cref{def:uniformity-congruence}, the relation \((\Ucong_{\Cat{C}})\) is a traced monoidal congruence because (ut) with \(u = v = \id{}\) imply that \(\trace(f) \Ucong \trace(g)\) whenever \(f \Ucong g\).
  For a traced monoidal category \(\Cat{C}\), define \(\fun{Q_{U}}(\Cat{C}) = \Cat{C} \slash \Ucong_{\Cat{C}}\).
  By \Cref{lem:traced-monoidal-congruence}, \(\Cat{C} \slash \Ucong_{\Cat{C}}\) is a traced monoidal category; by definition of \((\Ucong_{\Cat{C}})\) it is also uniformly traced.
  By \Cref{lem:traced-monoidal-congruence}, there are traced monoidal functors \(\eta_{\Cat{C}} \colon \Cat{C} \to \fun{\iota_{T}}(\fun{Q_{U}}(\Cat{C}))\) that are universal: for all traced monoidal functors \(H \colon \Cat{C} \to \fun{\iota_{T}}(\Cat{D})\), there is a unique traced monoidal functor \(H^{\#} \colon \fun{Q_{U}}(\Cat{C}) \to \Cat{D}\) such that \(\eta_{\Cat{C}} \dcomp \fun{\iota_{T}}(H^{\#}) = H\).
  This proves that \(\fun{Q_{U}}\) extends to a functor that is left adjoint to \(\fun{\iota_{T}}\) by the characterisation of adjunctions with universal arrows~\cite[Theorem~IV.2]{mac_lane_categories_1978}.
\end{proof}

Now that we have constructed the free uniformly traced monoidal category on a monoidal category, we restrict this adjunction to fb categories with idempotent convolution.
The next result constructs the object part of the left adjoint.
\Cref{prop:free-fib-uniform-trace} concludes by constructing its morphism part and applying \Cref{lem:restricting-adjunctions}.

Hereafter, we write $\UTr \colon \Cat{SMC} \to \Cat{UTMonCat}$ for the composition of $\freeTr$ with $\fun{Q_{U}}$ and $\fun{J_{1}}\colon \Cat{FIBCat} \to \Cat{SMC}$ and $\fun{J_{2}}\colon \Cat{UTFIBCat} \to \Cat{UTMonCat}$ for the obvious injections.

\begin{lem}%
  \label{lem:uniform-trace-preserves-biproducts}
  The free uniformly traced monoidal category over a fb category with idempotent convolution is also an fb category with idempotent convolution. %
  In other words, for a fb category with idempotent convolution \(\Cat{C}\), there is a uniformly traced fb category with idempotent convolution \(\fun{UTr_{B}}_{o}(\Cat{C})\) such that \(\UTr(\fun{J_{1}}(\Cat{C})) = \fun{J_{2}}(\fun{UTr_{B}}_{o}(\Cat{C}))\).
\end{lem}
\begin{proof}
  Let \((\Cat{C}, \piu, \zero)\) be a symmetric monoidal category whose monoidal product is a biproduct such that convolution is idempotent, \(\diag{X} \dcomp \codiag{X} = \id{X}\) for all objects \(X\).
  Consider a morphism \(f \colon X \to Y\) in \(\UTr(\Cat{C}) = \fun{Q_{U}}(\freeTr(\Cat{C}))\).
  By the definitions of \(\fun{Q_{U}}\) and \(\freeTr\), there is a morphism \(g \colon S \piu X \to S \piu Y\) in \(\Cat{C}\) whose trace is \(f\), \(f = \trace_{S}(g) = (g \mid S)\).
  By the universal property of products, the morphism \(g\) has two components: \(g = \diag{S \piu X} \dcomp (g_{1} \piu g_{2})\).
  The natural comonoid structure \((\diag{},\bang{})\) of \(\Cat{C}\) gives a comonoid structure \(((\diag{} \mid \zero), (\bang{} \mid \zero))\) in \(\UTr(\Cat{C})\) via the unit of the adjunction, \(\eta_{\Cat{C}}\).
  We show that this comonoid structure is natural in \(\UTr(\Cat{C})\).
  We can rewrite \(f \dcomp (\bang{Y} \mid \zero)\) using naturality of \(\bang{Y}\) in \(\Cat{C}\).
  \[
    

}

  \]
  By uniformity, we rewrite the trace of \(g_{1}\).
  \[
    
    %
}

  \]
  This shows that \(f \dcomp (\bang{Y} \mid \zero) = (\bang{X} \mid \zero)\), i.e.\ that the counit is natural in \(\UTr(\Cat{C})\).

  Similarly, we can rewrite \(f \dcomp (\diag{Y} \mid \zero)\) using naturality of \(\diag{Y}\) in \(\Cat{C}\).
  \[ 
    %
}
 \]
  By uniformity, we rewrite the trace of \(g_{1}\) with the comultiplication maps.
  \[ 
    %
}
 \]
  Then, \(f \dcomp (\diag{Y} \mid \zero) = (\diag{X} \mid \zero) \dcomp (f \piu f)\), i.e.\ the comultiplication is also natural in \(\UTr(\Cat{C})\).
  \[ 
    %
}
 \]
  This shows that \(\UTr(\Cat{C})\) is a finite-product category.
  Dually, it is also a finite-coproduct category and then a finite-biproduct category.
  Since convolution is idempotent in \(\Cat{C}\), it must be so in \(\UTr(\Cat{C})\) as well:
  \[(\diag{X} \mid \zero) \dcomp (f \piu f) \dcomp (\codiag{Y} \mid \zero) = f \dcomp (\diag{Y} \mid \zero) \dcomp (\codiag{Y} \mid \zero) = f \dcomp (\diag{Y} \dcomp \codiag{Y} \mid \zero) = f.\]
  Then, we set \(\fun{UTr_{B}}_{o}(\Cat{C}, \piu, \zero) = (\UTr(\Cat{C}), \piu, \zero)\) and this is well-defined.
\end{proof}

\begin{prop}\label{prop:free-fib-uniform-trace}
  The adjunction between symmetric monoidal categories and uniformly traced monoidal categories, \(\adjunction{\UTr}{\SMC}{\Cat{UTMonCat}}{\fun{U_{T}}}\), restricts to the adjunction  \(\adjunction{\fun{UTr_{B}}}{\Cat{FIBCat}}{\Cat{UTFIBCat}}{\fun{U_{TB}}}\).
\end{prop}
\begin{proof}
Observe that the functor \(\fun{J_{1}} \colon \Cat{FIBCat} \to \SMC\) is full and faithful. Indeed,  a morphism of fb categories with idempotent convolution is a symmetric monoidal functor preserving monoids and comonoids. Moreover, any symmetric monoidal functor between fb categories must preserve $\cobang{X} \colon \zero \to X$ and $\bang{X} \colon X \to \zero$ as $\zero$ is the initial and final object. Simple computations confirm that also $\codiag{X}$ and $\diag{X}$ are preserved. For analogous reasons, the functor \(\fun{J_{2}} \colon \Cat{UTFIBCat} \to \Cat{UTMonCat}\) is full and faithful. These functors are also injective on objects because two fb categories with idempotent convolution are equal whenever they are equal as monoidal categories with the biproduct as monoidal product.

  \Cref{lem:uniform-trace-preserves-biproducts} constructs the object part of \(\fun{UTr_{B}}\).
  For the morphism part, consider a morphism of fb categories with idempotent convolution \(F \colon \Cat{C} \to \Cat{D}\).
  Then, we obtain \(\UTr(\fun{J_{1}}(F)) \colon \fun{J_{2}}(\fun{UTr_{B}}(\Cat{C})) \to \fun{J_{2}}(\fun{UTr_{B}}(\Cat{D}))\) by the definition on objects of \(\fun{UTr_{B}}\).
  Since \(\fun{J_{2}}\) is full and faithful, there is a unique \(\fun{UTr_{B}}(F) \colon \fun{UTr_{B}}(\Cat{C}) \to \fun{UTr_{B}}(\Cat{D})\) such that \(\fun{J_{2}}(\fun{UTr_{B}}(F)) = \UTr(\fun{J_{1}}(F))\).
  The object part of the functor \(\fun{U_{TB}}\) simply forgets the trace structure, as \(\fun{U_{T}}\) does, so we obtain that \(\fun{J_{1}}(\fun{U_{TB}}(\Cat{C}, \piu, \zero, \trace)) = (\Cat{C}, \piu, \zero) = \fun{U_{T}}(\fun{J_{2}}(\Cat{C}, \piu, \zero, \trace))\) and similarly for functors.
  Since the functors \(\fun{J_{1}}\) and \(\fun{J_{2}}\) are full, the components of the unit \(\eta_{\fun{J_{1}}}\) and counit \(\epsilon_{\fun{J_{2}}}\) of the adjunction belong to the image of \(\fun{J_{1}}\) and \(\fun{J_{2}}\), respectively.

  Finally, we can apply \Cref{lem:restricting-adjunctions} to obtain that the functors \(\UTr \colon \SMC \to \Cat{UTMonCat}\) and \(\fun{U_{T}} \colon \Cat{UTMonCat} \to \SMC\) restrict to functors \(\fun{UTr_{B}} \colon \Cat{FIBCat} \to \Cat{UTFIBCat}\) and \(\fun{U_{TB}} \colon \Cat{UTFIBCat} \to \Cat{FIBCat}\), and that these are adjoint.
\end{proof}

We are left with proving the adjunction between uniformly traced fb categories with idempotent convolution and Kleene bicategories.
This adjunction is also a quotient because the posetal enrichment of Kleene bicategories is derived from the idempotent convolution.

\begin{rem}\label{rem:poset-from-convolution}
  The poset enrichment in fb categories with idempotent convolution is derived (\Cref{lem:idempfib}): the convolution monoid defines a poset enrichment by \(f \leq g\) whenever \(f \sqcup g = g\); and any poset enrichment needs to coincide with this one.
  This means that, while it is useful to think in terms of inequalities, we do not need to add inequalities to signatures.
  In particular, imposing posetal uniformity reduces to quotienting by a congruence.
  This result also simplifies the definition of morphisms of Kleene bicategories (\Cref{def:kleenebicategory}): they are traced monoidal functors preserving the monoid and comonoid structures.
\end{rem}

\begin{defi}%
  [Kleene congruence]%
  \label{def:kleene-congruence}
  For a uniformly traced fb category with idempotent convolution \((\Cat{C}, \piu, 0, \trace)\), generate a traced monoidal congruence \((\Kcong_{\Cat{C}})\) inductively by the rules below.
  We write \(f \Kleq[X,Y]_{\Cat{C}} g\) to indicate that the pair \((f \sqcup g, g)\) belongs to the equivalence relation \((\Kcong[X,Y]_{\Cat{C}})\).
  By \Cref{rem:poset-from-convolution}, quotienting by \((\Kcong)\) is equivalent to adding the corresponding inequalities to the native poset enrichment \((\leq)\) of \(\Cat{C}\).
  \begin{gather}
    \inferrule*[right=($\leq$)]{f \leq^{\scriptscriptstyle{X,Y}} g}{f \Kleq[X,Y] g}
    \qquad
    \inferrule*[right=($\id{}$)]{X \in \Cat{C}}{\trace(\codiag{};\diag{}) \Kleq[X,X] \id{}}
    \qquad
    \inferrule*[right=(t)]{f \Kleq[X,Y] g \quad g \Kleq[X,Y] h}{f \Kleq[X,Y] h}\nonumber
    \\[8pt]
    \inferrule*[right=($\dcomp$)]{f \Kleq[X,Y] f' \quad g \Kleq[Y,Z] g'}{f\dcomp g \Kleq[X,Z] f'\dcomp g'}
    \qquad
    \inferrule*[right=($\piu$)]{f \Kleq[X,Y] f' \quad g \Kleq[X',Y'] g'}{f\piu g \Kleq[X \piu X', Y \piu Y'] f' \piu g'}\label{eq:posetuniformcong}
    \\[8pt]
    \inferrule*[right=(ut-l)]{u \Kcong[S,T] v \qquad f \dcomp (u \piu \id{}) \Kleq[S \piu X, T \piu Y] (v \piu \id{}) \dcomp g}{\trace_{S}f \Kleq[X,Y] \trace_{T}g}\nonumber
    \\[8pt]
    \inferrule*[right=(ut-r)]{u \Kcong[S,T] v \qquad (v \piu \id{}) \dcomp g \Kleq[S \piu X, T \piu Y] f \dcomp (u \piu \id{})}{\trace_{T}g \Kleq[X,Y] \trace_{S}f}\nonumber
  \end{gather}
\end{defi}

\begin{prop}\label{prop:free-kleene}
  The quotient by the Kleene congruence \((\Kcong)\) defines a functor \(\fun{Q_{K}}\) that is the left adjoint of the inclusion, \( \adjunction{\fun{Q_{K}}}{\Cat{UTFIBCat}}{\Cat{KBicat}}{\fun{\iota_{K}}}\).
\end{prop}
\begin{proof}
  By \Cref{def:kleene-congruence}, the relation \((\Kcong_{\Cat{C}})\) is a traced monoidal congruence because (ut-l) and (ut-r) with \(u = v = \id{}\) imply that \(\trace(f) \Kcong \trace(g)\) whenever \(f \Kcong g\).
  We define a function \(\fun{Q_{K}}\) on objects of \(\Cat{UTFIBCat}\).
  For a uniformly traced fb category with idempotent convolution \(\Cat{C}\), define \(\fun{Q_{K}}(\Cat{C}) = \Cat{C} \slash \Kcong_{\Cat{C}}\).
  By \Cref{lem:traced-monoidal-congruence}, \(\Cat{C} \slash \Kcong_{\Cat{C}}\) is a traced monoidal category.
  The monoid and comonoid structures that determine biproducts are defined to be the equivalence classes of the monoid and comonoid structures, respectively;
  since \((\Kcong_{\Cat{C}})\) is a traced monoidal congruence, it preserves the monoid and comonoid axioms and their naturality; then the monoidal structure of \(\fun{Q_{K}}(\Cat{C})\) is also a biproduct.
  By definition of \((\Kcong_{\Cat{C}})\), \(\fun{Q_{K}}(\Cat{C})\) is also a Kleene bicategory.
  By \Cref{lem:traced-monoidal-congruence}, there are traced monoidal functors \(\eta_{\Cat{C}} \colon \Cat{C} \to \fun{\iota_{K}}(\fun{Q_{K}}(\Cat{C}))\) that are universal: for all traced monoidal functors \(H \colon \Cat{C} \to \fun{\iota_{K}}(\Cat{D})\), there is a unique traced monoidal functor \(H^{\#} \colon \fun{Q_{K}}(\Cat{C}) \to \Cat{D}\) such that \(\eta_{\Cat{C}} \dcomp \fun{\iota_{K}}(H^{\#}) = H\).
  We need to check that these functors are finite-biproduct functors, i.e.\ that they preserve the monoid and comonoid structures.
  The functor \(\eta_{\Cat{C}}\) preserves them by definition; the functor \(H^{\#}\) preserves them whenever \(H\) preserves them as well.
  This proves that \(\fun{Q_{K}}\) extends to a functor that is left adjoint to \(\fun{\iota_{K}}\) by the characterisation of adjunctions with universal arrows~\cite[Theorem~IV.2]{mac_lane_categories_1978}.
\end{proof}

\begin{thm}\label{th:free-kleene-on-signature}
  There is an adjunction that constructs the free Kleene bicategory \(\Cat{K}_{\Sigma}\) on a signature \(\Sigma\).
\end{thm}
\begin{proof}
  Compose the adjunctions in~\ref{eq:chain-adjunctions-free-kleene}.
  The first two adjunctions are known.
  We have proven the last two in \Cref{prop:free-fib-uniform-trace} and \Cref{prop:free-kleene}.
\end{proof}

\subsection{Free Kleene-Cartesian Bicategories}%
\label{ssec:free-kleene-cartesian}
This section concludes the proof of \Cref{thm:KleeneCartesiantapesfree} by restricting the adjunctions from the previous section (illustrated in the top row of~\eqref{eq:square-adjunctions-free-kcb}) and the adjunction that constructs free finite-biproduct rig categories on a monoidal signature~\cite[Theorem 5.11]{bonchi2023deconstructing}.
\Cref{prop:free-ut-fib-rig,prop:free-kleene-rig} prove the adjunctions in the second row; \Cref{prop:free-fb-cb-rig,prop:free-ut-fib-cb,prop:free-kleene-cartesian} prove those in the third row.

\begin{equation}%
  \label{eq:square-adjunctions-free-kcb}
\begin{tikzcd}
  {}
  & {\Cat{FBCat}} \arrow[bend left=10]{r}{\fun{Q}_{id}} \arrow[phantom]{r}{\bot}
  & {\Cat{FIBCat}} \arrow[bend left=10]{r}{\fun{UTr_{B}}} \arrow[phantom]{r}{\bot} \arrow[bend left=10]{l}{U}
  & {\Cat{UTFIBCat}} \arrow[bend left=10]{r}{\fun{Q_K}} \arrow[phantom]{r}{\bot} \arrow[bend left=10]{l}{U}
  & {\Cat{KBicat}} \arrow[bend left=10]{l}{U}\\
  {\Cat{MSig}} \arrow[bend left=10]{r}{\Cat{FBT}} \arrow[phantom]{r}{\text{\cite{bonchi2023deconstructing}}}
  & {\Cat{FBRig}} \arrow{u}{\fun{J_{FB}}} \arrow[bend left=7]{rr}{\fun{UTr_{R}}} \arrow[phantom]{rr}{\bot} \arrow[bend left=10]{l}{U}
  & { }
  & {\Cat{UTFIBRig}} \arrow{u}[swap]{\fun{J_{UT}}} \arrow[bend left=10]{r}{\fun{Q_{KR}}} \arrow[phantom]{r}{\bot} \arrow[bend left=7]{ll}{U}
  &  {\Cat{KRig}} \arrow{u}[swap]{\fun{J_K}} \arrow[bend left=10]{l}{U}\\
  {\Cat{MSig}} \arrow{u}{\fun{C}} \arrow[bend left=10]{r}{\Cat{FBCT}} \arrow[phantom]{r}{\bot}
  & {\Cat{FBCB}} \arrow{u}{\fun{K_{FB}}} \arrow[bend left=7]{rr}{\fun{UTr_{C}}} \arrow[phantom]{rr}{\bot}  \arrow[bend left=10]{l}{U}
  & { }
  &  {\Cat{UTFIBCB}} \arrow{u}[swap]{\fun{K_{UT}}} \arrow[bend left=10]{r}{\fun{Q_{KC}}} \arrow[phantom]{r}{\bot} \arrow[bend left=7]{ll}{U}
  &  {\Cat{KCB}} \arrow{u}[swap]{\fun{K_K}} \arrow[bend left=10]{l}{U}
\end{tikzcd}
\end{equation}

We begin with the following result that is useful to prove that a monoidal category \((\Cat{C}, \piu, \zero)\)  carries the structure of a rig category. %
\begin{lem}\label{lem:whisk}
  A symmetric monoidal category \((\Cat{C}, \piu, \zero)\) has an additional symmetric strict monoidal structure making it a right-strict rig category if and only if:
  \begin{enumerate}
    \item For all objects \(X\) of \(\Cat{C}\), there are functors \(\RW{X}{-}, \LW{X}{-} \colon \Cat{C} \to \Cat{C}\), called \emph{right} and \emph{left whiskering}, such that, for all objects \(X\) and \(Y\) and all morphisms \(f \colon X \to Y\) and \(f' \colon X' \to Y'\),
          \[\RW{Y}{X} = \LW{X}{Y} \text{ (denoted by } X \per Y{)} \qquad \RW{X'}{f} \dcomp \LW{Y}{f'} = \LW{X}{f'} \dcomp \RW{Y'}{f};\]
    \item There is an object \(\uno\) of \(\Cat{C}\) that is the unit of the whiskerings, \(\RW{1}{f} = f = \LW{1}{f}\);
    \item The whiskerings satisfy associativity,
          \[\LW{X}{\RW{Y}{f}} = \RW{Y}{\LW{X}{f}} \quad \LW{X \per Y}{f} = \LW{X}{\LW{Y}{f}} \quad \RW{Y \per X}{f} = \RW{X}{\RW{Y}{f}} ;\]
    \item There are natural transformations \(\symmt{-}{Y} \colon \RW{Y}{-} \to \LW{Y}{-}\) satisfying the hexagon identity and \(\symmt{X}{Y} \dcomp \symmt{Y}{X} = \id{X \per Y}\);
    \item The whiskerings interact with the monoidal structure \((\piu, \zero)\) as in \Cref{table:whisk-distr}, for isomorphisms \(\dl{X}{Y}{Z} \colon X \per (Y \piu Z) \to (X \per Y) \piu (X \per Z)\).
  \end{enumerate}
  \begin{table}[t]
    \centering
    {\begin{tabular}{lc lc}
      \toprule
       \multicolumn{2}{l}{$1. \ \ \LW{\zero}{f} = \id{\zero}$} & $2. \ \ \RW{\zero}{f} = \id{\zero}$ & (\newtag{W1}{eq:whisk:zero}) \\[0.3em]
       \multicolumn{3}{l}{$1. \ \ \LW{X}{f_1 \piu f_2} = \dl{X}{X_1}{X_2} ; (\LW{X}{f_1} \piu \LW{X}{f_2}) ; \Idl{X}{Y_1}{Y_2}$} & \multirow{2}{*}{(\newtag{W2}{eq:whisk:funct piu})} \\[0.3em]
       \multicolumn{3}{l}{$2. \ \ \RW{X}{f_1 \piu f_2} = \RW{X}{f_1} \piu \RW{X}{f_2}$}
      & {} \\[0.3em]
       \multicolumn{3}{l}{$1. \ \ \LW{X \piu Y}{f} = \LW{X}{f} \piu \LW{Y}{f}$}  & \multirow{2}{*}{(\newtag{W3}{eq:whisk:sum})} \\[0.3em]
       \multicolumn{3}{l}{$2. \ \ \RW{X \piu Y}{f} = \dl{Z}{X}{Y} ; ( \RW{X}{f} \piu \RW{Y}{f} ) ; \Idl{W}{X}{Y}$} & {} \\[0.3em]
       \multicolumn{3}{l}{$\RW X {\symmp{Y}{Z}} = \symmp{Y \per X}{Z \per X}$} & (\newtag{W4}{eq:whisk:symmp}) \\[0.3em]
       $\RW X {\dl{Y}{Z}{W}} = \dl{Y}{Z \per X}{W \per X}$ & (\newtag{W5}{eq:whisk:dl}) & $\LW X {\dl{Y}{Z}{W}} = \dl{X \per Y}{Z}{W} ; \Idl{X}{Y \per Z}{Y \per W}$ & (\newtag{W6}{eq:whisk:Ldl}) \\
        \bottomrule
    \end{tabular}}
    \caption{The algebra of whiskerings}%
    \label{table:whisk-distr}
  \end{table}
\end{lem}
\begin{proof}
  Conditions (1-4) ensure that there is a monoidal structure~\cite{POWER_ROBINSON_1997}.
  The laws in \Cref{table:whisk-distr} spell out the coherence conditions for distributivity in terms of left and right whiskerings.
\end{proof}

We define candidate whiskerings for the free uniformly traced finite-biproduct category $\UTr(\Cat{C})$ on a finite-biproduct rig category $\Cat{C}$.
\Cref{prop:free-ut-fib-rig} proves that they give a rig structure to $\UTr(\Cat{C})$.

\begin{defi}\label{def:utr-whisk}
  Let $(\Cat{C}, \piu, \zero, \per, \uno)$ be a finite-biproduct rig category, $\UTr(\Cat{C})$ the uniformly traced finite-biproduct category freely generated from $(\Cat{C}, \oplus, \zero)$, and $X$ an object of $\UTr(\Cat C)$.
  Then $\fun{L}_X, \fun{R}_X \colon \UTr(\Cat C) \to \UTr(\Cat C)$ are defined on objects as
  $\LW{X}{Y} \defeq X \per Y$ and $\RW{X}{Y} \defeq Y \per X$,
  and on arrows $\st{f}{S} \colon Y \to Z$ as
  \[ \LW{X}{f \mid S} \defeq \st{\symmt{X}{Y}}{\zero} ; \RW{X}{f \mid S} ; \st{\symmt{Z}{X}}{\zero} \qquad \text{ and } \qquad \RW{X}{f \mid S} \defeq \st{\RW{X}{f}}{S \per X}.\]
\end{defi}

\begin{prop}\label{prop:free-ut-fib-rig}
  The free uniformly traced fb category with idempotent convolution over a fb rig category is also a fb rig category with idempotent convolution.
\end{prop}
\begin{proof}
  Consider a finite-biproduct rig category \((\Cat{C}, \piu, \zero, \per, \uno)\) and the free uniformly traced fb category with idempotent convolution \(\fun{UTr_{B}}(\fun{Q}_{id}(\fun{J_{FB}}(\Cat{C})))\) over it.
  We want to show that it is a uniformly traced finite-idempotent-biproduct rig category.
  By construction, it is uniformly traced.
  By \Cref{prop:free-fib-uniform-trace}, it is a fb category with idempotent convolution.
  We show that it also has a monoidal structure that makes it a rig category.

  The monoidal structure is defined using the monoidal structure of \((\Cat{C}, \per, \uno)\).
  By definition, \(\fun{UTr_{B}}(\fun{Q}_{id}(\fun{J_{FB}}(\Cat{C})))\) has the same objects as \(\Cat{C}\), so the monoidal product on objects coincide with that of \(\Cat{C}\).
  The monoidal structure is more easily defined through whiskerings as in \Cref{def:utr-whisk}.
  The symmetries are also lifted from \((\Cat{C}, \per, \uno)\), \(\symmt{X}{Y} = \st{\symmt{X}{Y}}{\zero}\), and therefore satisfy the hexagon equation and \(\symmt{X}{Y} \dcomp \symmt{Y}{X} = \st{\symmt{X}{Y} \dcomp \symmt{Y}{X}}{\zero} = \st{\id{}}{\zero} = \id{}\).

  We check that the whiskerings preserve identities, using that \((\Cat{C}, \piu, \zero, \per, \uno)\) is a distributive category, that the symmetries are isomorphisms and \Cref{def:utr-whisk}.
  \begin{align*}
    & \RW{X}{\id{Y} \mid \zero} && \LW{X}{\id{Y} \mid \zero} \\
    &=  \st{\RW{X}{\id{Y}}}{\zero \per X} && =  \st{\symmt{X}{Y}}{\zero} ; \RW{X}{\id{Y} \mid \zero} ; \st{\symmt{Y}{X}}{\zero} \\
    &=  \st{\id{Y \per X}}{\zero \per X} && =  \st{\symmt{X}{Y}}{\zero} ; \st{\id{Y \per X}}{\zero} ; \st{\symmt{Y}{X}}{\zero} \\
    &=  \st{\id{Y \per X}}{\zero} && =  \st{\id{X \per Y}}{\zero}
  \end{align*}
  The whiskerings also preserve compositions.
  For $(f \mid S) \colon Y \to Z$ and $(g \mid T) \colon Z \to W$, the right whiskering preserves their composition.
  \begin{align*}
    & \RW{X}{\st{f}{S} ; \st{g}{T}} \\
    & = \RW{X}{(\symmp{S}{T} \piu \id{Y}) ; (\id{T} \piu f) ; (\symmp{T}{S} \piu \id{Z}) ; (\id{S} \piu g) \mid S \piu T} \tag{\ref{eq:seq in Utr}} \\
    & = \st{\RW{X}{(\symmp{S}{T} \piu \id{Y}) ; (\id{T} \piu f) ; (\symmp{T}{S} \piu \id{Z}) ; (\id{S} \piu g)}}{(S \piu T) \per X} \tag{Definition~\ref{def:utr-whisk}} \\
    & = (\RW{X}{(\symmp{S}{T} \piu \id{Y}) ; (\id{T} \piu f) ; (\symmp{T}{S} \piu \id{Z}) ; (\id{S} \piu g)} \tag{Table~\ref{tab:equationsonobject}} \\
    & \qquad \mid (S \per X) \piu (T \per X))\\
    & = ( \RW{X}{\symmp{S}{T} \piu \id{Y}} ; \RW{X}{\id{T} \piu f} ; \RW{X}{\symmp{T}{S} \piu \id{Z}} ; \RW{X}{\id{S} \piu g}  \tag{Functoriality in $\Cat{C}$} \\
    & \qquad \mid (S \per X) \piu (T \per X)) \\
    & = ((\symmp{S\per X}{T \per X} \piu \id{Y \per X}) ; (\id{T \per X} \piu \RW{X}{f}) ; (\symmp{T \per X}{S \per X} \piu \id{Z \per X})  \tag{\ref{eq:whisk:funct piu}, \ref{eq:whisk:symmp} in $\Cat{C}$} \\
    & \qquad ; (\id{S \per X} \piu \RW{X}{g}) \mid (S \per X) \piu (T \per X))\\
    & = \st{\RW{X}{f}}{S \per X} ; \st{\RW{X}{g}}{T \per X} \tag{\ref{eq:seq in Utr}} \\
    & = \RW{X}{f \mid S} ; \RW{X}{g \mid T} \tag{Definition~\ref{def:utr-whisk}}
  \end{align*}
  The left whiskering also preserves their composition.
  \begin{align*}
    & \LW{X}{\st{f}{S} ; \st{g}{T}}\\
    &= \st{\symmt{X}{Y}}{\zero} ; \RW{X}{\st{f}{S} ; \st{g}{T}} ; \st{\symmt{W}{X}}{\zero} \tag{Definition~\ref{def:utr-whisk}} \\
    &= \st{\symmt{X}{Y}}{\zero} ; \RW{X}{f \mid S} ; \RW{X}{g \mid T} ; \st{\symmt{W}{X}}{\zero} \tag{Functoriality of $\fun{R}$} \\
    &= \st{\symmt{X}{Y}}{\zero} ; \RW{X}{f \mid S} ; \st{\symmt{Z}{X}}{\zero} ; \st{\symmt{X}{Z}}{\zero} ; \RW{X}{g \mid T} ; \st{\symmt{W}{X}}{\zero} \tag{Symmetries} \\
    &= \LW{X}{f \mid S} ; \LW{X}{g \mid T} \tag{Definition~\ref{def:utr-whisk}}
  \end{align*}
  The symmetries are natural transformations \(\symmt{-}{Y} \colon \RW{Y}{-} \to \LW{Y}{-}\) because they are isomorphisms and by the definition of whiskerings (\Cref{def:utr-whisk}).
  \[\RW{X}{f \mid S} ; \st{\symmt{Z}{X}}{\zero}
    = \st{\symmt{Y}{X}}{\zero} ; \st{\symmt{X}{Y}}{\zero} ; \RW{X}{f \mid S} ; \st{\symmt{Z}{X}}{\zero} = \st{\symmt{Y}{X}}{\zero} ; \LW{X}{f \mid S}\]

  We prove that the whiskerings satisfy the interchange law using uniformity of the trace in \(\fun{UTr_{B}}(\fun{Q}_{id}(\fun{J_{FB}}(\Cat{C})))\).
  Let $f_1 \colon S_1 \piu X_1 \to S_1 \piu Y_1$ and $f_2 \colon S_2 \piu X_2 \to S_2 \piu Y_2$ be morphisms in $\Cat{C}$ and observe that the following holds by the interchange law in $\Cat{C}$:
  \[ \st{\LW{S_1 \piu X_1}{f_2} ; \RW{S_2 \piu Y_2}{f_1}}{\zero} = \st{\RW{S_2 \piu X_2}{f_1} ; \LW{S_1 \piu Y_1}{f_2}}{\zero}. \]
  Using string diagrams for $(\Cat{C}, \piu, \zero)$, the equality above translates into the equality between diagrams below:
  \begin{align*}
    & \Bigg( \scalebox{1.9}{

}
} \Bigg| \zero \Bigg) .
  \end{align*}
  By precomposing and postcomposing with appropriate distributors, the following holds:
  \begin{align*}
    & \Bigg( \scalebox{1.9}{
    %
}
} \Bigg| \zero \Bigg) .
  \end{align*}
  Using uniformity of the trace with $\st{{
    %
}
}}{\zero}$ as strict morphism, we obtain the following equality:
  \begin{align*}
    & \Bigg( \scalebox{1.8}{
    %
}
}  \Bigg| (S_1 \per S_2) \piu (S_1 \per X_2) \Bigg).
  \end{align*}
  By naturality of symmetry, the following holds:
  \begin{align*}
    & \Bigg( \scalebox{1.4}{
    %
}
}
      \Bigg| (S_1 \per S_2) \piu (S_1 \per X_2) \Bigg).
  \end{align*}
  Using uniformity of the trace with $\st{\scalebox{0.9}{
    %
}
}}{S_1 \per S_2}$ as strict morphism, we obtain the following equality:
  \begin{align*}
    & \Bigg( \scalebox{1.7}{
    %
}
}
      \Bigg| (S_1 \per X_2) \piu (Y_1 \per S_2) \Bigg),
  \end{align*}
  which, by~\eqref{eq:seq in Utr} and Definition~\ref{def:utr-whisk}, corresponds to the equality below:
  \begin{align*}
    & \st{\delta^l_{X_1, S_2, X_2} ; \LW{X_1}{f_2} ; \delta^l_{X_1,S_2,Y_2}}{X_1 \per S_2} ; \RW{Y_2}{f_1 \mid S_1} \\
    &= \RW{X_2}{f_1 \mid S_1} ; \st{\delta^{-l}_{Y_1, S_2, X_2} ; \LW{Y_1}{f_2} ; \delta^l_{Y_1,S_2,Y_2}}{Y_1 \per S_2}.
  \end{align*}
  To conclude the proof of the interchange law, observe that for every $(f \mid S) \colon Y \to Z$ it holds that
  \begin{equation}\label{eq:LRaux}
    \LW{X}{f \mid S} = \st{\Idl{X}{S}{Y} ; \LW{X}{f} ; \dl{X}{S}{Z}}{X \per S}
  \end{equation}
  as shown below:
  \begin{align*}
    & \LW{X}{f \mid S}\\
    &= \st{\symmt{X}{Y}}{\zero} ; \RW{X}{f \mid S} ; \st{\symmt{Z}{X}}{\zero} \tag{Definition~\ref{def:utr-whisk}} \\
    &= \st{\symmt{X}{Y}}{\zero} ; \st{\RW{X}{f}}{S \per X} ; \st{\symmt{Z}{X}}{\zero} \tag{Definition~\ref{def:utr-whisk}} \\
    &= \st{ (\id{S \per X} \piu \symmt{X}{Y}) ; \RW{X}{f} ; (\id{S \per X} \piu \symmt{Z}{X}) }{S \per X}  \tag{\ref{eq:seq in Utr}} \\
    &= \st{ (\symmt{X}{S} \piu \symmt{X}{Y}) ; \RW{X}{f} ; (\symmt{S}{X} \piu \symmt{Z}{X}) }{X \per S}  \tag{\Ref{ax:trace:sliding}} \\
    &= \st{ (\symmt{X}{S} \piu \symmt{X}{Y}) ; \symmt{S \piu Y}{X} ; \LW{X}{f} ; \symmt{X}{S \piu Z} ; (\symmt{S}{X} \piu \symmt{Z}{X}) }{X \per S}  \tag{Symmetries in $\Cat{C}$} \\
    &= \st{\Idl{X}{S}{Y} ; \LW{X}{f} ; \dl{X}{S}{Z}}{X \per S} \tag{\ref{eq:rigax1}}
  \end{align*}
  As monoidal unit, we take the monoidal unit \(\uno\) of \(\Cat{C}\) and prove that it serves as monoidal unit in \(\fun{UTr_{B}}(\fun{Q}_{id}(\fun{J_{FB}}(\Cat{C})))\) as well.
  For $(f \mid S) \colon X \to Y$, we use \Cref{def:utr-whisk} and that \(\uno\) is the unit for \((\per)\) in \(\Cat{C}\).
  \begin{align*}
    & \RW{\uno}{f \mid S} && \LW{\uno}{f \mid S} \\
    &= \st{\RW{\uno}{f}}{S \per \uno} && = \symmt{\uno}{X} \dcomp \RW{\uno}{f \mid S} \dcomp \symmt{Y}{\uno}\\
    &= \st{f}{S \per \uno} && = \symmt{\uno}{X} \dcomp \st{f}{S} \dcomp \symmt{Y}{\uno}\\
    &= \st{f}{S} && = \st{f}{S}
  \end{align*}
  We prove that the whiskerings satisfy associativity.
  For a morphism $(f \mid S) \colon Z \to W$, we check the three associativity equations.
  For the right whiskering, we use \Cref{def:utr-whisk} and associativity of the right whiskering in \((\Cat{C}, \per, \uno)\).
  \[\RW{X}{\RW{Y}{f \mid S}}= \RW{X}{\st{\RW{Y}{f}}{S \per Y}} = \st{\RW{X}{\RW{Y}{f}}}{S \per Y \per X} = \st{\RW{Y \per X}{f}}{S \per Y \per X} \]
  For the left and right whiskering, we use \Cref{def:utr-whisk}, associativity of the whiskerings in \((\Cat{C}, \per, \uno)\) and the sliding axiom of the trace.
  \begin{align*}
    &  \LW{X}{\RW{Y}{f \mid S}} \\
    & = \st{\symmt{X}{Z \per Y}}{\zero} ; \RW{X}{\RW{Y}{f \mid S}} ; \st{\symmt{W \per Y}{X}}{\zero} \tag{Definition~\ref{def:utr-whisk}} \\
    & = \st{\symmt{X}{Z \per Y}}{\zero} ; \RW{X}{\RW{Y}{f} \mid S \per Y} ; \st{\symmt{W \per Y}{X}}{\zero} \tag{Definition~\ref{def:utr-whisk}} \\
    & = \st{\symmt{X}{Z \per Y}}{\zero} ; \st{\RW{X}{\RW{Y}{f}}}{S \per Y \per X} ; \st{\symmt{W \per Y}{X}}{\zero} \tag{Definition~\ref{def:utr-whisk}} \\
    & = \st{\symmt{X}{Z \per Y}}{\zero} ; \st{ \symmt{(S \piu Z)\per Y}{X} ; \LW{X}{\RW{Y}{f}} ; \symmt{X}{(S \piu W) \per Y} }{S \per Y \per X} \tag{Symmetries in $\Cat{C}$} \\
    & \qquad ; \st{\symmt{W \per Y}{X}}{\zero} \\
    & = \st{\symmt{X}{Z \per Y}}{\zero} ; \st{ \symmt{(S \piu Z)\per Y}{X} ; \RW{Y}{\LW{X}{f}} ; \symmt{X}{(S \piu W) \per Y} }{S \per Y \per X} \tag{Whiskering in $\Cat{C}$} \\
    & \qquad ; \st{\symmt{W \per Y}{X}}{\zero} \\
    & = \st{\symmt{X}{Z \per Y}}{\zero} ; ((\symmt{S \per Y}{X} \piu \symmt{Z \per Y}{X}) ; \Idl{X}{S \per Y}{Z \per Y} ; \RW{Y}{\LW{X}{f}}  \tag{\ref{eq:rigax1}} \\
    & \qquad ; \dl{X}{S \per Y}{W \per Y} ; (\symmt{X}{S \per Y} \piu \symmt{X}{W \per Y}) \mid S \per Y \per X) ; \st{\symmt{W \per Y}{X}}{\zero} \\
    & = ((\id{S \per Y \per X} \piu \symmt{X}{Z \per Y}) ; (\symmt{S \per Y}{X} \piu \symmt{Z \per Y}{X}) ; \Idl{X}{S \per Y}{Z \per Y} ; \RW{Y}{\LW{X}{f}} \tag{\ref{eq:seq in Utr}} \\
    & \qquad  ; \dl{X}{S \per Y}{W \per Y} ; (\symmt{X}{S \per Y} \piu \symmt{X}{W \per Y}) ; (\id{S \per Y \per X} \piu \symmt{W \per Y}{X}) \mid S \per Y \per X )\\
    & = ((\symmt{S \per Y}{X} \piu \id{X \per Z \per Y}) ; \Idl{X}{S \per Y}{Z \per Y} ; \RW{Y}{\LW{X}{f}} \tag{Symmetries} \\
    & \qquad  ; \dl{X}{S \per Y}{W \per Y} ; ( \symmt{X}{S \per Y} \piu \id{X \per W \per Y} ) \mid S \per Y \per X) \\
    & = { \st{ \Idl{X}{S \per Y}{Z \per Y} ; \RW{Y}{\LW{X}{f}} ; \dl{X}{S \per Y}{W \per Y} }{X \per S \per Y} } \tag{\Ref{ax:trace:sliding}} \\
    & = { \st{ \RW{Y}{\Idl{X}{S}{Z}} ; \RW{Y}{\LW{X}{f}} ; \RW{Y}{\dl{X}{S}{W}} }{X \per S \per Y} } \tag{\ref{eq:whisk:dl} in $\Cat{C}$} \\
    & = { \st{ \RW{Y}{\Idl{X}{S}{Z} ; \LW{X}{f} ; \dl{X}{S}{W}} }{X \per S \per Y} } \tag{Functoriality in $\Cat{C}$} \\
    & = \RW{Y}{\Idl{X}{S}{Z} ; \LW{X}{f} ; \dl{X}{S}{W} \mid X \per S} \tag{Definition~\ref{def:utr-whisk}} \\
    & = \RW{Y}{\LW{X}{f \mid S}} \tag{\ref{eq:LRaux}}
  \end{align*}
  For the left whiskering, we use \Cref{def:utr-whisk}.
  \begin{align*}
    & \LW{X}{\LW{Y}{f \mid S}}\\
    &=  \st{\symmt{X}{Y \per Z}}{\zero} ; \RW{X}{\LW{Y}{f \mid S}} ; \st{\symmt{Y \per W}{X}}{\zero} \tag{Definition~\ref{def:utr-whisk}} \\
    &=  \st{\symmt{X}{Y \per Z}}{\zero} ; \LW{Y}{\RW{X}{f \mid S}} ; \st{\symmt{Y \per W}{X}}{\zero} \tag{Associativity of $\fun{LR}$} \\
    &=  \st{\symmt{X}{Y \per Z}}{\zero} ; \st{\symmt{Y}{Z \per X}}{\zero} ; \RW{Y}{\RW{X}{f \mid S}} \tag{Definition~\ref{def:utr-whisk}} \\
    & \qquad ; \st{\symmt{W \per X}{Y}}{\zero} ; \st{\symmt{Y \per W}{X}}{\zero} \\
    &=  \st{\symmt{X}{Y \per Z}}{\zero} ; \st{\symmt{Y}{Z \per X}}{\zero} ; \RW{X \per Y}{f \mid S}  \tag{Associativity of $\fun{RR}$} \\
    & \qquad ; \st{\symmt{W \per X}{Y}}{\zero} ; \st{\symmt{Y \per W}{X}}{\zero} \\
    &=  \st{\symmt{X \per Y}{Z}}{\zero} ; \RW{X \per Y}{f \mid S} ; \st{\symmt{W}{X \per Y}}{\zero} \tag{Symmetries} \\
    &=  \LW{X \per Y}{f \mid S} \tag{Definition~\ref{def:utr-whisk}}
  \end{align*}

  This shows that \((\fun{UTr_{B}}(\fun{Q}_{id}(\fun{J_{FB}}(\Cat{C}))), \per, \uno)\) is a symmetric strict monoidal category.

  We now need to show that the monoidal structure distributes over the biproducts.
  For this, we check the equations in \Cref{table:whisk-distr}.\\
  \textsc{Equation}~\eqref{eq:whisk:zero}.
  For a morphism $(f \mid S) \colon X \to Y$, we check that whiskering with \(\zero\) is annihilating, using \Cref{def:utr-whisk} and that \(\zero\) is annihilating in \(\Cat{C}\).
  \begin{align*}
    &\RW{\zero}{f \mid S} && \LW{\zero}{f \mid S} \\
    &= \st{\RW{\zero}{f}}{S \per \zero} && = \st{\symmt{X}{\zero}}{\zero} ; \RW{\zero}{f \mid S} ; \st{\symmt{\zero}{Y}}{\zero} \\
    &= \st{\id{\zero}}{S \per \zero} && = \st{\symmt{X}{\zero}}{\zero} ; \st{\id{\zero}}{\zero} ; \st{\symmt{\zero}{Y}}{\zero} \\
    &= \st{\id{\zero}}{\zero} && = \st{\id{\zero}}{\zero}
  \end{align*}
  \textsc{Equation}~\eqref{eq:whisk:funct piu}.
  For two morphisms, $(f_1 \mid S_1) \colon X_1 \to Y_1$ and $(f_2 \mid S_2) \colon X_2 \to Y_2$, we show that right whiskering preserves their biproduct.
  \begin{align*}
    & \RW{X}{\st{f_1}{S_1} \piu \st{f_2}{S_2}} \\
    & = \RW{X}{ (\id{S_1} \piu \symmp{S_2}{X_1} \piu \id{X_2}) ; (f_1 \piu f_2) ; (\id{S_1} \piu \symmp{Y_1}{S_2} \piu \id{Y_2}) \mid S_1 \piu S_2 } \tag{\ref{eq:per in Utr}} \\
    & = (\RW{X}{(\id{S_1} \piu \symmp{S_2}{X_1} \piu \id{X_2}) ; (f_1 \piu f_2) ; (\id{S_1} \piu \symmp{Y_1}{S_2} \piu \id{Y_2})} \tag{Definition~\ref{def:utr-whisk}} \\
    & \qquad \mid (S_1 \piu S_2) \per X)\\
    & = ((\id{S_1 \per X} \piu \symmp{S_2 \per X}{X_1 \per X} \piu \id{X_2 \per X}) ; (\RW{X}{f_1} \piu \RW{X}{f_2}) \tag{\ref{eq:whisk:funct piu}, \ref{eq:whisk:symmp} in $\Cat{C}$} \\
    & \qquad ; (\id{S_1 \per X} \piu \symmp{Y_1 \per X}{S_2 \per X} \piu \id{Y_2 \per X}) \mid (S_1 \piu S_2) \per X)\\
    & = ((\id{S_1 \per X} \piu \symmp{S_2 \per X}{X_1 \per X} \piu \id{X_2 \per X}) ; (\RW{X}{f_1} \piu \RW{X}{f_2}) \tag{Symmetries} \\
    & \qquad ; (\id{S_1 \per X} \piu \symmp{Y_1 \per X}{S_2 \per X} \piu \id{Y_2 \per X}) \mid (S_1 \per X) \piu (S_2 \per X))\\
    & = \RW{X}{f_1 \mid S_1} \piu \RW{X}{f_2 \mid S_2} \tag{\ref{eq:per in Utr}}
  \end{align*}
  Similarly, we show that left whiskering preserves their biproduct.
  \begin{align*}
    & \LW{X}{\st{f_1}{S_1} \piu \st{f_2}{S_2}} \\
    & = \st{\symmt{X}{X_1 \piu X_2}}{\zero} ; \RW{X}{\st{f_1}{S_1} \piu \st{f_2}{S_2}} ; \st{\symmt{Y_1 \piu Y_2}{X}}{\zero} \tag{Definition~\ref{def:utr-whisk}} \\
    & = \st{\symmt{X}{X_1 \piu X_2}}{\zero} ; (\RW{X}{f_1 \mid S_1} \piu \RW{X}{f_2 \mid S_2})  ; \st{\symmt{Y_1 \piu Y_2}{X}}{\zero} \tag{\ref{eq:whisk:funct piu}.2} \\
    & = \st{\dl{X}{X_1}{X_2} ; (\symmt{X}{X_1} \piu \symmt{X}{X_2})}{\zero} ; (\RW{X}{f_1 \mid S_1} \piu \RW{X}{f_2 \mid S_2}) \tag{\ref{eq:rigax1}} \\
    & \qquad ; \st{ (\symmt{Y_1}{X} \piu \symmt{Y_2}{X}) ; \Idl{X}{Y_1}{Y_2} }{\zero} \\
    & = \st{\dl{X}{X_1}{X_2}}{\zero} ; (\st{\symmt{X}{X_1}}{\zero} \piu \st{\symmt{X}{X_2}}{\zero}) ; (\RW{X}{f_1 \mid S_1} \piu \RW{X}{f_2 \mid S_2}) \tag{\ref{eq:seq in Utr}, \ref{eq:per in Utr}} \\
    & \qquad ; (\st{\symmt{Y_1}{X}}{\zero} \piu \st{\symmt{Y_2}{X}}{\zero}) ; \st{\Idl{X}{Y_1}{Y_2}}{\zero} \\
    & = \st{\dl{X}{X_1}{X_2}}{\zero} ; ( ( \st{\symmt{X}{X_1}}{\zero} ; \RW{X}{f_1 \mid S_1} ; \st{\symmt{Y_1}{X}}{\zero} )  \tag{Symmetries} \\
    & \qquad \piu (\st{\symmt{X}{X_2}}{\zero} ; \RW{X}{f_2 \mid S_2} ; \st{\symmt{Y_2}{X}}{\zero} ) ) ; \st{\Idl{X}{Y_1}{Y_2}}{\zero}\\
    & = \st{\dl{X}{X_1}{X_2}}{\zero} ; ( \LW{X}{f_1 \mid S_1} \piu \LW{X}{f_2 \mid S_2} ) ; \st{\Idl{X}{Y_1}{Y_2}}{\zero} \tag{Definition~\ref{def:utr-whisk}}
  \end{align*}
  \textsc{Equation}~\eqref{eq:whisk:sum}.
  For a morphism $(f \mid S) \colon Z \to W$, we show that right whiskering with a biproduct is, up to isomorphisms, the biproduct of right whiskerings.
  \begin{align*}
    & \RW{X \piu Y}{f \mid S} \\
    & = \st{\RW{X \piu Y}{f}}{S \per (X \piu Y)} \tag{Definition~\ref{def:utr-whisk}} \\
    & = \st{ \dl{S \piu Z}{X}{Y} ; (\RW{X}{f} \piu \RW{Y}{f}); \Idl{S \piu W}{X}{Y} }{S \per (X \piu Y)} \tag{\ref{eq:whisk:sum}.2 in $\Cat{C}$} \\
    & = ((\dl{S}{X}{Y} \piu \dl{Z}{X}{Y}) ; (\id{S \per X} \piu \symmp{S \per Y}{Z \per X} \piu \id{Z \per Y}) ; (\RW{X}{f} \piu \RW{Y}{f}) \tag{\ref{eq:rigax5}}  \\
    & \qquad ; (\id{S \per X} \piu \symmp{S \per Y}{W \per X} \piu \id{W \per Y}) ; (\Idl{S}{X}{Y} \piu \Idl{W}{X}{Y}) \mid S \per (X \piu Y)) \\
    & = ((\id{S \per X} \piu \id{S \per Y} \piu \dl{Z}{X}{Y}) ; (\id{S \per X} \piu \symmp{S \per Y}{Z \per X} \piu \id{Z \per Y}) ; (\RW{X}{f} \piu \RW{Y}{f}) \tag{\Ref{ax:trace:sliding}} \\
    & \qquad ; (\id{S \per X} \piu \symmp{S \per Y}{W \per X} \piu \id{W \per Y}) ; (\id{S \per X} \piu \id{S \per Y} \piu \Idl{W}{X}{Y}) \mid (S \per X) \piu (S \per Y)) \\
    & = \st{\dl{Z}{X}{Y}}{\zero} ; ((\id{S \per X} \piu \symmp{S \per Y}{Z \per X} \piu \id{Z \per Y}) ; (\RW{X}{f} \piu \RW{Y}{f}) \tag{\ref{eq:seq in Utr}} \\
    & \qquad ; (\id{S \per X} \piu \symmp{S \per Y}{W \per X} \piu \id{W \per Y}) \mid (S \per X) \piu (S \per Y)) ; \st{\Idl{W}{X}{Y}}{\zero} \\
    & = { \st{\dl{Z}{X}{Y}}{\zero} ;  ( \st{\RW{X}{f}}{S \per X} \piu \st{\RW{Y}{f}}{S \per Y} )  ; \st{\Idl{W}{X}{Y}}{\zero} } \tag{\ref{eq:per in Utr}} \\
    & = { \st{\dl{Z}{X}{Y}}{\zero} ;  ( \RW{X}{f \mid S} \piu \RW{Y}{f \mid S} )  ; \st{\Idl{W}{X}{Y}}{\zero} } \tag{Definition~\ref{def:utr-whisk}}
  \end{align*}
  Similarly, we show that left whiskering with a biproduct is the biproduct of left whiskerings.
  \begin{align*}
    & \LW{X \piu Y}{f \mid S} \\
    & = \st{\symmt{X \piu Y}{Z}}{\zero} ; \RW{X \piu Y}{f \mid S} ; \st{\symmt{W}{X \piu Y}}{\zero} \tag{Definition~\ref{def:utr-whisk}} \\
    & = \st{\symmt{X \piu Y}{Z}}{\zero} ; ( \st{\dl{Z}{X}{Y}}{\zero} ;  ( \RW{X}{f \mid S} \piu \RW{Y}{f \mid S} )  ; \st{\Idl{W}{X}{Y}}{\zero} ) ; \st{\symmt{W}{X \piu Y}}{\zero} \tag{\ref{eq:whisk:sum}.2} \\
    & = \st{ (\symmt{X}{Z} \piu \symmt{Y}{Z}) ; \Idl{Z}{X}{Y} }{\zero} ; ( \st{\dl{Z}{X}{Y}}{\zero} ;  ( \RW{X}{f \mid S} \piu \RW{Y}{f \mid S} )  ; \st{\Idl{W}{X}{Y}}{\zero} ) \tag{\ref{eq:rigax1}} \\
    & \qquad ; \st{\dl{W}{X}{Y} ; (\symmt{W}{X} \piu \symmt{W}{Y})}{\zero} \\
    & = \st{ (\symmt{X}{Z} \piu \symmt{Y}{Z}) ; \Idl{Z}{X}{Y} ; \dl{Z}{X}{Y} }{\zero} ; ( ( \RW{X}{f \mid S} \piu \RW{Y}{f \mid S} ) ) \tag{\ref{eq:seq in Utr}} \\
    & \qquad  ; \st{\Idl{W}{X}{Y} ; \dl{W}{X}{Y} ; (\symmt{W}{X} \piu \symmt{W}{Y})}{\zero} \\
    & = \st{ \symmt{X}{Z} \piu \symmt{Y}{Z}}{\zero} ; ( \RW{X}{f \mid S} \piu \RW{Y}{f \mid S} ) ; \st{\symmt{W}{X} \piu \symmt{W}{Y}}{\zero} \tag{$\delta^l$ iso} \\
    & = ( \st{\symmt{X}{Z}}{\zero} \piu \st{\symmt{Y}{Z}}{\zero} ) ; ( \RW{X}{f \mid S} \piu \RW{Y}{f \mid S} ) ; ( \st{\symmt{W}{X}}{\zero} \piu \st{\symmt{W}{Y}}{\zero} ) \tag{\ref{eq:per in Utr}} \\
    & = ( \st{\symmt{X}{Z}}{\zero} ; \RW{X}{f \mid S} ; \st{\symmt{W}{X}}{\zero} ) \piu ( \st{\symmt{Y}{Z}}{\zero} ; \RW{Y}{f \mid S} ; \st{\symmt{W}{Y}}{\zero} ) \tag{Symmetries} \\
    & = \LW{X}{f \mid S} \piu \LW{Y}{f \mid S} \tag{Definition~\ref{def:utr-whisk}}
  \end{align*}
  We prove the final three equations.\\
  \textsc{Equation}~\eqref{eq:whisk:symmp}.
  \begin{align*}
    \RW{X}{\symmp{Y}{Z} \mid \zero}
    &= \st{\RW{X}{\symmp{Y}{Z}}}{\zero \per X} \tag{Definition~\ref{def:utr-whisk}} \\
    &= \st{ \symmp{Y \per X}{Z \per X} }{\zero \per X} \tag{\ref{eq:whisk:symmp} in $\Cat{C}$} \\
    &= \st{ \symmp{Y \per X}{Z \per X} }{\zero} \tag{Table~\ref{tab:equationsonobject}}
  \end{align*}
  \textsc{Equation}~\eqref{eq:whisk:dl}.
  \begin{align*}
    \RW{X}{\dl{Y}{Z}{W} \mid \zero}
    &= \st{\RW{X}{\dl{Y}{Z}{W}}}{\zero \per X} \tag{Definition~\ref{def:utr-whisk}} \\
    &= \st{\dl{Y}{Z \per X}{W \per X}}{\zero \per X} \tag{\ref{eq:whisk:dl} in $\Cat{C}$} \\
    &= \st{\dl{Y}{Z \per X}{W \per X}}{\zero} \tag{Table~\ref{tab:equationsonobject}}
  \end{align*}
  \textsc{Equation}~\eqref{eq:whisk:Ldl}.
  \begin{align*}
    & \LW{X}{\dl{Y}{Z}{W} \mid \zero}\\
    &= \st{\symmt{X}{Y \per (Z \piu W)}}{\zero} ; \RW{X}{\dl{Y}{Z}{W} \mid \zero} ; \st{\symmt{(Y  \per Z) \piu (Y \per W)}{X}}{\zero} \tag{Definition~\ref{def:utr-whisk}} \\
    &= \st{\symmt{X}{Y \per (Z \piu W)}}{\zero} ; \st{\RW{X}{\dl{Y}{Z}{W}}}{\zero \per X} ; \st{\symmt{(Y  \per Z) \piu (Y \per W)}{X}}{\zero} \tag{Definition~\ref{def:utr-whisk}} \\
    &= \st{\symmt{X}{Y \per (Z \piu W)}}{\zero} ; \st{\RW{X}{\dl{Y}{Z \per X}{W \per X}}}{\zero} ; \st{\symmt{(Y  \per Z) \piu (Y \per W)}{X}}{\zero} \tag{Table~\ref{tab:equationsonobject}} \\
    &= \st{\symmt{X}{Y \per (Z \piu W)} ; \RW{X}{\dl{Y}{Z \per X}{W \per X}} ; \symmt{(Y  \per Z) \piu (Y \per W)}{X}}{\zero} \tag{\ref{eq:seq in Utr}} \\
    &= \st{\LW{X}{\dl{Y}{Z \per X}{W \per X}} ; \symmt{X}{(Y  \per Z) \piu (Y \per W)} ; \symmt{(Y  \per Z) \piu (Y \per W)}{X} }{\zero} \tag{Symmetries in $\Cat{C}$} \\
    &= \st{\LW{X}{\dl{Y}{Z \per X}{W \per X}} }{\zero} \tag{Symmetries} \\
    &= \st{\dl{X \per Y}{Z}{W} ; \Idl{Y}{X \per Z}{X \per W}}{\zero} \tag{\ref{eq:whisk:Ldl} in $\Cat{C}$} \\
    &= \st{\dl{X \per Y}{Z}{W}}{\zero} ; \st{\Idl{Y}{X \per Z}{X \per W}}{\zero} \tag{\ref{eq:seq in Utr}}
  \end{align*}
  We conclude by proving that the traced and monoidal structure interact as in \eqref{eq:missing} so that we obtain a traced rig category.
  Consider a morphism \((f \mid S) \colon T \piu X \to T \piu Y\) in \(\fun{UTr_{B}}(\fun{Q}_{id}(\fun{J_{FB}}(\Cat{C})))\).
  \begin{align*}
    (\trace_{T}(f \mid S)) \per \id{Z}
    & = (f \mid S \piu T) \per \id{Z}\\
    & = (f \per \id{Z} \mid (S \piu T) \per Z)\\
    & = \trace_{T \per Z}(\delta^{-r} \dcomp (f \per \id{Z} \mid S \per Z) \dcomp \delta^{r})
  \end{align*}
\end{proof}

\begin{prop}\label{prop:free-kleene-rig}
  The free Kleene bicategory on a uniformly traced finite-idempotent-biproducts rig category is also uniformly traced finite-idempotent-biproducts rig.
\end{prop}
\begin{proof}
  Consider a uniformly traced finite-idempotent-biproduct rig category \(\Cat{C}\) and the free Kleene bicategory \(\fun{Q_K}(\fun{J_{UT}}(\Cat{C}))\) over it.
  We want to show that it is a Kleene rig category.
  By \Cref{prop:free-kleene}, it is a Kleene bicategory.
  We show that it also has a monoidal structure that makes it a rig category by proving that the quotient by the congruence \((\Kcong)\) in \Cref{def:kleene-congruence} respects the monoidal structure.
  We need to prove that \(f \per g \Kcong f' \per g'\) whenever \(f \Kcong f'\) and \(g \Kcong g'\).
  It suffices to prove that \(f \per g \Kleq f' \per g\) whenever \(f \Kleq f'\) because: \(f \Kcong f'\) if and only if \(f \Kleq f'\) and \(f' \Kleq f\); if \(f \per g \Kleq f' \per g\) then \(g \per f = \symmt{X}{X'} \dcomp (f \per g) \dcomp \symmt{Y'}{Y} \Kleq \symmt{X}{X} \dcomp (f' \per g) \dcomp \symmt{Y'}{Y} = g \per f'\); if \(f \per g \Kleq f' \per g\) and \(f' \per g \Kleq f' \per g'\) then \(f \per g \Kleq f' \per g \Kleq f' \per g'\).
  So, suppose that \(f \Kleq f'\) and proceed by induction on \((\Kleq)\).
  \begin{itemize}
    \item[(\(\leq\))] If \(f \leq f'\), then \(f \sqcup f' = f'\).
          By the axioms of biproduct rig category, \((f \per g) \sqcup (f' \per g) = (f \sqcup f') \per g = f' \per g\).
          From this equality, we obtain that \(f \per g \leq f' \per g\) and then \(f \per g \Kleq f' \per g\).
    \item[(\(\id{}\))] If \(f = \trace(\codiag{X} \dcomp \diag{X})\) and \(f'= \id{X}\).
          By the axioms of traced finite-biproduct rig category and the definition of \((\Kleq)\), we obtain that \(f \per g = (\trace(\codiag{} \dcomp \diag{}) \per \id{Y}) \dcomp (\id{X} \per g) = trace(\Idr{X}{X}{Y} \dcomp ((\codiag{X} \dcomp \diag{X}) \per \id{Y}) \dcomp \dr{X}{X}{Y}) \dcomp (\id{X} \per g) = \trace(\codiag{X \per Y} \dcomp \diag{X \per Y}) \dcomp (\id{X} \per g) \Kleq \id{X \per Y} \dcomp (\id{X} \per g) = f' \per g\).
    \item[(t)] If \(f \Kleq h\) and \(h \Kleq f'\), then \(f \per g \Kleq h \per g\) and \(h \per g \Kleq f' \per g\) by induction.
          By transitivity of \((\Kleq)\), we obtain that \(f \per g \Kleq f' \per g\).
    \item[(\(\dcomp\))] If \(f = f_{1} \dcomp f_{2}\) and \(f' = f'_{1} \dcomp f'_{2}\) with \(f_{1} \Kleq f'_{1}\) and \(f_{2} \Kleq f'_{2}\), then \(f_{1} \per g \Kleq f'_{1} \per g\) and \(f_{2} \per \id{} \Kleq f'_{2} \per \id{}\) by induction.
          Since \((\Kleq)\) is a congruence, we obtain that \(f \per g = (f_{1} \per g) \dcomp (f_{2} \per \id{}) \Kleq (f'_{1} \per g) \dcomp (f'_{2} \per \id{}) = f' \per g\).
    \item[(\(\piu\))] If \(f = f_{1} \piu f_{2}\) and \(f' = f'_{1} \piu f'_{2}\) with \(f_{1} \Kleq f'_{1}\) and \(f_{2} \Kleq f'_{2}\), then \(f_{1} \per g \Kleq f'_{1} \per g\) and \(f_{2} \per g \Kleq f'_{2} \per g\) by induction.
          Since \((\Kleq)\) is a monoidal congruence, we obtain that \(f \per g = \delta \dcomp ((f_{1} \per g) \piu (f_{2} \per g)) \dcomp \delta^{-1} \Kleq \delta \dcomp ((f'_{1} \per g) \piu (f'_{2} \per g)) \dcomp \delta^{-1}  = f' \per g\).
    \item[(ut-l)] If \(f = \trace(h)\) and \(f' = \trace(h')\) with \(h \dcomp (u \piu \id{}) \Kleq (v \piu \id{}) \dcomp h'\) and \(u \Kcong v\), then \((h \dcomp (u \piu \id{})) \per \id{} \Kleq ((v \piu \id{}) \dcomp h') \per \id{}\) and \(u \per \id{} \Kcong v \per \id{}\) by induction.
          Since \((\Kleq)\) is a congruence, we obtain that \(\delta^{-r} \dcomp (h \per \id{}) \dcomp \delta^{r} \dcomp ((u \per \id{}) \piu \id{}) = \delta^{-r} \dcomp ((h \dcomp (u \piu \id{})) \per \id{}) \dcomp \delta^{r} \Kleq \delta^{-r} \dcomp (((v \piu \id{}) \dcomp h') \per \id{}) \dcomp \delta^{r} = ((v \per \id{}) \piu \id{}) \dcomp \delta^{-r} \dcomp (h' \per \id{}) \dcomp \delta^{r}\).
          By definition of \((\Kleq)\), we obtain that \(\trace(\delta^{-r} \dcomp (h \per \id{}) \dcomp \delta^{r}) \Kleq \trace(\delta^{-r} \dcomp (h' \per \id{}) \dcomp \delta^{r})\) and, by congruence and the axioms of traced rig category, that \(f \per g = (\trace h \per \id{}) \dcomp (\id{} \per g) = \trace(\delta^{-r} \dcomp (h \per \id{}) \dcomp \delta^{r})  \dcomp (\id{} \per g) \Kleq \trace(\delta^{-r} \dcomp (h' \per \id{}) \dcomp \delta^{r})  \dcomp (\id{} \per g) = f' \per g\).
    \item[(ut-r)] Analogous to (ut-l).
  \end{itemize}
  The coherence axioms for symmetric monoidal and rig categories hold because they hold in \(\Cat{C}\) and the quotient by \((\Kcong)\) preserves all the operations involved.
\end{proof}

\begin{prop}\label{prop:free-ut-fib-cb}
  The free uniformly traced monoidal category over a finite-biproduct Cartesian bicategory is also a Cartesian bicategory.
  In other words, for a finite-idempotent-biproduct Cartesian bicategory \(\Cat{C}\), the free uniformly traced monoidal category over \(\fun{J_{FB}}(\fun{K_{FB}}(\Cat{C}))\) lies in the image of \(\fun{K_{UT}} \dcomp \fun{J_{UT}}\).
\end{prop}
\begin{proof}
  Consider a finite-biproduct Cartesian bicategory \((\Cat{C}, \piu, \zero, \per, \uno)\) and the free uniformly traced monoidal category \(\fun{UTr_{B}} (\fun{Q}_{id} (\fun{J_{FB}} (\fun{K_{FB}} (\Cat{C}))))\) over \(\fun{Q}_{id}(\fun{J_{FB}}(\fun{K_{FB}}(\Cat{C})))\).
  By \Cref{prop:free-ut-fib-rig}, \(\fun{UTr_{B}} (\fun{Q}_{id} (\fun{J_{FB}} (\fun{K_{FB}} (\Cat{C}))))\) lies in the image of \(\fun{J_{UT}}\) because it can be endowed with monoidal structure making it a finite-idempotent-biproduct rig category.
  We need to show that this monoidal structure is also that of a Cartesian bicategory.
  The copy, discard, cocopy and codiscard morphisms are lifted from \(\Cat{C}\): \((\copier{} \mid \zero)\), \((\discharger{} \mid \zero)\), \((\cocopier{} \mid \zero)\) and \((\codischarger{} \mid \zero)\).
  For morphisms \(f\) and \(g\) in \(\Cat{C}\), \((f \mid 0) \dcomp (g \mid 0) = (f \dcomp g \mid 0)\) and \((f \mid 0) \per (g \mid 0) = (f \per g \mid 0)\) by the definition of composition and whiskerings.
  Then, the (co)copy and (co)discard structure satisfies all the equations of a Cartesian bicategory because it does so in \(\Cat{C}\).
\end{proof}

\begin{prop}\label{prop:free-kleene-cartesian}
  The free Kleene bicategory over a uniformly traced finite-idempotent-biproduct Cartesian bicategory is also a Cartesian bicategory.
  In other words, for a uniformly traced finite-idempotent-biproduct Cartesian bicategory \(\Cat{C}\), the free Kleene bicategory over \(\fun{J_{UT}}(\fun{K_{UT}}(\Cat{C}))\) lies in the image of \(\fun{K_{K}} \dcomp \fun{J_{K}}\).
\end{prop}
\begin{proof}
  Consider a uniformly traced finite-idempotent-biproduct Cartesian bicategory \(\Cat{C}\).
  \Cref{prop:free-kleene-rig} proves that \(\fun{Q_{K}}(\fun{J_{UT}}(\fun{K_{UT}}(\Cat{C})))\) lies in the image of \(\fun{J_{K}}\) by proving that the congruence \((\Kcong)\) is also a congruence for the monoidal product \((\per)\).
  This implies that the quotient by \((\Kcong)\) preserves all the equations between the structure of Cartesian bicategory, i.e.\ that \(\fun{Q_{K}}(\fun{J_{UT}}(\fun{K_{UT}}(\Cat{C})))\) is also a Cartesian bicategory.
\end{proof}

\begin{prop}%
  [{\cite[Theorem~5.11]{bonchi2023deconstructing}}]%
  \label{prop:free-fb-rig}
  There is an adjunction that constructs the free finite-biproduct rig category on a monoidal signature.
\end{prop}

\begin{prop}%
  \label{prop:free-fb-cb-rig}
  There is an adjunction that constructs the free finite-biproduct rig Cartesian bicategory on a monoidal signature.
\end{prop}
\begin{proof}
  There is a functor \(\fun{C} \colon \Cat{MSig} \to \Cat{MSig}\) such that \(\fun{C}(\Sigma, E)\) is the signature \((\Sigma, E)\) with the extra generators and equations for the structure of Cartesian bicategories: for every object in the signature \(\Sigma\), we add a multiplication, a comultiplication, a unit and a counit to the generators in \(\Sigma\); for every object and for every generator in the signature \(\sigma\), we add the equations of (co)commutative (co)monoids, the equations of a Frobenius algebra and idempotency of convolution to the equations in \(E\).
  The free finite-biproduct rig category over the signature \(\fun{C}(\Sigma, E)\) is also a Cartesian bicategory~\cite[Theorem~7.3]{bonchi2023deconstructing}.
  The unit and counit of the adjunction from \Cref{prop:free-fb-rig} satisfy the hypothesis of \Cref{lem:restricting-adjunctions}.
  The functors \(\fun{C}\) and \(\fun{K_{FB}}\) are faithful, so we apply \Cref{lem:restricting-adjunctions} to restrict the adjunction between monoidal signatures and finite-biproduct rig categories to finite-biproduct Cartesian bicategories.
\end{proof}

\begin{prop}%
  \label{prop:free-kleene-on-fbcb}
  There is an adjunction \(\adjunction{\fun{K_{CB}} }{\Cat{FBCB} }{\Cat{KCB} }{U}\) that constructs the free Kleene-Cartesian bicategory on a finite-biproduct Cartesian bicategory.
\end{prop}
\begin{proof}
  We use the results in this section to construct the free Kleene-Cartesian bicategory on a finite-biproduct Cartesian bicategory.
  By \Cref{prop:free-ut-fib-rig} and by faithfulness of \(\fun{J_{UT}}\), we can restrict the functor \(\fun{Q}_{id} \dcomp \fun{UTr_{B}} \colon \Cat{FBCat} \to \Cat{UTFIBCat}\) to a functor \(\fun{UTr_{R}} \colon \Cat{FBRig} \to \Cat{UTFIBRig}\) that satisfies \(\fun{UTr_{R}} \dcomp \fun{J_{UT}} = \fun{J_{FB}} \dcomp \fun{UTr_{B}}\).
  It is easy to check that the forgetful functors also make the diagram commute.
  A similar reasoning yields functors \(\fun{Q_{KR}} \colon \Cat{UTFIBRig} \to \Cat{KRig}\), \(\fun{UTr_{C}} \colon \Cat{FIBCB} \to \Cat{UTFIBCB}\) and \(\fun{Q_{KC}} \colon \Cat{UTFIBCB} \to \Cat{KCB}\) by applying \Cref{prop:free-kleene-rig}, \Cref{prop:free-ut-fib-cb} and \Cref{prop:free-kleene-cartesian}, and faithfulness of \(\fun{J_{K}}\), \(\fun{K_{UT}}\) and \(\fun{K_{K}}\).
  Then, we can apply \Cref{lem:restricting-adjunctions} to obtain the adjunctions in \Cref{eq:square-adjunctions-free-kcb}.
  In particular, we obtain an adjunction whose left adjoint constructs the free Kleene-Cartesian bicategory over a fb-cb rig category, \(\fun{K_{CB}} \colon \Cat{FBCB} \to \Cat{KCB}\).
\end{proof}

Finally, we can conclude the proof of our main result.
\begin{proof}[Proof of \Cref{thm:KleeneCartesiantapesfree}]
  \Cref{prop:free-fb-cb-rig} constructs the free fb-cb rig category \(\Cat{FBCT}_{\Sigma}\) on a monoidal signature \(\Sigma\).
  \Cref{prop:free-kleene-on-fbcb} constructs the free Kleene-Cartesian bicategory on a finite-biproduct Cartesian bicategory.
  We can, then, compose these two adjunctions to obtain the free Kleene-Cartesian bicategory on a monoidal signature.
\end{proof}

\newpage
\section{Appendix to Section \ref{sec:peano}}\label{app:Peano} %

\begin{proof}[Proof of Proposition \ref{lemma:Peano}]
    First, we prove that the axioms in Figure~\ref{fig:tape-theory-natural-numbers} entail those in Figure~\ref{fig:peano-theory-natural-numbers}.
    \begin{itemize}
        \item \eqref{ax:peano:succ:sv},\eqref{ax:peano:zero:sv} follow from \eqref{ax:peanoT:iso:1}, i.e.
        \[
            

}
 \leq 
    }
 \stackrel{\eqref{ax:peanoT:iso:1}}{=} 
    }
 .
        \]
        \item \eqref{ax:peano:succ:tot}, \eqref{ax:peano:succ:inj}, \eqref{ax:peano:zero:tot}, \eqref{ax:peano:bottom} follow from the matrix normal form of~\eqref{ax:peanoT:iso:2}, i.e.
        \begin{equation}\label{eq:peano:iso2-equivalence}
            
    %
 \right.  
        \end{equation}
        Observe that the last equality on the right imposes injectivity of $0$ as well, however this is always the case for morphisms $1 \to A$ that are total.
        \item \eqref{ax:peano:ind} holds by Theorem~\ref{thm:induction}.
    \end{itemize}
    Similarly, we prove that the axioms in Figure~\ref{fig:peano-theory-natural-numbers} entail those in Figure~\ref{fig:tape-theory-natural-numbers}. It is convenient to prove them in the following order.
    \begin{itemize}
        \item For \eqref{ax:peanoT:ind}, let $P \colon 1 \to A$ be the following morphism
        \[
    }
 \defeq  
    %
}
\]
        and observe that the first condition of~\eqref{ax:peano:ind} holds:
        \[
            
    }
 \stackrel{\eqref{ax:relBangCobang}}{\geq} 
    %
}
.
        \]
        For the second condition observe that $P$ is defined as $0 ; \kstar{s}$, and recall that in a Kleene algebra $\kstar{s} ; s \leq \kstar s$. Thus $P ; s = 0 ; \kstar{s} ; s \leq 0 ; \kstar{s} = P$. We conclude using~\eqref{ax:peano:ind}.
        \item For \eqref{ax:peanoT:iso:1} we prove the two inclusions separately, starting with the one below.
        \[
            
    }
 \stackrel{\eqref{ax:peano:succ:sv}, \eqref{ax:peano:zero:sv}}{\leq} 
    %
}
 \stackrel{\eqref{ax:relDiagCodiag}}{\leq} 
    }
 
        \]
        For the other inclusion, observe that by~\Cref{lemma:cb:adjoints}
        \[ 
            
    }
 \leq 
    }
  \iff  
    %
}
.
        \]
        Thus, we prove the following:
        \begin{align*}
            
    }
 &\leftstackrel{\eqref{ax:peanoT:ind}}{\geq} 
    %
}
 \\
                                                            &\leftstackrel{\eqref{eq:Kllenelaw}}{\geq} 
    }
 \stackrel{\eqref{ax:peanoT:ind}}{\geq} 
    }

        \end{align*}
        \item \eqref{ax:peanoT:iso:2} follow from \eqref{ax:peano:succ:tot}, \eqref{ax:peano:succ:inj}, \eqref{ax:peano:zero:tot}, \eqref{ax:peano:bottom} as shown in \eqref{eq:peano:iso2-equivalence}.
    \end{itemize}

\end{proof}

\newpage

\section{Appendix to Section~\ref{sec:hoare}}\label{app:hoare}
In this appendix we provide a detailed proof of Lemma \ref{lemma:encodingsubstPed}.

In the main text we avoid formally defining \emph{substitutions}, but for the purpose of our proof it is convenient to illustrate the inductive definition.
Given two expressions $t$ and $e$ and a variable $x$, the expression $e[t / x]$ is defined inductively as follows, where $y$ is a variable different from $x$.
\[x[t / x] \defeq t \qquad y[ t / x]\defeq y \qquad f(e_1, \dots, e_n)[t/x] \defeq f(e_1[t/x], \dots, e_n[t/x] )\]
Similarly to the case of expressions, one defines substitution on a variable $x$ of a term $t$ in a predicate $P$ inductively:
\[ R(e_1, \dots, e_n) [t/x] \defeq R(e_1[t/x], \dots, e_n[t/x]) \qquad \bar{R}(e_1, \dots, e_n) [t/x] \defeq \bar{R}(e_1[t/x], \dots, e_n[t/x]) \]
\[ \top [t/x] \defeq \top \quad \bot [t/x] \defeq \bot \quad (P \lor Q)[t/x]\defeq P[t/x] \lor Q[t/x] \quad (P \land Q)[t/x]\defeq P[t/x] \land Q[t/x] \]

The following result illustrates how our encoding deals with substitution of expressions.

\begin{lem}\label{lemma:encodingsubst}
Let $\Gamma ' = \Gamma, x\colon A, \Delta$ for some typing contexts $\Gamma$ and $\Delta$. If $\Gamma' \vdash e\colon B$ and $\Gamma ' \vdash t\colon A$, then
\[\encoding{ \Gamma' \vdash e[t /x] \colon B} = 
    \InputIfFileExists{hoare/subs/step1.tikz}{}{\input{./tikz/hoare/subs/step1.tikz}}
\] %
\end{lem}
\begin{proof}
The proof proceeds by induction on $\Gamma' \vdash e\colon B$.

If $e$ is the variable $x$, then by the rule (var) $A=B$. Moreover, by definition of $\encoding{\cdot}$, $\encoding{\Gamma, x\colon A, \Delta \vdash x \colon A} =  \discharger{\encoding{\Gamma}} \per \id{A} \per \discharger{\encoding{\Delta}}$. 
Thus
\[ 
    \InputIfFileExists{hoare/subs/step1.tikz}{}{\input{./tikz/hoare/subs/step1.tikz}}
 =_{\basicR} 
    \InputIfFileExists{hoare/subs/step2.tikz}{}{\input{./tikz/hoare/subs/step2.tikz}}
 \stackrel{\eqref{ax:copierun}}{=_{\basicR}} 
    \InputIfFileExists{hoare/subs/step3.tikz}{}{\input{./tikz/hoare/subs/step3.tikz}}
 = \encoding{\Gamma ' \vdash x [t / x]\colon A} \]

If $e$ is a variable $y$, different from $x$, then by the rule $(var)$, there are two possible cases: either $\Gamma = \Gamma_1, y\colon B , \Gamma_2$ for some typing contexts $\Gamma_1$ and $\Gamma_2$ or  $\Delta = \Delta_1, y\colon B , \Delta_2$.
We consider the first case, the second is symmetrical. Observe that, by definition of $\encoding{\cdot}$,

\[\encoding{\Gamma_1, y\colon B, \Gamma_2, x \colon A, \Delta \vdash y \colon B} =  \discharger{\encoding{\Gamma_1}} \per \id{B} \per  \discharger{\encoding{\Gamma_2}} \per  \discharger{A} \per  \discharger{\encoding{\Delta}} \]

Thus,
\[ 
    \InputIfFileExists{hoare/subs/step1.tikz}{}{\input{./tikz/hoare/subs/step1.tikz}}
 =_{\basicR} 
    \InputIfFileExists{hoare/subs/step4.tikz}{}{\input{./tikz/hoare/subs/step4.tikz}}
 \stackrel{\eqref{eq:deterministic-total}}{=_{\basicR}} 
    \InputIfFileExists{hoare/subs/step5.tikz}{}{\input{./tikz/hoare/subs/step5.tikz}}
 \stackrel{\eqref{ax:copierun}}{=_{\basicR}} 
    \InputIfFileExists{hoare/subs/step6.tikz}{}{\input{./tikz/hoare/subs/step6.tikz}}
 = \encoding{\Gamma ' \vdash y [t / x]\colon A} \]

If \(e\) is an application, $e=f(e_1, \dots , e_n)$, by definition of \(\encoding{\cdot}\) on operations, \(\encoding{\Gamma' \vdash f(e_{1}, \ldots, e_{n}) \colon A} \defeq \copier{\encoding{\Gamma'}}^n ; (\encoding{\Gamma' \vdash e_1} \per \dots \per \encoding{\Gamma' \vdash e_n} ) ; f\).
By naturality of copy, we obtain
\begin{align*}
  
    \InputIfFileExists{hoare/subs/step1.tikz}{}{\input{./tikz/hoare/subs/step1.tikz}}
 &=_{\basicR} 
    \InputIfFileExists{hoare/subs/step7.tikz}{}{\input{./tikz/hoare/subs/step7.tikz}}
 \stackrel{\eqref{eq:deterministic-total}}{=_{\basicR}} 
    \InputIfFileExists{hoare/subs/step8.tikz}{}{\input{./tikz/hoare/subs/step8.tikz}}
 \\
  &= \encoding{\Gamma' \vdash f(e_{1} [t / x],\ldots ,e_{n} [t / x])\colon A} = \encoding{\Gamma' \vdash f(e_{1},\ldots ,e_{n}) [t / x]\colon A}
\end{align*}
\end{proof}

With this, we can now prove Lemma \ref{lemma:encodingsubstPed}.
\begin{proof}[Lemma \ref{lemma:encodingsubstPed}]
  Proceed by induction on the typing rules for predicates.

  If \(P\) is \(\top\), then
  \[
    
    \InputIfFileExists{hoare/pred/step1.tikz}{}{\input{./tikz/hoare/pred/step1.tikz}}
 =_{\basicR}  
    \InputIfFileExists{hoare/pred/step2.tikz}{}{\input{./tikz/hoare/pred/step2.tikz}}
 \stackrel{\eqref{eq:deterministic-total}}{=_{\basicR}} 
    \InputIfFileExists{hoare/pred/step3.tikz}{}{\input{./tikz/hoare/pred/step3.tikz}}
 \stackrel{\eqref{ax:copierun}}{=_{\basicR}} 
    \InputIfFileExists{hoare/pred/step4.tikz}{}{\input{./tikz/hoare/pred/step4.tikz}}
  = \encoding{\Gamma' \vdash \top \colon 1} = \encoding{\Gamma' \vdash \top[t /x] \colon 1}.
  \]

  If \(P\) is \(\bot\), then
  \[  
    
    \InputIfFileExists{hoare/pred/step1.tikz}{}{\input{./tikz/hoare/pred/step1.tikz}}
 =_{\basicR} 
    \InputIfFileExists{hoare/pred/step5.tikz}{}{\input{./tikz/hoare/pred/step5.tikz}}
  \stackrel{\eqref{ax:bangnat}}{=_{\basicR}} 
    \InputIfFileExists{hoare/pred/step6.tikz}{}{\input{./tikz/hoare/pred/step6.tikz}}
 = \encoding{\Gamma' \vdash \bot \colon 1} = \encoding{\Gamma' \vdash \bot[t /x] \colon 1}.
  \]

  If \(P\) is a predicate symbol \(R\), then
  \begin{align*}
    
    \InputIfFileExists{hoare/pred/step1.tikz}{}{\input{./tikz/hoare/pred/step1.tikz}}
 &=_{\basicR} 
    \InputIfFileExists{hoare/pred/step7.tikz}{}{\input{./tikz/hoare/pred/step7.tikz}}
 \stackrel{\eqref{eq:deterministic-total}}{=_{\basicR}} 
    \InputIfFileExists{hoare/pred/step8.tikz}{}{\input{./tikz/hoare/pred/step8.tikz}}
 \\
    &= \encoding{\Gamma' \vdash R(e_{1} [t / x],\ldots ,e_{n} [t / x])\colon A} = \encoding{\Gamma' \vdash R(e_{1},\ldots ,e_{n}) [t / x]\colon A}.
  \end{align*}

  If \(P\) is a negated predicate symbol \(\bar{R}\), then
  \begin{align*}
    
    \InputIfFileExists{hoare/pred/step1.tikz}{}{\input{./tikz/hoare/pred/step1.tikz}}
 &=_{\basicR} 
    \InputIfFileExists{hoare/pred/step9.tikz}{}{\input{./tikz/hoare/pred/step9.tikz}}
 \stackrel{\eqref{eq:deterministic-total}}{=_{\basicR}} 
    \InputIfFileExists{hoare/pred/step10.tikz}{}{\input{./tikz/hoare/pred/step10.tikz}}
 \\
    &= \encoding{\Gamma' \vdash \bar{R}(e_{1} [t / x],\ldots ,e_{n} [t / x])\colon A} = \encoding{\Gamma' \vdash \bar{R}(e_{1},\ldots ,e_{n}) [t / x]\colon A}.
  \end{align*}

  For the conjunction case, \(P = Q \land R\),%
  \begin{align*}
    
    \InputIfFileExists{hoare/pred/step1.tikz}{}{\input{./tikz/hoare/pred/step1.tikz}}
 &=_{\basicR} 
    \InputIfFileExists{hoare/pred/step11.tikz}{}{\input{./tikz/hoare/pred/step11.tikz}}
 \stackrel{\eqref{eq:deterministic-total}}{=_{\basicR}} 
    \InputIfFileExists{hoare/pred/step12.tikz}{}{\input{./tikz/hoare/pred/step12.tikz}}
 \\
    &= \encoding{\Gamma' \vdash Q[t/x] \colon 1} \land \encoding{\Gamma' \vdash R[t/x] \colon 1} \\
    &= \encoding{\Gamma' \vdash Q[t/x] \land R[t/x] \colon 1} = \encoding{\Gamma' \vdash (Q \land R)[t/x] \colon 1}.
  \end{align*}

  For the disjunction case, \(P = Q \lor R\),%
  \begin{align*}
    
    \InputIfFileExists{hoare/pred/step1.tikz}{}{\input{./tikz/hoare/pred/step1.tikz}}
 &=_{\basicR} 
    \InputIfFileExists{hoare/pred/step13.tikz}{}{\input{./tikz/hoare/pred/step13.tikz}}
 \stackrel{\eqref{ax:diagnat}}{=_{\basicR}} 
    \InputIfFileExists{hoare/pred/step14.tikz}{}{\input{./tikz/hoare/pred/step14.tikz}}
 \\
    &= \encoding{\Gamma' \vdash Q[t/x] \colon 1} \lor \encoding{\Gamma' \vdash R[t/x] \colon 1} \\
    &= \encoding{\Gamma' \vdash Q[t/x] \lor R[t/x] \colon 1} = \encoding{\Gamma' \vdash (Q \lor R)[t/x] \colon 1}.
  \end{align*}
\end{proof}

\end{document}